\documentclass[10pt]{article}

\usepackage[colorlinks, linkcolor=blue,citecolor=blue]{hyperref}           
\usepackage{cmap}
\usepackage{color}
\usepackage[percent]{overpic}
\usepackage{graphicx,subfigure,amsmath,amssymb,amsfonts,bm,epsfig,epsf,url,dsfont,tcolorbox,bbm,microtype}
\usepackage{amsthm,mathrsfs}
\usepackage{epigraph}
\usepackage{tikz}

\usepackage{multirow}
\usepackage{bbm}      
\usepackage{booktabs}
\usepackage{cases}
\usepackage{enumitem}
\usepackage[small,bf]{caption}
\usepackage[top=1in,bottom=1in,left=1in,right=1in]{geometry}
\usepackage{fancybox}
\usepackage{hyperref}
\usepackage{scrhack}
\hypersetup{
    colorlinks,
    linkcolor={blue!80!black},
    citecolor={blue!80!black},
    urlcolor={blue!80!black}
}
\usepackage{algorithm}
\usepackage{verbatim}
\usepackage{algorithmic}
\usepackage{mathtools}
\usepackage{braket}
\newcommand{\graybox}[1]{%
\begin{tcolorbox}[
  width=\linewidth,
  left=1mm,
  right=1mm,
  boxsep=0pt,
  boxrule=0.5pt,
  sharp corners,
  colback=white!95!black]
#1
\end{tcolorbox}
}

\usepackage[normalem]{ulem}
\newcommand\redout{\bgroup\markoverwith{\textcolor{red}{\rule[0.5ex]{2pt}{0.8pt}}}\ULon}

\usepackage{scalerel,stackengine}

\stackMath
\newcommand\reallywidehat[1]{
\savestack{\tmpbox}{\stretchto{%
0  \scaleto{%
    \scalerel*[\widthof{\ensuremath{#1}}]{\kern-.6pt\bigwedge\kern-.6pt}%
    {\rule[-\textheight/2]{1ex}{\textheight}}
  }{\textheight}%
}{0.5ex}}%
\stackon[1pt]{#1}{\tmpbox}%
}

\newcommand{\qc}[1]{\left[\!\left[#1\right]\!\right]}

\newcommand{\floor}[1]{\left \lfloor #1 \right \rfloor}

\newtheorem{theorem}{Theorem}
\newtheorem{cor}[theorem]{Corollary}

\newtheorem{lemma}[theorem]{Lemma}
\newtheorem{prop}[theorem]{Proposition}

\newtheorem{introtheorem}{Theorem}

\theoremstyle{remark}
\newtheorem{definition}{Definition}
\newtheorem{remark}{Remark}
\newtheorem{construction}{Construction}
 \newtheorem{obs}{Observation} 

\newtheorem{example}{Example}

\newenvironment{fminipage}%
  {\begin{Sbox}\begin{minipage}}%
  {\end{minipage}\end{Sbox}\fbox{\TheSbox}}

\makeatletter
\newcommand*{\rom}[1]{\expandafter\@slowromancap\romannumeral #1@}
\makeatother

\DeclareMathOperator{\supp}{supp}
\DeclareMathOperator{\Span}{Span}

\usepackage{mathrsfs}

\newcommand{\nc}{\newcommand}
\nc{\us}{{\underline{s}}}
\nc{\un}{{\underline{n}}}  \nc{\ux}{{\underline{x}}} 
\nc{\uX}{{\underline{X}}}  \nc{\uY}{{\underline{Y}}}  \nc{\uv}{{\underline{v}}}     \nc{\uz}{{\underline{z}}} 
\nc{\uzero}{{\underline{0}}} 
\nc{\uTheta}{{\underline{\theta}}} 
\nc{\uy}{{\underline{y}} } \nc{\ue}{{\underline{e}}}  \nc{\uf}{{\underline{f}} } \nc{\ur}{{\underline{r}}} 
\nc{\ub}{{\underline{b}} } \nc{\ua}{{\underline{a}}} 
 \nc{\ui}{{\underline{i}}}  
 \nc{\wt}{{\mathrm{wt}}}  
  \nc{\lcs}{{\mathrm{LCS}}}  
  \nc{\uj}{{\underline{j}}}  
\nc{\Corr}{\mathrm{corr}}
\nc{\1}{{\mathbbm{1}}}
\nc\bfx{\boldsymbol x}
\nc{\Cov}{\mathrm{Cov}}

\newcommand{\Ind}{\mathbbm{1}}

\newcommand{\tr}{\mathrm{tr}}

\newcommand{\abs}[1]{\left|#1\right|}
\newcommand{\R}{\mathbb{R}} 
\newcommand{\N}{\mathbb{N}}

\newcommand{\Z}{\mathbb{Z}}

\newcommand{\C}{\mathbb{C}}

\newcommand{\E}{\mathbb{E}}

 \nc\sC{{\mathscr C}}
  
\def\P{{\mathbb P}}

\newcommand{\calB}{{\cal B}}

\newcommand{\calD}{{\cal D}}
\newcommand{\calE}{{\cal E}}
\newcommand{\calF}{{\cal F}}
\newcommand{\calG}{{\cal G}}
\newcommand{\calH}{{\cal H}}
\newcommand{\calI}{{\cal I}}
\newcommand{\calJ}{{\cal J}}
\newcommand{\calK}{{\cal K}}
\newcommand{\calL}{{\cal L}}
\newcommand{\calM}{{\cal M}}
\newcommand{\calN}{{\cal N}}

\newcommand{\calP}{{\cal P}}
\newcommand{\calQ}{{\cal Q}}
\newcommand{\calR}{{\cal R}}
\newcommand{\calS}{{\cal S}}
\newcommand{\calT}{{\cal T}}

\newcommand{\calX}{{\cal X}}

\DeclarePairedDelimiterX{\cond}[1]{[}{]}{\setargs{#1}}
\NewDocumentCommand{\setargs}{>{\SplitArgument{1}{;}}m}
{\setargsaux#1}
\NewDocumentCommand{\setargsaux}{mm}
{\IfNoValueTF{#2}{#1} {#1\,\delimsize|\,\mathopen{}#2}}

\newcommand{\be}{\begin{equation}}
\newcommand{\ee}{\end{equation}}
\newcommand{\beqna}{\begin{eqnarray}}
\newcommand{\eeqna}{\end{eqnarray}}

\newcommand{\p}[1]{\left(#1\right)}
\newcommand{\pp}[1]{\left[#1\right]}
\newcommand{\ppp}[1]{\left\{#1\right\}}
\newcommand{\norm}[1]{\left\|#1\right\|}
\newcommand{\innerP}[1]{\left\langle#1\right\rangle}

\usepackage{xspace}
\usepackage{array}
\usepackage{tabularx}
\usepackage{booktabs} 
\usepackage{adjustbox} 

\newcommand{\s}[1]{\mathsf{#1}}

\newcommand{\su}[1]{\underline{\mathsf{#1}}}

\makeatletter
\def\thanks#1{\protected@xdef\@thanks{\@thanks
        \protect\footnotetext{#1}}}
\makeatother



\usepackage{etoolbox}

\makeatletter
\patchcmd{\thebibliography}
  {\list}
  {\small
   \list}
  {}{}

\patchcmd{\thebibliography}
  {\sloppy}
  {\sloppy
   \setlength{\itemsep}{1pt}
   \setlength{\parsep}{3pt}
   \setlength{\parskip}{3pt}}
  {}{}
\makeatother

\usepackage{etoolbox}

\makeatletter
\patchcmd{\l@section}
  {\addvspace{1.0em \@plus\p@}}
  {\addvspace{0pt}}
  {}{}
\makeatother

\allowdisplaybreaks
\date{}

\begin{document}
\title{Theory of approximate quantum error correction\\and the error-set model}
 \author{Dor~Elimelech$^1$\hspace*{.4in} Victor~V.~Albert$^2$\hspace*{.4in}  Alexander~Barg$^{1,2,3}$}\thanks{$^1$Institute for Systems Research, University of Maryland, College Park, MD 20742. $^2$Joint Institute for Quantum Information and Computer Science, University of Maryland/NIST, College Park, MD 20742. $^3$Department of ECE, University of Maryland, College Park, MD 20742. Emails: 
  \{dor,vva,abarg\}@umd.edu. 
 }

\maketitle

\begin{abstract}
    We develop a theory of approximate quantum error correction (QEC) based on the error-set model, complemented by general methods for code construction. Exact QEC has a powerful error-set structure: by the Knill-Laflamme conditions, a code correcting a given error set automatically protects against every channel whose Kraus operators lie in their linear span. This linearity gives rise to code distance, the equivalence between erasures and general errors, and a theory of asymptotically good codes. A longstanding view has been that these features do not extend to AQEC, leaving the theory essentially channel-by-channel. We show instead that, although full Knill--Laflamme linearity fails, a restricted form survives and suffices to extend all three structural features to the approximate setting. Specifically, a common error-set criterion governs families of channels whose Kraus operators are linear combinations of a given error set and whose coefficient matrices satisfy a spectral constraint. Using the B\'eny-Oreshkov worst-case and Petz average-case frameworks, we derive uniform fidelity guarantees for these families in terms of two new code parameters--the \emph{environment-leakage distance}, controlling worst-case performance, and the \emph{Knill-Laflamme Hellinger distance}, characterizing the average-case performance of Petz recovery.

    To demonstrate the scope of this model, we develop partition-based constructions across diverse quantum systems and geometries, placing exact and approximate correction on equal footing. These constructions lead to a metric--error alignment hierarchy for Hilbert spaces, metrics, and error families, which in turn characterizes the resulting recovery guarantees. They yield the first known asymptotically good code families for fermionic systems, one-dimensional Rydberg-blockaded systems, and deletion errors, and extend to other physical platforms.

\end{abstract}

\clearpage
{\small
\tableofcontents}


\clearpage
\rightline{\begin{minipage}{.5\textwidth}{\small \hspace*{.2in}\textit{The connection between correcting general errors and erasure errors breaks down for approximate QECCs. This suggests there is no sensible notion of
distance for an approximate quantum error-correcting code.}\\[.1in]
\rightline {C. Cr{\'e}peau, D. Gottesman, and A. Smith, 2005, \cite{crepeau2005approximate}}\par} \end{minipage}}

\section{Introduction and overview}

Quantum information is carried by coherence and entanglement, but precisely these features are highly sensitive to uncontrolled interactions and decoherence \cite{zurek2003decoherence}. The possibility of protecting quantum states against such noise is therefore one of the central structural questions of quantum information theory. Quantum error correction provides such a protection mechanism \cite{shor1995scheme,knill1997theory}. It underlies the modern theoretical picture of scalable fault-tolerant quantum computation \cite{shor1996fault,aharonov2008fault}, gives a framework for reliable transmission through noisy quantum channels \cite{knill1997theory,bennett1996mixed}, and it provides the natural language for 
robust features of many-body systems and topological phases \cite{kitaev2003fault,dennis2002topological}. From the theoretical point of view, quantum error correction is therefore not only a way of combating noise, but also a language for describing highly organized sectors of the Hilbert space and operator structures that preserve them \cite{knill1997theory,kitaev2003fault,terhal2015quantum}.

The basic idea of quantum error correction is subtle. Because unknown quantum states cannot be copied, protection cannot rely on classical-style repetition \cite{wootters1982single}. Instead, the logical information is encoded into a larger Hilbert space in such a way that the relevant error operators act reversibly on the encoded subspace. In this way, decoherence and imperfect control are turned from an irreversible loss of information into a structured problem about subspaces, error operators, and recovery maps \cite{knill1997theory}.

The success of exact QEC lies not only in the possibility of recovery, but also in the structural framework it provides. Through the Knill--Laflamme conditions \cite{knill1997theory}, exact correctability becomes a linear condition on the action of errors on the code, and this linearity has powerful consequences: once a code corrects a prescribed set of errors, it automatically corrects every channel whose Kraus error operators lie in their linear span \cite{knill1997theory,knill2000theory}. This is a major advantage of the exact theory: one need not tailor the code to a specific noise channel, but can protect uniformly against an entire adversarial family of channels, constrained only by the type and amount of noise they are allowed to introduce \cite{knill2000theory,gottesman2010introduction}. In this way, exact QEC acquires a genuine error-set model, rather than a theory tied to one fixed channel. In particular, it leads to the notion of \textit{code distance}, which quantifies the greatest noise severity --- for example, the largest number of affected subsystems --- that the code can correct exactly \cite{knill1997theory,gottesman2010introduction}.

Exact correction is often too rigid for physically relevant noise models: there are always some uncorrectable errors. In approximate quantum error correction, one relaxes the requirement of perfect recovery and asks only that, after noise and decoding, the recovered state remain close to the original encoded state according to a suitable fidelity or distance measure. This relaxation is not merely technical. Already in early work, it was observed that approximate codes can outperform exact ones for important noise models such as amplitude damping \cite{leung1997approximate}, and more recent constructions show that AQEC can approach asymptotic coding limits that are inaccessible to exact QEC, including the quantum Singleton and Hamming bounds \cite{bergamaschi2024approaching,ma2025haar}. In this sense, AQEC is needed not only because exact correction may fail to exist, but because allowing a vanishingly small recovery error can fundamentally enlarge the achievable coding regime.

This point has led to a substantial line of research on AQEC. Foundational works developed quantitative criteria for approximate recoverability in terms of entanglement fidelity, coherent information, and near-optimal recovery maps \cite{schumacher1996sending,schumacher2001approximate,barnum2002reversing,beny2010general,klesse2007approximate}. More recent work has sharpened channel-level performance metrics and the role of the Petz/transpose map \cite{ng2010simple,noh2018quantum,zheng2024near,li2025optimality, kim2026optimal}, produced explicit code constructions and asymptotic existence results \cite{bergamaschi2024approaching,ma2025haar,xu2025letting}, and uncovered new connections between AQEC, many-body order, circuit complexity, and information masking \cite{yi2024complexity,yi2025lov,li2025random}. Together, these developments show that AQEC is not simply a perturbative variant of exact QEC, but a broad and active framework in its own right.

Yet for all this progress, existing AQEC theory remains overwhelmingly channel-based. The central formulations ask whether a code approximately corrects a fixed noise channel or quantify the performance of a prescribed decoder for a given channel model. What is missing is the structural framework that makes exact QEC so robust: in the approximate setting, correctability is no longer linear in the Knill--Laflamme sense, and one cannot simply pass from approximate correction of a prescribed error set to approximate correction of every channel whose Kraus operators lie in its linear span. As a result, while exact QEC has enjoyed a highly useful adversarial error-set formulation for nearly three decades, no analogous general theory has emerged for AQEC. This severely limits the practical reach of the existing channel-based theory, since one is forced to analyze channels individually rather than work uniformly with entire families of noise constrained only by the type and amount of errors they may introduce.

This absence has long shaped the prevailing intuition about AQEC. In particular, as evidenced by the epigraph, it was argued early on that even basic exact-QEC notions such as linearity and distance may fail to admit a sensible approximate analog, and more generally, that strong structural features of the exact theory --- such as the relation between correcting erasures and correcting general errors, from which the usual notion of distance arises --- may simply not survive in the approximate setting.

The goal of this work is to overturn this conclusion. We show that AQEC does admit a natural and useful adversarial error-set model. The key point is that, although the full linearity of the exact Knill--Laflamme theory fails in the approximate setting, a restricted form of linearity survives: one can pass from an error set to a family of channels generated by its linear span, provided the mixing coefficients satisfy a spectral constraint. This leads to a well-defined class of channels controlled by an error set of interest, and hence to an approximate analog of the adversarial error-set model of exact QEC. Within this model, we develop approximate counterparts of the main structural notions of the exact theory, including a notion of distance, general conditions for AQEC, and a quantitative relation between correcting erasures and general errors, which in turn reveals that the known connection between AQEC and quantum circuit complexity runs deeper than previously understood. We complement this structural theory with explicit code constructions across a range of quantum platforms and noise models. We further support these results with numerical examples that illustrate the usefulness of our error-set bounds and conditions and the performance of our code constructions; see Section~\ref{sec:NumericalExperiments}. Taken together, these results extend the existing channel-based theory of AQEC to an error-set-based one. More broadly, they open a new direction for adversarial AQEC, in which approximate correction can be studied uniformly over structured families of noise rather than one channel at a time.

To obtain a full adversarial model, one further point must be addressed: a single recovery map must work for every channel controlled by the error set, without knowledge of the particular channel. The conditions developed here guarantee approximate correctability uniformly across this family, but the channel-based arguments allow the recovery map to depend on the channel. In our companion paper \cite{elimelech2026universal}, we show that this model also admits universal recovery: the environment-leakage distance and the Knill--Laflamme Hellinger distance introduced here control the performance of a single recovery map across the entire family, under the worst-case and average-case criteria, respectively. Together, the two works extend both parts of the exact-QEC error-set structure to AQEC: uniform control over a family of noise channels and recovery independent of the channel.

\subsection{The missing error-set model for AQEC }
Our starting point is the standard channel-based formulation of approximate quantum error correction. Let $\calH$ be a finite-dimensional Hilbert space. A quantum code $Q\subseteq \calH$ is a $K$-dimensional subspace, and an encoded quantum state is a density operator $\rho\in D(\calH)$ whose support lies in $Q$. Noise is modeled by a quantum channel, namely a completely positive trace-preserving (CPTP) map $\calN:L(\calH)\to L(\calH')$. Any such channel admits a Kraus representation $\calN(\rho)=\sum_k E_k\rho E_k^\dag$,
where the operators $\ppp{E_k}_k$ satisfy the completeness relation $\sum_k E_k^\dag E_k=I$. The operators $E_k$ are called Kraus operators of the channel, and this representation is, in general, not unique.

A code $Q$ is said to be an $\varepsilon$-approximate quantum error-correcting code (or simply $\varepsilon$\textbf{-AQEC}) for a channel $\calN$ if there exists a CPTP recovery map $\calD$ such that the recovered channel $\calD\circ \calN$ is $\varepsilon$-close to the identity under a suitable metric function (see Definition~\ref{def: AQEC codes}). Throughout most of this work, we measure this closeness by the Bures distance between channels, defined from the worst-case entanglement fidelity,
\begin{align*}
d(\calN,\calM)& := \sqrt{1-\calF_e(\calN,\calM)}\\
\calF_e(\calN,\calM)& := \inf_{\rho\in D(\calH)}\calF\p{\calN(\ket{\psi_\rho}\bra{\psi_\rho}),\calM(\ket{\psi_\rho}\bra{\psi_\rho})},
\end{align*}
where $\calF(\rho,\tau)=\norm{\sqrt{\rho}\sqrt{\tau}}_1$ is the state fidelity and $\ket{\psi_\rho}$ is any purification of $\rho$; see Definition~\ref{def:Bures}. This criterion is well established in AQEC. It is equivalent to the diamond-norm distance for nearby channels (see Definition~\ref{def:CbDiamondNorm}) and is the distance measure used in the B\'eny--Oreshkov framework \cite{beny2010general}, which will serve as our main tool below. Later in the paper, we also consider a weaker, average-case notion of channel closeness based on channel fidelity and show that the error-set model developed here extends naturally to this setting.

In exact quantum error correction, the Knill--Laflamme conditions \cite{knill1997theory} give more than a criterion for reversing a fixed noise channel. If a code corrects an error set $\calE$, then every channel whose Kraus operators lie in $\Span(\calE)$ is exactly reversible on the code space.  This is the adversarial error-set model of exact QEC: rather than tailoring the code to a specific channel, one seeks to protect uniformly against all channels whose noise is generated by a prescribed family of errors. Building on this adversarial error-set perspective, we develop a corresponding formalism for the approximate setting.

For AQEC, a useful general framework
(for the worst-case criterion) is the channel-level theory of B\'eny and Oreshkov \cite{beny2010general}. It gives necessary and sufficient conditions for approximate correction of a fixed channel by characterizing when the corresponding complementary channel is close to a constant channel. Below, we restate this criterion in the operator form of Theorem~\ref{th:BO}. In hindsight, however, it is not clear from this formulation how to pass from one fixed channel to a genuine error-set model, since the approximate conditions do not enjoy the same linearity properties as the exact Knill--Laflamme theory.

This leads to the natural question that guides this work:
\begin{center}
    \textit{Given an error set of interest $\calE$, for which channels with Kraus operators in $\Span(\calE)$ can one guarantee AQEC properties of a given code?}
\end{center}
The B\'eny--Oreshkov conditions suggest the answer. An error set $\calE$ defines a completely positive map (generally not trace-preserving) and hence also an associated complementary map.
The key question is then which channels with Kraus operators in $\Span(\calE)$ have the property that closeness of the complementary map associated with $\calE$ to a constant map implies closeness of their complementary channel to a constant channel.
 This question gives rise to the notion of $\calE$-controlled channels, which leads to the approximate adversarial error-set model developed next.

Our controlled families of channels are defined by an operator set $\calE=\ppp{E_k}_{k}$ and a condition on the expansion coefficients of the channel's Kraus operators as linear combinations of the operator set elements.
A channel $\calN$ will be called 
$\calE$-\emph{controlled} if it admits a Kraus representation $\ppp{A_m}_{m}$ of the form
\[A_m=\sum_{k\in[M]}c_{m,k}E_k,
\]
where the coefficient matrix $C=(c_{m,k})$ satisfies $\norm{C}_\infty\leq 1$. Thus, the {\bf approximate error-set model}  consists of the set of all channels whose Kraus operators lie in $\Span(\calE)$ with a uniform spectral bound on the mixing coefficients. A code $Q$ is then said to be an $\varepsilon$-AQEC code for the error set $\calE$ if it is an $\varepsilon$-AQEC code for every $\calE$-controlled channel in the sense of Definition~\ref{def: AQEC codes}.

The B\'eny--Oreshkov theorem indicates how to control this family uniformly. Given a code $Q\subseteq\calH$ and an error set $\calE=\ppp{E_k}_{k}$, one associates to every matrix $\lambda=(\lambda_{k,l})$ the B\'eny--Oreshkov superoperator (see Definition~\ref{def:DiamondNormProxConst})
\[\calB^{\calE}_{\lambda,Q}(\rho)=\sum_{k,l}\tr\p{(P{E_l^\dag E_k}P-\lambda_{k,l}P)\rho}\ket{k}\bra{l},
\]
where $P$ is the orthogonal projector on the code space $Q$.
When $\calE$ is the Kraus set of a quantum channel, the B\'eny--Oreshkov theorem states that approximate correctability of $Q$ for that channel is equivalent to the corresponding complementary channel being close to a constant channel. When this complementary-channel criterion is transferred from a fixed channel to the CP map associated with an error set, it becomes the requirement that the superoperator $\calB^{\calE}_{\lambda,Q}$ be close to the zero map for some choice of $\lambda$. We measure this closeness using the diamond norm $\norm{\cdot}_{\diamond}$, see Definition~\ref{def:CbDiamondNorm}. Passing from the Bures-distance formulation of B\'eny and Oreshkov to the diamond norm introduces the usual square-root loss: an $O(\varepsilon)$-AQEC guarantee requires $O(\varepsilon^2)$ proximity of $\calB^{\calE}_{\lambda,Q}$ to the zero map for some choice of $\lambda$. To address this requirement, we define the {\bf environment-leakage distance} by optimizing this deviation over all possible choices of $\lambda$: 
   \[
   \zeta(\calE,Q) := \inf_{\lambda}\norm{\calB^{\calE}_{\lambda,Q}}_{\diamond}
   \]
(see Definition~\ref{def:DiamondNormProxConst}). This is the main metric we use to control the AQEC capabilities of a family of $\calE$-controlled channels. The following theorem (which is a restatement of Theorem~\ref{th:BOerrorset}) shows that this quantity indeed yields a uniform sufficient condition for AQEC in the error-set model.

\begin{introtheorem}[Sufficient AQEC conditions for error sets]\label{th:restatementSuffworst}
Let $\calE=\ppp{E_k}_{k}$ be a finite set of error operators, and let $Q\subseteq\calH$ be a quantum code. If $\zeta(\calE,Q)\leq {\varepsilon^2}/2$,
then $Q$ is an $\varepsilon$-AQEC code for the error set $\calE$. Consequently, $Q$ is an $\varepsilon$-AQEC code for any $\calE$-controlled channel.
\end{introtheorem}
In the special case where $\calE$ is a Kraus set of a single channel $\calN$, the theorem reduces to the sufficiency part of the B\'eny--Oreshkov characterization for $\calN$, up to the unavoidable $\frac{1}{\sqrt{2}}\sqrt{\cdot }$ \,loss that arises as a result of replacing the Bures distance with the diamond-norm criterion used here. This theorem can be complemented by necessary AQEC conditions for several important families of error sets, including unitary Hilbert--Schmidt orthogonal errors and amplitude damping errors; see Propositions~\ref{prop:orthUni} and \ref{prop:NecFock}. These necessary bounds generally involve constants depending on the size or structure of the error set, and may therefore become large. Nevertheless, Proposition~\ref{prop:exampleOptimality} shows that this loss is not merely an artifact of the proof technique: there exist examples for which the sufficient condition above is also necessary, up to a dimension-independent constant and the square-root loss mentioned above. In particular, this shows that, within the present framework, any substantial improvement of the sufficient condition in Theorem~\ref{th:restatementSuffworst} requires additional assumptions on the error model or the code.

The sufficient condition above already shows that the quantity $\zeta(\calE,Q)$ gives uniform control over all $\calE$-controlled channels and therefore defines a genuine error-set model for AQEC. For such a model to be useful, however, two further points must be checked: first, that the family of $\calE$-controlled channels has the structural properties associated with the adversarial noise model, and second, that it gives rise to natural and nontrivial examples, extending the familiar error-set viewpoint of exact QEC. Both issues are studied in detail in Section~\ref{sec:StructrualProp}. Here, we briefly summarize the main conclusions.

On the structural side, the definition of $\calE$-controlled channels is phrased in terms of a specific Kraus representation, so it is not a priori clear that it defines an intrinsic property of the channel. Observation~\ref{obs: independent} shows that this ambiguity is resolved as soon as the error set $\calE$ is linearly independent: in this case, the condition does not depend on the choice of the Kraus representation. Another basic requirement for an adversarial error model is closure under probabilistic mixtures. Observation~\ref{obs:Convexity} shows that if $\calN_1,\dots,\calN_\ell$ are $\calE$-controlled, then any convex combination $\sum_{i=1}^{\ell}p_i\calN_i$ is again $\calE$-controlled. Physically, this corresponds to a channel that applies one of the $\calE$-controlled channels at random, with probabilities $p_1,\dots,p_\ell$. Thus, the model is stable under both changes of representation and probabilistic mixing. 

\subsubsection*{Approximate distance, erasures/general errors equivalence, and quantum circuit complexity}
One of the most important examples arises from orthogonal unitary error sets. As shown in Proposition~\ref{prop:orthUni}, if $\calE$ consists of Hilbert--Schmidt-orthogonal unitary operators, then every channel whose Kraus operators lie in $\Span(\calE)$ is automatically $\calE$-controlled. In this case, the restricted linearity built into our AQEC model becomes the full linearity commonly associated with exact QEC. In particular, for the set $\calE_t$ of Pauli errors of weight at most $t$ on $\calH_q^{\otimes N}$ (where $\calH_q$ is a $q$-dimensional Hilbert space; see Example~\ref{ex:Paulis}), the set of $\calE_t$-controlled channels contains exactly the channels introducing at most $t$ local errors. This recovers the usual adversarial error-set picture of exact QEC in a natural approximate form.

  This makes it possible to extend the usual notion of code distance to the approximate setting, and gives rise to the \textbf{approximate code distance}. 
 For Pauli errors, a code has $\varepsilon$-distance at least $t$ if it is an $\varepsilon$-AQEC code for the error set $\calE_t$, that is, if it is $\varepsilon$-AQEC for all channels introducing at most $t$ local errors. More generally, Definition~\ref{def:dist} formulates approximate distance for an {\em arbitrary indexed family} of error sets $\mathscr{E}=(\calE_t)_t$, where the index $t$ captures the severity of the noise according to the particular model under consideration. This point of view applies not only to limited-weight errors, but also to amplitude damping noise (Example~\ref{ex:AD}), Majorana noise in fermionic systems (Example~\ref{ex:Majoranas}), erasures, deletions, and other models considered in the main text of the paper. In Section~\ref{sec:NumericalGNU}, we illustrate this point with GNU codes whose exact photon-loss distance remains two, while our bounds certify correction of more photon losses at a fixed small error tolerance as the total excitation increases.

Another structural feature of exact QEC that extends to the approximate setting is the relation between erasures and general errors. In the exact theory, correcting $2t$ erasures is equivalent to correcting $t$ general errors, and this equivalence is one of the main sources of the usual notion of code distance. Using the error-set model developed in this paper, we show that an analogous equivalence persists in AQEC: approximate correction of $2t$ erasures is equivalent to approximate correction of $t$ general errors, with an explicit relation between the associated approximation parameters. This is formalized in the following theorem, which is an informal version of Theorem~\ref{th:ErasureGeneralEr}.

\begin{introtheorem}[Connection between approximately correcting erasures and general errors]\label{th:restatementErasuresErrors}
If a code $Q$ $\varepsilon$-corrects $2t$ erasures, then it $\varepsilon'$-corrects $t$ general errors. Conversely, if $Q$ $\varepsilon$-corrects $t$ general errors, then it $\varepsilon''$-corrects $2t$ erasures. 
\end{introtheorem}

The parameters $\varepsilon'$ and $\varepsilon''$ above are related by an explicit dimension-dependent constant together with a square-root loss; see Theorem~\ref{th:ErasureGeneralEr}.
Although $\varepsilon$ and $\varepsilon'$ are not equal, and the constants relating them can grow with the system size, this result is still meaningful: it shows that the exact-QEC bridge between erasures and general errors is preserved in the approximate setting and can be expressed in quantitative form. In particular, it provides an approximate analog of one of the structural mechanisms underlying code distance, and shows that the error-set model captures part of the exact-QEC structure that is absent from a purely channel-based AQEC viewpoint.

This perspective also plays a role in recent connections between AQEC and quantum circuit complexity. A central quantity in this line of work is the subsystem variance introduced in \cite{yi2024complexity}, which measures how much information about the encoded state can be read from the reduced density operator on a prescribed subsystem, and thereby controls AQEC against replacement channels on that subsystem. In Proposition~\ref{prop:SSVzeta}, we show that subsystem variance also bounds the environment-leakage distance for erasure errors; combined with Theorem~\ref{th:restatementErasuresErrors}, this implies that the same quantity controls AQEC performance against general limited-weight errors. Thus, subsystem variance becomes a common parameter linking local indistinguishability, approximate correction of erasures and general errors, and results concerning circuit complexity, derived in \cite{yi2024complexity, yi2025lov}.

\subsubsection*{The error-set model for the average-case AQEC criterion}
Alongside worst-case entanglement fidelity, a standard, weaker notion of performance is the channel fidelity (denoted by $\calF_{\mathrm{ch}}$),  \cite{audenaert2002optimizing,kosut2009quantum,noh2018quantum,zheng2024near}. Instead of minimizing over all input states, it evaluates the entanglement fidelity only on a maximally entangled state; see \eqref{eq:ChaFidel} for the exact definition. This quantity is closely related to the Haar-average input-output fidelity, with the precise relation given in \eqref{eq:averagcase}, and is therefore naturally interpreted as an average-case criterion. To stay consistent with the Bures-distance viewpoint used throughout the paper, we measure the corresponding error by
$d_{\mathrm{ch}}\p{\calN,\calM} := \sqrt{1-\calF_{\mathrm{ch}}^2\p{\calN,\calM}}$,
so that the average-case and worst-case formulations are expressed on the same square-root scale.

Because channel fidelity is itself a natural and widely used notion of approximate correction, it provides a natural test for the robustness of our error-set model: if the model is truly structural, it should not be tied to the worst-case criterion alone. What is far from obvious, however, is that the same family of $\calE$-controlled channels should continue to govern approximate correction in this average-case setting. The point of this subsection is that it does. Remarkably, the same passage from a fixed channel to a uniform family of $\calE$-controlled channels survives in the channel-fidelity criterion, so the error-set viewpoint is not an artifact of the worst-case B\'eny--Oreshkov framework. In this way, the existing average-case AQEC theory can also be lifted from a channel-based formulation to an error-set-based one, with conditions again stated directly in terms of the error set $\calE$ rather than derived separately for each channel.

As in the worst-case setting, we now pass from a fixed channel to an error-set model. Given an error set $\calE$, we use the family of $\calE$-controlled channels from Definition~\ref{def:AQECC}, and say that a code is average-case ($\s{Av}$) $\varepsilon$-AQEC for $\calE$ if it $\varepsilon$-corrects every channel in this family with respect to the 
channel-fidelity-based metric $d_{\mathrm{ch}}$; see Definition~\ref{def:AvAQECErrorSet} for the precise statement. This definition goes beyond a formal analogy: not only does uniform controllability extend to the average-case criterion, but the channels generated by the linear span of the error set that admit uniform average-case control are again precisely the $\calE$-controlled channels.

A standard approach to AQEC (and in particular, under the channel fidelity criterion) is via the Petz recovery map, also known in the QEC literature as the transpose channel. It is exact whenever the Knill--Laflamme conditions hold, and remains near-optimal for approximate correction in channel fidelity \cite{barnum2002reversing,ng2010simple,zheng2024near}. More precisely, for a code $Q$ and a channel $\calN$, the optimal average-case AQEC error $\varepsilon_{\mathrm{ch}}^{\mathrm{opt}}$ is characterized, up to a factor of $1/\sqrt{2}$, by the performance of the transpose-channel decoder. Recently, \cite{zheng2024near} gave an explicit expression for this performance in terms of the corresponding QEC matrix. Namely, if $Q$ is a $K$-dimensional code with basis $\ppp{\ket{c_i}}_{i}$ and $\calN$ has Kraus set $\calE=\ppp{E_k}_{k}$ of size $M$, then the associated QEC matrix is the positive semidefinite operator $A_{\s{QEC}}\in L(\C^K\otimes\C^M)$ with entries
\[\p{A_{\s{QEC}}}_{ik,jl} :=  \bra{c_i}E_k^\dag E_l\ket{c_j},\]
and the optimal average-case AQEC error is given, up to the same universal factor, by the expression 
\begin{equation}
    \sqrt{\frac{1}{2}\Big({1-\frac{1}{K^2}\|{\tr_{K}({\sqrt{A_{\s{QEC}}})}}\|_2^2}\Big)}
\leq \varepsilon_{\mathrm{ch}}^{\mathrm{opt}}
\leq \sqrt{1-\frac{1}{K^2}\|{\tr_{K}({\sqrt{A_{\s{QEC}}})}}\|_2^2}.\label{eq:reeq}
\end{equation}

Our new observation is that this expression admits a geometric interpretation which makes sense well beyond the case where $\calE$ is the Kraus set of a fixed channel. Recall that the exact Knill--Laflamme conditions for an error set $\calE=\ppp{E_k}_{k\in[M]}$ are equivalent to the requirement that the corresponding QEC matrix has the form $A_{\s{QEC}}=I_K\otimes \lambda$, or equivalently, lies in the Knill--Laflamme space
\[\calH_{\s{KL}} := \ppp{I_K\otimes\lambda~:~\lambda\in L(\C^M)}\subseteq L(\C^K\otimes\C^M).\]

In Lemma~\ref{lem:orthogonalProjectionKLspace}, we show that the quantity appearing in \eqref{eq:reeq} is precisely the normalized Hellinger distance from $A_{\s{QEC}}$ to this Knill--Laflamme space. We call this quantity the {\bf Knill--Laflamme Hellinger distance} of the code and denote it by $\zeta_{\s{H}}(\calE,Q)$. This geometric interpretation is more than 
a mere rephrasing of the result of \cite{zheng2024near}. The key point is that the QEC matrix is defined for an arbitrary finite error set $\calE$ rather than only when $\calE$ happens to be a Kraus set of a quantum channel. We then prove that the Hellinger distance to the Knill--Laflamme space decreases under conjugation by contractions. Since passing from an error set $\calE$ to an $\calE$-controlled channel corresponds precisely to such a conjugation on the associated QEC matrix, this yields a sufficient condition that uniformly controls the average-case AQEC performance over the full family of $\calE$-controlled channels. This is formulated in the following theorem, which is an informal restatement of Theorem~\ref{th:AvAQECErrorSetSuff}.

\begin{introtheorem}[Sufficient $\s{Av}$-AQEC condition for error sets]
Let $Q$ be a $K$-dimensional quantum code and let $\calE$ be an error set. If
\[\zeta_{\s{H}}(\calE,Q)=\sqrt{\frac{1}{K}\tr(A_{\s{QEC}})-\frac{1}{K^2}\norm{\tr_K\p{\sqrt{A_{\s{QEC}}}}}_2^2}\leq \varepsilon,
\] 
then $Q$ is an $\s{Av}$-$\varepsilon$-AQEC code for the error set $\calE$. Consequently, $Q$ is an $\s{Av}$-$\varepsilon$-AQEC code for any $\calE$-controlled channel.
\end{introtheorem}

 When $\calE$ is the Kraus set of a single channel $\calN$, the corresponding QEC matrix satisfies $\tr(A_{\s{QEC}})=K$, and the displayed condition reduces to the near-optimal Petz-map criterion of \cite{zheng2024near}. As in the channel-based setting of that work, this condition is optimization-free and relatively easy to evaluate numerically since it is expressed directly in terms of the QEC matrix. In Section~\ref{sec:NumericalShor}, we evaluate this condition for the well-known square Shor codes and a weighted Pauli error set containing errors of every weight. The resulting upper bound is within a small multiplicative factor of a lower bound on the optimal average-case error over the full controlled family. In addition, Proposition~\ref{prop:AvAQECNecessary} proves complementary necessary conditions, parallel to those obtained earlier for the worst-case AQEC criterion.


\subsection{AQEC in Hilbert spaces indexed by metric spaces}
The strength of a new theory lies in its ability both to give rise to natural and useful examples and to extend, rather than replace, established results. The theory developed here has both features. Because our error-set model and the corresponding AQEC conditions are built on well-established channel-level frameworks—namely, the B\'eny--Oreshkov theory for the worst-case criterion and the transpose/Petz-map approach for the average-case criterion—existing analyses of individual codes and channels extend naturally to statements about entire error sets.

 However, examples inherited from the existing AQEC literature remain confined to particular codes and particular noise models.  To reveal the full scope of the error-set model, we introduce a broader construction that is not tied to a particular platform or channel, but applies uniformly across a wide range of quantum systems and geometries. Our approach is based on selecting a subset of basis states, partitioning it into disjoint blocks, and defining one basis codeword as a superposition over each block,
 an idea that goes back to early work on quantum error correction \cite{knill2000theory}.  Rather than merely revisiting that construction, we develop it into a general scheme for producing families of codes across different Hilbert spaces and error models, in a way that interacts naturally with the error-set formalism. This provides a broad class of examples that exposes the structural content of the theory and places exact and approximate correction on the same footing. More generally, it offers a unified perspective on how approximate error correction extends the principles of quantum coding beyond the exact setting, including in adversarial models where both regimes can be treated within a common framework.

This discussion sets the stage for the second part of the paper.  We consider Hilbert spaces whose distinguished basis is indexed by a discrete metric space $(X,d)$, so that basis vectors $\ket{x}$  inherit its natural geometric properties such as separation.  The noise model is given by an indexed family of error sets $\mathscr{E}=(\calE_t)_t$, where the parameter $t$ measures, in a model-dependent way, the severity of the noise.  This construction leverages the links between the natural noise
processes in quantum systems and the geometry of the discrete space, and this correspondence can be exploited to construct both exact and approximate quantum codes in a unified way.

To formalize this approach, in Definition~\ref{def:MEAhirarchy}, we introduce the \textit{metric--error alignment hierarchy}. The hierarchy consists of three levels 
 $$
 \mathscr{L}_2\subset \mathscr{L}_1\subset \mathscr{L}_0,
 $$
quantifying the extent to which the noise can align orthogonal basis states $\ket{x},\ket{x'}$ indexed by 
distant points in the underlying metric space, thereby measuring the deviation from the exact Knill--Laflamme conditions.
Level $\mathscr{L}_0$ corresponds to the absence of any such structure. At level $\mathscr{L}_1$, whenever $x$ and $x'$ are farther apart than the allowed noise level, no pair of errors from $\calE_t$ can create overlap between $\ket{x}$ and $\ket{x'}$. At level $\mathscr{L}_2$, this alignment is even more rigid: different errors remain orthogonal even when acting on the same basis state. This hierarchy provides a way to quantify the structure underlying our code constructions: the stronger the alignment between the metric and the noise, the stronger the exact and approximate correction guarantees that our partition-code scheme can achieve. 
    As it turns out, a very broad range of quantum systems and noise models studied in the literature fall into a nontrivial level of this hierarchy; we refer to Table~\ref{tab:metric_spaces_examples} for the examples analyzed in this work.

\begin{table}[t]
\centering
\small
\renewcommand{\arraystretch}{1.15}
\begin{tabular}{p{3.6cm}p{3.0cm}>{\centering\arraybackslash}p{2.5cm}>{\centering\arraybackslash}p{2.2cm}>{\centering\arraybackslash}p{1.5cm}}
\toprule
Hilbert space & Noise type & Indexing space & Underlying metric & Level in hierarchy \\
\midrule
Qudit & HW / Pauli & $[q]^N$ & Hamming & $\mathscr{L}_1 $ \\
Qudit & amplitude damping & $[q]^N$ & $\ell_1$ & $\mathscr{L}_2$ \\
Qudit & deletions & $[q]^N$ & deletion dist. & $\mathscr{L}_1$ \\
\midrule
Bosonic Fock state & amplitude damping & $\mathbb{Z}_{0}^{q}$ & $\ell_1$ & $\mathscr{L}_2$ \\
Bosonic Fock state & shift-rotation & $\mathbb{Z}_{0}^{q}$ & $\ell_1$ & $\mathscr{L}_1$ \\
Bosonic Fock state, constant excitation & amplitude damping & $\mathcal{S}_{q,N}$ & $\ell_1$ & $\mathscr{L}_2$ \\
Bosonic Fock state, constant excitation & shift-rotation & $\mathcal{S}_{q,N}$ & $\ell_1$ & $\mathscr{L}_1$ \\
\midrule
Fermionic & Majorana & $\{0,1\}^N$ & Hamming & $\mathscr{L}_1$ \\
\midrule
Permutation symmetric & deletions & $\mathcal{S}_{q,N}$ & $\ell_1$ & $\mathscr{L}_2$ \\
Permutation symmetric &HW / Pauli  & $\mathcal{S}_{q,N}$ & $\ell_1$ & $\mathscr{L}_0$\\
\bottomrule
\end{tabular}
\caption{Metric spaces and metric--error alignment levels for the main examples considered in this work.}
\label{tab:metric_spaces_examples}
\end{table}

\subsubsection*{Quantum codes from partitions}
The construction of quantum codes from partitions has a substantial history, with early appearances in \cite{knill2000theory} and several later developments across different settings \cite{calderbank1996good,movassagh2024constructing,aydin2026quantum,elimelech2026asymptotically}. The first three of these works focused on exact quantum error correction, while \cite{elimelech2026asymptotically} was the first to study approximate QEC for such codes, motivating the AQEC analysis in this work.  In Section~\ref{sec:PartitionCodes}, we formulate these ideas in the setting of Hilbert spaces indexed by discrete metric spaces and analyze them through the metric--error alignment hierarchy introduced above.

The construction itself is very simple. One begins with a subset $C\subseteq X$ (typically referred to as a \textit{classical code}) of basis labels and partitions it into disjoint blocks $C_0,\dots,C_{K-1}$. To construct a quantum code, we form 
superpositions of the basis vectors in each of the blocks:
\begin{equation}
    \ket{c_i}=\sum_{x\in C_i}\alpha_x\ket{x},\qquad \sum_{x\in C_i}\abs{\alpha_x}^2=1.\label{eq:partitionDefInt}
\end{equation}
The resulting quantum code is the $\C$-span of these states. We refer to Construction~\ref{const:partitionCodes} for the formal
definition. 

The role of the metric--error alignment hierarchy is already visible in the exact-QEC setting. The starting point is a classical code $C\subseteq X$ whose elements are well separated in the sense that the minimum distance between two distinct
elements of $C$ is bounded below by some number $t$. If $(\calH_X,d,\mathscr{E})$ is in $\mathscr{L}_1$ or $\mathscr{L}_2$, then errors from the set $\calE_t$ cannot result in an overlap between basis states indexed by different blocks of the partition: the corresponding off-diagonal (orthogonality) Knill--Laflamme conditions are automatically satisfied. Thus, once the minimum distance of $C$ is sufficiently large, we only need to take care of the non-deformation conditions inside each block. The next theorem (informal restatement of Theorem~\ref{th:exactQECpatition})
 shows when these conditions can be fulfilled. 

\begin{introtheorem}[Existence of exact QEC partition codes]
Let $C\subseteq X$ be a classical code with $d(C)>t$. 
Then $C$ gives rise to a $K$-dimensional exact QEC partition code for $\calE_t$ whenever
\[(\calH_X,d,\mathscr{E})\in\mathscr{L}_i
\qquad\text{and}\qquad
|C|\geq (K-1)\bigl(|\calE_t|^{3-i}+1\bigr)+1,
\qquad i=1,2.\] 
\end{introtheorem}

The exact-QEC existence result above is based on a convex-geometric argument: once the off-diagonal error-correction conditions are guaranteed by the metric separation of the underlying classical code, the remaining conditions are reduced to the existence of a suitable partition. A common tool for finding it is given by the Tverberg theorem, see, e.g., \cite{aydin2026quantum}. As explained in Remark~\ref{rem:Tver}, this argument can in principle be made constructive, but the resulting procedure for finding the partition and the corresponding coefficients has prohibitively large complexity even for modest code parameters. This highlights the need for a simple and genuinely usable construction. Such a construction becomes available in the approximate setting: one may choose the underlying classical code at random, partition it into equal-size blocks in an arbitrary way, and take uniform superpositions over each block. The averaging inherently
present in this random partition construction is precisely what makes AQEC accessible, and, through the error-set model developed above, the resulting codes still come with uniform adversarial guarantees for all channels controlled by the relevant error set.

The following theorem is an informal asymptotic restatement of Theorem~\ref{th:AQECrandomPar}, in which the approximation error is required to vanish. It gives the error-set AQEC guarantees for random partition codes relevant to our asymptotic constructions, whereas the full theorem applies more generally.
\begin{introtheorem}[AQEC guarantees for random partition codes]  
Let $\calH_X$ be a Hilbert space indexed by $(X,d)$ and let $\mathscr{E}=(\calE_t)_t$ be an error-set family such that $(\calH_X,d,\mathscr{E})\in\mathscr{L}_i$, for $i=1,2$. Suppose that $C\subseteq X$ is a random classical code obtained by drawing $\abs{C}$ i.i.d. points from a probability measure $\mu$ on $X$ such that $d(C)>t$ with high probability, and let $\{C_1,\dots,C_K\}$ be an arbitrary partition of $C$ into equal-size blocks.

Let $Q$ be the $K$-dimensional quantum code with basis vectors \eqref{eq:partitionDefInt} formed as uniform superpositions. With high probability, the code $Q$ is an $\varepsilon$-AQEC code for the error set $\calE_t$ with $\varepsilon\to 0$ as long as
\[
 K=o(|C|^{\frac13}/{|\calE_t|^{2-\frac{2i}3}}).
\]
\end{introtheorem}
Our conclusions extend beyond asymptotic claims. Numerical examples in Section~\ref{sec:NumericalPartition} show that a fixed, unoptimized i.i.d. partition code can support competitive worst-case AQEC performance under amplitude damping already at modest values of the parameters. We compare this code with two instances of the Chuang--Leung--Yamamoto bosonic construction having the same number of modes, total excitation, and code dimension.

 \subsubsection*{Asymptotically good families of partition codes}

Since AQEC requires the approximation error $\varepsilon$ to vanish asymptotically, the natural setting for our constructions is the asymptotic regime. Accordingly, we study families of partition codes for several central quantum systems and noise models, and ask when they are asymptotically good in the sense of Definition~\ref{def:AsymGood}, that is, when both the code rate and the relative approximate distance 
remain bounded away from zero. This is the relevant regime for robust quantum coding because, in most models of interest, the severity of the noise grows linearly with the system size. The quantum systems and noise models treated below, together with their positions in the metric--error alignment hierarchy, are summarized in Table~\ref{tab:metric_spaces_examples}.

For each family, we establish existential claims for exact-QEC partition codes and also analyze explicit random AQEC partition codes. The existential exact-QEC guarantees typically give better rate--distance tradeoffs because they amount to conditions for the existence of very delicate and potentially rare combinatorial objects, such as Tverberg partitions, selecting the best partition among many possibilities. Their drawback is that this construction is generally computationally infeasible. The random AQEC constructions are weaker quantitatively since they rely on averaging and concentration, but they are simple, explicit, and practically deployable. Taken together, the two analyses show that the partition-code framework is broad enough to produce asymptotically good codes in many settings and provide complementary perspectives on the tradeoffs between performance and constructibility.

\begin{enumerate}
\item \textbf{$q$-ary codes on tensor-product Hilbert spaces (Section~\ref{sec:QuditsConst}).} We study bounded-weight errors, deletions, and amplitude damping noise.
For bounded-weight Pauli errors, the exact partition construction matches the asymptotic CSS rate, while the random construction gives AQEC with positive rate. For deletion errors, our framework yields the first asymptotically good quantum deletion codes. For amplitude damping, it gives the first evidence that correcting linearly many amplitude damping errors can be asymptotically easier than correcting the same number of general bounded-weight errors. 

\item \textbf{Rydberg atom chains (Section~\ref{sec:Rydberg}).} We give the first construction of codes for Fibonacci Rydberg chains, whose Hilbert space is constrained by the blockade condition that no two consecutive atoms can be simultaneously in the excited (up-spin) state. Our partition-code framework naturally respects this constraint by working directly within the allowed subspace, yielding both exact and approximate codes that are intrinsically compatible with the Fibonacci structure of the system.

\item \textbf{Majorana fermionic codes (Section~\ref{sec:majo}).} The application of the partition-code
framework to fermionic Fock space with Majorana noise yields asymptotically good codes in the setting where existing constructions are only stabilizer/subsystem-based. The resulting codes provide non-stabilizer examples and show that the general metric-indexed viewpoint extends naturally to fermionic systems.

\item \textbf{Constant-excitation Fock-state codes (Section~\ref{sec:CEFCex}).} We consider both amplitude damping noise and number-shift/phase-rotation noise. The analysis yields asymptotically good exact and approximate codes in the high-excitation regime, and also shows how different sampling distributions for the underlying simplex codes trade coding rate against physically relevant properties such as bounded per-mode occupancy.

\item \textbf{Permutation-invariant codes (Section~\ref{sec:PIex}).} Here, the discrete simplex again provides the underlying metric space, now for symmetric subspaces. As a result, we obtain asymptotically good exact and approximate partition codes against deletions and erasures.
\end{enumerate}

\subsection*{Main definitions}
For the reader's convenience, we list the main concepts defined or used in this study.
\begin{itemize}[itemsep=-1pt,label=$\triangleright$]
    \item $\calE$-controlled channels, Definition \ref{def:AQECC}, p.~\pageref{def:AQECC};
    \item Approximate quantum error correction for the error-set model, Definition \ref{def:AQECC}, p.~\pageref{def:AQECC};
    \item Approximate code distance, Definition \ref{def:dist}, p.~\pageref{def:dist};
    \item Metric--error alignment hierarchy, Definition \ref{def:MEAhirarchy}, p.~\pageref{def:MEAhirarchy};
    \item Codes from partitions, Construction \ref{const:partitionCodes}, p.~\pageref{const:partitionCodes}.
\end{itemize}

\subsection{Conclusion and outlook}\label{sec: conclusion}

In this work, we initiate the study of an error-set model for approximate quantum error correction, replacing channel-by-channel guarantees with a uniform theory for families of channels, in direct analogy with exact QEC. The resulting theory shows that, although the mathematical structure of AQEC is more delicate than that of the exact Knill--Laflamme framework, the basic organizing ideas of exact adversarial QEC carry over to the approximate regime. In particular, we identify a restricted form of linearity, introduce a corresponding family of $\calE$-controlled channels, derive uniform AQEC conditions for this model under both the worst-case and average-case criteria, define the notion of approximate distance, and establish a quantitative relation between erasures and general errors. Together with the partition-code framework developed in the second part of the paper, this gives a broad collection of examples showing that the error-set viewpoint is not only conceptually natural, but also practically productive across a range of quantum systems and noise models.

The present work is only a first step toward an adversarial theory of AQEC. Although we believe it lays the foundations, much remains to be understood before such a theory becomes as mature and useful as its exact-QEC counterpart. We conclude by highlighting several directions that seem especially important for its further development.

\begin{enumerate}
\item \textbf{Fault-tolerant AQEC recovery.}
In our companion paper \cite{elimelech2026universal} we develop universal error-set recovery, which forms the second component of an adversarial AQEC theory. The next step toward integrating AQEC into physical platforms is to implement such recovery maps fault tolerantly, accounting also for errors introduced during the recovery procedure itself~\cite{shor1996fault}. This raises a central question: under which assumptions can universal AQEC recovery maps be implemented fault tolerantly while retaining useful correction guarantees over the entire controlled family? Addressing this question would be an important step toward making the adversarial AQEC framework practically useful.

\item \textbf{An AQEC theory for stabilizer codes.}  
Stabilizer codes are by far the most prominent family in quantum error correction, owing to their algebraic structure, their many explicit constructions, and their central role in fault-tolerant quantum computation. Yet, in contrast to exact QEC, there is still no comparably satisfactory AQEC theory for stabilizer codes: in particular, one lacks a convenient characterization of their performance against general approximate noise families. A plausible reason is that, without an error-set model, AQEC for stabilizer codes has largely been studied on a channel-by-channel basis. The approach developed here suggests a new route. It is therefore natural to ask whether one can build a genuine theory of stabilizer AQEC for error sets, with structural criteria and distance-like notions that parallel the exact setting.

\item \textbf{Other mechanisms for error-set AQEC.}  
The main limitations of our framework are inherited from the current channel-based theory of AQEC. Even after organizing the theory around error sets, the resulting conditions remain difficult to analyze in general, especially under the worst-case criterion, where the B\'eny--Oreshkov superoperators are hard to control, and in the average-case criterion, where sharp estimates for the Petz-map expressions are still missing.  Although these tools are sufficient to extract a coherent and useful theory, they do not yet yield a fully satisfactory picture. For example, we do not know how to formulate general necessary and sufficient conditions for 
error-set AQEC, even though Proposition~\ref{prop:exampleOptimality} shows that our sufficient condition is, in some cases, already essentially optimal. This raises a broader question: is there another mechanism, beyond the present B\'eny-Oreshkov/Petz-based approach, that would lead to genuinely sharp error-set AQEC criteria in full generality?
\end{enumerate}

These questions suggest that the theory introduced here is only the first step toward a broader understanding of approximate quantum error correction. If developed further, the error-set perspective on AQEC could provide a common language for robust approximate coding across many quantum platforms, much as the adversarial error-set model does for exact QEC.
 
\section{Preliminaries}

Throughout this work, we consider finite-dimensional Hilbert spaces, often denoted by $\calH$. The spaces of linear operators and density operators on $\calH$ are denoted by $L(\calH)$ and $D(\calH)$, respectively. We denote the state fidelity by $\calF(\rho,\tau) := \norm{\sqrt{\rho}\sqrt{\tau}}_1
$. 
Let $\calH$ and $\calH'$ be Hilbert spaces. A quantum channel from $\calH$ to $\calH'$ is a completely positive trace-preserving (CPTP) map $\calN:L(\calH)\to L(\calH')$. It is well known that any CPTP map can be represented as $\calN(\rho)=\sum_k E_k\rho E_k^\dag$,
where $\calE=\ppp{E_k}_k$ is a set of Kraus operators for $\calN$, satisfying the completeness relation $\sum_k E_k^\dag E_k=I$. We often use the trace norm $\norm{ }_1$, Frobenius norm $\norm{ }_2$, and spectral norm $\norm{ }_\infty$, defined as
\begin{align}
    \norm{\rho}_1 := \tr\p{\sqrt{\rho^\dag \rho}}=\sum_{\lambda\in \s{S}(\rho)}\abs{\lambda}, \quad
    \norm{\rho}_2 :=  \sqrt{\tr\p{\rho^\dag \rho}}=\sqrt{\sum_{\lambda\in \s{S}(\rho)}\abs{\lambda}^2},\quad\norm{\rho}_\infty := \max_{\lambda\in \s{S}(\rho)}\abs{\lambda},\label{eq:normsDef}
\end{align}
where $\s{S}(\rho)$ denotes the set of singular values of $\rho$. The Hilbert--Schmidt inner product on $L(\calH)$ is defined as $\innerP{T,P}_{\s{HS}}=\tr(T^\dag P)$. Chapter 1 in Watrous \cite{watrous2018theory} is a good general reference for the norms; certain specific properties that we use are collected in 
Appendix~\ref{app: Norms}.

We use the standard notation $[n]=\ppp{0,1,\dots,n-1}$ for an integer $n\in\N$ and write $\Z_0$ for the set of nonnegative integers. Lowercase underlined letters, such as $\ux,\uy,\un$, denote words or vectors over a finite alphabet. Sans-serif letters, such as $\s{X}$, typically denote random variables, and their underlined counterparts, such as $\su{X}$, denote random vectors.

For a positive integer $N$ we use the standard notation $\binom{N}{t}=\frac{N!}{(N-t)!t!}$ for $t\leq N$, and set $\binom{N}{t}=0$ otherwise. By abuse of notation, we write $\binom{[N]}{t}$ for the collection of all subsets of $[N]$ of size exactly $t$. For $\un=(n_0,\dots,n_{q-1})\in\Z_0^q$ with $\sum_{i=0}^{q-1}n_i=N$, we use the multinomial notation
\[
\binom{N}{\un} := \frac{N!}{n_0! \cdots n_{q-1}!}.
\]
Given positive integers $N$ and $q$, we define the {\em discrete simplex} by
\begin{equation}
    \calS_{q,N} :=  \ppp{(n_0,\dots,n_{q-1})\in \Z_0^q ~:~ \sum_{i=0}^{q-1}n_i=N}.
    \label{eq:simplex}
\end{equation}
The $q$-ary entropy function is defined as
\begin{equation}
    \calH_q(\delta) :=  \delta\log_q(q-1)-\delta\log_q\delta-(1-\delta)\log_q(1-\delta),\label{eq:qaryEntropy}
\end{equation} 
with the usual convention $0\log 0=0$.  We use standard asymptotic notation such as $o( )$, $O( )$, and $\Theta( )$ in the limit where the relevant blocklength parameter tends to infinity.

\subsection{Qudit codes, erasures, deletions, and bounded-weight errors}\label{sec:qudits}

Let $\calH_q$ be a $q$-dimensional Hilbert space spanned by a computational basis $\ket{0},\dots,\ket{q-1}$. In analogy with classical $q$-ary codes, we refer to subspaces of $\calH_q^{\otimes N}$ as $q$-ary quantum codes. Thus, an $N$-qudit system over $\calH_q$ is given by the tensor product $\calH_q^{\otimes N}$, and a code on this system is a subspace $Q\subseteq \calH_q^{\otimes N}$. In this section, we consider bounded-weight errors and deletion errors on such $N$-qudit systems. Throughout the paper, {\em weight} refers to the Hamming weight of a vector in $[q]^N$ or, when applied to an $n$-fold tensor product of HW operators, the number of non-identity factors.
\subsubsection{Generalized Pauli errors}
 The Heisenberg--Weyl (HW) operators provide a natural generalization of Pauli operators to qudit systems. Let $\omega=e^{\frac{2\pi i}{q}}$ and let $X$ and $Z$ on the $\calH_q$ be defined as
\begin{align*}
    X\ket{j}=\ket{j+1~\mathrm{mod}~q}, \quad Z\ket{j}=\omega^{j}\ket{j}. 
\end{align*}
The HW operators are the unitaries given by
\begin{align}
    W_{a,b}=X^aZ^b, \quad a,b\in [q].\label{eq: Wab}
\end{align}
the set $\ppp{W_{a,b}/\sqrt{q}}_{a,b\in [q]}$ is an orthonormal basis of $L(\calH_q)$ with respect to the Hilbert-Schmidt inner product. In the case of $N$ qudit systems, 
tensor products $\otimes_{i=0}^{N-1} W_{a_i,b_i}$ play the role of generalized Pauli errors. Below we write such operators as $W_\ux$, where $\ux\in [q^2]^N$ is the indicator vector of nonzero pairs $(a_i,b_i)$. In particular, error operators of support size at most $t$ are spanned by HW operators $W_\ux$ such that the Hamming weight $\s{wt}(\ux)\le t$.

\subsubsection{Erasure errors}
Erasure errors model noise processes in which the information in a subsystem is erased, where \textit{the identity of this subsystem is known}. Such noise processes are often modeled by channels in which certain qudits are replaced by a distinguished orthogonal flag state $\ket{\perp}$, orthogonal to $ \calH_q$, thereby marking the locations of the errors while destroying the corresponding local quantum information. When an erasure error is introduced, the affected coordinates are revealed to the receiver through the appearance of the erasure flag, and the resulting state remains in an \(N\)-fold tensor product, but over the enlarged local space \(\calH_q^{\perp} :=  \calH_q \oplus \Span\{\ket{\perp}\}\). A \(t\)-erasure error on a \(q\)-qudit system \(\calH_q^{\otimes N}\) can thus be defined by Kraus operators that replace the qudits in positions \(I\subseteq [N]\) by \(\ket{\perp}\). Formally, a \(t\)-erasure channel is any channel which admits a Kraus representation with operators from the span of the set
\begin{equation}\label{eq:erasure_operator}
\calE_{t}^{\s{Er}} := \ppp{
  E_{I,\ux} : I\in \binom{[N]}{t},\ \ux\in [q]^t},
\qquad
E_{I,\ux} := 
\ket{\perp^{t}}_{I}\bra{\ux}_{I}\otimes I_{[N]\setminus I},
\end{equation}
where \(\ket{\perp^{t}}_{I}\) denotes the tensor product of \(t\) copies of \(\ket{\perp}\) placed in the coordinates indexed by \(I\).

\subsubsection{Deletion errors}\label{sec:DelErrors}
Quantum deletion errors model noise processes in which an unknown subsystem of a composite quantum system is removed, resulting in a state on a smaller Hilbert space. Operationally, this corresponds to applying a partial trace over an unknown subsystem. In contrast to erasure errors, the identity of the removed subsystem is not revealed to the receiver (by classical or quantum side information), and the resulting state therefore lies in a tensor product space of smaller dimension. Quantum error correction for deletion errors has been studied in several works, including explicit constructions of deletion-correcting codes for qubit systems \cite{nakayama2020first,hagiwara2020four,shibayama2021construction,hagiwara2022quantum,shibayama2021equivalence, nakamura2026insertion}, and permutation-invariant quantum codes \cite{ouyang2021permutation, aydin2023family, aydin2026quantum, bulled2026equivalence}. A $t$-deletion error on a $q$-qudit system $\calH_q^{\otimes N}$ can be defined by the set of Kraus operators of the form $E_{I,\ux}$, which acts on a computational basis state by removing qudits in positions $I\subseteq [N]$ if they are given by $\ket{\ux}$ (and by $0$ otherwise). Formally, a $t$-deletion channel is any channel that admits a Kraus representation with operators from the linear span of the set
\begin{equation}
    \calE_{t}^{\s{Del}} := \ppp{     D_{I,\ux} ~:~ I\in  \binom{[N]}{t}, \ux \in [q]^{t}}, \quad D_{I,\ux} :=  \bra{\ux}_I\otimes I_{[N]\setminus I}.\label{eq:DelErrors}
\end{equation}

\subsection{Permutation-invariant codes}\label{sec:PIcodes}

Permutation-invariant (PI) codes encode quantum information into the symmetric subspace of $N$ qudits. In other words, a PI code $Q^{\mathrm {PI}}\subseteq \calH_q^{\otimes N}$ is stabilized by the action of the symmetric group: for every $\ket{\psi}\in Q^{\mathrm {PI}}$ and every $g\in {\mathfrak S}_N$, 
$
    g\ket{\psi}=\ket{\psi},
$
where a permutation $g$ acts on the $\calH_q^{\otimes N}$ by permuting the qudits of the quantum state.
PI codes were introduced by Ruskai \cite{ruskaiExchange,ruskai-polatsek} to protect against particle exchange errors, which arise from the indistinguishability of identical particles. 
Beyond their original motivation, PI codes have natural physical realizations as ground states of the ferromagnetic Heisenberg model \cite{ouyangPI} and as collective excitations of atom-cavity systems (see, e.g., Refs.~\cite{chaudhury2007quantum,haas2014entangled,strobel2014fisher,lucke2014detecting,mcconnell2015entanglement,pezze2018quantum}).

A convenient way to describe symmetric codewords is through \emph{Dicke states}. To define them, we recall that the composition of a string $\ux$ of length $N$ over a $q$-ary alphabet is a $q$-tuple
\begin{equation}
    \mathrm{C}(\ux) =(n_0(\ux),n_1(\ux),\ldots,n_{q-1}(\ux)),\label{eq:Copositionof}
\end{equation}
where $n_i(\ux)$ is the number of occurrences of the character $i$ in the string $\ux$. Note that $\sum_x\mathrm{C}(\ux)=N$. A qudit {\em Dicke state} is the linear combination of all qudit states with the same composition, i.e., 
\begin{equation}
    \ket{D_{\un}} = \frac{1}{\sqrt{\binom{N}{\un}}}\sum_{\mathrm{C}(\ux)=\un}\ket{\ux}.\label{eq:Dickestates}
\end{equation}
Note that a Dicke state is completely determined by the composition $\un$, which is an element on the discrete simplex $\calS_{q,N}$ defined in \eqref{eq:simplex}.  A code $Q^{\mathrm{PI}}$ is a PI code if and only if it is contained in the \textit{permutation-symmetric spaces} given by the linear span of all Dicke states:
\begin{equation}
    \label{eq:symmetricSpaceX}
    \mathrm{Sym}({q,N})=\Span\ppp{\ket{D_{\un}}~:~ \un\in \calS_{q,N}}\subseteq \calH_q^{\otimes N}.
\end{equation}
\subsection{Constant-excitation bosonic Fock state codes}\label{sec:FockStateCodes}
In this subsection, we briefly introduce constant-excitation bosonic Fock-state codes, a family of quantum codes defined within the Hilbert space of bosonic modes with fixed total excitation number. These codes are naturally suited to optical and superconducting architectures, where information is encoded in photon-number states. For  fixed positive integers $N,q\in N$, we consider the $q$-mode Fock state space with a total excitation $N$ as 
\begin{equation}
    \calH_{q,N}=\Span\ppp{\ket{\un}=\ket{n_1}\otimes  s \otimes \ket{n_q} ~:~ \un\in \calS_{q,N}},\label{eq:COnstExSpace}
\end{equation} where $\ppp{\ket{\un}}_{\un}$ is an orthonormal basis indexed by the vertices of the simplex.

\subsubsection{ Amplitude damping noise}
A dominant noise mechanism in such systems is photon loss, which is accurately modeled by the amplitude damping channel. This channel captures the stochastic loss of excitations to the environment and thus provides the standard physical noise model for analyzing the error-correcting properties of constant-excitation bosonic codes. The amplitude damping channel with transmissivity $\gamma\in(0,1)$, denoted by $\calN_{\gamma}$, acting on $\calH_{q,N}$, is defined by the Kraus operators $\ppp{A_{\ur} ~:~ \ur \in \Z_0^{q}}$:
\begin{align}
    A_{\ur}=A_{r_1}\otimes A_{r_2}\otimes  \cdots \otimes A_{r_q}, \quad A_{r_i}\ket{n}=\begin{cases}
    \sqrt{\binom{n}{r_i}}\gamma^{\frac{r_i}{2}}(1-\gamma)^{\frac{n-r_i}{2}}\ket{n-r_i} & r_i\leq n,\\
    0 & r_i > n.
\end{cases}\label{eq:photonLossErr}
\end{align}

\subsubsection{Number-shift and phase-rotation noise}

Beyond amplitude damping, one may also consider bosonic noise models generated by bounded photon-number shifts together with bounded phase rotations. Models of this type appear in the literature on number-phase and rotation-symmetric bosonic codes. In the single-mode setting, \cite{albert2025bosonic} introduces unilateral number-shift operators as the bosonic analog of rotor momentum kicks. Discrete $N$-fold phase rotations together with number-shift operators are studied in the rotation-symmetric bosonic setting in \cite{marinoff2024explicit}, and the associated tradeoff between number-shift and phase-shift resilience is analyzed in \cite{ouyang2021trade}. 

 For integers $r,\theta\in \Z$ and $\beta\in [0,1]$ let $E_{r,\theta,\beta}$ denote the single-mode shift-rotation operator
\begin{equation*}
    E_{r,\theta,\beta}\ket{n}=
    \begin{cases}
        \exp\!\p{i\frac{2\pi\beta}{N}\theta n}\ket{n+r} & n+r\in\Z_0,\\
        0 & \text{otherwise}.
    \end{cases}
\end{equation*}
For parameters $t_{\mathrm{s}},t_{\mathrm{r}}\in\Z_0$ and $\beta\in[0,1]$, define the multimode shift-rotation error set
\begin{equation}
 \calE^{\s{SR}}_{t_{\mathrm{s}},t_{\mathrm{r}},\beta}=\ppp{E_{\ur,\uTheta,\beta}~:~\norm{\ur}_1\leq t_{\mathrm{s}},~\norm{\uTheta}_1\leq t_{\mathrm{r}}},\label{eq:BosonicSR}
\end{equation}
where $\ur,\uTheta\in\Z^q$ and
\begin{equation*}
    E_{\ur,\uTheta,\beta}:=\bigotimes_{i=1}^q E_{r_i,\theta_i,\beta}.
\end{equation*}
Equivalently, the action on a $q$-mode state $\ket{\un}=\ket{n_1}\otimes \cdots\otimes\ket{n_q}$ is defined as
\begin{equation}
    E_{\ur,\uTheta,\beta }\ket{\un}=
    \begin{cases}
        \exp\p{i\frac{2\pi\beta}{N}\sum_{i=1}^q \theta_i n_i}   \ket{\un+\ur} & \un+\ur\in\Z_0^q,\\
        0 & \text{otherwise}
    \end{cases}\label{eq:ShiftRotationErrorSet}
\end{equation}
where the second case refers to tuples with one or more negative entries.
 Here $\norm{  }_1$ denotes the vector $\ell_1$ norm on $\Z^q$. The shift vector $\ur$ models bounded photon-number gain/loss, while the diagonal factor models bounded relative phase rotations across the modes. The parameter $\beta>0$ truncates the per-mode rotation scale: the basic discrete phase increment in mode $i$ is $2\pi\beta\theta_i/N$, so reducing $\beta$ restricts the maximal rotation that can occur in each mode, and therefore also the total accumulated rotation across all modes. At the same time, parameter $t_{\mathrm{r}}$ controls the total rotation budget through the constraint $\norm{\uTheta}_1\le t_{\mathrm{r}}$ and may itself scale beyond $N$ if desired. In this sense, $\beta$ generalizes the limited-rotation viewpoint considered in \cite{albert2025bosonic} from the single-mode number-phase picture to the multimode setting.

\subsection{Fermionic Fock space codes}\label{sec:Fermionic}
We consider fermionic Fock-space codes, namely quantum codes defined within the Hilbert space of a finite number of fermionic modes. A fermionic mode is a two-level degree of freedom whose occupation number is either $0$ or $1$. Thus, for a fixed positive integer $N\in\N$, the fermionic Fock space of $N$ modes is
\begin{equation}
    \calF_N=\Span\ppp{\ket{\ux}~:~\ux=(x_1,\ldots,x_N)\in\ppp{0,1}^N},\label{eq:FermionicFockSpace}
\end{equation}
where $\ppp{\ket{\ux}}_{\ux}$ is the occupation-number basis. Equivalently, after fixing this basis, $\calF_N$ is identified with the computational basis Hilbert space of $N$ qubits. The coordinate $x_j$ records whether the $j$th fermionic mode is empty or occupied, and the vector $\ket{0,\ldots,0}$ is the vacuum state.
\subsubsection{Majorana noise}
Majorana fermion codes were introduced in \cite{bravyi2010majorana} to protect quantum information against low-weight fermionic errors. Majorana surface-code architectures for fault-tolerant quantum computation were developed in \cite{vijay2015majorana}, and a broader framework for quantum computation with Majorana surface codes and Majorana color codes was developed in \cite{litinski2018quantum}. We follow the support-size Majorana error model of \cite{bravyi2010majorana}.

Let $a_j$ and $a_j^\dag$ denote the annihilation and creation operators of the $j$th fermionic mode. They satisfy the canonical anticommutation relations $a_i a_j+a_j a_i=0$, and $ a_i a_j^\dag+a_j^\dag a_i=\delta_{i,j}I$. The associated $2N$ Majorana operators are defined by
\begin{equation*}
    \gamma_{2j}:=a_j+a_j^\dag,\quad \gamma_{2j+1}:=i(a_j^\dag-a_j),\quad j\in[N].\label{eq:MajoranaFromCreation}
\end{equation*}
These operators are Hermitian and satisfy
\begin{equation*}
    \gamma_u^\dag=\gamma_u,\quad \gamma_u\gamma_v+\gamma_v\gamma_u=2\delta_{u,v}I.\label{eq:MajoranaCAR}
\end{equation*}
Equivalently, under the Jordan-Wigner representation used in \cite{bravyi2010majorana}, one has
\begin{equation}
    \gamma_{2j}=Z_0 \cdots Z_{j-1}X_j,\quad \gamma_{2j+1}=Z_0\cdots Z_{j-1}Y_j,\label{eq:JordanWignerMajorana}
\end{equation}
where the Pauli operators act on the computational-basis representation of $\calF_N$. For a subset $I\subseteq[2N]$, we write $I=\ppp{u_1,\dots, u_m}$, $u_1< u_2< \cdots<u_m$ and define the Majorana monomial 
\begin{align}
\gamma_I := \gamma_{u_1}\gamma_{u_2} \cdots\gamma_{u_m},\label{eq:MajoranaI}
\end{align} with the convention that $\gamma_{\varnothing}=I$. The support of $\gamma_I$ is the subset $I$ of Majorana modes, and its Majorana weight is $\abs{I}$.

\section{Approximate quantum error correction}
We begin by introducing the notion of approximate quantum error correction (AQEC), in which the requirement of perfect recovery is relaxed and the decoded state is allowed to approximate the original encoded state up to a controlled error. Concretely, instead of demanding exact satisfaction of the Knill–Laflamme conditions, one permits a small deviation in the recovery map so that, after the action of the noise and decoding, each codeword is recovered only approximately.

Several rigorous frameworks for AQEC have been developed in the literature. We adopt the operator-algebraic framework of B\'eny and Oreshkov, as it provides necessary and sufficient conditions for the approximate correctability of a quantum code at the channel level, thereby giving a conceptually complete formulation of AQEC for channels. Building on this viewpoint, we formulate an approximate error-set model and show that it is both meaningful and workable for our purposes.

Prior work on Pauli-type error models already suggests that an error-set viewpoint should be possible, but only in rather specialized settings. In \cite{bergamaschi2024approaching}, the authors constructed a specific family of AQEC codes approaching the quantum Singleton bound and correcting adversarial errors on a linear number of registers; in the language used here, this corresponds to channels whose Kraus operators lie in the span of bounded-weight Pauli errors for that particular construction. More recently, \cite{ma2025haar} studied Haar-random codes for large structured unitary error sets, but their conditions rely on strong Hilbert-Schmidt orthogonality assumptions on the error operators together with approximate nondegeneracy properties of the corrupted basis codewords. These results provide important evidence that AQEC can support an error-set interpretation, but they do not develop such an interpretation as a general model. Our goal is to make this structure explicit: we introduce $\calE$-controlled channels for an arbitrary finite error set $\calE$ and derive AQEC conditions directly from the action of $\calE$ on the code.

Let us begin by briefly introducing distinguishability measures for quantum channels, which we use to quantify the performance of approximate quantum error correction. 
\begin{definition}[Entanglement fidelity, Bures distance, diamond norm]\label{def:Bures}
    Let $\calN,\calM :L(\calH)\to L(\calH')$ be quantum channels. The worst-case entanglement fidelity of $\calN$ and $\calM$ is  defined as
     \[\calF_{e}(\calN,\calM) := \inf_{\rho \in D(\calH)} \calF\p{\calN(\ket{\psi_\rho}\bra{\psi_\rho}),\calM(\ket{\psi_\rho}\bra{\psi_\rho})},\]
     where $\ket{\psi_{\rho}}$ is any purification of $\rho$. The Bures distance is defined as 
     \begin{equation}
         \label{eq:BuresDist}
         d\p{\calM,\calN} :=  \sqrt{1-\calF_{e}(\calN,\calM)}. 
     \end{equation}
 $d( , )${Bures distance}   
     The diamond norm distance between $\calN$ and $\calM$ is defined as: 
     \[\norm{\calN-\calM}_{\diamond} :=  \sup_{\rho\in D(\calH^{\otimes 2})}\norm{(I_{L(\calH)}\otimes\calN)(\rho)-(I_{L(\calH)}\otimes\calM)(\rho)}_1.\]

\end{definition}
From the Fuchs–van de Graaf inequalities (see \cite[Section 6.2]{khatri2020principles}) it follows that 
\begin{align}
    d(\calN,\calM)^2\leq \frac{1}{2}\norm{\calN-\calM}_{\diamond}\leq d(\calN,\calM)\sqrt{2-d(\calN,\calM)^2}\leq \sqrt{2}\,d(\calN,\calM).\label{eq:diamondBures}
\end{align}
Since $\sqrt{1+x}=1+\frac{1}{2}x+O(x^2)$, whenever two superoperators are close with respect to the Bures distance, they are also close with respect to the diamond norm distance. In particular, proximity in the sense of the Bures distance and the sense of diamond norm distance are equivalent. Let us now define approximate quantum error-correcting codes.
\graybox{
\begin{definition}[Approximate quantum error-correction codes for quantum channels]\label{def: AQEC codes}
    A quantum code $Q\subseteq \calH$ is said to be an $\varepsilon$-approximate quantum error-correcting (AQEC) code for a channel $\calN:L(\calH)\to L(\calH')$ if there exists a CPTP decoding operation $\calD:L(\calH')\to L(\calH)$ such that 
    \[d\p{\calD\circ \calN|_{L(Q)},I_{L(Q)}}\leq \varepsilon.\]
    A sequence of codes $Q_N\subseteq \calH_N$ is said to approximately correct the sequence of channels $\calN_N$ if $Q_N$ is $\varepsilon_N$-AQEC code  for $\calN_N$, where $\varepsilon_N=o(1)$.  
\end{definition}}
    
    \subsection{Error-set model and code distance for AQEC via the B\'eny-Oreshkov framework}\label{sec:ErrorSetModel}

    Among the general formulations of approximate quantum error correction, the framework of B\'eny and Oreshkov \cite{beny2010general} is especially appealing because it gives necessary and sufficient conditions for approximate correctability of a \emph{fixed noise channel} in terms of the Bures distance. Their result characterizes when a code approximately corrects a channel by analyzing how the operators in a Kraus representation act on the code space, and shows that approximate correction is equivalent to these operators acting nearly as scalars on the code. Our goal in this section is to move beyond this channel-level viewpoint and formulate an \emph{error-set model} for AQEC: rather than studying one channel at a time, we fix a set of error operators $\calE$ and analyze approximate correction implemented uniformly over the family of channels generated by $\calE$ (below we formally define the family of $\calE$-controlled channels). We show that the B\'eny-Oreshkov framework naturally leads to such a model, and that it retains the main structural features one would expect from an error-set formulation, including a corresponding notion of approximate distance.
    \begin{theorem}[B\'eny-Oreshkov \cite{beny2010general}]\label{th:BO}
        Let $\calN:\calH\to \calH'$ be a quantum channel with Kraus representation $\ppp{E_k}_{k=1}^M$, and let $Q\subseteq \calH$ be a quantum code. Then $Q$ is $\varepsilon$-AQEC code  for $\calN$ if and only if there exists a density matrix $\lambda=(\lambda_{k,l})_{k,l=1}^{M}$ such that  $d(\Lambda,\Lambda+\calB)\leq \varepsilon$, where\footnote{This corrects an unfortunate misprint in \cite[Corollary~2]{beny2010general} where the first term in $B_{k,l}$ is incorrectly recorded as $P E_k^\dag E_l P$.}
        \begin{align}
            \Lambda(\rho)=\sum_{k,l=1}^M \lambda_{k,l}\tr(P\rho)  \ket{k}\bra{l}, \quad \calB(\rho)=\sum_{k,l=1}^M \tr(B_{k,l}\rho)  \ket{k}\bra{l}, \quad B_{k,l}=P {E_l^\dag E_k} P - \lambda_{k,l} P.\label{eq:BOoperators}
        \end{align}
        where $\ppp{\ket{k}}_k$ is an orthonormal basis for $\C^{M}$, and $P$ is the projector on the code space $Q$.
    \end{theorem}
    
One of the central tenets of exact QEC is the universality of the Knill--Laflamme (KL) conditions as a criterion for error correction. If a code corrects an error set $\calE$, then the conclusion is not tied to a particular channel or to a particular Kraus representation: the code protects against noise on any channel whose Kraus operators lie in $\Span(\calE)$. This principle is the defining feature that distinguishes an error set as an adversarial noise model, rather than simply a description of a fixed channel.

For AQEC, this universal error-set interpretation is far less obvious. The B\'eny-Oreshkov conditions provide a powerful criterion for approximate correction of a fixed channel, but they do not by themselves allow one to pass from the Kraus representation to arbitrary linear combinations of the same errors. The obstruction is that approximate correction conditions do not enjoy the same linearity properties as the exact KL conditions. Therefore, to obtain a meaningful error-set formulation, one must restrict the channel family so that it remains rich enough to model adversarial noise, yet sufficiently structured for the approximate analysis to apply.

To this end, we fix a finite set of error operators $\calE$ and associate with it a family of channels whose Kraus operators lie in $\Span(\calE)$ and whose coefficient matrix satisfies a spectral constraint. This constraint replaces the missing linearity of the exact theory. It allows us to formulate approximate Knill--Laflamme-type conditions directly for the operators in $\calE$, and to deduce AQEC uniformly for every channel in the resulting $\calE$-controlled family.

\graybox{
    \begin{definition}[$\calE$-controlled channels]\label{def:AQECC}
    Let $\calE=\ppp{E_k}_k$ $\calE${set of error operators} be a finite set of error operators $E_k:\calH\to \calH'$. A channel  is said to be an $\calE$-\textit{controlled} if it admits a Kraus representation with operators $\ppp{A_m}_{m}$ of the form $A_m=\sum c_{m,k} E_k$, where the matrix $C=(c_{m,k})$ is an $\ell_2$  contraction, i.e., 
    \begin{equation}
        \norm{C}_\infty =\sup_{\norm{\ux}_2=1}\norm{Cx}_{2}\leq 1.\label{eq:controlled}
    \end{equation}
    We denote the set of all $\calE$-controlled channels by $\mathscr{N}(\calE)$. 
 
    A code is said to be an \textbf{$\varepsilon$-AQEC} code for the error set $\calE$ if it is an $\varepsilon$-AQEC code
in the sense of Definition~\ref{def: AQEC codes} for any \(\cal E\)-controlled noise channel.

\end{definition}}

We remark that the condition of Definition~\ref{def:AQECC} is independent of the Kraus representation as long as the error set $\calE$ is linearly independent. We elaborate on this in Observation~\ref{obs: independent} below.
In what follows, we provide an error-set version of the B\'eny-Oreshkov conditions that implies AQEC uniformly for 
$\calE$-controlled channels. 

\begin{definition}[B\'eny-Oreshkov superoperators, environment-leakage distance]\label{def:DiamondNormProxConst}
    Let $\calE=\ppp{E_k}_{k=0}^{M-1}$ be a finite set of error operators $E_k:\calH\to \calH'$ and $Q\subseteq \calH$ be a code space with projector $P$. For any $M\times M$ matrix $\lambda$, define a \textit{B\'eny-Oreshkov superoperator} $\calB^{\calE}_{\lambda,Q}:L(\calH)\to L(\C^M)$ for $\calE$ 
   $\calB^{\calE}_{\lambda,Q}${B\'eny-Oreshkov superoperator} as follows:
    \begin{equation}
        \calB^{\calE}_{\lambda,Q}(\rho)=\sum_{k,l}\tr(B_{k,l}\rho) \ket{k}\bra{l},\quad B_{k,l}=P{E_l^\dag E_k}P-\lambda_{k,l}P \label{eq:defBElambd}
    \end{equation}
 Further, define the \textit{environment-leakage distance} of $\calE$ as follows: $\zeta(\calE,Q)${environment-leakage distance}
    \[\zeta(\calE,Q)=\inf_{\lambda} \norm{\calB_{\lambda,Q}^{\calE}}_{\diamond}.\]
    This is the distance from the complementary map to the space of constant superoperators defined by the code and the error set.
\end{definition}

We start our investigation of the environment-leakage distance with a simple but important observation: $\zeta(\calE,Q)$ can be approximated up to a factor of $2$ by restricting to a specific well-behaved B\'eny-Oreshkov superoperator. The proof of the following lemma appears in Appendix~\ref{app:lemREPLACE}.
\begin{lemma}\label{lem:replaceOpt}
     Let $\calE=\ppp{E_k}_{k=0}^{M-1}$ be a finite set of error operators $E_k:\calH\to \calH'$ and $Q\subseteq \calH$ be a code space. For a  (normalized) codeword $\ket{c_0}\in Q$ consider the matrix $\lambda$ given by $\lambda_{k,l}=\bra{c_0}E_l^\dag E_k \ket{c_0}$. Then 
     \[\norm{\calB_{\lambda,Q}^{\calE}}_{\diamond}\leq 2 \zeta(\calE,Q).\]
\end{lemma}

Using Lemma~\ref{lem:replaceOpt}, we can prove a sufficient condition for AQEC for an error set, involving the environment-leakage distance: 

\begin{theorem}[Sufficient AQEC conditions for error sets]\label{th:BOerrorset}
     Let $\calE=\ppp{E_k}_{k=0}^{M-1}$ be a finite set of error operators, where $E_k:\calH\to \calH'$, and let $Q\subseteq \calH$ be a code. If $\zeta(\calE,Q)\leq \varepsilon^2/2$, then $Q$ is an $\varepsilon$-AQEC code  for $\calE$. 
\end{theorem}
\begin{proof}
    Let $\lambda$ be a matrix given by $\lambda_{k,l}=\bra{c_0}{E_l^\dag E_k} \ket{c_0}$, where  $\ket{c_0}\in Q$ is a normalized codeword. By Lemma~\ref{lem:replaceOpt} and the assumption of the theorem, 
    \[\norm{\calB^{\calE}_{\lambda,Q}}_{\diamond}\leq 2\zeta(\calE,Q)\leq  \varepsilon^2.\]
     Let $\calN\in \mathscr{N}(\calE)$ be an $\calE$-controlled channel with Kraus operators $\ppp{A_m}_{m=0}^{M'-1}$, where $A_m=\sum_{k}c_{m,k} E_k$ such that $\norm{C}_{\infty}\leq 1$. 
      For $m,h\in [M']$ define
    \[B^{(A)}_{m,h}=P{A_h^\dag A_m} P - \lambda^{(A)}_{m,h}{P},\]
    with $\lambda^{(A)}=\p{\lambda_{m,h}^{(A)}}_{m,h=0}^{M'-1}$ given by
    \begin{equation}
        \lambda_{m,h}^{(A)}=\sum_{k,l=0}^{M-1}c_{m,k} c_{h,l}^* \lambda_{k,l}=\sum_{k,l=0}^{M-1}c_{m,k} c_{h,l}^* \bra{c_0}E_l^\dag E_k \ket{c_0}=\bra{c_0}A_h^\dag A_m \ket{c_0}.\label{eq:newLambda}
    \end{equation} 
    Let $\Lambda_A,\calB_A$ be the corresponding B\'eny-Oreshkov superoperators as defined in \eqref{eq:BOoperators}, with respect to $\ppp{A_m}$ and the matrix $\lambda^{(A)}$.  We begin by observing that $\lambda^{(A)}$ is a density matrix. Indeed, let $\ket{v}\in \C^{M'}$ be any vector. By \eqref{eq:newLambda} we have
    \begin{align*}
        \bra{v}\lambda^{(A)} \ket{v}={\sum_{m,h}v_m^* \lambda^{(A)}_{m,h} v_h=\sum_{m,h}v_m^* \bra{c_0}A_h^\dag A_m\ket{c_0} v_h=\norm{\sum_m v_m^* A_m\ket{c_0}}_2^2}\geq 0,
        \end{align*}
 so $\lambda^{(A)}$ is positive semidefinite. It also has unit trace:
        \[\tr(\lambda^{(A)})=\sum_m \bra{c_0}A_m^\dag A_m \ket{c_0}=\braket{c_0|c_0}=1,\]
  where we used the fact that $(A_m)_m$ satisfies the completeness condition. 
    
     By \eqref{eq:diamondBures}, we have $d(\Lambda_A,\Lambda_A+\calB_A)\leq \sqrt{\norm{\calB_A}_{\diamond}}$ and therefore, by Theorem~\ref{th:BO}, to show that $Q$ is an $\varepsilon$-AQEC code  for $\calN$ it is sufficient to show that $\norm{\calB_A}_{\diamond}\leq \varepsilon^2$. 
    For any $m,h\in [M']$ we have 
    \[ B^{(A)}_{m,h}={PA_h^\dag A_m P - \lambda^{(A)}_{m,h}P=\sum_{k,l} c_{m,k}c_{h,l}^*\p{PE_l^\dag E_k P -\lambda_{k,l}P}=\sum_{k,l=0}^{M-1} c_{m,k}c_{h,l}^* B_{k,l}.} \]
    Here, the $B_{k,l}=P {E_l^\dag E_k} P -\lambda_{k,l}P$, $k,l\in [M]$ are the deviation-from-constant operators associated with the error set B\'eny-Oreshkov superoperator $\calB_{\lambda,Q}^{\calE}$.    Thus, for any $\rho\in L(\calH)$ we have
    \begin{align}
        \calB_A (\rho)&={\sum_{m,h=0}^{M'-1} \tr(\rho B^{(A)}_{m,h})  \ket{m}\bra{h}=\sum_{m,h=0}^{M'-1}\sum_{k,l=0}^{M-1}c_{m,k}c_{h,l}^* \tr(\rho B_{k,l})  \ket{m}\bra{h}.} \label{eq:repCalBa}
    \end{align}
    Consider the superoperator $\calT_C:L(\C^{M})\to L(\C^{M'})$ given by $\calT_C(\rho)= {C} \rho {C}^\dag$ where  
    \[{C}=\sum_{m=0}^{M'-1}\sum_{k=0}^{M-1} {c_{m,k}} \ket{m}\bra{k}.\]
    We observe that $\calB_{A}=\calT_C \circ \calB_{\lambda,Q}^{\calE}$. Indeed, for any $\rho\in L(\calH)$ 
    \begin{align*}
        \calT_C(\calB_{\lambda,Q}^{\calE}(\rho))&={C \p{\sum_{k,\ell=0}^{M-1} \tr(B_{k,l}\rho)\ket{k}\bra{l}} C^{\dag}}\\
        &={\sum_{m,h=0}^{M'-1} \sum_{k,l=0}^{M-1}\sum_{k',l'=0}^{M-1}\tr(B_{k,l}\rho)c_{m,k'} c_{h,l'}^*   \braket{k' | k}\braket{l | l'}   \ket{m}\bra{h}}\\
        &={\sum_{m,h=0}^{M'-1} \sum_{k,l=0}^{M-1}\tr(B_{k,l}\rho)c_{m,k} c_{h,l}^*   \ket{m}\bra{h}=\calB_A(\rho),}
    \end{align*}
    where in the last line we used \eqref{eq:repCalBa} and the fact that $(\ket{k})_{k\in[M]}$ is an orthonormal basis.
    Note that $\lambda=(\lambda_{k,l})_{k,l}=(\bra{c_0}{E_l^\dag E_k} \ket{c_0})_{k,l}$ is a Hermitian matrix, and therefore the superoperator $\calB_{\lambda,Q}^{\calE}$ is Hermitian preserving. Indeed, for any  Hermitian  $\rho$:
    \begin{align*}
        \calB_{\lambda,Q}^{\calE}(\rho)^\dag&= \Big(\sum_{k,l} \tr(B_{k,l}\rho)  \ket{k}\bra{l}\Big)^\dag =\sum_{k,l} \tr(B_{k,l}\rho)^*  \ket{l}\bra{k}\\
        &=\sum_{k,l} \tr(\rho^\dag B_{k,l}^\dag)  \ket{l}\bra{k}=\sum_{k,l} \tr(\rho B_{l,k})  \ket{l}\bra{k}=\calB_{\lambda,Q}^{\calE}(\rho),
    \end{align*}
    where the next-to-last inequality uses the relation
    \[B_{k,l}^\dag ={\p{P E_l^\dag E_k P -\lambda_{k,l}P}^\dag=P E_k^\dag E_l P -\lambda_{k,l}^*P=P E_k^\dag E_l P -\lambda_{l,k}P=B_{l,k}.}\]
Further, observe that $\calT_C$ is also Hermitian preserving since for any Hermitian $\rho$,
    \[\calT_{C}(\rho)^\dag= (C \rho C^\dag )^\dag=C\rho^\dag C^\dag =C\rho C^\dag =\calT_C(\rho).\]
We now use the equivalence of the diamond norm and the completely bounded norm for Hermitian-preserving maps, together with the submultiplicativity of the completely bounded norm with respect to composition (see Lemma~\ref{lem:cbDiamondProp})  to obtain
    \begin{align*}
        \norm{\calB_A}_{\diamond}=\norm{\calT_C\circ \calB}_{\diamond}\leq   \norm{\calT_C}_{\diamond}  \norm{\calB}_{\diamond}.
    \end{align*}
    By Lemmas~\ref{lem:cbDiamondProp} and \ref{lem:compoDiamond}, since $\calT_C$ is Hermitian preserving, we have  $\norm{\calT_C}_\diamond={\norm{C}_\infty^2}$. Combining this with the assumption $\norm{C}_{\infty}\leq 1$, we obtain the inequality 
    \[d(\Lambda_A,\Lambda_A+\calB_A)\leq \sqrt{\norm{\calB_A}_{\diamond}}\leq \sqrt{\norm{\calB}_{\diamond}}\leq \varepsilon.\]
    This completes the proof.
\end{proof}

{The sufficient AQEC conditions of Theorem~\ref{th:BOerrorset} allow us to define an approximate analog of the code distance. Although code distance is most commonly used in the context of limited-weight errors, the same idea naturally extends, both in the exact and approximate settings, to any indexed family of error sets. In this general formulation, the distance of a code measures the largest index $t$ for which the code corrects, exactly or approximately, all errors in $\calE_t$.}  The parameter $t$ should be understood as a model-dependent measure of error magnitude: for bounded-weight errors, it is the number of affected subsystems, while in other settings it may quantify, for example, the number of deletions or the photon loss count, or the total shift budget.

\graybox{
\begin{definition}[Approximate code distance]\label{def:dist}
Let $\calH$ be a Hilbert space, and let $\mathscr{E}=(\calE_t)_{t>0}$ be a family of error sets on $\calH$. For $\varepsilon\geq 0$ and a quantum code $Q\subseteq \calH$, the $\varepsilon$-distance of $Q$ with respect to $\mathscr{E}$, denoted by
$d_{\varepsilon}(Q,\mathscr{E}),$ $d_{\varepsilon}(Q,\mathscr{E})${distance of the code}
is the maximum $t$ such that $Q$ is $\varepsilon$-AQEC for $\calE_t$. If no such $t$ exists we say that $d_\varepsilon(Q,\mathscr{E})=0$. 
\end{definition}}
Several prior works \cite{bergamaschi2024approaching,ma2025haar} have discussed the concept of code distance for AQEC. However, their results are confined to the specific settings and constructions considered there, stopping short of providing a general, model-independent definition.

The concept of approximate distance extends naturally to the asymptotic setting. Namely, a sequence of codes is considered asymptotically good if its rate stays bounded away from zero while it approximately corrects an error set whose magnitude grows linearly with the system size {\em and} the approximation error vanishes asymptotically.
    \begin{definition}[Asymptotically good AQEC codes]\label{def:AsymGood}
    Assume that $(Q_N)_N$ is a sequence of quantum codes such that $Q_N\subset\calH_N$ and $\mathscr{E}_N$ is an error-set family acting on $\calH$. We say that $(Q_N)_N$ is an $\mathscr{E}_N$-asymptotically good AQEC code sequence if  there exists a sequence $\varepsilon_N\to 0$ such that both the code rate and relative approximate distance are separated from zero, i.e.,
    \[\liminf_{N\to \infty} \min\p{\frac{\log \dim(Q_N)}{\log \dim(\calH_N)}, \frac{d_{\varepsilon_{_N}}(Q_N,\mathscr{E}_N)}{N}}>0.\]
\end{definition}

\subsubsection{Structural properties and examples of the error-set model}\label{sec:StructrualProp}
Having established sufficient conditions for AQEC in the error-set model through the environment-leakage distance, we now further examine the meaning and scope of this model. In particular, we show that this formulation is not merely technically convenient, but also structurally well behaved: it enjoys natural properties and gives rise to a meaningful adversarial interpretation of several standard noise models in quantum error correction. We illustrate this by showing how the well-known QEC noise models fit naturally into the present framework. We start with a simple observation:

\begin{obs}\label{obs: independent}
    If the set $\calE$ is linearly independent, condition \eqref{eq:controlled} is independent of the Kraus representation; in other words, for any channel $\calN\in \mathscr{N}(\calE)$, any Kraus operator set for $\calN$ lies in $\Span(\calE)$ and its coefficient matrix $C$ satisfies \eqref{eq:controlled}. 
\end{obs}
\begin{proof}
    Assume that there exists a Kraus representation for $\calN$ formed of operators $\ppp{A_m}_m$ such that the corresponding coefficient matrix satisfies $\norm{C}_{\infty}\leq 1$. Let $\ppp{B_l}_l$ be another Kraus representation. Without loss of generality, assume that these two representations have the same size 
    (otherwise, pad the shorter one with zeros). As is well known, any two Kraus representations are related by a unitary transformation \cite[Theorem~8.2]{nielsen2010quantum}, i.e.,  there exists a unitary matrix $U$ such that for all $l$ 
    \begin{equation}
        B_l=\sum_m u_{l,m} A_m=\sum_m \sum_k u_{l,m} c_{m,k} E_k=\sum_k \p{\sum_m u_{l,m} c_{m,k} }   E_k.\label{eq:unique_rep}
    \end{equation}
Together with the linear independence of $\calE$, this implies that the unique representation of $\ppp{B_l}_l$ in the basis $\ppp{E_k}$ is given by the matrix $C'=UC$. Since multiplying by a unitary matrix does not change the spectral norm, we have 
    \[\norm{C'}_\infty=\norm{UC}_{\infty}=\norm{C}_{\infty}\leq 1.\]
    In particular, $\ppp{B_{l}}\subseteq\Span(\calE)$, and since $\calE$ is an independent set, representation \eqref{eq:unique_rep} is unique. Note that \eqref{eq:unique_rep} also shows that the coefficient matrix of $\ppp{B_l}_l$ is given by $UC$. Since multiplication by a unitary matrix preserves singular values, we have $\norm{UC}_\infty=\norm{C}_\infty \leq 1$.
\end{proof}

\begin{obs}\label{obs:Convexity} $\mathscr{N}(\calE)$ is a convex set for any $\calE$. 
In particular, if $\calN_1,\dots, \calN_\ell \in \mathscr{N}(\calE)$, then also any channel acting by randomly applying any of the channels $\calN_1,\dots, \calN_\ell $ with probabilities $p_1,\dots,p_\ell$ is in $\mathscr{N}(\calE)$.

\end{obs}
\begin{proof} Let $\calN_1\dots,\calN_{\ell}\in \mathscr{N}(\cal E)$, and let $p_1,\dots,p_{\ell}$ such that $\sum_ip_i=1$. Let  $\calN_i$ be given by Kraus operators $\{A^{(i)}_l\}_{l=1}^{L_i}$ with a coefficient matrix $Q_i$ satisfying $\norm{Q_i}_\infty\leq1$ (cf. Def.~\ref{def:AQECC}). Note that $\bigcup_{i=1}^m \{\sqrt{p_i} A^{(i)}_{l_i}\}_{l_i=1}^{L_i} $ is a Kraus  representation of $\calN$ with coefficient matrix
\[C=\begin{bmatrix}
\sqrt{p_1} Q_1\\
\vdots\\
\sqrt{p_m}  Q_{\ell}
\end{bmatrix}. \]
The spectral norm of $C$ satisfies
\begin{align*}
    \norm{C}_\infty^2&=\sup_{\norm{ x}_2= 1} \norm{Cx}_2=\sup_{\norm{x}_2=1}\sum_{i=1}^\ell p_i \norm{Q_ix }_2^2\\
& \leq \sum_{i=1}^\ell p_i\norm{Q_i }_\infty^2 \leq \sum_{i=1}^\ell p_i=1.  \qedhere
\end{align*}
    
\end{proof}

    \begin{prop}[Necessary condition and controlled channels for orthogonal unitary errors]\label{prop:orthUni}
    Let $\calE =\ppp{U_0,\dots ,U_{M-1}}$ be operators on a $d$-dimensional Hilbert space $\calH$,  orthogonal with respect to the Hilbert-Schmidt inner product, and satisfying $\tr(U_k^\dag U_k)=d$ for all $k$. 
    \begin{enumerate}
        \item The set of $\calE$-controlled channels, $\mathscr{N}(\calE)$, is the set of all channels which admit a Kraus representation contained in $\Span(\calE)$.
        \item If in addition $U_{k}$ are unitary for all $k$, any code  $Q\subseteq \calH$ that is an $\varepsilon$-AQEC code  for $\calE$, admits a positive matrix $\lambda$ such that $\|\calB_{\lambda,Q}^{\calE}\|_{\diamond}\leq {2\sqrt{2}M\varepsilon}$. In particular, $\zeta(\calE,Q)\leq {2\sqrt{2}M\varepsilon}$.
    \end{enumerate}
    \end{prop}

    \begin{proof}
         \underline{Proof of (1):} One direction of the inclusion is obvious, and it therefore remains to show that if $\calN$ is a channel with Kraus operators  $\ppp{A_m}_{m}\subseteq \Span(\calE)$ there exists a representation whose coefficient matrix $C=(c_{mk})$ satisfies $\norm{C}_{\infty}\leq 1$. 
        Indeed, note that since $U_1,\dots ,U_M$ are Hilbert-Schmidt orthogonal and have squared Hilbert--Schmidt norm $d$, the collection $(U_1/\sqrt{d},\dots ,U_M/\sqrt{d})$ 
        forms an orthonormal basis of $\Span\p{\calE}$ with respect to the Hilbert-Schmidt inner product.
        By the completeness condition, we have 
        \begin{align*}
            d&=\sum_m \tr(A_m^\dag A_m)=\sum_m\innerP{A_m,A_m}_{\s{HS}}=\sum_m\sum_{k,l}c_{mk}^*c_{ml} \overset{d \delta_{k,l}}{\overbrace{\innerP{U_k^\dag U_l}}}_{\s{HS}}=d \sum_{m,k}|c_{m,k}|^2.
        \end{align*}
This implies that 
        \begin{align*}
            \norm{C}_{\infty}\leq \norm{C}_{2}=\sqrt{\sum_{m,k}|c_{mk}|^2}= 1,
        \end{align*}
since the spectral norm is dominated by the Frobenius norm.

        \underline{Proof of (2):} Assume that $U_k$ is unitary for all $k$, and consider the quantum channel $\calN$ defined by the Kraus operators $\ppp{A_k:=U_k/{\sqrt{M}}}_{k=1}^M$. Since $U_k$ are unitary,
        \[
        \sum_{k=1}^{M} A_k^\dag A_k =\frac{1}{M}\sum_{k=1}^{M}  U_k^\dag U_k= I,
        \]
 implying that the completeness condition is satisfied and $\calN$ is indeed a quantum channel. 
 By Part (1), $\calN$ is also $\calE$-controlled. 
 We have assumed that $Q$ is an  $\varepsilon$-AQEC code for the channel $\calN$ and so by Theorem~\ref{th:BO}, there exists a density matrix $\lambda$ such that 
$d(\Lambda,\Lambda+\calB_{\lambda})\leq \varepsilon$ for $\calB_\lambda$ defined as in \eqref{eq:BOoperators}. Now using
 \eqref{eq:diamondBures}, we obtain that $\norm{\calB_{\lambda}}_{\diamond}\leq {2\sqrt{2}\varepsilon}$. Consider the matrix $\lambda_M=M
 \lambda$, and let $P$ be the projection on $Q$. Observe that
        \begin{align*}
            \calB_{\lambda_M,Q}^{\calE}(\rho)&=\sum_{k,l}\tr((P{U_l^\dag U_k}P-M\lambda_{k,l}P)\rho ) \ket{k}\bra{l}\\
            &= M  \sum_{k,l}\tr((P{A_l^\dag A_k}P-\lambda_{k,l}P)\rho )  \ket{k}\bra{l}=M \calB_{\lambda}(\rho),
        \end{align*}
so
        \[\norm{  \calB_{\lambda_M,Q}^{\calE}}_{\diamond}=\norm{  M\calB_{\lambda}}_{\diamond}\leq {2\sqrt{2} M\varepsilon}. \qedhere
        \]
    \end{proof}
   \begin{remark}\label{rem:CantImprove}
        An important question that immediately arises from Proposition~\ref{prop:orthUni} is whether the necessary condition, or some linear scaling of it, is also sufficient for error sets with this {special} structure. In other
        words, there might exist a constant $\alpha$ such that the condition
        \begin{equation}
            \zeta(\calE,Q)\leq \alpha M \varepsilon \label{eq:falseCond}
        \end{equation}
        is sufficient to guarantee that $Q$ is an $\varepsilon$-AQEC code for $\calE$. This may also seem intuitive, since this necessary condition is obtained by considering the channel which is, in some sense, the most chaotic among all controlled channels generated by this error set, the uniform-error channel. However, this line of reasoning turns out to be false. We show this by constructing a sequence of codes and unitary Hilbert-Schmidt-orthogonal error sets for which the condition
        \[
            \zeta(\calE,Q)\leq \alpha \varepsilon
        \]
        is necessary, which implies that for large $M$, condition \eqref{eq:falseCond} cannot be sufficient for $\varepsilon$-AQEC. Furthermore, it shows that there exist examples for which the sufficient AQEC condition of Theorem~\ref{th:BOerrorset} is also necessary, up to a linear constant and a square-root factor, the latter being inevitable due to the transition from the Bures distance to the diamond norm. On the practical side, this shows that, without assuming additional structure, the condition of Theorem~\ref{th:BOerrorset} cannot be improved at least as long as we rely on the environment-leakage distance. 
   \hfill $\triangleleft$ \end{remark}
 This argument is formalized in the following proposition, whose proof appears in Appendix~\ref{app:optimalCodeExamle}.
    \begin{prop}\label{prop:exampleOptimality}
        There exists a sequence of codes on $Q_N\subseteq \calH_2^{\otimes N}$ and a sequence of Pauli error sets $\calE_N$  with $|\calE_N|\to \infty$ such that if $Q_N$ is an $\varepsilon$-AQEC for $\calE_N$ then it must hold that
        \[\zeta(\calE_N,Q_N)\leq 4 \sqrt{2}\varepsilon.\]
    \end{prop}
    
\begin{example}[Limited weight errors]\label{ex:Paulis}
    Consider the set of HW operators of weight at most $t$ on $H_q^{\otimes N}$:
    \begin{equation}
        \calE_t=\ppp{W_{\ua,\ub}=\bigotimes_{i\in [N]}W_{a_i,b_i} ~:~ (\ua,\ub)\in [q^2]^{N}, \wt(\ua,\ub)\leq t},\label{eq:HWoperatorsT} 
    \end{equation}
    where $\wt(\ua,\ub)$ is the number of indices $i$ such that $(a_i,b_i)\neq(0,0)$. The resulting error-set family is $\mathscr{E}=(\calE_{t})_{t=0}^N$. It is well known that $\ppp{W_{\ua,\ub}}_{\ua,\ub}$ are Hilbert-Schmidt orthogonal and unitary. In particular, they satisfy the conditions of Proposition~\ref{prop:orthUni}, where $\mathscr{N}(\calE_t)$ includes all quantum channels introducing at most $t$ errors. In the context of approximate code distance (see Definition~\ref{def:dist}), we conclude that a code with distance $d_{\varepsilon}(Q,\mathscr{E})\geq t$ is $\varepsilon$-AQEC for any channel introducing at most $t$ qudit errors.
\hfill$\triangleleft$ \end{example} 
Using Theorem~\ref{th:BOerrorset} and Proposition~\ref{prop:orthUni} we derive both necessary and sufficient conditions for AQEC for $t$-limited error channels:
\begin{cor}\label{cor:generalErrorSuff}
    A code $Q\subseteq \calH_q^{\otimes N}$ is an $\varepsilon$-AQEC code  for any channel introducing $t$-limited errors if 
    \[\zeta(\calE_t, Q)\leq \frac{1}{2}\varepsilon^2.\]
    Conversely, if $Q\subseteq \calH_q^{\otimes N}$ is an $\varepsilon$-AQEC code  for any channel introducing $t$-limited errors then 
    \[\zeta(\calE_t, Q)\leq {2\sqrt{2}|\calE_t|\varepsilon=2\sqrt{2}\varepsilon  \sum_{i=0}^t \binom{N}{i}(q^2-1)^i} \]
\end{cor}

\begin{example}[Limited weight Majorana errors]\label{ex:Majoranas}
    Consider the fermionic Fock space $\calF_N$ of $N$ fermionic modes. For $t\in \Z_0$ define
     the limited-weight Majorana error set
    \begin{equation}
        \calE^{\s{Maj}}_t=\ppp{\gamma_I~:~I\subseteq[2N],~\abs{I}\leq t},\label{eq:MajoranaErrorsT}
    \end{equation}
    where $\gamma_I$ is the Majorana operator associated with the set $I$ defined in \eqref{eq:MajoranaI}. Under the Jordan-Wigner representation (see \eqref{eq:JordanWignerMajorana}), each $\gamma_I$ is mapped to a Pauli string up to a phase. Hence the operators $\ppp{\gamma_I}_{I\subseteq[2N]}$ are Hilbert-Schmidt orthogonal and unitary. In particular, the set $\calE^{\s{Maj}}_t$ satisfies the conditions of Proposition~\ref{prop:orthUni}, with $d=2^N$ and
    \begin{equation}
        \abs{\calE^{\s{Maj}}_t}=\sum_{\ell=0}^{t}\binom{2N}{\ell}.\label{eq:MajoranaErrorsCount}
    \end{equation}
    Therefore, $\mathscr{N}(\calE^{\s{Maj}}_t)$ is exactly the set of all quantum channels admitting a Kraus representation in $\Span(\calE^{\s{Maj}}_t)$, namely the channels generated by Majorana errors of support size at most $t$. In the context of approximate code distance, we conclude that a code with distance $d_{\varepsilon}(Q,\mathscr{E}^{\s{Maj}})\geq t$ is an $\varepsilon$-AQEC code for any channel introducing at most $t$ Majorana errors. Moreover, using Theorem~\ref{th:BOerrorset} and Proposition~\ref{prop:orthUni}, we obtain the following sufficient and necessary conditions in terms of $\zeta$. A code $Q\subseteq \calF_N$ is an $\varepsilon$-AQEC code for any channel introducing at most $t$ Majorana errors if
    \[ \zeta(\calE^{\s{Maj}}_t,Q)\leq \frac{1}{2}\varepsilon^2.\]
    Conversely, if $Q\subseteq \calF_N$ is an $\varepsilon$-AQEC code for any channel introducing at most $t$ Majorana errors, then
    \[\zeta(\calE^{\s{Maj}}_t,Q)\leq {2\sqrt{2}\abs{\calE^{\s{Maj}}_t}\varepsilon
        =2\sqrt{2}\varepsilon\sum_{\ell=0}^{t}\binom{2N}{\ell}}.\]
\hfill$\triangleleft$ \end{example}

\begin{remark}[Normalization of the error-set operators]\label{rem:scaling}
{The definition of $\calE$-controlled channels depends on the normalization of the operators in $\calE$. Indeed, rescaling the elements of $\calE$ changes the coefficient matrix associated with a Kraus representation and may therefore change the resulting family $\mathscr{N}(\calE)$. Thus, the choice of scaling forms a part of the error-set model, and the appropriate normalization may depend on the family of channels one wishes to control.}

{Example~\ref{ex:Paulis} shows that, for HW errors, the natural unitary normalization is canonical in this sense. By Proposition~\ref{prop:orthUni}, with this normalization, the $\calE$-controlled channels are exactly the channels admitting a Kraus representation in $\Span(\calE)$. On the other hand, if one rescales even one HW operator by a factor smaller than $1$, then the isometric channel $\calN(\rho)=W\rho W^\dag$
associated with that operator will have a compensating coefficient larger than one and will thus no longer be included in the corresponding controlled family. Thus, for this error model, the unitary scaling is essentially tight.}

{In general, however, the correct scaling should be chosen according to the intended channel family. In Example~\ref{ex:AD}, we choose a normalization that captures all truncated amplitude damping channels simultaneously. We then justify this choice by proving a necessary condition showing that the resulting sufficient condition is tight up to a factor of order $1/(t+1)$, where $t$ is the number of amplitude damping errors.}
\hfill$\triangleleft$ \end{remark}
\begin{example}[Photon loss errors] \label{ex:AD} Consider the constant excitation Fock state space $\calH_{q,N}$ and the set of at most $t$-photon-loss errors, defined as: 
\begin{equation}
    \calE_{\leq t}^{\s{AD}}= \ppp{\Tilde{A}_{\ur} ~:~ \ur \in \calS_{q,r}, r\leq t}.\label{eq:ADerrorSet}
\end{equation}
where $\Tilde{A}_{\ur}$ is a normalized version of $A_{\ur}$ defined in \eqref{eq:photonLossErr}:
\begin{equation}
    \Tilde{A}_{\ur}=\frac{1}{\sqrt{\binom{N}{r}}}   \Tilde{A}_{r_1}\otimes \Tilde{A}_{r_2}\otimes  \cdots \otimes \Tilde{A}_{r_q}, \quad \Tilde{A}_{r_i}\ket{n}=\begin{cases}
  \sqrt{\binom{n}{r_i}}   \ket{n-r_i} & r_i\leq n,\\
    0 & r_i > n.\end{cases}\label{eq:ADnormalizedOP}
\end{equation}
    Consider the truncated bosonic amplitude damping channel  introduced in \cite{elimelech2026asymptotically}, given by 
\begin{equation}\label{eq: AD channel}
    \calN_{\gamma,\leq t}(\rho)=\sum_{\substack{\ur\in \calS_{q,r}\\ r\leq t}} \frac{1}{p_{N,t,\gamma}}  A_{\ur} \rho A_{\ur}^{\dag}, 
\end{equation}
    where $ p_{N,t,\gamma}=\P\pp{\xi_\gamma\leq t}$ and $\xi_\gamma\sim\s{Binom}(N,\gamma)$ is a binomial random variable. The channel $\calN_{\gamma,\leq t}$ is given by the Kraus representation 
    \(\ppp{({p_{N,t,\gamma}})^{-\frac12}{A_{\ur}} ~:~ \ur \in \calS_{q,r}, r\leq t}.\)
    For every $\gamma\in(0,1)$, a simple calculation shows that $\calN_{\gamma,\leq t}\in \mathscr{N}(\calE_{\leq t}^{\s{AD}})$ as for any $\ur\in \calS_{q,r}$, $r<t$ we have 
    \begin{align*}
        \frac{1}{\sqrt{p_{N,t,\gamma}}}A_{\ur}&= \sqrt{\frac{\binom{N}{r}\gamma^r(1-\gamma)^{N-r}}{p_{N,t,\gamma}}} \Tilde{A}_{\ur}
        =\P(\xi=r\mid\xi\le t)^\frac12 \Tilde{A}_{\ur}
    \end{align*}
Since the multipliers are at most 1, the corresponding (diagonal) coefficient matrix $Q_\gamma$ satisfies $\norm{Q_{\gamma}}_\infty \leq 1$. This shows that $\calN_{\gamma,\leq t}\in \mathscr{N}(\calE_{\leq t}^{\s{AD}})$. 
\hfill$\triangleleft$ \end{example}
\begin{remark}[Asymptotic performance of asymptotically good Fock state codes]
    In a previous work, we established a connection between the performance of a constant excitation Fock state code used on the AD channel $\calN_\gamma$ and the same code used on the truncated channel $\calN_{\gamma,t}$; see Theorem 5 in \cite{elimelech2026asymptotically}. 
This theorem implies that if a Fock state code $Q_N$ with constant excitation $N$ is $\varepsilon$-AQEC for $\calN_{\gamma,\leq t}$, then it is an $\varepsilon'$-AQEC code for $\calN_{\gamma}$ with $\varepsilon'=\sqrt{1-p_{N,t,\gamma}(1-\varepsilon^2)}$.  In particular, if a sequence of codes $Q_{N}\subseteq \calH_{q_N,N}$ is $\mathscr{E}^{\s{AD}}_N$-asymptotically good (see Definition~\ref{def:AsymGood}), with $\mathscr{E}_{N}^{\s{AD}}=(\calE_{\leq t}^{\s{AD}})_t$ formed of photon loss errors on $q_N$ modes, and
    \[
    \liminf_{N\to \infty}\frac{d_{\varepsilon_N}(Q_N,\mathscr{E}_N^{\s{AD}})}{N}=\delta >0, \quad \varepsilon_N\to 0,
    \]
then for any $\gamma<\delta$, $Q_N$ is an $\varepsilon_N'$-AQEC for $\calN_\gamma$ with $\varepsilon'_N\to 0$. 
Here we use the fact that if a code state with total excitation $N$ is submitted to the amplitude damping channel, then
with high probability, the channel incurs $\approx \gamma N$ photon losses,  which means that $p_{N,t,\gamma}=1+o(1)$ for $t>N(\gamma+\alpha)$, for any $\alpha>0$. This conclusion can be informally summarized in a very intuitive way: 
        \textit{A Fock state code that can approximately correct $\delta N$ photon losses protects against the noise on the AD channel as long as the loss parameter $\gamma$ is bounded above by $\delta$}.
\hfill$\triangleleft$ \end{remark}
\begin{prop}[Necessary condition for AD noise]\label{prop:NecFock}
    A code $Q\in \calH_{q,N}$ with $t\le d^{\s{AD}}_\varepsilon(Q)\leq N$ must satisfy: 
    \[\zeta(\calE_{\leq t}^{\s{AD}},Q)\leq {2\sqrt{2}(t+1)\varepsilon}.\]
\end{prop}
\begin{proof}
We begin with a simple observation: for any $\ket{\psi}\in \calH_{q,N}$ we have 
\begin{equation}
    \sum_{r\leq t}\sum_{\ur\in\calS_{q,r}} \bra{\psi}\Tilde{A}_{\ur}^\dag \Tilde{A}_{\ur}\ket{\psi}=t+1\label{eq:tracePresFock}
\end{equation}
Indeed, for basis states $\ket{\un},\ket{\un'}\in \calH_{q,N}$ and for $\ur\in \calS_{q,r}$ we have 
\[\bra{\un'}\Tilde{A}_{\ur}^\dag \Tilde{A}_{\ur}\ket{\un}=\frac{1}{\binom{N}{r}}\sqrt{\prod_{i=0}^{q-1}\binom{n_i}{r_i}\prod_{i=0}^{q-1}\binom{n'_i}{r_i}}\braket{\un'-\ur|\un -\ur}=\frac{1}{\binom{N}{r}}\prod_{i=0}^{q-1}\binom{n_i}{r_i}  \delta_{\un,\un'}.\]
Using this identity for a general quantum state $\ket{\psi}=\sum_{\un \in \calS_{q,N}}\alpha_{\un}   \ket{\un }$, we compute 
\begin{align*}
    \sum_{r\leq t}\sum_{\ur\in\calS_{q,r}} \bra{\psi}\Tilde{A}_{\ur}^\dag \Tilde{A}_{\ur}\ket{\psi}&=\sum_{r\leq t}\sum_{\ur\in\calS_{q,r}} \sum_{\un,\un'\in \calS_{q,N}}\alpha_{\un}\alpha_{\un'}^*\bra{\un'}\Tilde{A}_{\ur}^\dag \Tilde{A}_{\ur}\ket{\un}&\\
    &=\sum_{r\leq t}\sum_{\ur\in\calS_{q,r}} \sum_{\un,\un'\in \calS_{q,N}}\alpha_{\un}\alpha_{\un'}^*\frac{1}{\binom{N}{r}}\prod_{i=0}^{q-1}\binom{n_i}{r_i}  \delta_{\un,\un'}\\
&= \sum_{r\leq t}\sum_{\un\in \calS_{q,N}}|\alpha_{\un}|^2\frac{1}{\binom{N}{r}}\sum_{\ur\in\calS_{q,r}} \prod_{i=0}^{q-1}\binom{n_i}{r_i}\\
&= \sum_{r\leq t}\sum_{\un\in \calS_{q,N}}|\alpha_{\un}|^2\\
&=t+1.
\end{align*}
Thus, operation $\calN_{\leq t}^{\s{U}}$ defined by the error operators 
\[ \calE_{\leq t}^{\s{U}}=\ppp{\Tilde{A}_{\ur}/\sqrt{t+1} ~:~\ur\in \calS_{q,r}, r\leq t}=\frac{1}{\sqrt{t+1}}  \calE_{\leq t}^{\s{AD}}.\]is CPTP (complete positivity is immediate, and the trace-preserving property follows from \eqref{eq:tracePresFock}). Also note that $\calN_{\leq t}^{\s{U}}\in \mathscr{N}(\calE^{\s{AD}})$ as the coefficient matrix for the operators in $\calE_{\leq t}^{\s{U}}$ (with respect to the operators in $\calE_{\leq t}^{\s{AD}}$) is a diagonal matrix with entries $1/\sqrt{t+1}\leq 1$. In particular, if a code $Q$ is $\varepsilon$-AQEC for $\calE_{\leq t}^{\s{AD}}$, then it is an $\varepsilon$-AQEC code  for the channel $\calN_{\leq t}^{\s{U}}$. By Theorem~\ref{th:BO} and \eqref{eq:diamondBures}, there exists a matrix $\lambda$ such that the  operator $\calB_{\lambda }$ defined in Theorem~\ref{th:BO} satisfies $\norm{\calB_{\lambda}}_{\diamond}\leq {2\sqrt{2}\varepsilon}$.
 Consider the matrix $\lambda'=\lambda  (t+1)$ and observe that the corresponding error-set B\'eny-Oreshkov operator satisfies $\calB_{\lambda',Q}^{\calE_{\leq t}^{\s{AD}}}=(t+1)  \calB_{\lambda}$. In particular, 
 \[\zeta(\calE_{\leq t}^{\s{AD}},Q)\leq \|\calB_{\lambda',Q}^{\calE_{\leq t}^{\s{AD}}}\|_{\diamond}=(t+1)\norm{\calB_{\lambda}}_{\diamond}\leq {2\sqrt{2}(t+1)\varepsilon}.  \qedhere
 \]
\end{proof}

  \subsection{The AQEC error-set model extends to the channel fidelity criterion}\label{subsec:ChannelFidelityErrorSet}

In Section~\ref{sec:ErrorSetModel}, we proved that the error-set formalism gives a universal model for approximate correction: once an error set $\calE$ is fixed, the performance of a code against all $\calE$-controlled channels, in the sense of Definition~\ref{def:AQECC}, can be controlled by conditions depending only on the action of $\calE$ on the code. The performance criterion used there was the worst-case entanglement fidelity, or equivalently, the corresponding Bures distance. In this section, we show that the same controlled-channel viewpoint is not tied to this particular worst-case criterion. For the channel-fidelity criterion studied in \cite{zheng2024near}, and closely related to earlier optimization-based approaches to AQEC \cite{audenaert2002optimizing,noh2018quantum}, the family of $\calE$-controlled channels again admits conditions that depend only on the error-set QEC matrix. This provides further evidence that $\calE$-controlled channels are the natural channel family generated by an error set rather than an artifact of the worst-case formulation.

Alongside worst-case entanglement fidelity, a widely used weaker criterion is channel fidelity, also called process fidelity \cite{audenaert2002optimizing,kosut2009quantum,noh2018quantum,zheng2024near}. For quantum channels $\calN,\calM:L(\calH)\to L(\calH')$, let
$\ket{\Phi_{\calH}}$ be a maximally entangled state on $\calH\otimes \calH$. 
\begin{equation}
    \calF_{\mathrm{ch}}\p{\calN,\calM} :=  \calF\p{(I_{L(\calH)}\otimes\calN)(\ket{\Phi_{\calH}}\bra{\Phi_{\calH}}),(I_{L(\calH)}\otimes\calM)(\ket{\Phi_{\calH}}\bra{\Phi_{\calH}})}^2.\label{eq:ChaFidel}
\end{equation}
Thus, $\calF_{\mathrm{ch}}$ tracks the entanglement fidelity only for the maximally mixed logical input rather than minimizing over all input states. It is therefore weaker than the worst-case entanglement-fidelity criterion, although it is a widely accepted and meaningful measure of channel performance. In particular, when $\calN:L(\calH)\to L(\calH)$ is compared with the identity channel $\calI$, the channel fidelity is equivalent, up to a standard normalization factor, to the Haar-average input-output fidelity (see \cite[Sec. 6.4]{khatri2020principles}):
\begin{equation}
    \calF_{\s{avg}}\p{\calN,\calI}
 :=  \intop_{\ket{\phi}}{\rm d\phi}\bra{\phi}\calN(\ket{\phi}\bra{\phi})\ket{\phi}
=\frac{d \calF_{\mathrm{ch}}\p{\calN,\calI}+1}{d+1},\label{eq:averagcase}
\end{equation}
where ${\rm d\phi}$ is the Haar measure on $\calH$ and $d=\dim(\calH)$.
Equivalently,
\[ 1-\calF_{\mathrm{ch}}\p{\calN,\calI}
=\frac{d+1}{d}\p{1-\calF_{\s{avg}}\p{\calN,\calI}}. \]
For this reason, we refer to the resulting correction criterion as the average-case, or $\s{Av}$, AQEC criterion. Following the Bures-distance normalization used for worst-case AQEC in \cite{beny2010general}, we measure the channel-fidelity error by the square root of the channel infidelity,
\begin{equation}
    \label{eq:chBuresDist}
    d_{\mathrm{ch}}\p{\calN,\calM} :=  \sqrt{1-\calF_{\mathrm{ch}}\p{\calN,\calM}}.
\end{equation}
The square root in this expression is introduced for compatibility with the Bures-distance viewpoint, and we will use it
for comparing the resulting conditions with those of Section~\ref{sec:ErrorSetModel}.

\begin{definition}[AQEC for quantum channels via channel fidelity]\label{def:AQECchennelFidelity}
A quantum code $Q\subseteq\calH$ is said to be an $\s{Av}$-$\varepsilon$-AQEC code for a channel $\calN:L(\calH)\to L(\calH')$ if there exists a CPTP decoding operation $\calD:L(\calH')\to L(\calH)$ such that
\begin{equation}
d_{\mathrm{ch}}\p{\calD\circ\calN|_{L(Q)},I_{L(Q)}}\leq\varepsilon.
\label{eq:ChannelBures}
\end{equation}

\end{definition}

We now lift Definition~\ref{def:AQECchennelFidelity} from a fixed channel to the error-set model, in parallel to the worst-case entanglement-fidelity formulation of Section~\ref{sec:ErrorSetModel}. Given an error set $\calE$, we use the family $\mathscr{N}(\calE)$ of $\calE$-controlled channels from Definition~\ref{def:AQECC} as the corresponding channel family. Thus, average-case AQEC for an error set means average-case AQEC uniformly over all channels controlled by that error set. Our main aim in this subsection is to show that this error-set definition is more than a mere formality: the fixed-channel channel-fidelity conditions can be converted into conditions stated directly on the error set $\calE$, just as in the worst-case formulation.

\begin{definition}[Average-case AQEC for error sets]\label{def:AvAQECErrorSet}
Let $\calE=\ppp{E_k}_k$ be a finite set of error operators $E_k:\calH\to\calH'$, and let $Q\subseteq\calH$ be a quantum code. For $\varepsilon\geq0$, we say that $Q$ is an $\s{Av}$-$\varepsilon$-AQEC code for the error set $\calE$ if $Q$ is an $\s{Av}$-$\varepsilon$-AQEC code, in the sense of Definition~\ref{def:AQECchennelFidelity}, for every $\calE$-controlled channel $\calN\in\mathscr{N}(\calE)$.
\end{definition}

A standard approach to AQEC is guided by the Petz recovery map, also known in the QEC literature as the transpose channel, which is exact whenever the KL conditions hold, and remains near-optimal for approximate correction \cite{barnum2002reversing,ng2010simple,zheng2024near}. Under the channel-fidelity criterion above, for a code $Q$ and a channel $\calN$, the corresponding transpose-channel decoder $\calD_{\s{TC}}$ satisfies
\begin{equation}
\frac{1}{\sqrt{2}}d_{\mathrm{ch}}\p{\calD_{\s{TC}}\circ \calN|_{L(Q)},I_{L(Q)}}\leq \varepsilon_{\mathrm{ch}}^{\mathrm{opt}}(Q,\calN)\leq d_{\mathrm{ch}}\p{\calD_{\s{TC}}\circ \calN|_{L(Q)},I_{L(Q)}}.
\label{eq:PetzEps}
\end{equation}
Here $\varepsilon_{\mathrm{ch}}^{\mathrm{opt}}(Q,\calN)$ denotes the value obtained by minimizing the left-hand side of \eqref{eq:ChannelBures} over all CPTP decoders $\calD$   (we will omit the channel $\calN$ and the code $Q$ from the notation whenever they are understood from the context, and simply write $\varepsilon_{\mathrm{ch}}^{\mathrm{opt}}$). Thus, up to the factor $\sqrt{2}$, the optimal channel-fidelity error is captured by the transpose channel.

Zheng et al.~\cite{zheng2024near} analyzed the channel fidelity obtained from the transpose-channel decoder and gave an explicit expression in terms of the QEC matrix. Let $Q$ be a $K$-dimensional code with orthonormal basis $\ppp{\ket{c_i}}_{i\in[K]}$ and let $\calN$ be given by Kraus operators $\ppp{E_k}_{k\in[M]}$. The associated QEC matrix is a $KM\times KM$ positive semidefinite matrix $A_{\s{QEC}}\in L\p{\C^K\otimes\C^M}$
whose entries are \[\p{A_{\s{QEC}}}_{ik,jl} :=  \bra{c_i}E_k^\dag E_l\ket{c_j},\quad i,j\in[K],\quad k,l\in[M].\]
It is shown in \cite{zheng2024near} that
\begin{equation}
\calF_{\mathrm{ch}}\p{\calD_{\s{TC}}\circ\calN|_{L(Q)},I_{L(Q)}}=\frac{1}{K^2}\norm{\tr_{K}\p{\sqrt{A_{\s{QEC}}}}}_2^2,
\label{eq:ZhengCalc}
\end{equation}
where $\norm{ }_2$ is the Frobenius norm (see Definition~\ref{def:statenorms}) and $\tr_{K}$ denotes the partial trace over the code-basis tensor factor $\C^{K}$ when $A_{\s{QEC}}$ is considered as an operator on ${\C^K\otimes\C^M}$.

Combining \eqref{eq:ZhengCalc} with the near-optimality of the transpose channel in \eqref{eq:PetzEps}, the authors of \cite{zheng2024near} obtained the following explicit two-sided estimate for the optimal channel-fidelity error:
\begin{prop}{\rm (\cite[Theorem~1]{zheng2024near})} For a code $Q$ with basis $(\ket{c_i})_i$ and a channel characterized by a Kraus set $\calE=\ppp{E_k}_k$:
\begin{equation}
\sqrt{\frac{1}{2}\Big({1-\frac{1}{K^2}\|{\tr_{K}({\sqrt{A_{\s{QEC}}})}}\|_2^2}\Big)}
\leq \varepsilon_{\mathrm{ch}}^{\mathrm{opt}}
\leq \sqrt{1-\frac{1}{K^2}\|{\tr_{K}({\sqrt{A_{\s{QEC}}})}}\|_2^2}.
\label{eq:ZhengBoundBures}
\end{equation}
\end{prop}

Expanding on this result, we next present a simple yet powerful geometric interpretation of the expression in \eqref{eq:ZhengBoundBures}. First, note that for a code $Q$ with basis $\ppp{\ket{c_i}}_{i\in[K]}$ and a general error set $\calE=\ppp{E_k}_{k\in[M]}$ (not necessarily associated with a channel), the KL conditions imply that $Q$ is an exact QEC code for $\calE$ if and only if $\bra{c_i}E_k^\dag E_l\ket{c_j}=\delta_{ij}\lambda_{kl}
$ for some $M\times M$ matrix $\lambda$. Equivalently, the QEC matrix $A_{\s{QEC}}$ has the form $A_{\s{QEC}}=I_K\otimes\lambda$ or simply lies in the Knill--Laflamme space:
\begin{equation}
    \calH_{\s{KL}} :=  \ppp{I_K\otimes \lambda~:~\lambda\in L(\C^M)}\subseteq L(\C^K\otimes \C^M).
\label{eq:KLspaceS}
\end{equation}
Our observation is that the expression in \eqref{eq:ZhengBoundBures} is the normalized Hellinger distance of the QEC matrix to this KL space. 
Before stating this formally, let us give some definitions. For $A\succeq0$, let 
\[ D_{\s{H}}\p{A,\calH_{\s{KL}}}
 := 
\min_{\substack{B \in \calH_{\s{KL}} \\ B\succeq0}}D_{\s{H}}\p{A,B} \]
where for  $A,B\succeq0$, the quantum Hellinger distance is
\[ D_{\s{H}}\p{A,B} :=  \|{\sqrt A-\sqrt B}\|_2. \]

\begin{definition}\label{def: KL Hellinger distance}
The {\em Knill--Laflamme Hellinger distance} between a $K$-dimensional quantum code $Q$ and an error set $\calE$ is defined as
  $$
   \zeta_{\s{H}}(\calE,Q):= \frac 1{\sqrt{K}} D_{\s{H}}\p{A_{\s{QEC}},\calH_{\s{KL}}},
   $$
where $A_{\s{QEC}}$ is an error-correction matrix for $Q$.
\end{definition}
First note that $\zeta_{\s{H}}$ is well defined in the sense that it is independent of the basis chosen to represent the matrix $A_{\s{QEC}}$ which is defined up to a conjugation by a unitary matrix, and the Hellinger distance to the KL space is invariant under such conjugations; see Lemma~\ref{lem:HellingerContaction}.

In the next lemma, proved in Appendix~\ref{app:ProofLemGeometric}, we express the quantity in Eq.~\eqref{eq:ZhengBoundBures} in terms of $\zeta_{\s{H}}$. 

\begin{lemma}\label{lem:orthogonalProjectionKLspace}
Let $Q$ be a quantum code with basis $\ppp{\ket{c_i}}_{i\in[K]}$, let $\calE=\ppp{E_k}_{k\in[M]}$ be an error set, and let $A_{\s{QEC}}$ be a corresponding QEC matrix. Then 
\[ \frac{1}{K}\tr(A_{\s{QEC}})-\frac{1}{K^2}\norm{\tr_K\p{\sqrt{A_{\s{QEC}}}}}_2^2 =
\zeta_{\s{H}}^2(\calE,Q) \]
\end{lemma}

Combining the geometric observation of Lemma~\ref{lem:orthogonalProjectionKLspace} with the channel-fidelity characterization of the transpose-channel decoder in \eqref{eq:ZhengBoundBures}, we obtain a near-optimal geometric estimate for average-case AQEC. Namely, when $\calE$ is a Kraus set of a channel $\calN$, the optimal channel-fidelity error is controlled, up to the same universal factor $\sqrt{2}$, by the Hellinger distance of the QEC matrix $A_{\s{QEC}}$ from the KL space $\calH_{\s{KL}}$.

\begin{cor}[Geometric channel-fidelity AQEC bound]\label{cor:HellingerChannelFidelity}
Let $Q\subseteq\calH$ be a $K$-dimensional quantum code with orthonormal basis $\ppp{\ket{c_i}}_{i\in[K]}$, and let $\calN:L(\calH)\to L(\calH')$ be a quantum channel with Kraus set $\calE_\calN$. Then
   \begin{equation}\label{eq: two sided}
\frac 1{\sqrt 2}\zeta_{\s{H}}(\calE_{\calN},Q)
\leq \varepsilon_{\mathrm{ch}}^{\mathrm{opt}}(Q,\calN)\leq \zeta_{\s{H}}(\calE_{\calN},Q)
   \end{equation}
\end{cor}

Indeed, since $\calN$ is trace-preserving,
\[
\frac{1}{K}\tr(A_{\s{QEC}})=\frac{1}{K}\sum_{i\in[K]}\sum_{k\in[M]}\bra{c_i}E_k^\dag E_k\ket{c_i}=\frac{1}{K}\sum_{i\in[K]}\tr\p{\calN(\ket{c_i}\bra{c_i})}=1.
\]
Substituting this identity into Lemma~\ref{lem:orthogonalProjectionKLspace} gives
\[
1-\frac{1}{K^2}\norm{\tr_K\p{\sqrt{A_{\s{QEC}}}}}_2^2=\zeta_{\s{H}}^2(\calE_{\calN},Q).
\]
Together with \eqref{eq:ZhengBoundBures}, this gives the two-sided estimate in \eqref{eq: two sided}.

We now give the average-case analog of the sufficient error-set conditions of Theorem~\ref{th:BOerrorset}.  When the error set $\calE$ itself is a Kraus set of a channel, the statement reduces to the near-optimal transpose-channel sufficient condition of \cite{zheng2024near}, written in the geometric form of Corollary~\ref{cor:HellingerChannelFidelity}.

\begin{theorem}[Sufficient $\s{Av}$-AQEC condition for error sets]\label{th:AvAQECErrorSetSuff}
Let $Q\subseteq\calH$ be a $K$-dimensional quantum code with orthonormal basis $\ppp{\ket{c_i}}_{i\in[K]}$, and let $\calE=\ppp{E_k}_{k\in[M]}$ be a finite error set. If
\[ \sqrt{\frac{1}{K}\tr(A_{\s{QEC}})-\frac{1}{K^2}\norm{\tr_K\p{\sqrt{A_{\s{QEC}}}}}_2^2}
=\zeta_{\s{H}}(\calE,Q)\leq\varepsilon, \]
then $Q$ is an $\s{Av}$-$\varepsilon$-AQEC code for the error set $\calE$, in the sense of Definition~\ref{def:AvAQECErrorSet}.
\end{theorem}
\begin{proof}
    The proof follows by combining Lemma~\ref{lem:orthogonalProjectionKLspace} and Corollary~\ref{cor:HellingerChannelFidelity} with a contraction argument for the Hellinger distance under conjugation by a contraction. Let $\calE=\ppp{E_k}_{k\in [M]}$ and let $\calN\in \mathscr{N}(\calE)$ be an $\calE$-controlled channel with Kraus operators $\calE_{\calN}=\ppp{F_l}_{l\in [M']}$ from $\Span(\calE)$ and linear coefficient matrix $C$ with $\norm{C}_\infty\leq 1$. Denote by $A_{\s{QEC}}(\calE)$ and $A_{\s{QEC}}(\calE_{\calN})$ the QEC matrices of the basis $\ppp{\ket{c_i}}_{i\in [K]}$ with respect to $\calE$ and $\calE_{\calN}$, respectively. From the definition of the QEC matrix, we immediately have
    \begin{equation} A_{\s{QEC}}(\calE_{\calN})=(I_K\otimes \Bar{C})A_{\s{QEC}}(\calE)(I_K\otimes \Bar{C}^\dag),\label{eq:chatCorrected} \end{equation}
  where $\Bar{C}$ is the matrix obtained by taking the entrywise complex adjoint on $C$. Denote $R=I_K\otimes \Bar{C}$ and note that 
    \begin{equation}
        \norm{R}_{\infty}=\norm{I_K}_\infty \norm{\Bar{C}}_{\infty}=\norm{C}_{\infty}\leq 1,\label{eq:conjugation}
    \end{equation}
    as the spectral norm is multiplicative with respect to tensor products. 
    
    Assume that $\s{D}_{\s{H}}(A_{\s{QEC}}(\calE),\calH_{\s{KL}})/\sqrt{K}\leq \varepsilon$. By Corollary~\ref{cor:HellingerChannelFidelity} it is sufficient to show that 
    $$
    \s{D}_{\s{H}}(A_{\s{QEC}}(\calE_{\calN}),\calH_{\s{KL}})/\sqrt{K}\leq \varepsilon,
    $$
    or that
    \[\s{D}_{\s{H}}(A_{\s{QEC}}(\calE_{\calN}),\calH_{\s{KL}})\leq \s{D}_{\s{H}}(A_{\s{QEC}}(\calE),\calH_{\s{KL}}).\] 
    Indeed, we have
    \begin{align}
        \s{D}_{\s{H}}(A_{\s{QEC}}(\calE_{\calN}),\calH_{\s{KL}})&\nonumber=\min_{\substack{B\in\calH_{\s{KL}}\\ B\succeq 0}} \s{D}_{\s{H}}(A_{\s{QEC}}(\calE_{\calN}),B)\\
        &=\min_{\substack{B\in\calH_{\s{KL}}\\ B\succeq 0}} \s{D}_{\s{H}}(RA_{\s{QEC}}(\calE)R^{\dag},B)\label{eq:justRef}\\
        &\leq\min_{\substack{B\in\calH_{\s{KL}}\\ B\succeq 0}} \s{D}_{\s{H}}(RA_{\s{QEC}}(\calE)R^{\dag},R B R^\dag)\label{eq:subsetOpt}\\
        & \leq \min_{\substack{B\in\calH_{\s{KL}}\\ B\succeq 0}} \s{D}_{\s{H}}(A_{\s{QEC}}(\calE), B )\label{eq:contractuse}\\
        &\nonumber=\s{D}_{\s{H}}(A_{\s{QEC}}(\calE),\calH_{\s{KL}}),
    \end{align}
 where \eqref{eq:justRef} follows from \eqref{eq:chatCorrected}. The inequality in \eqref{eq:subsetOpt} follows since conjugation by $R=I_{K}\otimes \Bar{C}$ sends positive operators in $\calH_{\s{KL}}$ to  positive operators in $\calH_{\s{KL}}$, and therefore expression \eqref{eq:subsetOpt} equals the minimum on the subset $\ppp{RBR^{\dag}~:~ B\in \calH_{\s{KL}},~ B\succeq 0}\subseteq \ppp{B\in \calH_{\s{KL}}   ~:~ B\succeq 0}$. Finally, \eqref{eq:contractuse} follows since the Hellinger distance contracts under the map $A\to RA R^\dag$ for a contraction $R$ (see Lemma~\ref{lem:HellingerContaction}).
\end{proof}

We next turn to the converse direction. The sufficient condition above shows that small Hellinger distance from the KL space guarantees $\s{Av}$-AQEC for the whole $\calE$-controlled family. We now show that, for unitary error sets, this distance is also controlled by the optimal average-case AQEC error. The argument follows the same general strategy as in the worst-case setting, with the channel-fidelity estimate replacing the worst-case Bures estimate. Thus, the following proposition is an $\s{Av}$-AQEC analog of Proposition~\ref{prop:orthUni} and Corollary~\ref{cor:generalErrorSuff}.

\begin{prop}[Necessary conditions for $\s{Av}$-AQEC]\label{prop:AvAQECNecessary}
Let $\calE=\ppp{U_0,\dots,U_{M-1}}$ be a finite set of unitary operators on a Hilbert space $\calH$, and let $Q\subseteq\calH$ be a $K$-dimensional code with orthonormal basis $\ppp{\ket{c_i}}_{i\in[K]}$. Suppose that $Q$ is an $\s{Av}$-$\varepsilon$-AQEC code for the error set $\calE$ in the sense of Definition~\ref{def:AvAQECErrorSet}. Then the corresponding QEC matrix satisfies
\[
\frac{1}{\sqrt{K}}D_{\s{H}}\p{A_{\s{QEC}},\calH_{\s{KL}}}\leq \sqrt{2M}\varepsilon.
\]
In particular, if $Q\subseteq\calH_q^{\otimes N}$ is an $\s{Av}$-$\varepsilon$-AQEC code for the family of channels generated by $t$-limited errors, then
\[
\frac{1}{\sqrt{K}}D_{\s{H}}\p{A_{\s{QEC}},\calH_{\s{KL}}}\leq \varepsilon  \sqrt{2\sum_{i=0}^{t}\binom{N}{i}(q^2-1)^i}.
\]
\end{prop}

\begin{proof}
    We follow the strategy of Proposition~\ref{prop:orthUni}: the uniform superoperator $\calN_{\s{U}}$  defined by the Kraus set $\frac
    {1}{\sqrt{M}}\calE=\ppp{\frac{1}{\sqrt{M}} U_k}_{k\in [M]}$ is an $\calE$-controlled CPTP map. Let $A_{\s{QEC}}(\calN)$ be the corresponding QEC matrix.  Note that $A_{\s{QEC}}(\calN)=\frac{1}{M}A_{\s{QEC}}(\calE)$. In particular,
    \begin{align}
        D_{\s{H}}(A_{\s{QEC}}(\calE),\calH_{\s{KL}})&\nonumber=D_{\s{H}}(M   A_{\s{QEC}}(\calN),\calH_{\s{KL}})\\
        &\nonumber=\min_{\substack{B \in \calH_{\s{KL}}\\ B\succeq 0}}D_{\s{H}}\p{M   A_{\s{QEC}}(\calN),B }\\
        &\nonumber=\min_{\substack{B \in \calH_{\s{KL}}\\ B\succeq 0}}\norm{\sqrt{M   A_{\s{QEC}}(\calN)}-\sqrt{B} }_2\\
        &\nonumber=\min_{\substack{B \in \calH_{\s{KL}}\\ B\succeq 0}}\sqrt{M}\norm{\sqrt{  A_{\s{QEC}}(\calN)}-\sqrt{  B/M} }_2\\
        &\nonumber=\sqrt{M} \min_{\substack{B \in \calH_{\s{KL}}\\ B\succeq 0}}\norm{\sqrt{  A_{\s{QEC}}(\calN)}-\sqrt{  B} }_2\\
        &=\sqrt{M}      D_{\s{H}}(A_{\s{QEC}}(\calN),\calH_{\s{KL}}).\label{eq:ScaleHellinger}
    \end{align}
    On the other hand, by Corollary~\ref{cor:HellingerChannelFidelity} and the AQEC assumption, we have
    \begin{equation}
        \frac{1}{\sqrt{2K}}D_{\s{H}}\p{A_{\s{QEC}}(\calN),\calH_{\s{KL}}}\leq \varepsilon_{\mathrm{ch}}^{\mathrm{opt}}(Q,\calN)\leq \varepsilon\label{eq:OPtunitaryepsilon}
    \end{equation}
  The proof is completed by combining \eqref{eq:ScaleHellinger} and \eqref{eq:OPtunitaryepsilon}.
\end{proof}

In contrast to the worst-case fidelity formulation, the $|\calE|$-dependent loss in Proposition~\ref{prop:AvAQECNecessary} is only of order $\sqrt{M}$, rather than $M$. The reason is that the channel-fidelity/Petz characterization is governed by the Hellinger distance of the QEC matrix to the KL subspace. This quantity enters linearly in the final AQEC error, while the contribution of several blocks combines at the level of squared Hellinger distances. Consequently, multiplying the number of relevant components by $M$ produces only a $\sqrt{M}$ loss after taking the square root.

The same proof technique also extends the necessary condition of Proposition~\ref{prop:NecFock} to the $\s{Av}$-AQEC setting. The resulting statement has the same form, except that the linear factor $t+1$ is replaced by $\sqrt{t+1}$. Thus, for the average-case channel-fidelity criterion, the corresponding necessary conditions retain the same structural content but with a milder dependence on the number of relevant errors.


\subsection{Equivalence of AQEC for erasures and general limited-weight errors}

 A fundamental structural feature of exact quantum error correction is the equivalence between correcting $t$ arbitrary qudit errors and correcting $2t$ erasures. In particular, this equivalence underlies the usual notion of distance for exact quantum codes, since a code of distance $2t+1$ is precisely one that corrects $t$ general errors\footnote{Here {\em general limited-weight errors} refers to the set of error operators in the span of the HW set $\calE_t$ that act on at most $t$ qudits.}, or equivalently $2t$ erasures. By contrast, the approximate setting lacks a comparable principle, as highlighted by the epigraph to this paper. In this section, we show that, at least in the setting considered here, this pessimistic outlook is not warranted: we prove that approximate correction of erasures \textit{does} imply approximate correction of general errors, and vice versa.

We start by proving a necessary condition for correcting $2t$ erasures: 
\begin{prop}\label{prop:necSerasure}
    If a code $Q\subseteq \calH_{q}^{\otimes N}$ is an $\varepsilon$-AQEC code  for $\calE_{t}^{\s{Er}}$, then 
    \[\zeta(\calE_{t}^{\s{Er}},Q)\leq {2\sqrt{2}\binom{N}{t}\varepsilon}.\]
\end{prop}
\begin{proof}
    The proof follows the same idea as in the proofs of Proposition~\ref{prop:orthUni} and Proposition~\ref{prop:NecFock}: we show that by scaling the operators in $\calE_{ t}^{\s{Er}}$ we obtain a Kraus error set satisfying the completeness relation, and then using B\'eny-Oreshkov conditions together with \eqref{eq:diamondBures} to obtain an upper bound on $\zeta(\calE_{t}^{\s{Er}},Q)$. Formally, consider the uniform $t$-erasure channel $\calN_{t,\s{U}}^{\s{Er}}$, in which a subset of $t$ qudits is chosen uniformly at random and erased, given by the Kraus operators
\[
\calE_{t,\s{U}}^{\s{Er}} := \biggl\{
\frac{1}{\sqrt{\binom{N}{t}}}   E_{I,\ux} : I\in \binom{[N]}{t},\ \ux\in [q]^{t},\ \alpha\in\C
\biggr\},
\qquad
E_{I,\ux} := 
\ket{\perp^{t}}_{I}\bra{\ux}_{I}\otimes I_{[N]\setminus I}
\]

A code $Q$ is $\varepsilon$-AQEC code  for $\calE_{ t}^{\s{Er}}$ then it is $\varepsilon$-AQEC code  for the channel $\calN_{t,\s{U}}^{\s{Er}}$. Indeed, the errors of the Kraus set $\calE_{t,\s{U}}^{\s{Er}}$ are spanned by the elements $\calE_{t}^{\s{Er}}$ and the corresponding connection matrix $C=\binom{N}{t}^{-\frac{1}{2}}I,$ implying that $\norm{C}_\infty<1$ and $\calN_{t,\s{U}}^{\s{Er}}\in \mathscr{N}(\calE_{t}^{\s{Er}})$. By Theorem~\ref{th:BO} and \eqref{eq:diamondBures}, there exists a matrix $\lambda$ such that the  operator $\calB_{\lambda }$ (defined in Theorem~\ref{th:BO}) satisfies $\norm{\calB_{\lambda}}_{\diamond}\leq {2\sqrt{2}\varepsilon}$.
 Consider the matrix $\lambda'=\lambda \binom{N}{t}$ and observe that the corresponding error-set B\'eny-Oreshkov operator satisfies $\calB_{\lambda',Q}^{\calE_{ t}^{\s{Er}}}=\binom{N}{t}  \calB_{\lambda}$. In particular, 
     \[
     \zeta(\calE_{ t}^{\s{Er}},Q)\leq \|\calB_{\lambda',Q}^{\calE_{ t}^{\s{Er}}}\|_{\diamond }=\binom{N}{t}\norm{\calB_{\lambda}}_{\diamond}\leq {2\sqrt{2}\binom{N}{t}\varepsilon}. \qedhere
     \]
\end{proof}
Below we denote by $B_q(n,s)$ the volume of the Hamming ball of radius $s$ in the $n$-dimensional Hamming space $[q]^n$:
 $B_q(n,s)${volume of the Hamming ball}
    $$
    B_q(n,s)=\sum_{i=0}^s \binom ni (q-1)^i.
    $$
Our main technical result in this section is stated in the next theorem.
\begin{tcolorbox}[width=1\linewidth, left=1mm, right=1mm, boxsep=0pt, boxrule=0.5pt, sharp corners=all, colback=white!95!black] 
\begin{theorem}\label{th:equivErGen}
    Let $Q\subseteq \calH_{q}^{\otimes N}$ be a code, and let $t\leq N/2$ be an integer. Then:
    \[
    \Big(2\binom{N}{2t}B_{q^2}(2t,t)\Big)^{-1} \zeta(\calE_{2t}^{\s{Er}},Q) \leq \zeta(\calE_t,Q)\leq 2B_{q^2}(2t,t)^2 \zeta(\calE_{2t}^{\s{Er}},Q), 
    \]
where $\calE_t$ is the set of HW operators of  weight at most $t$ defined in \eqref{eq:HWoperatorsT}.
\end{theorem}
\end{tcolorbox}
The proof of Theorem~\ref{th:equivErGen} appears in Appendix~\ref{app:ThErGen}, starting on p.~\pageref{app:ThErGen}. At a high level, the Bény--Oreshkov superoperators quantify the violation of the approximate KL conditions for a prescribed matrix of constants $\lambda$. The key observation is that, for both erasures and HW errors, this violation is controlled by the degree to which information deletion acts as a constant operator on an appropriate subsystem of size t in the erasure setting and 2t in the HW setting. The proof in both directions is therefore based on showing that these two kinds of B\'eny-Oreshkov data encode the same underlying ``deletion-to-constant'' information, and can be converted into one another through explicit factor maps, namely $\calB_{\s{Er}}=\calF\circ  \calB_{\s{HW}}$ in one direction and a reverse factorization in the other. The resulting bounds then follow by controlling the norms of these factor maps. Combining this theorem with the sufficient AQEC conditions of Theorem~\ref{th:BOerrorset} and the necessary conditions of Propositions~\ref{prop:orthUni}  and \ref{prop:necSerasure}, we deduce the relation between AQEC for $2t$ erasures and $t$ general limited errors.

\begin{theorem}
\label{th:ErasureGeneralEr}
    Let $Q\subseteq \calH_q^{\otimes N}$ be a quantum code. Then, if $Q$ is an $\varepsilon$-AQEC code  for $2t$ erasures, then it is  an $\varepsilon'$-AQEC code  for $t$ general errors, with 
    \begin{equation}
        \varepsilon'={2^{7/4}\sqrt{\binom{N}{2t}}B_{q^2}(2t,t)  \sqrt{\varepsilon}}.
    \end{equation}
    Conversely,  if $Q$ is an $\varepsilon$-AQEC code  for $t$ general errors, then it is  an $\varepsilon'$-AQEC code  for $2t$ erasures, with 
    \begin{equation}
        \varepsilon'={2^{7/4}\sqrt{\binom{N}{2t}B_{q^2}(2t,t)B_{q^2}(N,t)}  \sqrt{\varepsilon}}.
    \end{equation}
\end{theorem}
\begin{proof}
    Assume that $Q$ is an $\varepsilon$-AQEC code  for $2t$ erasures. By Proposition~\ref{prop:necSerasure} we have 
    \[\zeta(\calE^{\s{Er}}_{2t},Q)\leq {2\sqrt{2}\binom{N}{2t}\varepsilon}. \]
    By Theorem~\ref{th:equivErGen} we have 
    \[\zeta(\calE_{t},Q)\leq 2B_{q^2}(2t,t)^2    \zeta(\calE^{\s{Er}}_{2t},Q)\leq {4\sqrt{2}\binom{N}{2t}B_{q^2}(2t,t)^2\varepsilon}.\]
    Using Theorem~\ref{th:BOerrorset} we have that $Q$ is an $\varepsilon'$ AQEC for $\calE_t$ with 
    \[\varepsilon'=\sqrt{2\zeta(\calE_{t},Q)}\leq {2^{7/4}\sqrt{\binom{N}{2t}}B_{q^2}(2t,t)\sqrt{\varepsilon}}.\]
    
    The proof of the second part of the claim follows the same steps. By Proposition~\ref{prop:orthUni}, assuming  that $Q$ is an $\varepsilon$-AQEC code  for $t$ general errors, we have
    \[\zeta(\calE_{t},Q)\leq {2\sqrt{2} B_{q^2}(N,t)\varepsilon},\]
    where here $B_{q^2}(N,t)$ corresponds to the number of HW operators on $N$ qudits with weight at most $t$. Using the lower bound of Theorem~\ref{th:equivErGen} we have 
    \[\zeta(\calE^{\s{Er}}_{2t},Q)\leq 2\binom{N}{2t}B_{q^2}(2t,t)  \zeta(\calE_{t},Q) \leq   {4\sqrt{2}\binom{N}{2t}B_{q^2}(2t,t) B_{q^2}(N,t)\varepsilon}, \]
    which gives that $Q$ is an $\varepsilon'$ AQEC for $\calE^{\s{Er}}_{2t}$ with 
    \[\varepsilon'=\sqrt{2\zeta(\calE^{\s{Er}}_{2t},Q)}\leq {2^{7/4}\sqrt{\binom{N}{2t}B_{q^2}(2t,t) B_{q^2}(N,t)\varepsilon}}.
   \qedhere 
    \]
\end{proof}

Theorem~\ref{th:ErasureGeneralEr} converts an erasure-AQEC guarantee into a general-error AQEC guarantee, at the price of a multiplicative factor in the error parameter depending on the number of erased locations and local error patterns. Since this factor grows with the system size, it is important to understand how this reduction behaves for code families of increasing length. For fixed $t$, or more generally, when the relevant error sets grow only polynomially with $N$, the loss is only logarithmic in the dimension of a positive-rate code. Even in the linear-distance regime $t=\delta N$, where the factors 
may grow exponentially in $N$, they remain polynomial in the code dimension for positive-rate families. Thus, the theorem gives a quantitative route that links approximate erasure correction to approximate correction of general errors, provided the erasure-AQEC guarantee is strong enough to absorb the explicit loss.

This point can be seen concretely through the random partition codes for deletion errors analyzed below in Section~\ref{sec:ExDelEr}. Fix $q$ and write the quantum rate as $\calK=\log_q K/N$, where $K$ is the dimension of the $N$-qudit code. Since deletion correction implies erasure correction, Theorem~\ref{th:ErasureGeneralEr} can be used to obtain AQEC guarantees for general errors from sufficiently strong deletion-AQEC guarantees. To correct $\delta N$ general errors, it suffices to correct $2\delta N$ erasures with error parameter $\varepsilon_N$ satisfying
\begin{equation}
2\sqrt{\binom{N}{2\delta N}}B_{q^2}(2\delta N,\delta N)\sqrt{\varepsilon_N}\to0.\label{eq:ErasurestoG}
\end{equation}
Equivalently, if $\varepsilon_N=q^{-aN}$, then it is sufficient that $a>2\calG_q(\delta)$, where
    \[ 
    \calG_q(\delta) := \lim_{N\to\infty}\frac{1}{N}\log_q\bigg({\sqrt{\binom{N}{2\delta N}}B_{q^2}(2\delta N,\delta N)}\bigg)=\frac12\frac{H_2(2\delta)}{\log_2 q}+4\delta H_{q^2}\p{\frac{1}{2}},
    \]
where $H_q(\cdot)$ is defined in \eqref{eq:qaryEntropy}. Thus, the blow-up in Theorem~\ref{th:ErasureGeneralEr} gives a concrete target: the erasure-AQEC error must decay faster than $q^{-2\calG_q(\delta)}$. 

This stronger decay requirement translates into a slightly stronger rate constraint for the random partition construction. Without imposing a prescribed exponential decay rate on $\varepsilon_N$, the deletion-AQEC analysis of Section~\ref{sec:ExDelEr} gives $\varepsilon_N\to0$ for correction of $2\delta N$ deletions whenever
\begin{equation}
     \calK<\frac{1}{3}\p{\frac{R^{\s{Del}}(2\delta)}{2}-4M^{\s{Del}}(2\delta)},\label{eq:rateDelEx}
\end{equation}
where $R^{\s{Del}}( \cdot)$ and $M^{\s{Del}}(\cdot )$ are defined in \eqref{eq:Rdel} and \eqref{eq:Mdel}, respectively. Repeating the concentration analysis used to derive \eqref{eq:rateDelEx}, while requiring $\varepsilon_N=q^{-aN}$ with $a>2\calG_q(\delta)$, gives the condition
\[ \calK<\frac{1}{3}\p{\frac{R^{\s{Del}}(2\delta)}{2}-4M^{\s{Del}}(2\delta)-8\calG_q(\delta)}.\label{eq:ErasurestoGRate} \]
Since $R^{\s{Del}}(0)=1$, $M^{\s{Del}}(0)=0$, and $\calG_q(0)=0$, the right-hand side remains positive for all sufficiently small $\delta>0$. Therefore, the constants in Theorem~\ref{th:ErasureGeneralEr} reduce the achievable rate, but they do not make the erasure-to-general reduction vacuous.

The size of this loss also indicates where sharper results could improve the tradeoff. The constants come from the diamond-norm comparison in Theorem~\ref{th:equivErGen} and from the gap between the necessary and sufficient AQEC conditions in Proposition~\ref{prop:necSerasure} and Corollary~\ref{cor:generalErrorSuff}; tightening either of these steps would strengthen the link between erasure AQEC and general-error AQEC.

\subsection{AQEC, subsystem variance, and quantum circuit complexity}

Recent works have revealed a close connection between approximate quantum error correction and quantum circuit complexity \cite{yi2024complexity,yi2025lov,li2025random}. A central quantity in this connection is the subsystem variance introduced in \cite{yi2024complexity}. Informally, the subsystem variance of a code on a subsystem $I$ measures how much information about the encoded state can be seen from the reduced density matrix on $I$. If all code states have the same marginal on $I$, then the subsystem carries no logical information, and replacement noise on $I$ is exactly correctable. If the subsystem variance is small but nonzero, then the subsystem carries only a small amount of logical information, and the code is approximately correctable against replacement channels acting on that subsystem.
The significance of subsystem variance stretches beyond AQEC: namely, \cite{yi2024complexity} relates small subsystem variance 
to lower bounds for quantum circuit complexity of code states. Thus, subsystem variance also serves as a link between the indistinguishability of code states and circuit complexity. 

Here we connect this quantity to the error-set framework developed above. Specifically, we compare the environment-leakage distance $\zeta(\calE_{t}^{\s{Er}},Q)$ for erasure errors with the subsystem variance of a quantum code $Q$ on subsets of size $t$. Combining these bounds with Proposition~\ref{prop:necSerasure} and Theorem~\ref{th:ErasureGeneralEr} shows that subsystem variance not only controls correction against replacement channels on a fixed subsystem; through the erasure-to-general-error reduction, it also controls the AQEC capabilities of the code against general limited-weight errors. This strengthens the link between circuit complexity and AQEC: the same local-indistinguishability parameter that implies circuit complexity lower bounds also governs, quantitatively, the ability to correct physically natural families of errors. For simplicity, we follow \cite{yi2024complexity} and focus on the case of $q=2$.

\begin{definition}[Subsystem variance]\label{def:SubsystemVariance}
Let $Q\subseteq\calH_2^{\otimes N}$ be a $K$-dimensional quantum code with orthonormal basis $\ppp{\ket{c_i}}_{i\in[K]}$, and let $P_Q$ be the orthogonal projection on $Q$. Define the maximally mixed code state
\[ \Gamma_Q := \frac{1}{K}P_Q=\frac{1}{K}\sum_{i\in[K]}\ket{c_i}\bra{c_i}. \]
For a subsystem $I\subseteq[N]$, the subsystem variance of $Q$ on $I$ is
\[ V_{\s{sub}}(Q,I) := \max_{\rho\in\calD(Q)}\norm{\rho_I-(\Gamma_Q)_I}_1, \]
where $\calD(Q)$ denotes the set of density operators supported on $Q$, and $\rho_I=\tr_{[N]\setminus I}(\rho)$. Since the trace norm is convex, the maximum can equivalently be taken over pure code states $\rho=\ket{\psi}\bra{\psi}$ with $\ket{\psi}\in Q$. For $t\in[N]$, we define the size-$t$ subsystem variance by
\[ V_{\s{sub}}(Q,t) := \max_{I\in \binom{[N]}{t}}V_{\s{sub}}(Q,I). \]
\end{definition}

The quantity $V_{\s{sub}}(Q,t)$ measures the maximal distinguishability of code states from the maximally mixed code state 
based on the observation of only $t$ physical subsystems. In \cite{yi2024complexity}, this local distinguishability parameter is used to relate AQEC-type properties to circuit complexity lower bounds. We now show that, in our error-set language, the same quantity is directly related to the erasure performance of $Q$ as measured by $\zeta(\calE_t^{\s{Er}},Q)$. Consequently, through the erasure/general-error equivalence proved above, subsystem variance also controls the AQEC performance of the code against general limited-weight errors.
The next proposition, proved in Appendix~\ref{app:Complexity}, p.~\pageref{app:Complexity}, quantifies this connection by relating
the subsystem variance on subsystems of size $t$ to the environment-leakage distance for $t$-erasure errors. 

\begin{prop}\label{prop:SSVzeta}
    Let $Q\subseteq \calH_2^{\otimes N}$ be a $K$-dimensional quantum code, and let $t\in \N$ be fixed. Then 
    \begin{align}
        \frac{1}{2}V_{\s{sub}}(Q,t)\leq \zeta(\calE_t^{\s{Er}},Q)\leq  2K\binom{N}{t}V_{\s{sub}}(Q,t).
    \end{align}
\end{prop}
Thus, up to explicit dimension and combinatorial factors, small subsystem variance is equivalent to a small erasure proximity parameter in our error-set formulation. 

\subsection{Sufficient inner product conditions for AQEC}
As shown above, the environment-leakage distance provides meaningful information on the AQEC capabilities of a code, and the sufficient condition of Theorem~\ref{th:BOerrorset} may even be tight, up to dimension-independent factors (see Proposition~\ref{prop:exampleOptimality}). At the same time, for ease of use, it is preferable to replace this criterion with simpler, more directly verifiable conditions. In this subsection, we therefore derive a relaxed version of Theorem~\ref{th:BOerrorset}, formulated in terms of inner product conditions. We tailor these conditions to three scenarios, corresponding to three levels of structural assumptions on the action of the noise operators on basis elements. These scenarios are precisely the ones that 
are satisfied by the partition codes of Section~\ref{sec:PartitionCodes}, which we introduce later. These codes are indexed
by partitions of finite metric spaces with varying geometric assumptions. Once those are defined, it will be evident that
the three cases in Proposition~\ref{prop:AQECcond} are designed to match their geometric counterparts. We term the suite of geometric
constraints the {\em metric--error alignment hierarchy}, and moving lower along it, we can afford to make the assumptions on error
correction progressively less restrictive. 

\begin{prop}[Sufficient AQEC conditions for error sets]\label{prop:AQECcond}
    Let $\calE=\ppp{E_k}_{k=0}^{M-1}$ be a finite error set and let $Q$ be a code space with an orthonormal basis $\ppp{\ket{c_i}}_{i=0}^{K-1}$. The code $Q$ is $\varepsilon$-AQEC code  for the error set $\calE$ if there exists a set of numbers $\p{\lambda_{k,l}}_{k,l}$ such that any one of the following conditions holds:
    \begin{enumerate}
        \item[{\rm 1.}] \begin{equation}
            \max_{\substack{i,j\in [K]\\k,l\in [M]} } \abs{\bra{c_i} E_k^\dag E_l \ket{c_j}-\lambda_{kl} \delta_{ij}}\leq \frac{\varepsilon^2}{2K^2M^2};\label{eq:BOcond1}
        \end{equation}
        \item[{\rm 2.}]  $\bra{c_i} E_k^\dag E_l \ket{c_j}=0$ whenever $i\neq j$ and 
        \begin{equation}
            \max_{\substack{i\in [K]\\k,l\in [M]} } \abs{\bra{c_i} E_k^\dag E_l \ket{c_i}-\lambda_{k,l}}\leq \frac{\varepsilon^2}{2K M^2};\label{eq:BOcond2}
        \end{equation}
        \item[{\rm 3.}] $\bra{c_i} E_k^\dag E_l \ket{c_j}=0$ whenever $(i,k)\neq (j,l)$ and
        \begin{equation}
            \max_{\substack{i\in [K]\\k\in [M]} } \abs{\bra{c_i} E_k^\dag E_k \ket{c_i}-\lambda_{kk}}\leq \frac{\varepsilon^2}{2K M}.\label{eq:BOcond3}
        \end{equation}
    \end{enumerate}
\end{prop}
The proof is given in Appendix~\ref{app:InnerProdCond}, p.~\pageref{app:InnerProdCond}.


\section{QEC in Hilbert spaces indexed by discrete metric spaces }\label{sec:Hierarchy}
In this section, we study quantum coding in Hilbert spaces whose distinguished basis is indexed by a discrete metric space. The noise model is given by a family of error sets $\mathscr{E}=(\calE_t)_t$, where the parameter $t$ quantifies, in a model-dependent way, the amount of noise. We identify three levels of alignment between the action of the error operators and the metric structure of the underlying indexing space, and formulate them as a hierarchy. We then introduce a general quantum code construction in this setting based on an underlying classical code in the indexing space, and show that its exact and approximate error-correction guarantees are determined by the level of this hierarchy.

Let $(X,d)$ be a discrete space equipped with a metric function $d:X\times X \to \R_+$. Let $\calH_X$ be a Hilbert space indexed by the elements of $X$, namely, a Hilbert space admitting an orthonormal basis $\ppp{\ket{x}}_{x\in X}$. In many cases, the noise model of interest is compatible, to some extent, with the geometry of $X$, and this feature can be exploited to construct both exact and approximate quantum codes. We formalize this connection by distinguishing between three levels of alignment:

\graybox{\begin{definition}[{\bf Metric–error alignment hierarchy}]\label{def:MEAhirarchy} We define a three-level metric--error alignment hierarchy $\mathscr{L}_0\supset \mathscr{L}_1\supset\mathscr{L}_2$. 
        Let $\calH_X$ be indexed by a discrete space $(X,d)$, and let $\mathscr{E}=(\calE_t)_{t\in\N}$  be a collection of error sets on $\calH_X$, with $t$ representing the noise level. 
        \begin{enumerate}
            \item If nothing can be assumed about the relation between $\mathscr{E}$ and $d$, we say that $(\calH_X,d,\mathscr{E})$ is in  \textit{level zero} of the hierarchy, denoted as $(\calH_X,d,\mathscr{E})\in\mathscr{L}_0$.
            \item $(\calH_X,d,\mathscr{E})$ is in the \textit{first level} of the hierarchy, $(\calH_X,d,\mathscr{E})\in \mathscr{L}_1$, if   for any $t>0$ and $x,x'\in X$ 
            \begin{equation}
                d(x,x')> t \implies \bra{x}E^\dag F \ket{x'}=0 \text{ for all $E,F\in \calE_t$}.\label{eq:firstlevel}
            \end{equation}
            \item $(\calH_X,d,\mathscr{E})$ is in the \textit{second level} of the hierarchy, $(\calH_X,d,\mathscr{E})\in \mathscr{L}_2$, if $(\calH_X,d,\mathscr{E})\in \mathscr{L}_1$, and for all $x\in X$, 
            \begin{equation}
                \bra{x}E^\dag F \ket{x}=0 \text{ for all $E,F\in \bigcup_{t}\calE_t$ such that $E\neq F$}.\label{eq:secondlavel}
            \end{equation}
        \end{enumerate}
\end{definition}
}

In the remainder of this subsection, we present several examples of Hilbert spaces arising in coding-theoretic settings, together with their natural indexing metric spaces, and determine their level in the metric--error alignment hierarchy. These examples illustrate how different noise models interact with the underlying metric structure, and show that the same physical Hilbert space may belong to different levels depending on the chosen indexing and error family. For the reader's convenience, the main examples considered below are summarized in Table~\ref{tab:metric_spaces_examples}.

        \begin{example}[Qudit Pauli errors are in \(\mathscr{L}_1\)]\label{ex:HWsecondlevel}
        Let $(X=[q]^N,d)$ be the Hamming space equipped with the scaled Hamming metric defined as 
              \begin{equation}
                d(\ux,\uy)=\frac{1}{2}\wt(x-y)=\frac{1}{2}\abs{\ppp{i\in [N]~:~ x_i\neq y_i}},\label{eq:HammingScaled}
            \end{equation}
        and let $\calH_X=\calH_q^{\otimes N}$ with the standard computational basis.
Consider $\calE_t$, the set of HW operators of weight up to $t\in \N$ defined in \eqref{eq:HWoperatorsT}. 
Note that $(\calH_X,d,\mathscr{E})\in \mathscr{L}_1\backslash \mathscr{L}_2$. Indeed, assume $d(\ux,\uy)>t $ and $E,F\in \calE_t$. Note that $E^\dag F$ acts on at most $2t$ distinct qudits, and therefore $E^\dag F \ket{y}$ is spanned by basis elements of the form $\ppp{\ket{y+e} ~:~ \wt(e)\leq 2t}$ which by the distance assumption, does not include $\ket{x}$. Since $\braket{x|y+e}=0$ for any such elements, \eqref{eq:firstlevel} holds. On the other hand, let $\ux$ be the all $1$'s vector and let $Z_1\in \calE$ be the $Z$ operator on the first qudit. Clearly, $\bra{x}IZ_1\ket{x}=\omega\neq 0$, and \eqref{eq:secondlavel} does not hold.
        \hfill$\triangleleft$ \end{example}
        
\begin{example}[Bosonic loss errors are in \(\mathscr{L}_2\)]\label{ex:BosonicL1L2} Consider $X_{\s{F}}=\Z_0^{q}$ be the space of $q$-tuples with nonnegative integer entries, equipped with the $\ell_1$ distance:
\begin{equation}
    \label{eq:ell1Dist}d_{1}(\ux,\uy)=\frac{1}{2}\sum_{i=0}^{q-1} |x_i-y_i|=\frac{1}{2}\norm{\ux-\uy}_1.
\end{equation}
Note that $X$ naturally indexesthe $q$-mode Fock state space $\calH_{X_{\s{F}}}=\bigcup_{N\in \N}\calH_{q,N}$, where $\calH_{q,N}$ denotes the space of Fock states with excitation $N$, defined in \eqref{eq:COnstExSpace}. Let $\mathscr{E}^{\s{AD}}=(\calE_{\leq t}^{\s{AD}})_t$ be the set of (normalized) amplitude damping noise operators as defined in \eqref{eq:ADerrorSet} and \eqref{eq:ADnormalizedOP}, where $\calE_{\leq t}^{\s{AD}}$ represents all errors corresponding to a total loss of at most $t$ photons. A straightforward calculation reveals that for any $\un,\un'\in X$ with $\sum_i n_i=N$, $\sum n_i'=N'$, and for all $\ur,\us$ with $\sum_i r_i=r$ and $\sum_i s_i=s $:
\[
\bra{\un'}\Tilde{A}_{\ur}^\dag \Tilde{A}_{\us}\ket{\un}=\begin{cases}
    \frac{1}{\sqrt{\binom{N'}{r}\binom{N}{s}}}\sqrt{\prod_{i=0}^{q-1}\binom{n_i}{r_i}\prod_{i=0}^{q-1}\binom{n'_i}{s_i}}\braket{\un'-\ur|\un -\us} & N\geq s, N'\geq r\\
    0 & \text{otherwise.}
\end{cases}\]
 Note that if $d_1(\un,\un')>t$ and $r,s\leq t$, then $\un'-\ur\neq \un-\us$, so $\bra{\un'}\Tilde{A}_{\ur}^\dag \Tilde{A}_{\us}\ket{\un}=0$. Furthermore, if $\un=\un'$ and $\ur\neq \us$ then $\un-\ur\neq \un-\us$ and again $\bra{\un}\Tilde{A}_{\ur}^\dag \Tilde{A}_{\us}\ket{\un}=0$. This shows that $(X_{\s{F}},d_1,\mathscr{E}^{\s{AD}})\in \mathscr{L}_2$.\hfill$\triangleleft$ \end{example} 

\begin{example}[Bosonic shift-rotation errors are in \(\mathscr{L}_1\)]\label{ex:BosonicL1L21}

An argument similar to that of Example~\ref{ex:BosonicL1L2} is used to show that the bosonic Fock space is in $\mathscr{L}_1$ with respect to the number-shift and phase rotation errors. Let $t\in N$ and $f:\N\to \N$ be any function and $\beta\in[0,1]$ be fixed. Consider the error family $\mathscr{E}^{\s{SR}}=(\calE_{t,f(t),\beta}^{\s{SR}})_t$, as defined in \eqref{eq:ShiftRotationErrorSet}.  As in the case of AD errors, for any operators $E_{\ur,\uTheta,\beta},E_{\ur',\uTheta',\beta}\in \calE_{t,f(t),\beta}^{\s{SR}}$ and number states $\ket{\un},\ket{\un'}$ we have
\[\bra{\un} E_{\ur,\uTheta,\beta}^\dag E_{\ur',\uTheta',\beta}\ket{\un'}= \begin{cases}
    \exp\p{i\frac{2\pi\beta}{N}\sum_{i=1}^q (\theta_i' n_i'-\theta_i n_i)}  \braket{\un+\ur|\un'+\ur'} & N\geq r, N'\geq r',\\
    0& \text{otherwise.}
\end{cases}\]
The above expression vanishes if $d_1(\un,\un')>t$ as $\un+\ur\neq \un'+\ur'$, which implies $(X_{\s{F}},d_1,\mathscr{E}^{\s{SR}})\in \mathscr{L}_1$.
\hfill$\triangleleft$ \end{example}

\begin{example}[Qudit amplitude damping is in \(\mathscr{L}_2\)]\label{ex:QuditAD}
A similar example is obtained from qudit amplitude damping noise. Let $X=[q]^N$, equipped with the $\ell_1$ metric $d_1(\ux,\uy)$, so that $X$ indices the computational basis of $\calH_X=\calH_q^{\otimes N}$. Mathematically, this space can be embedded in the $N$-mode Fock space by identifying a basis vector $\ket{\ux}$, $\ux=(x_1,\dots,x_N)\in[q]^N$, with the occupation vector having $x_i$ excitations in mode $i$. Under this identification, $\calH_q^{\otimes N}\subseteq \Span\p{\bigcup_{i=0}^{(q-1)N}\calH_{N,i}}$, where $\calH_{N,i}$ denotes the $N$-mode subspace of total excitation number $i$, as in \eqref{eq:COnstExSpace}. Hence, the usual bosonic amplitude damping lowering operators restrict naturally to the qudit space, giving the standard qudit amplitude damping model; see, e.g., \cite{grassl2018quantum}. Let $\mathscr{E}^{\s{AD}}=(\calE_{\leq t}^{\s{AD}})_t$, where $\calE_{\leq t}^{\s{AD}}$ consists of all qudit amplitude damping errors corresponding to a total loss of at most $t$ excitations. Since these errors are obtained by restricting the corresponding Fock-state amplitude damping errors to the bounded-occupancy subset $[q]^N$, the same $\ell_1$-metric argument applies: an error that loses at most $t$ excitations can only move a basis state by $\ell_1$-distance at most $t$. Therefore, exactly as in the Fock-state setting, the qudit amplitude damping model satisfies the second-level alignment condition with the scaled $\ell_1$ distance, and hence $(X,d_1,\mathscr{E}^{\s{AD}})\in\mathscr{L}_2$. Thus, the Fock-state amplitude damping example extends directly to qudit amplitude damping noise.
\hfill$\triangleleft$ \end{example}

        \begin{example}[Deletions are in \(\mathscr{L}_1\)]\label{ex:DelL1} Let $X=[q]^N$, and let $d_{\s{Del}}$ denote the deletion distance on $X$. For two sequences $\ux,\uy\in X$, this distance is the minimum number of deletions that must be applied to each sequence to make the resulting sequences identical. Equivalently,
        \begin{equation}
            d_{\s{Del}}(\ux,\uy)= N-\lcs(\ux,\uy),\label{eq:dDel}
        \end{equation}
        where $\lcs(\ux,\uy)$ denotes the longest common subsequence of $\ux$ and $\uy$. 
        Let $\mathscr{E}^{\s{Del}}=(\calE_t^{\s{Del}})$. Arguing similarly to Example~\ref{ex:HWsecondlevel}, we have $(\calH_X,d_{\s{Del}},\mathscr{E}^{\s{Del}})\in \mathscr{L}_1$ but $(\calH_X,d_{\s{Del}},\mathscr{E}^{\s{Del}})\notin \mathscr{L}_2$. 

        We now observe an interesting phenomenon:  reducing to a subsystem and changing the indexing metric space may result in a higher metric--error alignment level. Recall the permutation-symmetric space $\mathrm{Sym}(q,N)\subseteq \calH_X$ defined in \eqref{eq:symmetricSpaceX}. As mentioned in Section~\ref{sec:PIcodes}, $\mathrm{Sym}(q,N)$ is naturally indexed by the space $\hat{X}=\calS_{q,N}$ (with the basis of Dicke states \eqref{eq:Dickestates}), equipped with the $\ell_1$ metric \eqref{eq:ell1Dist}. Recall that $t$-deletion errors are generally indexed by pairs $(I,\ux)$ where $I$ are the deleted coordinates and $\ux\in [q]^t$. It was shown in \cite{aydin2026quantum} that the action of a deletion error on a Dicke state is determined by the composition of $\ux$ alone, and that in particular:
        \begin{align}
            \bra{D_{\un}}E_{I,\ux}^\dag E_{I',\uy} \ket{D_{\un'}}=\sqrt{\frac{\binom{N-t}{\un -\ue_x}\binom{N-t}{\un' -\ue_y}}{\binom{N}{\un}\binom{N}{\un'}}}  \braket{D_{\un-\ue_x}|D_{\un'-\ue_y}},\label{eq:levelPI}
        \end{align}
        where $\ue_x=C(\ux)\in \calS_{q,t}$ and $\ue_y=C(\uy)\in \calS_{q,t}$   denote the composition of $\ux$ and $\uy$ as defined in \eqref{eq:Copositionof}. That is, deletion errors can be indexed using elements in $\calS_{q,t}$. Since $\ue_x$ and $\ue_y$ sum to $t$, if $d_1(\un,\un')>t$ the we have $\un-\ue_x\neq \un'-\ue_y$ and the inner product of \eqref{eq:levelPI} becomes $0$. Additionally, if $\un= \un'$ we clearly have $\un-\ue_x\neq \un-\ue_y$. Let $\hat{\mathscr{E}}^{\s{Del}}=(\hat{\calE}_t^{\s{Del}})_t$ be the set of simplex-indexed deletion errors:
       \begin{equation} \hat{\calE}_t^{\s{Del}}=\ppp{\sqrt{\binom{t}{\ue}}  E_{\ue} ~:~ \ue \in \calS_{q,t}}, \label{eq:DelPIset}\end{equation}
        describing the action of deletions on the permutation-symmetric space. The normalization is required since each simplex-indexed operator captures the action of all deletion operators $E_{[t],\ux}$ with the same composition $C(\ux)$. Note that we can limit ourselves to the case that $I=[t]$ since symmetry implies that the action of $E_{I,\ue}$ on the permutation-symmetric space is independent of $I$. We have shown that $(\calH_{\hat{X}},d_1,\hat{\mathscr{E}}^{\s{Del}})\in \mathscr{L}_2$.           
        \hfill$\triangleleft$ \end{example}

\begin{example}[Majorana errors are in \(\mathscr{L}_1\)]\label{ex:MajoranaL1} Let $X=\set{0,1}^{N}$ and $d$ be the scaled Hamming metric defined in \eqref{eq:HammingScaled}. Consider $\mathscr{E^{\s{Maj}}}=(\calE^{\s{Maj}}_t)_t$, the family of limited-support Majorana errors \eqref{eq:MajoranaErrorsT} on the fermionic Fock space $\calF_N$, which is naturally indexed by $X$ with the occupation-number basis $\ppp{\ket{x} ~:~x\in X}$. Recall that under the Jordan-Wigner representation, one has \eqref{eq:JordanWignerMajorana}
\begin{equation}
    \gamma_{2j}=Z_0 \cdots Z_{j-1}X_j,\quad \gamma_{2j+1}=Z_0 \cdots Z_{j-1}Y_j,
\end{equation}
where the Pauli operators act on the computational-basis representation of $\calF_N$. 
Therefore, for $\ux\in X$,
\begin{equation*}
    \gamma_{2j}\ket{\ux}=(-1)^{\sum_{i<j}x_i}\ket{\ux\oplus \ue_j},\quad
    \gamma_{2j+1}\ket{\ux}=i(-1)^{\sum_{i\leq j}x_i}\ket{\ux\oplus \ue_j},\label{eq:SingleMajoranaAction}
\end{equation*}
where $\ue_j$ is the $j$th standard basis vector and $\oplus$ denotes bitwise addition modulo $2$. Thus, each single Majorana operator flips one occupation bit, up to a phase determined by the occupations of the preceding modes. We conclude that a nontrivial monimal Majorana operator $\gamma_I$, $I\subseteq [2N]$, acts on an occupation-basis element $\ket{\ux}$ as 
\[\gamma_{I}\ket{\ux}=\psi_{I,\ux}   \ket{\ux'}\]
where $\psi_{I,\ux}$ is a phase factor in $\ppp{\pm i,\pm 1}$ and $\ux'$ is a binary vector that agrees with $\ux$ on $[N]\setminus I'$, and $I'=\ppp{\floor{j/2} ~:~ j\in I }$. In particular, the coordinates at which $\ux$ and $\ux'$ disagree must lie in $I'$, whose cardinality is bounded by $|I|$. It follows that if $\ux,\uy\in X$ are such that $d(\ux,\uy)> t$ and $I_1,I_2$ are subsets of $[2N]$ of size at most $t$,  
\[\bra{\ux}\gamma_{I_1}^{\dag} \gamma_{I_2}\ket{\uy} =\psi_{I_1,\ux}^*\psi_{I_2,\uy}\braket{\ux'|\uy'}=0,\]
where the last equality follows since $\ux'\neq \uy'$ by the distance assumption. The above proves that $(X,d,\mathscr{E}^{\s{Maj}})\in \mathscr{L}_1$.

This example highlights one advantage of the partition-code framework developed in the next section. One main difficulty with Majorana QECCs is that Majorana operators are not local as qubit operators under the Jordan-Wigner representation: even a single Majorana monomial may contain a long $Z$ string. However, this non-locality appears only in the diagonal phase action. The action on the occupation label is still local in the following sense: each single Majorana operator flips exactly one occupation bit, and each Majorana monomial changes only the occupation bits associated with its support. The $\mathscr{L}_1$ condition asks when two error operators can make two different basis states non-orthogonal after the errors are applied. This depends only on the resulting occupation labels, not on the phases produced by the long $Z$ strings. Therefore, the non-local $Z$ strings do not prevent Majorana noise from respecting the geometry of the indexing space $X=\set{0,1}^{N}$. This is the property we exploit when constructing quantum codes for Majorana noise from classical codes on the Hamming space.
\hfill$\triangleleft$ \end{example}

\subsection{{Quantum codes from partitions}} \label{sec:PartitionCodes}
The idea of constructing quantum codewords by partitioning a larger collection of basis vectors goes back to Knill, Laflamme, and Viola~\cite{knill2000theory}. In their general noise framework, after constructing a classical transmission basis adapted to a finite-dimensional error space, one partitions this basis and chooses nonnegative coefficients so that the corresponding convex hulls intersect; the resulting superpositions form a quantum code. Codes obtained by this method are often called KLV codes. This line of thought was further developed in the setting of quantum metric spaces and Lie-theoretic error models, including for instance~\cite{bumgardner2012codes}, building on the broader $W^*$-metric viewpoint of~\cite{bumgardner2012codes}. Related convex-geometric formulations also appear in the theory of higher-rank matricial ranges and hybrid quantum error correction~\cite{cao2021higher}. More recent constructions are closer in spirit to the geometric setting considered here. The authors of \cite{movassagh2024constructing} constructed qudit partition codes protecting against bounded-weight errors relying on classical codes, generalizing the CSS construction. This paradigm was later adapted in \cite{aydin2026quantum,elimelech2026asymptotically} to constant-excitation Fock-state codes, permutation-invariant codes, and spin codes, using structured classical codes as the underlying combinatorial objects.

In this section, we place these constructions in a common framework by considering Hilbert spaces indexed by discrete metric spaces and error sets whose action is aligned with the underlying geometry in the sense of Definition~\ref{def:MEAhirarchy}. While our exact-QEC existence argument uses the same convex-geometric principle as the KLV construction, the role it plays here is different. Rather than applying Tverberg's theorem directly to an unstructured transmission basis, we start from an underlying classical code constructed for the metric that is aligned with the 
quantum error model. This allows the partition argument to inherit distance and structural properties from the classical code, and leads to concrete families in several settings, including qudit systems, constrained Hilbert spaces, Majorana systems, and bosonic Fock spaces. On the approximate side, our main new ingredient is the averaging-based framework for random partition codes, first used in the Fock-state setting in \cite{elimelech2026asymptotically}, which we extend here to general metric-indexed Hilbert spaces and organize according to the type of metric--error alignment.

\begin{construction}[Codes from partitions]\label{const:partitionCodes}
    Let $(X,d)$ be a discrete metric space which indexes the Hilbert space $\calH_X=\Span\ppp{\ket{x} ~:~ x\in X}$, 
    and let  $C\subseteq X$ be a classical code. A quantum code $Q\subset \calH_X$ is called a $C$-partition code if there exists a partition of $C$ into pairwise disjoint subsets $C_0,\dots, C_{K-1}$ and an orthonormal basis of $Q$, $\ket{c_0},\dots, \ket{c_{K-1}}$ such that 
   
    \begin{equation}
        \ket{c_i}=\sum_{x\in C_i}\alpha_{x} \ket{x}, \quad \sum_{x\in C_i}|\alpha_x|^2=1.\label{eq:partitionCode}
    \end{equation}
    A basis for a partition code of the form \eqref{eq:partitionCode} is called a \textit{canonical basis}. 
\end{construction}

\begin{example}[CSS codes are partition codes]\label{ex:CSS} As observed in \cite{movassagh2024constructing}, this construction generalizes the CSS construction, which
is based on a pair of classical codes $C_1,C_2$ that satisfy the dual code containment property $C_2^\perp\subseteq C_2$. The ambient space is $X=[q]^N$, equipped with the Hamming metric. The partition elements are given by the cosets of the form $C_{\ux}=C_2^\bot+\ux$, $\ux\in C_1$, and the coefficients are $\alpha_{\ux}=1/\sqrt{|C_2^\bot|}$. 

\hfill$\triangleleft$ \end{example}
\subsubsection{Partition codes in the metric--error alignment hierarchy}
We study structural properties and error-correction performance of partition codes whose underlying classical code satisfies an appropriate distance assumption. We show that, as the metric--error alignment level increases, the resulting quantum partition code acquires increasingly strong structural properties. These properties can then be exploited to obtain improved exact and approximate QEC guarantees. 
\begin{definition}
    The distance of a classical code $C$ in a discrete metric space $(X,d)$ is defined to be
    \[d(C)=\min_{\substack{x,y\in C\\ x\neq y}} d(x,y).\]
\end{definition}

\begin{lemma}\label{lem:PartitionCodesProp}
    Let $\calH_X$ be a Hilbert space indexed by a discrete metric space $(X,d)$ with a family of error sets $\mathscr{E}=(\calE_t)_t$. Assume that  $Q$ is a quantum partition in $\calH_X$ obtained from a classical code $C\subset X$ with a canonical basis $\ppp{\ket{c_i}}_i$ as defined in \eqref{eq:partitionCode}.
    \begin{enumerate}
        \item If $(\calH_X,d,\mathscr{E})\in \mathscr{L}_1$ then for any $t<d(C)$  and $E,F\in \calE_t$ we have 
        \[\bra{c_i} E^\dag F \ket{c_j}=0 \quad \forall i\neq j, \text{ and all } E,F\in \calE_t.\]
        \item 
        If $(\calH_X,d,\mathscr{E})\in \mathscr{L}_2$ then we also have
         \[\bra{c_i} E^\dag F \ket{c_i}=0 \quad \forall i \text{ and } E,F\in \bigcup_{t< d(C)} \calE_t \text{ such that  } E\neq F.\]
    \end{enumerate}
\end{lemma}
\begin{proof}
    The proof is straightforward: Let $C_1,\dots,C_K$ be the partition that defines $Q$ and $\alpha_x$ be the coefficients corresponding to the basis elements. Assume that $t<d(C)$, and $i\neq j$. We have 
    \begin{align*}
        \bra{c_i} E^\dag F \ket{c_j}=\sum_{x\in C_i}\sum_{x'\in C_j}\alpha_x^* \alpha_{x'} \bra{x}E^\dag F \ket{x'}=0,
    \end{align*}
    where in the last equality we used that $d(x,x')\geq d(C)>t$ (as $C_i\cap C_j=\varnothing$ which implies $x\neq x'$), and the $\mathscr{L}_1$ assumption. Further assuming $\mathscr{L}_2$ we have 
    \begin{align}
        \bra{c_i} E^\dag F \ket{c_i}&=\nonumber\sum_{x,x'\in C_i}\alpha_x^* \alpha_{x'} \bra{x}E^\dag F \ket{x'}\\
        &\nonumber=\sum_{\substack{x,x'\in C_i\\x\neq x'}}\alpha_x^* \alpha_{x'} \bra{x}E^\dag F \ket{x'}+\sum_{x\in C_i} \bra{x}E^\dag F \ket{x}=0,
    \end{align}
    where in the last equality, the terms in the first sum are all zeros by the code distance and $\mathscr{L}_1$ assumption, and the remaining terms are all zeros due to the $\mathscr{L}_2$ assumption.
\end{proof}

\subsection{Exact QEC using partition codes}
An important implication of Lemma~\ref{lem:PartitionCodesProp} concerns the existence of exact-QEC partition codes derived from a prescribed underlying classical code with a given distance. In the following theorem, we generalize the existence result of \cite{aydin2026quantum} to arbitrary codes satisfying the second level of the metric--error alignment hierarchy.
\begin{theorem}[Existence of exact QEC partition codes in first and second levels]\label{th:exactQECpatition} Let $\calH_X$ be a Hilbert space indexed by a discrete metric space $(X,d)$ with a family of error sets $\mathscr{E}=(\calE_t)_t$. A classical code $C\subseteq X$ with distance $d(C)>t$ admits a quantum partition code $Q$ with dimension $K$ that is an exact QEC for $\calE_t$ under each of the following conditions:
\begin{enumerate}
    \item $(\calH_X,d,\mathscr{E})\in\mathscr{L}_1$ and \begin{equation}
    \abs{C}\geq (K-1)(|\calE_t|^2+1)+1.\label{eq:TverbergCosnt2}
\end{equation}
\item $(\calH_X,d,\mathscr{E})\in\mathscr{L}_2$ and \begin{equation}
    \abs{C}\geq (K-1)(|\calE_t|+1)+1.\label{eq:TverbergCosnt}
\end{equation}
\end{enumerate}

\end{theorem}
The proof follows the convex-geometric approach of \cite{aydin2026quantum}. We briefly sketch the main idea below and give the complete proof in Appendix~\ref{app:tverberg}, page~\pageref{app:tverberg}.
\begin{proof}[Proof sketch:] We sketch the proof for the case $\calH_X,d,\mathscr{E})\in \mathscr{L}_2$, as the proof for the first level of the hierarchy follows the same idea. Recall that by the well-known KL conditions, a code $Q$ with basis $(\ket{c_i})_{i\in K}$ is QEC for an error set $\calE_t$ if and only if there exist constants $(\lambda_{E,F})_{E,F\in \calE_t}$ such that for any $i,j$ 
\begin{equation}
    \bra{c_i}E^\dag F \ket{c_j}=\lambda_{E,F}  \delta_{i,j}. \label{eq:KLpartition}
\end{equation}
We note that for any partition code $Q$, by Lemma~\ref{lem:PartitionCodesProp} the left-hand side of \eqref{eq:KLpartition} vanishes for any $i\neq j$ as long as $d(C)>t$ and $(\calH_X,d,\mathscr{E})\in \mathscr{L}_1$, and therefore the orthogonality conditions are fulfilled automatically. If $(\calH_X,d,\mathscr{E})\in \mathscr{L}_2$, then   \eqref{eq:KLpartition} becomes zero for $E\neq F$.  Thus, it remains to find a partition of $C$ and coefficients $(\alpha_{x})_{x\in C}$ such that for all $i,j\in [K]$ and $E\in \calE_{t}$ we have 
\[\bra{c_i}E^\dag E\ket{c_i}=\bra{c_j}E^\dag E\ket{c_j}.\]
Expanding the above expression, the problem reduces to proving the existence of a partition $C_1,\dots, C_K$ such that the system of equations 
\[\sum_{x\in C_1}a_{x}  \bra{x}E^\dag E \ket{x}=\sum_{x\in C_2}a_{x}  \bra{x}E^\dag E \ket{x}= \cdots=\sum_{x\in C_K}a_{x}  \bra{x}E^\dag E \ket{x},\quad 
E\in \calE_t \]
has a \textit{nonnegative }solution which satisfies 
\[\sum_{x\in C_1}a_{x}=\sum_{x\in C_2}a_{x}= \cdots=\sum_{x\in C_K}a_{x}.\]
We prove the existence of such a solution using the Tverberg theorem argument of \cite{aydin2026quantum}, which also implies the constraint \eqref{eq:TverbergCosnt}.
\end{proof}

\begin{remark}\label{rem:Tver}
Theorem~\ref{th:exactQECpatition} offers sufficient conditions for the existence of exact-QEC partition codes constructed from a classical code satisfying a suitable distance assumption. This existence result can be made constructive by applying a constructive version of Tverberg's theorem: its proof reduces the construction of the required partition to the problem of finding a Tverberg partition for a set of points in the ambient space whose dimension is determined by the number of relevant error operators. Using the constructive procedure discussed in \cite{agarwal2008algorithms}, this leads to an algorithm whose complexity scales as $|C|^{O(|\calE_t|^{3-i})}$ where $(\calH_X,d,\mathscr{E})\in\mathscr{L}_i$, $i=1,2$.
Consequently, although the result is constructive in principle, the resulting procedure is generally computationally infeasible. In particular, when the number of correctable errors grows with the number of physical subsystems, the quantity $|\calE_t|$ itself grows rapidly, and the above complexity becomes doubly exponential in the relevant system-size parameters.
\hfill$\triangleleft$ \end{remark}

\begin{remark}\label{rem:QuariLimited} The cardinality conditions in  \eqref{eq:TverbergCosnt} and \eqref{eq:TverbergCosnt2} are stated in a worst-case form: they count all linear constraints that may be needed in order to enforce the KL conditions. In structured settings, this count can be unnecessarily pessimistic, since some of the corresponding equations are satisfied automatically. More precisely, in the $\mathscr{L}_1$ scenario, the factor $|\calE_t|^2$ in \eqref{eq:TverbergCosnt} may be replaced by the number of operators $W=E^\dag F$, with $E,F\in\calE_t$, for which the function $x\mapsto \bra{x}W\ket{x}$ is not constant on $X$. Indeed, if this function is constant, then the corresponding KL constraint is already scalar on every partition element and imposes no additional restrictions on the Tverberg partition. Similarly, in the $\mathscr{L}_2$ scenario, the factor $|\calE_t|$ in \eqref{eq:TverbergCosnt2} may be replaced by the number of operators $W=E^\dag E$, with $E\in\calE_t$, for which $x\mapsto \bra{x}E^\dag E\ket{x}$ is not constant on $X$.

As a concrete example, consider the binary qubit case with $X=\{0,1\}^N$, and let $\calE_t$ be the set of Pauli errors of weight at most $t$. Write $W$, up to a phase, as $W=\bigotimes_{i=0}^{N-1} P_i$ with $P_i\in\{I,X,Y,Z\}$. If at least one $P_i$ is equal to $X$ or $Y$, then $W$ flips the $i$th computational-basis bit, and therefore $W\ket{x}$ is proportional to a basis vector $\ket{x'}$ with $x'\neq x$. Hence $\bra{x}W\ket{x}=0$ for every $x\in\{0,1\}^N$, so this constraint is vacuous. The only remaining case is when $P_i\in\{I,Z\}$ for every $i$, in which case, if $S=\{i:P_i=Z\}$, then $\bra{x}W\ket{x}=\pm(-1)^{\sum_{i\in S}x_i}$. This function is constant only when $S=\varnothing$, and is nonconstant otherwise. Hence, among operators of the form $W=E^\dag F$ with $E,F\in\calE_t$, the only nontrivial constraints come from Pauli operators consisting of $Z$'s on at most $2t$ coordinates and identities elsewhere. Thus, the relevant number of constraints is $B_2(N,2t)=\sum_{j=0}^{2t}\binom Nj$, up to the harmless inclusion of the identity operator. In particular, for $t=\delta N$ and $N\to\infty$, the exponential growth rate of this constraint count is $H_2(2\delta)$. Therefore, choosing the underlying classical binary code with the Gilbert--Varshamov rate $1-H_2(2\delta)$ yields a quantum code of asymptotic rate $1-2H_2(2\delta)$. This recovers the rate obtained in \cite{movassagh2024constructing}. The present construction has two advantages. First, once the required Tverberg partition is found, the resulting exact QEC conditions hold deterministically, rather than only with high probability over the choice of a random classical code. Second, the partition step is based on finding a Tverberg partition for the relevant constraint vectors. In the linear-distance regime, this remains computationally expensive, but the standard constructive Tverberg approach scales substantially better than the recursive Dines-based implementation used in \cite{movassagh2024constructing}, whose worst-case column growth is triply exponential in $N$.
\hfill$\triangleleft$ \end{remark}

\subsection{Approximate QEC by random partition codes}\label{sec:AQECrandPar}

In this section, we present a simple yet powerful construction of partition codes based on random underlying classical codes, and analyze their AQEC performance. This construction generalizes the construction of \cite{elimelech2026asymptotically} for constant-excitation Fock state codes to general Hilbert spaces indexed by metric spaces. Our performance guarantees improve as we move up the metric--error alignment hierarchy, eventually yielding better rate/error-correction tradeoffs, as demonstrated in Section~\ref{sec:Exampels}.

The random code construction has two main advantages. First, it replaces the costly and impractical procedure of finding a Tverberg partition, discussed in Remark~\ref{rem:Tver}, with a low-complexity averaging process. Second, the resulting quantum codes naturally inherit structural properties of the underlying random classical code. This allows us to tailor the classical code distribution so that the induced quantum code satisfies desired properties. This section develops these ideas in a general and abstract form; concrete examples are presented in Section~\ref{sec:Exampels}.

We start by providing sufficient AQEC conditions for partition codes on the first and second metric--error alignment levels. This is formulated in the following Lemma, which is a direct consequence of Lemma~\ref{lem:PartitionCodesProp} combined with Proposition~\ref{prop:AQECcond}: 
\begin{lemma}[Sufficient AQEC conditions for partition codes]\label{lem:SufCondHira}
Let $\calH_X$ be a Hilbert space indexed by a discrete metric space $(X,d)$ with a family of error sets $\mathscr{E}=(\calE_t)_t$. Assume that  $Q$ is a quantum partition code 
in $\calH_X$ obtained from a classical code $C\subset X$ with a canonical basis $\ppp{\ket{c_i}}_i$. Assume that $t<d(C)$.
\begin{enumerate}\label{cor:SuffcientAQECpartition}
    \item  If $(\calH_X,d,\mathscr{E})\in \mathscr{L}_1$, then $Q$ is $\varepsilon$-AQEC for the error set $\calE_t$ provided that there exist numbers $(\lambda_{E,F})_{E,F\in \calE_t}$
    \begin{equation}
        \max_{\substack{i\in [K]\\E,F\in \calE_t} } |\bra{c_i} E^\dag F \ket{c_i}-\lambda_{E,F}|\leq \frac{\varepsilon^2}{2K |\calE_t|^2}\label{eq:FirstLevAQEC}.
    \end{equation}
    \item  If $(\calH_X,d,\mathscr{E})\in \mathscr{L}_2$ then $Q$ is $\varepsilon$-AQEC for the error set $\calE_t$ provided that there exist numbers $(\lambda_{E})_{E\in \calE_t}$
    \begin{equation}
        \max_{\substack{i\in [K]\\E\in \calE_t} } |{\bra{c_i} E^\dag E \ket{c_i}-\lambda_{E}}|\leq \frac{\varepsilon^2}{2K |\calE_t|}\label{eq:SecondLevAQEC}.
    \end{equation}
\end{enumerate}
\end{lemma}

\begin{construction}\label{const:randomParition}[Random partition codes] 
Let $(X,d)$ be a discrete metric space and $\calH_X$ be indexed by $X$. Let $\mu$ be a probability measure on $X$ (equipped with the sigma-algebra of all subsets $\sigma=2^X$). Consider the random classical code $\s{C}\subseteq X$ of size $L=T  K$ generated by $L$ i.i.d. random elements $\s{x}_0,\dots, \s{x}_{L-1}$ distributed according to $\mu$. The corresponding random quantum partition code $\s{Q}^{\s{P}}_{K,L}(\mu)$ is obtained from the partition $\s{C}_0,\dots,\s{C}_{K-1}$ and coefficients $\alpha_{\s{x}}$ given by 
\[\s{C}_k=\ppp{\s{x}_{T  k+j} ~:~j=0,\dots, T-1}, \quad \alpha_{\s{x_i}}=\frac{1}{\sqrt{T}}.\]
Concretely, $\s{Q}^{\s{P}}_{K,L}(\mu)$ is the code spanned by the vectors of the form
\[\ket{\s{c}_k}=\frac{1}{\sqrt{T}}\sum_{\s{x}\in \s{C}_k} \ket{\s{x}}, \quad k=0,1,\dots, K-1.\]
\end{construction}

While the construction above is based on random codes, it also extends to the case where the classical underlying code is deterministic. For example, as mentioned in Example~\ref{ex:CSS}, if the random code in Construction~\ref{const:randomParition} is replaced by a linear code, and the partition elements are chosen to be cosets induced by another linear code, then one recovers the CSS construction. From this perspective, Construction~\ref{const:randomParition} can be viewed as an i.i.d. random analog of the CSS construction: the linear/coset structure is replaced by an i.i.d. sampling procedure, making the construction available in arbitrary Hilbert spaces indexed by metric spaces even when no linear structure is present.

\begin{definition}[Packing probability, metric-basis error intensity]\label{def:parametersForApprox} Let $(X,d)$ be a discrete metric space with a probability measure $\mu$ and $\mathscr{E}=(\calE_t)_t$ be an 
error-set family on $\calH_X$. We define the following quantities 
\begin{enumerate}
    \item The $t$-\textit{packing probability} of $\mu$ is defined to be
    \[p_\mu(L,t)=\P\pp{d(\s{C}) >t}, \quad \s{C}=\set{\s{x}_0,\dots,\s{x}_{L-1}},\]
    where $\s{x}_0,\dots,\s{x}_{L-1}$ are $\mu$-i.i.d. random elements. By a slight abuse of notation, for $t=0$ we denote by $p_\mu(L,0)$ the probability that $\s{x}_0,\dots \s{x}_{L-1}$ are all distinct. 
    \item The \textit{metric-basis error intensity} of $\calE_t$ is defined as
    \[\kappa_t(\mathscr{E},X)=\sup_{\substack{E\in \calE_t\\ x\in X}}\bra{x}E^\dag E \ket{x}.\]
\end{enumerate}
    
\end{definition}

\begin{theorem}[AQEC guarantees for random partition codes]\label{th:AQECrandomPar}
Let $(X,d)$ be a discrete metric space equipped with a probability measure $\mu$, and let $\mathscr{E}=(\calE_t)_t$ be an error-set family on $\calH_X$ with $|\calE_t|=M_t$. Let $\varepsilon>0$, $K,T\in \N $ be fixed and $L=KT$. Then $\s{Q}^{\s{P}}_{K,L}(\mu)$ is a $K$-dimensional $\varepsilon$-AQEC w.r.t. $\calE_t$ with probability at least $p_t$ where: 
\begin{enumerate}
   \item If $(\calH_X,d,\mathscr{E})\in \mathscr{L}_2$: 
   
   \begin{equation}
       \label{eq:probgoodL2}
        p_t=p_{\mu}(L,t)-2M_tK  \exp\p{-\frac{L\varepsilon^4 }{2K^3M_t^2\kappa_t(\mathscr{E},X)^2}}.
   \end{equation}
   \item If $(\calH_X,d,\mathscr{E})\in \mathscr{L}_1$:

   \begin{equation}
       \label{eq:probgoodL1} p_t=p_{\mu}(L,t)-2M_t^2K  \exp\p{-\frac{L\varepsilon^4 }{8K^3M_t^4\kappa_t(\mathscr{E},X)^2}}.\end{equation}
\end{enumerate}
    
\end{theorem}
\begin{proof}[Proof sketch] We give the complete proof in Appendix~\ref{app:RandomPartit} and sketch the main idea here. By Lemma~\ref{lem:SufCondHira}, as long as the underlying classical random code satisfies $d(\s{C})>t$, it is sufficient to show that there exist scalars $\lambda_{E,F}$ such that the $\lambda_{E,F}$ and $\bra{\s{c}_i }E^\dag F \ket{\s{c}_i}$ are sufficiently close (where for the $\mathscr{L}_2$ case only the case $E=F$ needs to be checked). We show that under the distance assumption, these inner products reduce to averages of i.i.d. random variables of the form $\bra{\s{x}}E^\dag F \ket{\s{x}}$, which naturally concentrate around the expectation $\E_{\s{x}\sim\mu}\pp{\bra{\s{x}}E^\dag F \ket{\s{x}}}$, which we define to be $\lambda_{E,F}$. In the absence of further assumptions on the distribution $\mu$, we control the concentration around these expectations using Hoeffding's inequality, which gives the error exponents in \eqref{eq:probgoodL2} and \eqref{eq:probgoodL1}.
\end{proof}

\begin{remark}[Variance-sensitive refinements]\label{rem:varianceRefine}The proof of Theorem~\ref{th:AQECrandomPar} uses only minimal assumptions on the sampling measure $\mu$. In particular, the concentration step is based on Hoeffding's inequality and therefore depends only on a uniform bound on the random variables $\bra{\s{x}}E^\dag F\ket{\s{x}}$, where $\s{x}\sim\mu$. This makes the result broadly applicable, but also somewhat conservative. If additional information on $\mu$ is available, and in particular if one can control the $\mu$-variance of the variables $\bra{\s{x}}E^\dag F\ket{\s{x}}$, then the Hoeffding step can potentially be replaced by sharper variance-sensitive concentration bounds. Such refinements may substantially improve the exponents in \eqref{eq:probgoodL2} and \eqref{eq:probgoodL1}, and consequently enlarge the range of parameters for which the random partition construction succeeds with high probability. In particular, they may lead to better rate--error-correction tradeoffs. \hfill$\triangleleft$ \end{remark}

\begin{remark}[Role of the alignment assumptions]\label{rem:alignmentNeeded}
The assumptions $(\calH_X,d,\mathscr{E})\in\mathscr{L}_1$ or $(\calH_X,d,\mathscr{E})\in\mathscr{L}_2$ are essential for the concentration argument in Theorem~\ref{th:AQECrandomPar}. Indeed, the minimum-distance assumption on the underlying classical code is not sufficient by itself. Without the $\mathscr{L}_1$ or $\mathscr{L}_2$ structure, the inner products $\bra{\s{c}_i}E^\dag F\ket{\s{c}_j}$ may contain many nonzero cross terms of the form $\bra{\s{x}}E^\dag F\ket{\s{y}}$, with $\s{x}\in C_i$ and $\s{y}\in C_j$. Since each codeword is normalized by a factor $1/\sqrt{T}$, these terms appear with an overall $1/T$ normalization, but there are $T^2$ possible cross terms. Thus, their total contribution may scale as $\Theta(T)=\Theta(L/K)$, with typical deviations of order $\Theta(\sqrt{T})=\Theta(\sqrt{L/K})$. Consequently, these quantities do not generally concentrate around the desired scalars as $K$ and $L$ grow. The alignment assumptions are precisely the conditions that prevent this uncontrolled accumulation of cross terms and allow the proof to reduce the relevant quantities to averages of i.i.d. diagonal terms. \rightline{$\triangleleft$} \end{remark}

\subsubsection*{Inherited typicality of random partition codes} 

The random partition construction has an additional feature beyond its AQEC guarantees: it can preserve structural properties of the underlying random classical labels. This is useful because, in many settings, one wants the basis states supporting the code to satisfy some additional regularity or physical constraints that extend beyond the code distance. The definitions below formalize this idea. A classical property is simply a subset of admissible labels in the indexing space $X$, and the corresponding quantum notion asks that the code be supported, exactly or approximately, on basis vectors indexed by the labels with that property. The main point is that if a property holds with high probability under the sampling distribution used to generate the random partition, then the resulting random partition code will approximately satisfy the corresponding quantum property with high probability.

\begin{definition}[Classical properties]
 Let $X$ be a discrete space. A (classical) property in $X$ is a subset of elements $\mathscr{P}\subseteq X$. If $\mu$ is a measure on $X$, we denote the $\mu$-measure of $\mathscr{P}$ by $p_{\mu}(\mathscr{P})$.
\end{definition}

\begin{definition}[Quantum properties]\label{def: Quantum properties}
    Let $\calH_X$ be a Hilbert space indexed by a space $X$ and let $\mathscr{P}\subseteq X$ be a property.  A quantum code $Q\subset \calH_X$ is said to satisfy $\mathscr{P}$ if 
\[Q\subseteq \Span\ppp{\ket{x} : x\in \mathscr{P}} :=  \calH_{\mathscr{P}}.\]
    For $\varepsilon>0$, a code $Q$ is said to $\varepsilon$-satisfy $\mathscr{P}$ if 
    \begin{equation}
        \inf_{\ket{c}\in Q }\bra{c} \Pi_{\mathscr{P}}\ket{c}\geq1- \varepsilon,\label{eq:propCondition}
    \end{equation}
  where $\Pi_{\mathscr{P}}$ is the orthogonal projection on $\calH_{\mathscr{P}}$ and minimization runs over normalized states. Equivalently, every normalized codeword has overlap at least $1-\varepsilon$ with the subspace $\calH_{\mathscr{P}}$. In operational terms, when measuring a codeword using the two-outcome measurement $\ppp{\Pi_{\mathscr{P}},I-\Pi_{\mathscr{P}}}$, the outcome corresponding to $\mathscr{P}$ occurs with probability at least $1-\varepsilon$. A sequence of codes $(Q_N)_N$ is said to asymptotically approximately satisfy the properties $(\mathscr{P}_N)_N$ if there exists $\varepsilon_N\to 0$ such that $Q_N$ $\varepsilon_N$-satisfies $\mathscr{P}_N$.
\end{definition}
\begin{example}
    A basic example to keep in mind is the balancedness of $q$-ary strings. Let $X=[q]^N$, and for $\ux\in[q]^N$ and $a\in[q]$, let $N_a(\ux)$ denote the number of coordinates of $\ux$ equal to $a$. For fixed $\eta>0$, one may consider the property
\[
\mathscr{P}_N^{\s{bal}} := \ppp{\ux\in[q]^N~:~\abs{N_a(\ux)-N/q}\leq N^{1/2+\eta}\text{ for all }a\in[q]}.
\]
Thus, a quantum code satisfies $\mathscr{P}_N^{\s{bal}}$ if it is spanned by computational basis states whose symbol frequencies are all close to uniform, and it $\varepsilon$-satisfies $\mathscr{P}_N^{\s{bal}}$ if every codeword has at least $1-\varepsilon$-fraction of its mass on such balanced basis states.
This is operationally meaningful because it prevents the code from concentrating on highly biased basis configurations, and is closely related to regularity conditions such as constant-composition or constant-excitation constraints that often play a useful role in code design. We will return to this example in Section~\ref{sec:QuditsConst}, where the random qudit partition codes are generated from the uniform distribution on $[q]^N$.
\hfill$\triangleleft$ \end{example}
Other examples involving different ambient spaces appear in Sec.~\ref{sec:Exampels} below. In particular, for constant-excitation Fock-state codes and permutation-invariant codes, an important property is bounded per-mode occupancy, which was already used in \cite{elimelech2026asymptotically} and will be discussed further in Section~\ref{sec:CEFCex}. The following theorem gives a general mechanism for transferring such high-probability classical properties into approximate quantum properties of the resulting random partition code.

\begin{theorem}[Typicality of random partition codes]\label{th:TypicallityProperties}
Let $\calH_X$ be a Hilbert space indexed by a discrete space $(X,d)$, equipped with a probability measure $\mu$ and let $0<\xi<\frac{1}{2}$ be a constant. For a property $\mathscr{P}\subseteq X$ and fixed $K,T\in \N $ such that $L=KT$, the random partition code $\s{Q}^{\s{P}}_{K,L}(\mu)$  $\varepsilon$-satisfies $\mathscr{P}$ with probability at least $p_\mu(L,0)-2Ke^{-2T^{1-2\xi}}$, where 
\[\varepsilon=1-p_{\mu}(\mathscr{P}) +T^{-\xi}.\]
\end{theorem}

The proof of Theorem~\ref{th:TypicallityProperties} (given in Appendix~\ref{app:Typicallityproof}) relies on a concentration argument. Given that all elements of the underlying random code $\s{C}$ are distinct, the projection coefficients $\bra{c} \Pi_{\mathscr{P}}\ket{c}$ of codewords reduce to convex combinations of those of the basis elements, which implies that it suffices to verify \eqref{eq:propCondition} for the basis of the code. For a basis state $\ket{\s{c}_i}$, we show that $\bra{\s{c}_i} \Pi_{\mathscr{P}}\ket{\s{c}_i}$ is exactly the fraction of elements in $\s{C}_i$ that satisfy property $\mathscr{P}$, and it concentrates around $p_{\mu}(\mathscr{P})$.

\section{Asymptotically good families from partition codes}\label{sec:Exampels}
In this section, we apply the general partition-code framework to construct asymptotically good families of quantum codes in several settings, including $q$-ary tensor-product systems, Majorana fermionic systems, multi-mode constant-excitation Fock spaces, and permutation-invariant subspaces, under a variety of noise models. We focus on the asymptotic regime in which both the rate and the correctable distance scale linearly with the natural system-size parameter, such as the number of qudits, modes, or Majorana degrees of freedom.

For each model, we derive existential exact QEC guarantees from Theorem~\ref{th:exactQECpatition}, and complement them with AQEC guarantees for the explicit random partition construction. As expected, the random constructions generally yield weaker rate/distance tradeoffs than the existential exact QEC results since they rely on a robust averaging mechanism and relatively coarse concentration bounds. Their advantage, however, is that they are efficient and explicit, whereas the exact constructions obtained through the Tverberg-type existence argument are computationally infeasible in general, as discussed in Remark~\ref{rem:Tver}.

Thus, the goal of this section is not to approach the best possible rate/distance tradeoff in each individual model, but rather to demonstrate the generality and robustness of quantum partition codes across a broad range of Hilbert-space geometries and noise mechanisms. More refined, model-specific analyses can potentially improve both the exact and approximate QEC guarantees. Nevertheless, the general framework already yields several new consequences, including asymptotically good 
codes for quantum deletions, non-stabilizer Majorana fermionic codes, codes for Rydberg atom chains, and constant-excitation Fock-state codes protecting against number-shift and phase-rotation noise. In certain cases, it also provides evidence of separation: for qudit codes, correcting linearly many amplitude damping errors is easier than correcting general bounded-weight errors.

\subsection{Asymptotic analysis}\label{sec: Asymptotic analysis}
We now describe the asymptotic arguments that enter all of our examples listed below. In each case, we consider a sequence of Hilbert spaces $\calH_{X_N}$ indexed by discrete metric spaces $(X_N,d)$, equipped with probability measures $\mu_N$ which are typically uniform. The noise model is given by an indexed error-set family $\mathscr{E}=(\calE_t)_t$ satisfying either $(\calH_{X_N},d,\mathscr{E})\in\mathscr{L}_1$ or $(\calH_{X_N},d,\mathscr{E})\in\mathscr{L}_2$. We study the linear-distance regime $t=\delta N$, where $\delta>0$ is fixed and $N$ tends to infinity. 

We consider quantum partition codes of dimension $K$ based on underlying classical codes of size $L$.  For each model, we use existence results for classical codes listed in Appendix~\ref{app:classicalCodes}. 
Letting $\delta$ be a value of the relative distance, let $R(\delta)$ denote the corresponding classical rate guarantee, defined formally in the Appendix.  In the estimates used below, there are two closely related random-code statements. The first is an {\em expurgated version}: starting with an i.i.d. random classical code and removing ``bad codewords'', with high probability we obtain a classical code of distance at least $\delta N$ and rate $R(\delta)$. The second is a pure i.i.d. version, with no expurgation, which gives the same distance guarantee at half this rate. This distinction is important for what follows. The Tverberg-based exact-QEC construction of Theorem~\ref{th:exactQECpatition} only requires the existence of an underlying classical code with a specified distance, and therefore we use the expurgated rate $R(\delta)$. In contrast, the random-AQEC proof of Theorem~\ref{th:AQECrandomPar} relies on concentration over independent samples inside each partition block; the expurgation step destroys this i.i.d. structure and introduces dependencies. Therefore, the AQEC random-partition analysis uses the pure i.i.d. rate $R(\delta)/2$. Thus, for the exact-QEC analysis, the classical distance constraint is 
\[ \frac{\log L}{N}\leq R(\delta)+o(1), \]
whereas for the random-AQEC analysis, it is
\[ \frac{\log L}{N}\leq \frac{R(\delta)}{2}+o(1). \]

The partition-code construction takes $L=KT$, and throughout this section, we parametrize the size of the underlying classical code
by $L=K^{\beta}$ for some $\beta\geq1$. Thus, if $\calK :=  \frac{\log K}{N},$ then the classical distance constraint becomes
\[ \beta\calK\leq \Big\{R(\delta),\frac{R(\delta)}{2}\Big\}+o(1) \]
for the exact and random QEC constructions, respectively.

The second constraint comes from the QEC condition itself. For exact QEC, it is obtained from Theorem~\ref{th:exactQECpatition}; for AQEC, it is obtained from Theorem~\ref{th:AQECrandomPar}. In all models considered below, the metric-basis error intensity from Definition~\ref{def:parametersForApprox} satisfies $\kappa_t(\mathscr{E},X_N)=1$. Hence, the relevant noise parameter is the normalized error-set exponent
\[
M_{\delta} :=  \frac{\log|\calE_{\delta N}|}{N},
\]
and its asymptotic limit. 
Thus, each example reduces to two ingredients: the classical Gilbert--Varshamov rate $R(\delta)$ and the error-set growth exponent $M_{\delta}$. Combining the classical distance constraint with the exact or approximate QEC constraint gives an upper bound on the achievable quantum rate $\calK=\log K/N$. This is the form in which we state the results below. To match standard conventions, logarithms are taken in base $q$ for $q$-ary qudit codes, and in base $e$ for the remaining models; this convention is used consistently for $R(\delta)$, $\calK$, and $M_{\delta}$.

\subsubsection*{Existence of exact QEC partition codes} 
In this subsection and the one that follows, we formulate conditions for asymptotically good codes that will be repeatedly used for the concrete code families and noise processes considered in the examples below. Suppose that $(\calH_{X_N},d,\mathscr{E})\in\mathscr{L}_i$, with $i\in\ppp{1,2}$. The existence result of Theorem~\ref{th:exactQECpatition} requires the underlying classical code to have size at least
\[ L\geq (K-1)\p{|\calE_t|^{3-i}+1}+1. \]
In the linear-distance regime $t=\delta N$, with $L=K^\beta$ and $\calK=\log K/N$, this condition becomes
\[ \beta\calK\geq \calK+(3-i)M_\delta+o(1).\]
On the other hand, the classical distance requirement gives
\[ \beta\calK=\frac{\log L}{N}\leq R(\delta)+o(1). \]
These two constraints are simultaneously satisfied whenever $ \calK+(3-i)M_\delta<R(\delta),$
up to vanishing $o(1)$ terms. Equivalently, the exact partition code construction gives a positive asymptotic quantum rate for every
\begin{equation}
     (\calH_X,d,\mathscr{E})\in \mathscr{L}_i,\quad  \text{and}\quad \calK\leq R(\delta)-(3-i)M_\delta.\label{eq:AsympRateExact}
\end{equation} Thus, at alignment level $\mathscr{L}_2$ one obtains the exact QEC rate condition $\calK<R(\delta)-M_\delta$, while at alignment level $\mathscr{L}_1$ one obtains $\calK<R(\delta)-2M_\delta$. 

\subsubsection*{AQEC guarantees for random partition codes } For the random AQEC construction, the argument is similar but the interpretation is different. We are no longer using Theorem~\ref{th:exactQECpatition} to prove the existence of an exact-QEC partition through a Tverberg-type argument. Instead, we analyze the explicit random partition code $\s{Q}^{\s{P}}_{K,L}(\mu)$ using Theorem~\ref{th:AQECrandomPar}. The first requirement is that the underlying i.i.d. classical code has distance at least $\delta N$. For the pure i.i.d. ensemble, without expurgation, the estimates of Appendix~\ref{app:classicalCodes} give $p_{\mu}(L,\delta N)=1-o(1)$ provided that
\[ \frac{\log L}{N}=\beta\calK\leq \frac{R(\delta)}{2}+o(1). \]
The second requirement is that the concentration error term in Theorem~\ref{th:AQECrandomPar} vanishes. Since in all examples below $\kappa_t(\mathscr{E},X_N)=1$ and $|\calE_{\delta N}|=\exp\p{NM_\delta+o(N)}$, this imposes a second constraint relating $\beta\calK$, $\calK$, and $M_\delta$. At level $\mathscr{L}_2$, the concentration term vanishes, even for an exponentially decaying AQEC parameter $\varepsilon_N$, whenever
\[(\beta-3)\calK\geq 2M_\delta+o(1). \]
Combining this with $\beta\calK\leq R(\delta)/2+o(1)$ gives the AQEC rate condition
\[3\calK+2M_\delta\leq \frac{R(\delta)}{2}. \]
Equivalently, for every arbitrarily small $\xi>0$,  with probability at least $1-e^{-\exp\p{c\,  \xi}}$ (for some universal constant $c$), 
 the random partition construction gives $\varepsilon_N$-AQEC codes with $\varepsilon_N\to0$ exponentially fast whenever
\begin{equation}
    (\calH_X,d,\mathscr{E})\in \mathscr{L}_2,\quad  \text{and}\quad \calK\leq \frac{1}{3}\p{\frac{R(\delta)}{2}-2M_\delta}-\xi, \quad \xi>0. \label{eq:AsympAQECcondL2}
\end{equation}
At level $\mathscr{L}_1$, the same calculation uses a stronger concentration term in \eqref{eq:probgoodL1}. In this case, the required condition is
\[ (\beta-3)\calK\geq 4M_\delta+o(1), \]
and therefore the combined rate condition becomes
\[ 3\calK+4M_\delta\leq \frac{R(\delta)}{2}. \]
Thus, for every arbitrarily small $\xi>0$, with probability $1-e^{-\exp\p{\Theta (\xi)}}$, the random partition construction gives $\varepsilon_N$-AQEC codes with exponentially decaying $\varepsilon_N$ whenever
\begin{equation}
    (\calH_X,d,\mathscr{E})\in \mathscr{L}_1,\quad  \text{and}\quad \calK\leq \frac{1}{3}\p{\frac{R(\delta)}{2}-4M_\delta}-\xi,  \quad \xi>0. \label{eq:AsympAQECcondL1}
\end{equation}

Throughout the asymptotic statements below, we suppress vanishing $o(1)$ terms in the rate constraints. We also omit the arbitrarily small slack parameter $\xi>0$: an inequality of the form $\calK<R$ means that $\calK$ is chosen strictly below $R$ asymptotically, so that the corresponding $o(1)$ terms and $\xi$ can be absorbed. Similarly, in the random coding statements, we do not explicitly repeat the quantitative success probability, such as bounds of the form $1-e^{-\exp(\Theta(\xi))}$ arising from the concentration estimates. Instead, we simply say that the corresponding random partition code satisfies the claimed AQEC guarantee \textit{with high probability.}

\subsection{Qudit codes}\label{sec:QuditsConst}
We begin with $q$-ary quantum codes on $\calH_q^{\otimes N}$, considering the noise models introduced in Section~\ref{sec:qudits}. In all cases, the indexing space is $X=[q]^N$, while the relevant metric depends on the noise model under consideration. Since exact and approximate rates for bounded-weight errors were already studied in \cite{movassagh2024constructing}, and further discussed in Remark~\ref{rem:QuariLimited}, we mainly focus below on deletion errors, amplitude damping noise, and the special case of codes for Rydberg atom chains. For completeness, we first recall the bounded-weight case. Let $d$ be the scaled Hamming metric, as in Example~\ref{ex:HWsecondlevel}, and let $\mathscr{E}=(\calE_t)_t$ be the family of support-limited HW errors defined in \eqref{eq:HWoperatorsT}. By Lemma~\ref{lem:random-qary-hamming}, i.i.d. random classical codes over $[q]^N$ attain the classical Gilbert--Varshamov bound with high probability. More generally, any family of classical codes attaining this bound can be used as the underlying code in Theorem~\ref{th:exactQECpatition}. As shown in \cite{movassagh2024constructing} and explained in Remark~\ref{rem:QuariLimited}, such codes give exact QEC partition codes correcting $\delta N$ bounded-weight errors with asymptotic rate
\begin{equation}
     \calK\leq 1-2H_q(2\delta). \label{eq:PartitionGeneralErrors}
\end{equation}
For $q=2$, this rate matches the tradeoff attained by the CSS construction of \cite{calderbank1996good}. Thus, in the bounded-weight setting, even structureless GV-type classical codes give partition codes with the same asymptotic rate as the highly structured CSS codes, which correspond to the nested linear codes coset
construction described in Example~\ref{ex:CSS}. 

For the random AQEC construction, note that
\[|\calE_{\delta N}|=\sum_{i=0}^{\delta N}\binom{N}{i}(q^2-1)^i, \]
and therefore, with logarithms taken in base $q$,
\[
M_\delta=\frac{\log_q|\calE_{\delta N}|}{N}=2H_{q^2}(\delta)+o(1),
\]
where $H_{q^2}(\delta)$ is the $q^2$-ary entropy function defined in \eqref{eq:qaryEntropy}. 
Since Example~\ref{ex:HWsecondlevel} shows that $(X,d,\mathscr{E})\in\mathscr{L}_1$, the general AQEC condition \eqref{eq:AsympAQECcondL1}, together with the GV-rate estimate from Lemma~\ref{lem:random-qary-hamming}, gives random partition codes with asymptotic AQEC rate
\[
\calK<\frac{1}{3}\p{\frac{R^{\s{H}}(\delta)}{2}-8H_{q^2}(\delta)}=\frac{1}{3}\p{\frac{1-H_q(2\delta)}{2}-8H_{q^2}(\delta)},
\]
whenever the right-hand side is positive. By continuity, this condition is nonvacuous for $\delta>0$ in some
neighborhood of zero. As in the exact-QEC setting (see Remark~\ref{rem:QuariLimited}), this AQEC rate is not expected to be optimal: a more delicate analysis of the bounded-weight model could improve the constants and the resulting tradeoff, but we do not pursue this direction here. 

\subsubsection*{Balancedness and other typical properties}
The random partition codes considered in this qudit subsection are generated from i.i.d. samples from the uniform measure on $[q]^N$. Hence, beyond their AQEC guarantees, they also inherit any property that holds with high probability for a uniformly random word. This follows directly from Theorem~\ref{th:TypicallityProperties}: if $\mathscr{P}_N\subseteq[q]^N$ satisfies $p_{\mu_N}(\mathscr{P}_N)=1-o(1)$, then the corresponding random partition code asymptotically approximately satisfies $\mathscr{P}_N$ with high probability, since the projection of each logical codeword onto $\calH_{\mathscr{P}_N}$ is controlled by the fraction of sampled labels satisfying $\mathscr{P}_N$.
A natural example is balancedness. Fix any $\varepsilon>0$. For $a\in[q]$ and $\ux\in[q]^N$, let $N_a(\ux)=|\ppp{i\in[N]:x_i=a}|$, and define
\[\mathscr{P}_N^{\s{bal}} := \ppp{\ux\in[q]^N~:~\abs{N_a(\ux)-N/q}\leq N^{1/2+\varepsilon}\text{ for all }a\in[q]}.\]
A uniformly random word in $[q]^N$ satisfies $\mathscr{P}_N^{\s{bal}}$ with probability $1-o(1)$. Therefore, by Theorem~\ref{th:TypicallityProperties}, the random AQEC partition codes constructed below asymptotically approximately satisfy $\mathscr{P}_N^{\s{bal}}$ with high probability. Operationally, measuring whether a quantum codeword is supported on balanced basis vectors 
succeeds with probability $1-o(1)$ for every normalized codeword.

This type of inherited typicality is useful because it allows the random construction to retain structural features of the uniform ensemble extending beyond distance and AQEC properties. Other high-probability properties can be inherited in the same way. For example, one may require local-pattern typicality, namely that every fixed word $\ua\in[q]^\ell$ appears with frequency $q^{-\ell}+o(1)$ among length-$\ell$ windows, for any fixed $\ell$. One may also require the absence of anomalously long constant runs, e.g., that the longest run of any fixed symbol is $O(\log N)$. These properties are not consequences of balancedness alone, but they are typical for uniform words and therefore are inherited by the random partition code in the same approximate sense.

\subsubsection{Deletion errors}\label{sec:ExDelEr}
We next consider the deletion noise model introduced in Section~\ref{sec:DelErrors}. Recall that the $t$-deletion error set is
\[ \calE_{t}^{\s{Del}}=\ppp{D_{I,\ux}~:~I\in\binom{[N]}{t},\ \ux\in[q]^t},\quad D_{I,\ux}=\bra{\ux}_I\otimes I_{[N]\setminus I}. \]
For a fixed subset $I\subseteq[N]$, the operators $\ppp{D_{I,\ux}}_{\ux\in[q]^t}$ are precisely the Kraus operators of the channel that traces out the coordinates in $I$. Hence, by Observation~\ref{obs:Convexity}, any random $t$-deletion channel, which deletes a subset $I\in\binom{[N]}{t}$ according to an arbitrary probability distribution, is $\calE_t^{\s{Del}}$-controlled. The relevant metric on $X=[q]^N$ is the deletion metric $d_{\s{Del}}$ defined in \eqref{eq:dDel}, and by Example~\ref{ex:DelL1} we have $(\calH_X,d_{\s{Del}},\mathscr{E}^{\s{Del}})\in\mathscr{L}_1$. For $t=\delta N$, the size of the deletion error set is
\[|\calE_{\delta N}^{\s{Del}}|=\binom{N}{\delta N}q^{\delta N}, \]
and therefore, with logarithms in base $q$,
\begin{equation}
    M^{\s{Del}}(\delta) := \frac{\log_q|\calE_{\delta N}^{\s{Del}}|}{N}=\delta+\frac{H_2(\delta)}{\log_2 q}+o(1). \label{eq:Mdel}
\end{equation}
The corresponding classical deletion rate from Lemma~\ref{lem:iidCodesDel} is
\begin{equation}
    R^{\s{Del}}(\delta) := \max\p{\frac{\max\p{1-\delta-\gamma_q,0}^2}{\ln q},1+\delta-2H_q  (\delta)},\label{eq:Rdel}
\end{equation}
where $(\gamma_q)_q$ are the Chv\'atal--Sankoff constants; see Definition~\ref{def:Chvvatal-SankoffConst}. Combining this with the exact-QEC asymptotic condition \eqref{eq:AsympRateExact} at level $\mathscr{L}_1$ gives exact-QEC partition codes correcting $\delta N$ deletions with asymptotic rate 
   \[
   \calK<R^{\s{Del}}(\delta)-2M^{\s{Del}}(\delta).
   \]
Similarly, applying the random-AQEC condition \eqref{eq:AsympAQECcondL1} gives random partition codes which are $\varepsilon_N$-AQEC for $\calE_{\delta N}^{\s{Del}}$, with $\varepsilon_N\to0$ exponentially fast and with high probability, whenever
\[
\calK<\frac{1}{3}\p{\frac{R^{\s{Del}}(\delta)}{2}-4M^{\s{Del}}(\delta)}=\frac{1}{3}\p{\frac{1}{2}\max\p{\frac{\max\p{1-\delta-\gamma_q,0}^2}{\ln q},1+\delta-2H_q  (\delta)}-4\delta-\frac{4H_2(\delta)}{\log_2 q}}.
\]
Since $R^{\s{Del}}(0)=1$ and $M^{\s{Del}}(0)=0$, both displayed rate bounds are positive for all sufficiently small $\delta>0$. In particular, the random partition construction gives asymptotically good quantum deletion codes with high probability. Figure~\ref{fig:DelRatesQary} illustrates the exact achievable rate bound for $q=2,3,5$. In each case, the curve remains positive on a nontrivial interval of deletion fractions, and the marked point on the horizontal axis indicates the first zero of the corresponding exact rate function.

\begin{figure}[t]
\centering
\begin{overpic}[width=.72\linewidth]{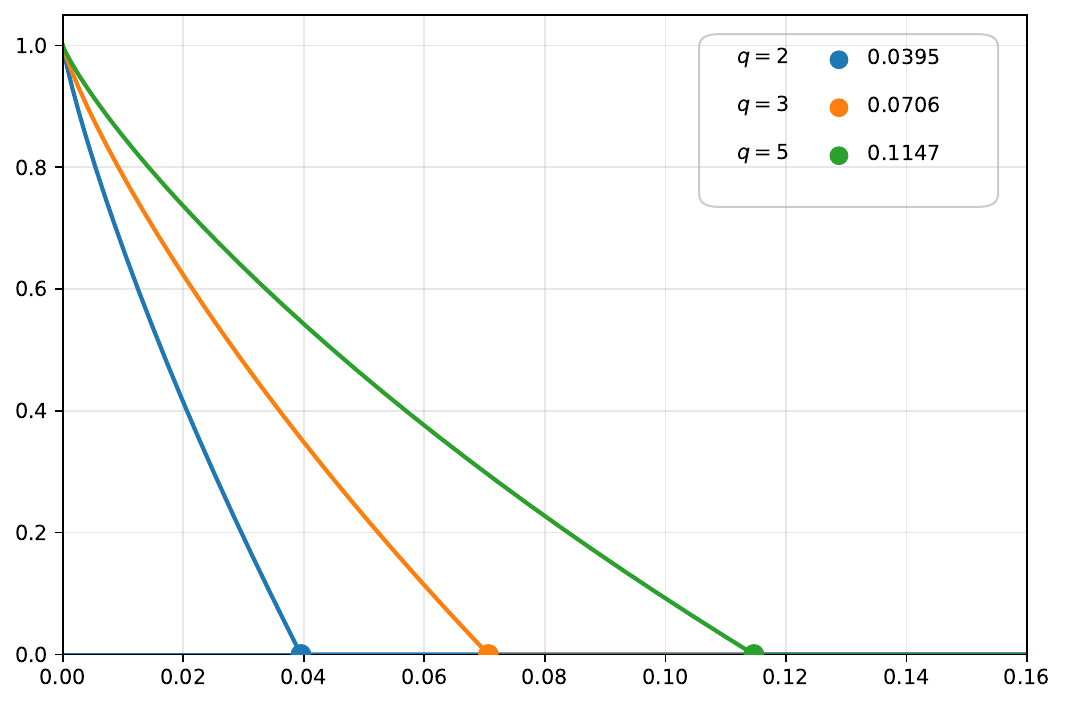}
\put(10,67){\small Exact achievable rate bound for $q$-ary quantum deletion codes}
\put(38,-2){\small{ deletion fraction $\delta$}}
\put(-3,26){\rotatebox{90}{\small{possible rate $\calK$}}}
\end{overpic}
\caption{Exact achievable rate bound for $q$-ary quantum deletion codes. The colored dots marked on the horizontal axis indicate the first positive zero of the corresponding exact rate function for $q=2,3,5$.}
\label{fig:DelRatesQary}
\end{figure}

\subsubsection{Amplitude damping errors}
We now consider qudit amplitude damping errors. Recall from Example~\ref{ex:AD} that the amplitude damping error set $\calE_t^{\s{AD}}$ given in \eqref{eq:ADerrorSet} controls all truncated amplitude damping channels, in the sense of Definition~\ref{def:AQECC}. When these operators are restricted to the $q$-ary space $\calH_q^{\otimes N}$, no mode can lose more than $q-1$ excitations; equivalently, every operator corresponding to the loss of at least $q$ excitations in one of the modes vanishes on $\calH_q^{\otimes N}$. Thus, the relevant error set for $q$-ary codes is
\[ 
\calE_{t,q}^{\s{AD}} := \ppp{\tilde{A}_{\ur}\in\calE_t^{\s{AD}}~:~\norm{\ur}_{\infty}\leq q-1}. 
\]
In particular, the nonzero error operators in $\calE_{t,q}^{\s{AD}}$ are indexed by loss vectors
\[ \Big\{\ur\in[q]^N~:~\sum_{i=1}^N r_i\leq t\Big\},\]
which form the $\ell_1$ ball $B_{q,N}^{\ell_1}(t/2,\uzero)$ with respect to the scaled $\ell_1$ metric of \eqref{eq:ell1Dist}. As shown in Example~\ref{ex:QuditAD}, the qudit amplitude damping model satisfies the second-level alignment condition, namely $(\calH_X,d_1,\mathscr{E}^{\s{AD}})\in\mathscr{L}_2$, where $X=[q]^N$, and $\mathscr{E}^{\s{AD}}=(\calE_{t,q}^{\s{AD}})_t$.

For $t=\delta N$, the exponential growth rate of this ball is computed in Lemma~\ref{lem:ell1BallZero}. With logarithms in base $q$, we write
\[ M^{\s{AD}}(\delta) := \frac{\log_q|\calE_{\delta N,q}^{\s{AD}}|}{N}=B_q(\delta/2)+o(1), \]
where $B_q( )$ is the $\ell_1$ ball exponent derived via the associated entropy maximization problem in Lemma~\ref{lem:ell1BallZero}; see Eq.~\eqref{eq:Bqfunction} . The corresponding classical $\ell_1$ rate function,
denoted by $R^{\ell_1}(\delta)$, is found in Lemma~\ref{lem:ell1randomCodes}; see Eq.~\eqref{eq:ell1GVrate}. 

Combining these quantities with the exact-QEC asymptotic condition \eqref{eq:AsympRateExact} at level $\mathscr{L}_2$ gives exact-QEC partition codes correcting $\delta N$ amplitude damping errors with asymptotic rate
\begin{equation}
    \calK< R_q^{\s{AD}}:= R^{\ell_1}(\delta)-M^{\s{AD}}(\delta)=R^{\ell_1}(\delta)-B_q(\delta/2).\label{eq:ADquareRate}
\end{equation}
Similarly, applying the random-AQEC condition \eqref{eq:AsympAQECcondL2} gives random partition codes which are $\varepsilon_N$-AQEC for $\calE_{\delta N,q}^{\s{AD}}$, with $\varepsilon_N\to0$ exponentially fast and with high probability, whenever
\[ \calK<\frac{1}{3}\p{\frac{R^{\ell_1}(\delta)}{2}-2M^{\s{AD}}(\delta)}=\frac{1}{3}\p{\frac{R^{\ell_1}(\delta)}{2}-2B_q(\delta/2)}.\]
Since $R^{\ell_1}(0)=1$ and $B_q(0)=0$, both displayed rate bounds are positive for all sufficiently small $\delta>0$. In particular, the random partition construction gives asymptotically good $q$-ary amplitude damping codes with high probability.

Observe that the existence of asymptotically good $q$-ary codes against amplitude damping errors is not surprising: any code that corrects $\delta N$ arbitrary bounded-weight errors also corrects $\delta N$ amplitude damping errors. The more meaningful question is whether amplitude damping is strictly easier to handle than general bounded-weight noise. In other words, can one correct $\delta N$ amplitude damping errors at rates that are impossible, or at least not guaranteed, for codes correcting $\delta N$ arbitrary errors? The rate bound in \eqref{eq:ADquareRate} gives a positive answer. Already for $q=3$, the curve $R_q^{\s{AD}}(\delta)$ exceeds the standard $q$-ary quantum Gilbert--Varshamov rate for correcting $\delta N$ arbitrary errors on a nontrivial interval of values of $\delta$; see Figure~\ref{fig:ADq3GV}. For $q=9$, the same amplitude damping rate can even exceed the nondegenerate quantum Hamming bound for correcting $\delta N$ arbitrary errors; see Figure~\ref{fig:ADq9Hamming}. Thus, up to the usual nondegeneracy caveat inherent in the Hamming bound, the partition-code construction produces rates against amplitude damping that cannot be achieved by nondegenerate codes correcting the same number of arbitrary errors. This demonstrates that the structured nature of amplitude damping noise is genuinely reflected in the achievable asymptotic rate.

\begin{figure}[t]
\centering
\begin{overpic}[width=.72\linewidth]{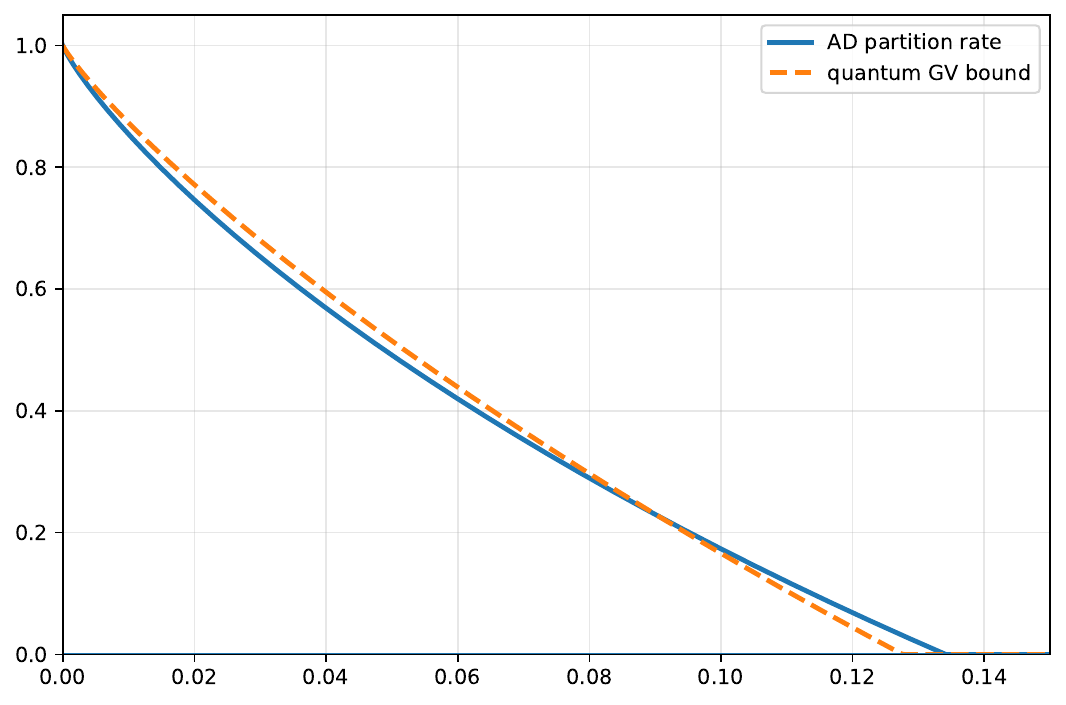}
\put(8,67){\small Exact amplitude damping rate versus quantum Gilbert--Varshamov bound for $q=3$}
\put(38,-2){\small{ error fraction $\delta$}}
\put(-3,26){\rotatebox{90}{\small Possible rate $\calK$}}
\end{overpic}
\caption{Comparison between the exact partition-code amplitude damping rate $R_3^{\s{AD}}(\delta)$ and the standard ternary quantum Gilbert--Varshamov bound for correcting $\delta N$ arbitrary errors. }
\label{fig:ADq3GV}
\end{figure} 

\begin{figure}[t]
\centering
\begin{overpic}[width=.72\linewidth]{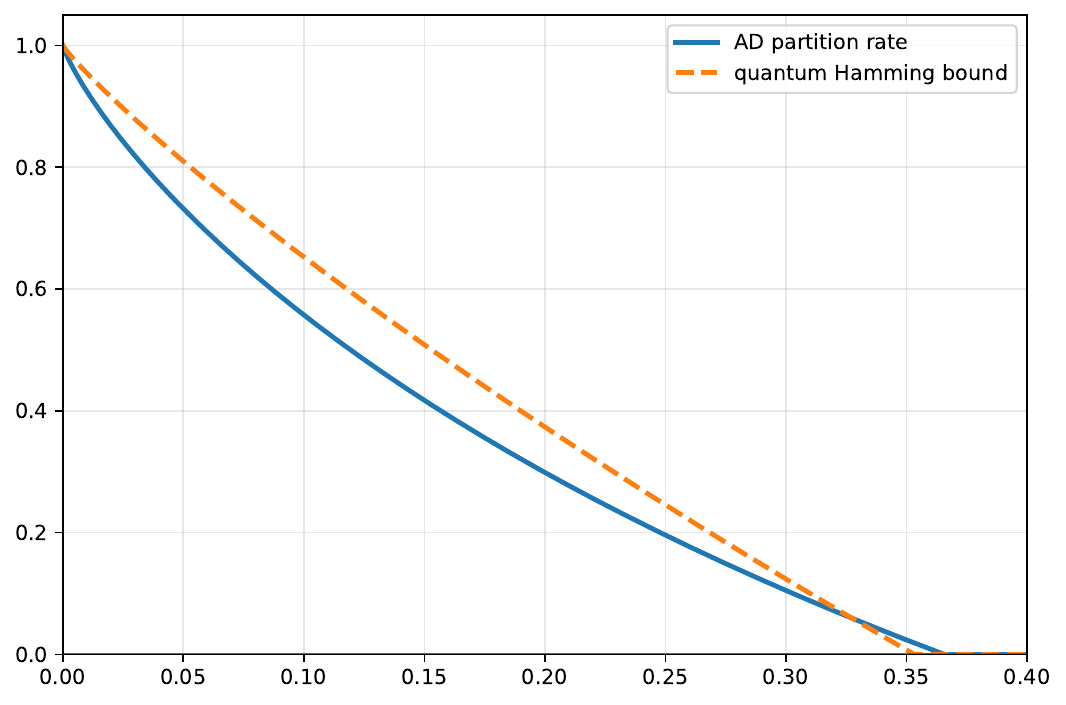}
\put(3,67){\small Exact amplitude damping rate versus quantum Hamming bound for $q=9$}
\put(38,-2){\small{ error fraction $\delta$}}
\put(-3,26){\rotatebox{90}{\small Possible rate $\calK$}}
\end{overpic}
\caption{Comparison between the exact partition-code amplitude damping rate $R_9^{\s{AD}}(\delta)$ and the nondegenerate quantum Hamming bound for correcting $\delta N$ arbitrary errors.}
\label{fig:ADq9Hamming}
\end{figure}
\subsection{Codes for Rydberg atom chains}\label{sec:Rydberg}
In this section, we construct partition codes inside constrained Hilbert spaces motivated by one-dimensional Rydberg atom chains. The general starting point is the quantum subspace correction framework of \cite{pawlak2024quantum}, which studies constraint-satisfying Hilbert spaces as subspaces that can be stabilized, checked, and recovered using QEC-like tools. In particular, their framework includes Ising-type and hard-core constraints, where admissible configurations are defined by local exclusion rules. This naturally raises the question of constructing quantum codes whose basis support already lies inside such a constrained subspace.

A simple and physically meaningful instance of this Ising/hard-core framework is obtained by taking the constraint graph to be a path. The resulting nearest-neighbor exclusion rule says that two adjacent sites cannot be simultaneously occupied. In a one-dimensional Rydberg chain, the same constraint appears naturally from the blockade mechanism, which forbids two adjacent atoms from being simultaneously excited. Thus, the admissible basis states are binary strings with no consecutive 1's, yielding the Fibonacci (or golden-mean-constrained) Hilbert space. Equivalently, if $\Pi_{i,i+1}$ denotes the local projector forbidding simultaneous excitation on sites $i$ and $i+1$, then the physical space is the projected subspace $\Phi_N :=  \Pi (\calH_2^{\otimes N})$, where $\Pi=\prod_{i=1}^{N-1}\Pi_{i,i+1}$. This space is no longer a tensor-product Hilbert space: its dimension grows as $\dim\Phi_N=(\varphi+o(1))^N$, where $\varphi=(1+\sqrt{5})/2$ is the golden ratio (see Lemma~\ref{lem:RLLavBall}). Such constrained chains arise naturally in Rydberg blockade physics and are closely related to hard-core models and Fibonacci-anyon fusion spaces~\cite{feiguin2007interacting,lesanovsky2012interacting,corcoran2025integrable}.
In our notation, the corresponding indexing space is the golden-mean constraint
\[ X_{N}^{\s{GM}} :=  \ppp{\ux\in\ppp{0,1}^{N}~:~(x_i,x_{i+1})\neq(1,1)\text{ for all }i\in[N-1]}, \]
as in \eqref{eq:GMspace}, and $\Phi_N=\calH_{X_N^{\s{GM}}}$. We equip $X_N^{\s{GM}}$ with the scaled Hamming metric \eqref{eq:HammingScaled}. We consider bounded-weight Pauli errors restricted to this constrained space, namely the family $\mathscr{E}=(\calE_t)_t$ from \eqref{eq:HWoperatorsT}. Since $X_N^{\s{GM}}\subseteq\ppp{0,1}^N$, the $\mathscr{L}_1$ alignment property follows directly from the corresponding bounded-weight Pauli example on the full binary Hamming space, and hence $(\calH_{X_N^{\s{GM}}},d,\mathscr{E})\in\mathscr{L}_1$.
We now apply the partition-code analysis in the regime $t=\delta N$. The classical input is supplied by Lemma~\ref{lem:RLLcodes}, which gives random classical codes inside $X_N^{\s{GM}}$ with scaled Hamming distance $2\delta N$ and rate $R^{\s{GM}}(\delta)$ given in \eqref{eq:RLLrateBinary}.
For the bounded-weight Pauli family, the number of errors of support at most $\delta N$ is
    {\setlength{\abovedisplayskip}{3pt}
 \setlength{\belowdisplayskip}{3pt}
 \[
   |\calE_{\delta N}|=\sum_{i=0}^{\delta N}\binom{N}{i}3^i,
   \]
   }
and hence, taking logarithms,
   {\setlength{\abovedisplayskip}{3pt}
 \setlength{\belowdisplayskip}{3pt}
    \[
    M^{\s{P}}(\delta) := \frac{1}{N}\log_2|\calE_{\delta N}|=2H_4(\delta)+o(1). 
    \]
}
Applying the general $\mathscr{L}_1$ exact-QEC rate condition \eqref{eq:AsympRateExact} gives exact-QEC partition codes inside $\Phi_N$ with asymptotic rate
\[ \calK<R^{\s{GM}}(\delta)-2M^{\s{P}}(\delta)=R^{\s{GM}}(\delta)-4H_4(\delta). \]
Similarly, applying the random-AQEC condition \eqref{eq:AsympAQECcondL1} gives random partition codes which are $\varepsilon_N$-AQEC for $\calE_{\delta N}$, with $\varepsilon_N\to0$ exponentially fast and with high probability, whenever
\[ \calK<\frac{1}{3}\p{\frac{R^{\s{GM}}(\delta)}{2}-4M^{\s{P}}(\delta)}=\frac{1}{3}\p{\frac{R^{\s{GM}}(\delta)}{2}-8H_4(\delta)}. \]
For the exact-QEC construction, one can sharpen the first bound by using the actual number of independent KL constraints, as in the unconstrained qubit case discussed in Remark~\ref{rem:QuariLimited}. Although the Pauli model is formally an $\mathscr{L}_1$ model, the pairwise KL constraints depend only on Pauli products of effective support at most $2\delta N$, and their exponential growth is governed by $H_2(2\delta)$. This gives the improved exact-QEC rate
\[\calK<R^{\s{GM}}(\delta)-H_2(2\delta). \]
At $\delta=0$, the improved exact-QEC bound gives $\calK<\log_2\varphi$, matching the exponential dimension of the constrained Hilbert space $\Phi_N$. The AQEC bound is also positive for all sufficiently small $\delta>0$. Thus, the partition-code construction gives both exact and approximate asymptotically good codes inside the Rydberg-blockaded Hilbert space.

\subsection{Majorana fermionic codes}\label{sec:majo}
Beyond $q$-ary quantum codes, partition codes give rise to a family of codes in the fermionic Fock space under Majorana noise, introduced in Section~\ref{sec:Fermionic}. Existing explicit constructions of Majorana error-correcting codes are primarily Majorana stabilizer or subsystem constructions, starting from the work of Bravyi, Leemhuis, and Terhal~\cite{bravyi2010majorana}, and continuing through surface-code, color-code, and high-rate fault-tolerant Majorana constructions~\cite{litinski2018quantum,mudassar2026fault}. To the best of our knowledge, no general non-stabilizer constructions of asymptotically good Majorana codes, exact or approximate, were previously known. The partition-code framework gives such a construction by lifting ordinary classical binary codes, in particular, the i.i.d. random classical codes from Lemma~\ref{lem:random-qary-hamming}, directly to Majorana codes.
Let $X=\ppp{0,1}^N$ index the occupation-number basis of the $N$-mode fermionic Fock space, equipped with the scaled Hamming metric $d$ from \eqref{eq:HammingScaled}. Let $\mathscr{E}^{\s{Maj}}=(\calE_{\leq t}^{\s{Maj}})_t$ be the family of Majorana errors of support size at most $t$. As shown in Example~\ref{ex:MajoranaL1}, the Majorana noise falls in the first level of the hierarchy: $(\calH_X,d,\mathscr{E}^{\s{Maj}})\in\mathscr{L}_1$. We work in the regime $t=\delta N$ and use binary logarithms, so $\calK=\log_2K/N$. The relevant classical rate is the binary Hamming Gilbert--Varshamov rate
\[ R^{\s{H}}(\delta)=1-H_2(2\delta). \]
The number of Majorana errors of support at most $\delta N$ is
\[ |\calE_{\leq\delta N}^{\s{Maj}}|=\sum_{i=0}^{\delta N}\binom{2N}{i}, \]
and hence
\[ M^{\s{Maj}}(\delta) := \frac{\log_2|\calE_{\leq\delta N}^{\s{Maj}}|}{N}=2H_2(\delta/2)+o(1). \]
The general exact-QEC condition \eqref{eq:AsympRateExact} at level $\mathscr{L}_1$ therefore gives exact-QEC partition codes correcting $\delta N$ Majorana errors with asymptotic rate
\[
\calK<R^{\s{H}}(\delta)-2M^{\s{Maj}}(\delta)=1-H_2(2\delta)-4H_2(\delta/2).
\]
Similarly, the random-AQEC condition \eqref{eq:AsympAQECcondL1} gives random partition codes which are $\varepsilon_N$-AQEC for $\calE_{\leq\delta N}^{\s{Maj}}$, with $\varepsilon_N\to0$ exponentially fast and with high probability, whenever
\[ \calK<\frac{1}{3}\p{\frac{R^{\s{H}}(\delta)}{2}-4M^{\s{Maj}}(\delta)}=\frac{1}{3}\p{\frac{1-H_2(2\delta)}{2}-8H_2(\delta/2)}. \]

The exact-QEC rate can be further improved by exploiting the fact that many pairwise Majorana KL constraints are linearly dependent. Indeed, for a Majorana monomial $\gamma_A$, let $s(A)\in\ppp{0,1}^N$ be the mode-flip pattern defined by $s(A)_j=|A\cap\ppp{2j,2j+1}|\bmod 2$. Then $\gamma_A$ maps each occupation vector $\ket{\ux}$ to a phase times $\ket{\ux\oplus s(A)}$. Hence, after the classical distance condition removes all cross terms, only pairs $A,B$ with $s(A)=s(B)$ contribute to the diagonal KL constraints. For such pairs, $\gamma_A^\dag\gamma_B$ is, up to a phase, a product of mode-parity operators $\prod_{j\in S}\gamma_{2j}\gamma_{2j+1}$ with $|S|\leq t$. Thus, the number of distinct diagonal constraints is at most $\sum_{s=0}^{t}\binom{N}{s},$
and therefore
\[ \frac{1}{N}\log_2\sum_{s=0}^{\delta N}\binom{N}{s}=H_2(\delta)+o(1). \]
Consequently, the exact-QEC partition-code rate found previously can be improved to
\[ \calK<1-H_2(2\delta)-H_2(\delta). \]

The refined Majorana rate should also be compared with the general partition-code rate for correcting $\delta N$ arbitrary qubit errors in \eqref{eq:PartitionGeneralErrors}. The bound
$\calK<1-H_2(2\delta)-H_2(\delta)$
is strictly larger than the corresponding general-error partition-code rate. This reflects the fact that Majorana errors have less freedom than arbitrary Pauli errors. Indeed, under the Jordan--Wigner representation, a Majorana monomial is mapped to a Pauli operator whose $X$-support determines the accompanying $Z$-string structure, up to local factors on the flipped modes. Thus, once the induced bit-flip pattern is fixed, the remaining phase structure is highly constrained.

\subsection{Constant excitation Fock state codes}\label{sec:CEFCex}
We now turn to constant-excitation bosonic Fock state codes introduced in Section~\ref{sec:FockStateCodes}. We consider the asymptotic regime in which the number of modes grows linearly with the total excitation, namely $q=\alpha N$ for fixed $\alpha>0$. This is the regime studied in \cite{elimelech2026asymptotically, aydin2026quantum}, and it is a natural high-excitation many-mode limit for bosonic codes. Since the dimension of $\calH_{q,N}$ is not of the form $\tilde{q}^N$ for an integer alphabet size, throughout this subsection, rates are measured using natural logarithms, i.e., $\calK=\ln(K)/N$.
Recall from Example~\ref{ex:BosonicL1L2} that the constant-excitation Fock space $\calH_{q,N}$ is indexed by the discrete simplex $\calS_{q,N}$, equipped with the scaled $\ell_1$ metric $d_1$ from \eqref{eq:ell1Dist}. This metric is aligned with the bosonic noise models considered below: amplitude damping errors satisfy the $\mathscr{L}_2$  condition, while number-shift and phase-rotation errors satisfy the $\mathscr{L}_1$ condition.

We study partition codes obtained from i.i.d. classical codes sampled according to two distributions on $\calS_{q,N}$ considered in \cite{elimelech2026asymptotically}: the uniform distribution, discussed in Section~\ref{sec:uniformSimples}, and the multinomial distribution, discussed in Section~\ref{sec:Multinomial}. The corresponding classical rate functions for i.i.d. simplex codes\footnote{Throughout, by a simplex code we mean a subset of points of the discrete simplex $\calS_{q,N}$ with the $\ell_1$ metric.} are denoted by $R_{\alpha}^{\s{U}}(\delta)$ and $R_{\alpha}^{\s{M}}(\delta)$, respectively, as defined in \eqref{eq:UniformSimplexrate} and \eqref{eq:MultinomialRate}.

These two ensembles also illustrate the usefulness of the inherited-typicality viewpoint. A physically important property is bounded per-mode occupancy: for a sequence $b_N$, define 
   \begin{equation} 
   \mathscr{P}_{N}^{\s{occ}}(b_N) := \ppp{\un\in\calS_{q,N}~:~\norm{\un}_{\infty}\leq b_N}. \label{eq:PMOprop}
   \end{equation} 
Bounding $\norm{\un}_{\infty}$ prevents the code support from concentrating too many photons in a single mode, which improves robustness to photon loss and increases the state’s coherence lifetime. Using Theorem~\ref{th:TypicallityProperties} together with Lemmas 19 and 23 from \cite{elimelech2026asymptotically}, the uniform simplex distribution gives bounded occupancy $O(\ln N)$ with high probability, while the multinomial distribution gives the stronger $O(\ln N/\ln\ln N)$ occupancy bound, at the cost of a weaker rate function. Thus, by choosing the underlying classical sampling distribution, the random partition construction can trade rate for additional structural properties of the resulting quantum code. Below we state the rate guarantees in terms of $R_{\alpha}^{*}$, with $*\in\ppp{\s{U},\s{M}}$, which covers both distribution choices.

\subsubsection{Number-shift and phase-rotation noise}
We consider the shift-rotation error family $\calE^{\s{SR}}_{t_{\mathrm{s}},t_{\mathrm{r}},\beta}$ from \eqref{eq:BosonicSR}. Here $t_{\mathrm{s}}$ bounds the total occupation-number shift, while $t_{\mathrm{r}}$ bounds the total discretized phase-rotation vector for the fixed maximal per-mode rotation parameter $\beta$. Throughout this subsection, we take $t_{\mathrm{s}}=\delta_{\mathrm{s}}N$ and $t_{\mathrm{r}}=\delta_{\mathrm{r}}N$. To apply the general asymptotic rate formulas, we first compute the error-set growth exponent
\[ M_{\alpha}^{\s{SR}}(\delta_{\mathrm{s}},\delta_{\mathrm{r}}) := \lim_{N\to\infty}\frac{1}{N}\ln\abs{\calE^{\s{SR}}_{\delta_{\mathrm{s}}N,\delta_{\mathrm{r}}N,\beta}}. \]
The shift labels $\ur$ and the rotation labels $\uTheta$ are signed integer vectors in $\Z^q$, with $q=\alpha N$, satisfying $\norm{\ur}_1\leq\delta_{\mathrm{s}}N$ and $\norm{\uTheta}_1\leq\delta_{\mathrm{r}}N$. Thus, their exponential growth is governed by the signed $\ell_1$ ball exponent $B_{\Z,\alpha}$ from Lemma~\ref{lem:signedL1BallZero}:
\[B_{\Z,\alpha}(\delta)=2\delta\ln\frac{2\delta}{\alpha}+\alpha\ln\p{\frac{2\delta}{\alpha}+\sqrt{1+\frac{4\delta^2}{\alpha^2}}}-2\delta\ln\p{\sqrt{1+\frac{4\delta^2}{\alpha^2}}-1}.\] 
Since $B_{\Z,\alpha}(\delta)$ is defined for balls of radius $\delta N$ in the scaled metric $d_1(\ux,\uy)=\frac{1}{2}\norm{\ux-\uy}_1$, the ordinary constraint $\norm{\ur}_1\leq\delta_{\mathrm{s}}N$ corresponds to scaled radius $\delta_{\mathrm{s}}N/2$, and similarly for $\uTheta$. Therefore
\[
M_{\alpha}^{\s{SR}}(\delta_{\mathrm{s}},\delta_{\mathrm{r}})=B_{\Z,\alpha}(\delta_{\mathrm{s}}/2)+B_{\Z,\alpha}(\delta_{\mathrm{r}}/2).
\]

 According to Example~\ref{ex:BosonicL1L21}, the shift-rotation model is an $\mathscr{L}_1$ model with respect to $d_1$. Applying the exact-QEC asymptotic condition \eqref{eq:AsympRateExact} gives exact-QEC partition codes obtained from the classical codes of Lemma~\ref{lem:iidCodesSimplex} correcting $t_{\mathrm{s}}=\delta_{\mathrm{s}}N$ shifts and $t_{\mathrm{r}}=\delta_{\mathrm{r}}N$ discretized rotations with rate
\[
\calK<R_{\alpha}^{*}(\delta_{\mathrm{s}})-2M_{\alpha}^{\s{SR}}(\delta_{\mathrm{s}},\delta_{\mathrm{r}})=R_{\alpha}^{*}(\delta_{\mathrm{s}})-2B_{\Z,\alpha}(\delta_{\mathrm{s}}/2)-2B_{\Z,\alpha}(\delta_{\mathrm{r}}/2).
\]
Similarly, applying the random-AQEC condition \eqref{eq:AsympAQECcondL1} gives random partition codes which are $\varepsilon_N$-AQEC for $\calE^{\s{SR}}_{\delta_{\mathrm{s}}N,\delta_{\mathrm{r}}N,\beta}$, with $\varepsilon_N\to0$ exponentially fast and with high probability whenever
\[
\calK<\frac{1}{3}\p{\frac{R_{\alpha}^{*}(\delta_{\mathrm{s}})}{2}-4M_{\alpha}^{\s{SR}}(\delta_{\mathrm{s}},\delta_{\mathrm{r}})}=\frac{1}{3}\p{\frac{R_{\alpha}^{*}(\delta_{\mathrm{s}})}{2}-4B_{\Z,\alpha}(\delta_{\mathrm{s}}/2)-4B_{\Z,\alpha}(\delta_{\mathrm{r}}/2)}.
\]
Since $B_{\Z,\alpha}(0)=0$ and $R_{\alpha}^{*}(0)>0$ for both ensembles, these bounds are positive whenever $\delta_{\mathrm{s}}$ and $\delta_{\mathrm{r}}$ are sufficiently small. Thus, the partition-code construction gives asymptotically good exact and approximate Fock-state codes against combined number-shift and phase-rotation noise.

\subsubsection{Amplitude damping errors}
Amplitude damping errors in the constant-excitation Fock setting are described by the normalized loss operators in \eqref{eq:ADerrorSet} and \eqref{eq:ADnormalizedOP}. These errors model the loss of at most $t$ photons across the $q$ modes, and the truncated amplitude damping channel of \cite{elimelech2026asymptotically} is controlled by the corresponding family $\calE_{\leq t}^{\s{AD}}$. By Example~\ref{ex:BosonicL1L2}, this model satisfies the second-level alignment condition with respect to the scaled $\ell_1$ metric.
For $t=\delta N$, the loss vectors are indexed by the union of simplices $\bigcup_{r\leq \delta N}\calS_{\alpha N , r}$. Therefore
\begin{align}
    M_{\alpha}^{\s{AD}}(\delta)& := \lim_{N\to\infty}\frac{1}{N}\ln\abs{\calE_{\leq \delta N}^{\s{AD}}}=\lim_{N\to\infty}\frac{1}{N}\ln\sum_{r=0}^{\delta N} \abs{\calS_{\alpha N,r}}\label{eq:label1}\\
    &\leq \lim_{N\to\infty}\p{\frac{1}{N}\ln \binom{\alpha N +\delta N-1}{\delta N}+\frac{\ln \delta N}{N}}=\ln(2)(\alpha+\delta)H_2\p{\frac{\alpha}{\alpha+\delta}},\label{eq:lable2}
\end{align}
where the inequality in the second line follows since $|\calS_{\alpha N,r}|$ is monotone increasing with $r$. A matching lower bound is obtained by lower bounding the sum in the r.h.s. of \eqref{eq:label1} by the largest summand, which implies that \eqref{eq:lable2} is in fact equality.
 Applying the exact-QEC asymptotic condition \eqref{eq:AsympRateExact} at level $\mathscr{L}_2$ gives exact-QEC partition codes correcting $\delta N$ photon losses with rate
   {\setlength{\abovedisplayskip}{3pt}
 \setlength{\belowdisplayskip}{3pt}
\[
\calK<R_{\alpha}^{*}(\delta)-M_{\alpha}^{\s{AD}}(\delta)=R_{\alpha}^{*}(\delta)-\ln(2)(\alpha+\delta)H_2\p{\frac{\alpha}{\alpha+\delta}}.
\]
}
Similarly, applying the random-AQEC condition \eqref{eq:AsympAQECcondL2} gives random partition codes which are $\varepsilon_N$-AQEC for $\calE_{\leq \delta N}^{\s{AD}}$, with $\varepsilon_N\to0$ exponentially fast and with high probability, whenever
\[
\calK<\frac{1}{3}\p{\frac{R_{\alpha}^{*}(\delta)}{2}-2M_{\alpha}^{\s{AD}}(\delta)}=\frac{1}{3}\p{\frac{R_{\alpha}^{*}(\delta)}{2}-2\ln(2)(\alpha+\delta)H_2\p{\frac{\alpha}{\alpha+\delta}}}.
\]
Since $M^{\s{AD}}(0)=0$ and $R_{\alpha}^{*}(0)>0$, both the exact and random-AQEC bounds are positive for all sufficiently small $\delta>0$, recovering asymptotically good constant-excitation Fock-state codes against linear photon loss in the high-excitation regime.

\subsection{Permutation-invariant codes}\label{sec:PIex}
We finally consider permutation-invariant partition codes, introduced in Section~\ref{sec:PIcodes}, against deletions and erasures. As in the constant-excitation Fock-state setting, we work in the asymptotic regime $q=\alpha N$ for fixed $\alpha>0$, where the dimension of the permutation-symmetric subspace grows exponentially in $N$. Since this dimension is not naturally of the form $\tilde q^N$ for a fixed alphabet size, we measure rates using natural logarithms throughout this subsection.
Recall from Example~\ref{ex:DelL1} that the symmetric space $\mathrm{Sym}(q,N)$, defined in \eqref{eq:symmetricSpaceX}, is indexed by the discrete simplex $\hat{X}=\calS_{q,N}$. With the scaled $\ell_1$ metric $d_1$ defined in \eqref{eq:ell1Dist}, the deletion error family takes the form
\[\hat{\calE}_{t}^{\s{Del}}=\ppp{\sqrt{\binom{t}{\ue}}E_{\ue}~:~\ue\in\calS_{q,t}}, \]
as in \eqref{eq:DelPIset}, and satisfies $(\calH_{\hat{X}},d_1,\hat{\mathscr{E}}^{\s{Del}})\in\mathscr{L}_2$. Moreover, by permutation invariance, deleting $t$ subsystems and erasing $t$ subsystems have the same action on the code space up to the classical side information specifying the erased coordinates. Thus, the deletion guarantees below also give erasure guarantees of the same order.
For $t=\delta N$, the error-set growth exponent is the same simplex exponent that appears for amplitude damping in constant-excitation Fock space:
\[ M_{\s{PI}}^{\s{Del}}(\delta) := \lim_{N\to\infty}\frac{1}{N}\ln\abs{\hat{\calE}_{\delta N}^{\s{Del}}}=\lim_{N\to\infty}\frac{1}{N}\ln\abs{\calS_{\alpha N,\delta N}}=\ln(2)(\alpha+\delta)H_2\p{\frac{\alpha}{\alpha+\delta}}.
\]

As in the Fock-state setting, we use the two simplex ensembles introduced in \cite{elimelech2026asymptotically}: the uniform simplex distribution from Section~\ref{sec:uniformSimples} and the multinomial distribution from Section~\ref{sec:Multinomial}. Their corresponding classical rate functions are denoted by $R_{\alpha}^{\s{U}}(\delta)$ and $R_{\alpha}^{\s{M}}(\delta)$, respectively, as defined in \eqref{eq:UniformSimplexrate} and \eqref{eq:MultinomialRate}. We write $R_{\alpha}^{*}(\delta)$ for either choice, with $*\in\ppp{\s{U},\s{M}}$.  

 The bounded-occupancy property $\mathscr{P}_{N}^{\s{occ}}(b_N)$ from \eqref{eq:PMOprop} has a natural interpretation here: a Dicke state $\ket{D_{\un}}$ is the uniform superposition over all computational-basis strings with composition $\un$, as in \eqref{eq:Copositionof}. Thus, requiring $\un\in\mathscr{P}_{N}^{\s{occ}}(b_N)$ means that no symbol appears more than $b_N$ times in the strings supporting the Dicke state. In this sense, bounded occupancy becomes a balancedness condition for the corresponding permutation-invariant codewords. As mentioned in Section~\ref{sec:CEFCex}, the resulting partition codes satisfy this balancedness property with high probability, with $b_N=O(\ln N)$ for the uniform distribution and $b_N=O(\ln N/\ln\ln N)$ for the multinomial distribution.
 
We now apply the general rate analysis. Since the deletion model is an $\mathscr{L}_2$ model, the exact-QEC asymptotic condition \eqref{eq:AsympRateExact} gives exact-QEC PI partition codes correcting $\delta N$ deletions, and hence also $\delta N$ erasures, with rate
\[ \calK<R_{\alpha}^{*}(\delta)-M_{\s{PI}}^{\s{Del}}(\delta)=R_{\alpha}^{*}(\delta)-\ln(2)(\alpha+\delta)H_2\p{\frac{\alpha}{\alpha+\delta}}.\]
Similarly, applying the random-AQEC condition \eqref{eq:AsympAQECcondL2} gives random PI partition codes which are $\varepsilon_N$-AQEC for $\hat{\calE}_{\delta N}^{\s{Del}}$, with $\varepsilon_N\to0$ exponentially fast and with high probability, whenever
\[\calK<\frac{1}{3}\p{\frac{R_{\alpha}^{*}(\delta)}{2}-2M_{\s{PI}}^{\s{Del}}(\delta)}=\frac{1}{3}\p{\frac{R_{\alpha}^{*}(\delta)}{2}-2\ln(2)(\alpha+\delta)H_2\p{\frac{\alpha}{\alpha+\delta}}}. \]
The same AQEC guarantee applies to erasures by the deletion-erasure equivalence for permutation-invariant codes. Since $M_{\s{PI}}^{\s{Del}}(0)=0$ and $R_{\alpha}^{*}(0)>0$, both bounds are positive for all sufficiently small $\delta>0$. Thus, the partition-code construction gives asymptotically good exact and approximate permutation-invariant codes against a linear number of deletions and erasures.
 
\section{Numerical examples}\label{sec:NumericalExperiments}
We give numerical examples that highlight the error-set guarantees developed in this paper. These examples address three questions: how closely the KL--Hellinger distance captures the average-case AQEC performance of square Shor codes against a family of stochastic Pauli errors; whether i.i.d. partition codes can outperform known codes of the same length and dimension for selected error models; and how approximate correction extends the exact-correction capabilities of GNU codes. For each experiment we specify the code and error set, describe the calculation, and state what is certified by the resulting bounds.

To compare upper and lower bounds on the same quantity, we extend the fixed-channel notation for the optimal error by taking the supremum over the controlled channels. That is, for a code $Q$ and an error set $\calE$, we write
\begin{align}
    \label{eq:NEshorFamilyError}
 &\varepsilon_{\mathrm{wc}}^{\mathrm{opt}}(Q,\calE)
 :=\sup_{\calN\in\mathscr{N}(\calE)}
 \varepsilon_{\mathrm{wc}}^{\mathrm{opt}}(Q,\calN)=\sup_{\calN\in\mathscr{N}(\calE)} \inf_{\calR}
 d(\calR\circ\calN|_{L(Q)},I_{L(Q)}),\\
 &\varepsilon_{\mathrm{ch}}^{\mathrm{opt}}(Q,\calE)\label{eq:NEchFamilyError}
 :=\sup_{\calN\in\mathscr{N}(\calE)}
 \varepsilon_{\mathrm{ch}}^{\mathrm{opt}}(Q,\calN)=\sup_{\calN\in\mathscr{N}(\calE)} \inf_{\calR}
 d_{\mathrm{ch}}(\calR\circ\calN|_{L(Q)},I_{L(Q)}),
\end{align} 
Here the infimum is taken over all CPTP recovery maps, and $d$ and $d_{\mathrm{ch}}$ are the worst-case distance~\eqref{eq:BuresDist} and average-case distance~\eqref{eq:chBuresDist}, respectively. Thus, the average/worst-case criterion is used for the encoded input, while the channel is chosen adversarially from the controlled family. 

\subsection{Error-set AQEC for the square Shor code}\label{sec:NumericalShor}
Shor's nine-qubit code was the first quantum code to protect a data qubit against an single errors in the encoded physical block~\cite{shor1995scheme}. Its construction combines repetition codes for bit and phase errors. In a general form, this construction gives a family of square Shor codes of increasing length, while retaining a simple description of the stabilizer syndromes and logical error classes \cite{agarwal2026pauli}. These features make it possible to calculate the KL--Hellinger distance and optimal recovery errors without constructing matrices of dimension exponential in the number of qubits. We use this family to examine the tightness of the average-case error-set guarantee in Theorem~\ref{th:AvAQECErrorSetSuff}.

Let $L\geq3$ be odd and set $N=L^2$. The $\qc{L^2,1,L}$ square Shor  code $Q_L$ is the two-dimensional subspace spanned by
\begin{equation}\label{eq:NEshorCode}
 \ket{c_0}=\p{\frac{\ket{0}^{\otimes L}+\ket{1}^{\otimes L}}{\sqrt2}}^{\otimes L},
 \qquad
 \ket{c_1}=\p{\frac{\ket{0}^{\otimes L}-\ket{1}^{\otimes L}}{\sqrt2}}^{\otimes L}.
\end{equation}
 We consider the error set formed of Pauli operators of every weight, appropriately normalized. Write $\calP_N=\ppp{I,X,Y,Z}^{\otimes N}$, with one Hermitian representative for each Pauli string, and let $\wt(W)$ denote the number of its nonidentity factors. For $0<p<1$, set $a=p/3$ and define
\begin{equation}\label{eq:NEshorErrorSet}
 \calE_p=\ppp{a^{\wt(W)/2}W:W\in\calP_N}
 =\ppp{I,\sqrt aX,\sqrt aY,\sqrt aZ}^{\otimes N}.
\end{equation}

A Pauli channel $\calN(\rho)=\sum_{W\in\calP_N}\pi_W W\rho W^\dag$ is $\calE_p$-controlled precisely when $\pi_W\leq a^{\wt(W)}$ for every $W$. In particular, the family contains independent Pauli noise with each nonidentity single-qubit probability at most $a$, as well as correlated Pauli channels satisfying these bounds. It also contains arbitrary channels acting on a random set $S$ of qubits via
{
\[
 \calN=\sum_{S\subseteq[N]}p_S\,I_{L(\calH_2^{\otimes S^c})}\otimes\calN_S,
 \qquad \calN_S\text{ is CPTP on }L(\calH_2^{\otimes S}),
\]
where $(p_S)_{S\subseteq[N]}$ is a probability distribution satisfying
\[
 \P[T\subseteq S]=\sum_{S:\,T\subseteq S}p_S\leq a^{|T|},
 \qquad T\subseteq[N].
\]}  This last condition is the spatial local-stochastic condition used in fault-tolerance analyses~\cite[Sec.~III.D, Definition~1]{bombin2016resilience}. It describes a useful subclass of the set of $\calE_p$-controlled channels.
 
{Theorem~\ref{th:AvAQECErrorSetSuff} implies the upper bound $\varepsilon_{\mathrm{ch}}^{\mathrm{opt}}(Q_L,\calE_p)\le \zeta_{\s{H}}(\calE_p,Q_L)$. For the lower bound we can fix an $\calE_p$-controlled Pauli channel $\calN_*$, so overall
\begin{equation}\label{eq:NEshorBracket}
 \varepsilon_{\mathrm{ch}}^{\mathrm{opt}}(Q_L,\calN_*)
 \leq\varepsilon_{\mathrm{ch}}^{\mathrm{opt}}(Q_L,\calE_p)
 \leq\zeta_{\s{H}}(\calE_p,Q_L).
\end{equation}  
A good lower bound is obtained by taking a stochastic Pauli channel $\calN_*$ that assigns high Pauli error probabilities to distinct logical classes with the same stabilizer syndrome, subject to the individual bounds $\pi_W\leq a^{\wt(W)}$ and placing the unused
probability mass on the identity. The channel $\calN_*$ is 
$\calE_{p}$-controlled (we do not believe that  $\calN_*$ is the worst channel in the full family). To calculate the optimal average-case error for $\calN_{*}$, we find the channel fidelity obtained by the optimal recovery map for stochastic Pauli channels and stabilizer codes, presented in~\cite[Appendix, Theorem~3]{fletcher2008structured}. We evaluate these optimal errors for $L=3,5,\ldots,17$, with $p=0.01$. The calculation includes errors of every weight and does not use sampling or a weight cutoff. As a benchmark, we also calculate the optimal recovery error for the controlled i.i.d. depolarizing channel with probabilities $(1-p,p/3,p/3,p/3)$ on each qubit (which is also $\calE_{p}$ controlled).

\begin{figure}[t]
\centering

\begin{overpic}[width=0.90\linewidth]{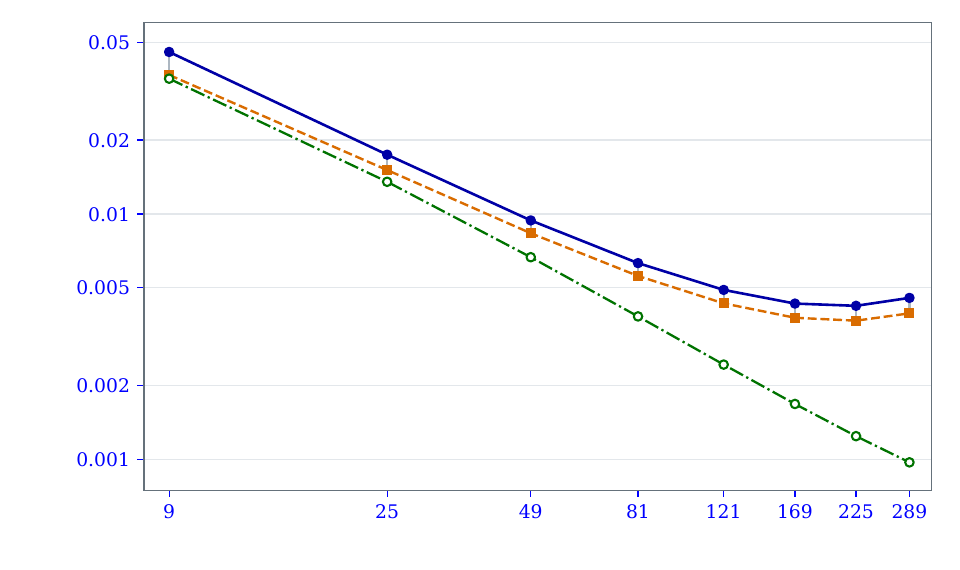}
 \put(55,1.5){\makebox(0,0){\small Number of physical qubits, $N=L^2$}}
 \put(1.5,34){\rotatebox{90}{\makebox(0,0){\small Average-case correction error, $d_{\mathrm{ch}}$}}}
\end{overpic}
\par\smallskip
{\small
\begin{tabular}{@{}ll@{}}
\tikz[baseline=-0.5ex]{\draw[blue!65!black,line width=0.9pt] (0,0)--(.65,0);\fill[blue!65!black] (.325,0) circle (1.7pt);}
& Upper bound: $\zeta_{\s{H}}(\calE_p,Q_L)$\\
\tikz[baseline=-0.5ex]{\draw[orange!85!black,dashed,line width=0.9pt] (0,0)--(.65,0);\fill[orange!85!black] (.275,-.05) rectangle (.375,.05);}
& Lower bound: $\varepsilon_{\mathrm{ch}}^{\mathrm{opt}}(Q_L,\calN_*)$\\
\tikz[baseline=-0.5ex]{\draw[green!45!black,dash dot,line width=0.9pt] (0,0)--(.65,0);\draw[green!45!black,fill=white] (.325,0) circle (1.8pt);}
& Benchmark: optimal recovery for i.i.d. depolarizing noise
\end{tabular}}
\caption{Average-case AQEC of square Shor codes for the error set~\eqref{eq:NEshorErrorSet}, with $p=0.01$. The horizontal axis gives $N=L^2$ for $L=3,5,\ldots,17$; the vertical axis gives $d_{\mathrm{ch}}=\sqrt{1-\calF_{\mathrm{ch}}}$. Both axes are logarithmic. The upper and lower curves bracket the family parameter \eqref{eq:NEchFamilyError}, with an upper-to-lower ratio of approximately $1.128$--$1.241$. The third curve is the optimum for a single i.i.d. channel in the same controlled family. }
\label{fig:NEshorBounds}
\end{figure}

Figure~\ref{fig:NEshorBounds} shows that the KL--Hellinger guarantee is close to the lower bound throughout the tested range: their ratio is at most $1.241$.
These finite-length results show that the error-set guarantee can capture the average-case AQEC error within a small multiplicative factor for a standard stabilizer family. They also distinguish the error-set question from correction of a single i.i.d. noise model. At $N=289$, the controlled witness channel has optimal error $0.00393722$, compared with $0.00097376$ for i.i.d. depolarizing noise, a factor of about $4.04$. The upper and lower bounds both increase slightly between $N=225$ and $N=289$ (they still leave room for the true error decrease through the entire region).

\subsection{A fixed i.i.d. partition code under amplitude damping}\label{sec:NumericalPartition}
We next examine whether the i.i.d. ensemble in Construction~\ref{const:randomParition} can give competitive AQEC performance under amplitude-damping noise given the same physical resources. We consider the performance of this construction against amplitude damping noise both in the channel and error-set setting. We start with the former, where we consider i.i.d. partition codes with the same number of modes, total excitation, and code dimension as a known bosonic construction. We compare a fixed partition code with two instances of the Chuang--Leung--Yamamoto (CLY) construction~\cite{chuang1997bosonic}. All the three codes are $128$-dimensional subspaces of $\calH_{14,28}$ of \eqref{eq:COnstExSpace}; thus, each encodes seven logical qubits using $28$ photons in $14$ modes. The comparison concerns the full AD channel as its loss probability varies.

For the partition code, we sample $L=65{,}536$ occupation vectors independently and uniformly from the simplex $\calS_{14,28}$ and partition them into $K=128$ sets of size $T=512$. The resulting basis states are
\[
 \ket{c_i}=\frac1{\sqrt{512}}\sum_{a=0}^{511}\ket{\un_{i,a}},
 \qquad i=1,2,\dots,128.
\]
All sampled occupation vectors are distinct, so these states are orthonormal. The performance analysis below is carried out for one fixed realization of this code. Its basis states are not optimized or resampled during the experiment.

The CLY codes we use for reference rely on the cyclic-orbit construction of~\cite[Sec.~VI, Eqs.~(6.9)--(6.13)]{chuang1997bosonic}, which we briefly recall. Let $R$ be the cyclic shift operator on finite-length strings. Choose $128$ vectors $\ux_i\in\calS_{14,14}$ with distinct orbits, each of size $14$, and set
\begin{equation}\label{eq:NEpartitionCLY}
 \ket{c_i}=\frac1{\sqrt{14}}\sum_{s=0}^{13}\ket{2R^s\ux_i},
 \qquad i=1,2,\dots,128.
\end{equation}
The resulting code corrects (exactly) a single photon loss.  
The freedom to choose the orbit representatives gives different codes with the same resources and exact-correction guarantee. We use two choices to examine how this freedom affects the comparison.

For the first reference, we sample the representatives uniformly from $\calS_{14,14}$, rejecting repeated orbits and orbits of size less than $14$. We generate $20$ such codes. For the second reference, we restrict the representatives to permutations of a vector containing $b$ zeros, $b$ twos, and $14-2b$ ones. We test $b\in\ppp{2,3,4,7}$, with three codes for each value. This gives $12$ additional codes with more balanced occupations. In each group, we retain the code with the smallest achieved average-case error at $\gamma=0.01$. This selection is held fixed for the worst-case comparison below and for all values of $\gamma$. We refer to these two instances as CLY with uniform representatives and CLY with balanced occupations, respectively.  All codes are subjected to independent AD noise with the same photon-loss probability $\gamma$ in every mode, as in Example~\ref{ex:AD}.

%
 We bound the optimal worst-case error $\varepsilon_{\mathrm{wc}}^{\mathrm{opt}}(Q,\calN_\gamma)$ of the full AD channel $\calN_\gamma$ from bellow and from above using the truncated amplitude channel framework of \cite{elimelech2026asymptotically}. We first consider its normalized truncation $\calN_{\gamma,\leq3}$, obtained by conditioning on at most three photon losses, as in~\cite[Eq.~(13)]{elimelech2026asymptotically}. The normalization is the code-independent probability \eqref{eq: AD channel}
\[
p_{28,3,\gamma}=\sum_{r=0}^{3}\binom{28}{r}\gamma^r(1-\gamma)^{28-r}.
\]
We construct this channel using the normalized loss operators of~\eqref{eq:ADnormalizedOP}, with $q=14$ and $N=28$. For the truncated channel, we use the B\'eny--Oreshkov characterization of Theorem~\ref{th:BO} through the complementary channel and a constant channel. The obtained bounds apply to arbitrary density matrices and therefore bound the optimum over all CPTP recoveries. We use the truncation argument of~\cite[Theorem~5]{elimelech2026asymptotically}, with the fidelity convention of~\eqref{eq:BuresDist}, to obtain lower and upper bounds on the optimal worst-case error for the full AD channel from the optimal error of its truncated version. These lower and upper bounds are presented in Figure~\ref{fig:NEpartition}.

\begin{figure}[t]
\centering
\definecolor{NEpartitionorange}{HTML}{CF6E12}
\definecolor{NEclyblue}{HTML}{245C9F}
\definecolor{NEclygreen}{HTML}{248354}
\captionsetup{labelfont={bf,color=blue}}
\begin{overpic}[width=.6\linewidth]{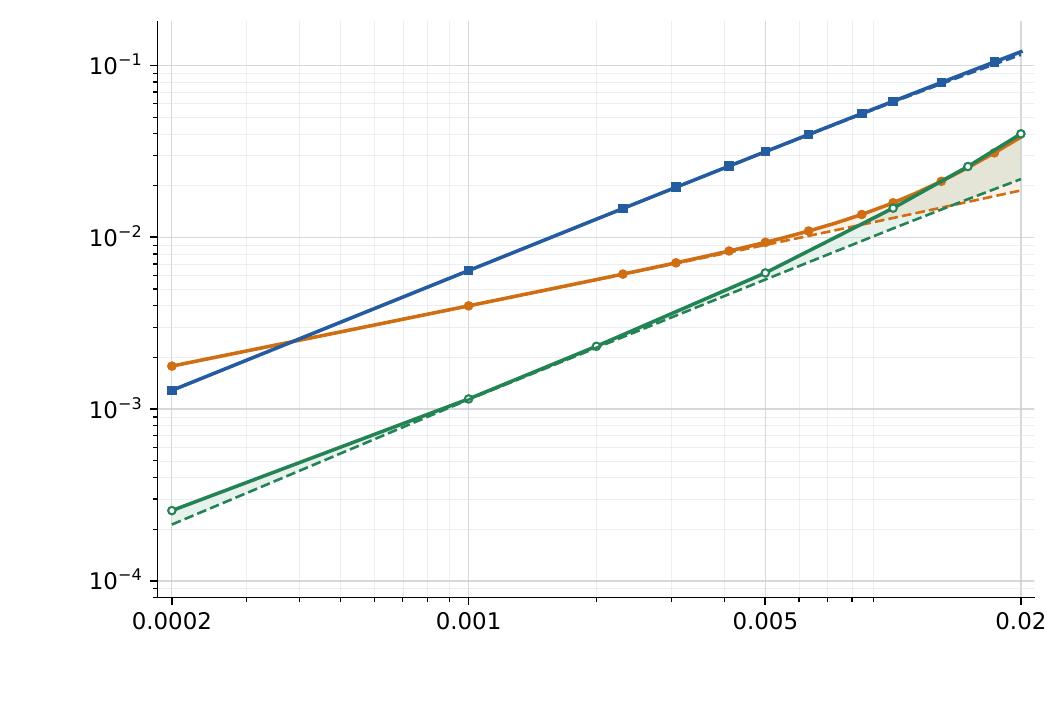}
 \put(57,2.0){\makebox(0,0){\small Photon-loss probability, $\gamma$}}
 \put(2.2,40){\rotatebox{90}{\makebox(0,0){\scriptsize Worst-case error, $\varepsilon_{\mathrm{wc}}^{\mathrm{opt}}$}}}
\end{overpic}
\par\smallskip
{\small\begin{tabular}{@{}ll@{}}
\tikz[baseline=-.5ex]{\draw[NEpartitionorange,line width=.9pt] (0,0)--(.65,0);\fill[NEpartitionorange] (.325,0) circle (1.7pt);}
& Fixed i.i.d. partition code\\
\tikz[baseline=-.5ex]{\draw[NEclyblue,line width=.9pt] (0,0)--(.65,0);\fill[NEclyblue] (.275,-.05) rectangle (.375,.05);}
& CLY: uniform representatives\\
\tikz[baseline=-.5ex]{\draw[NEclygreen,line width=.9pt] (0,0)--(.65,0);\draw[NEclygreen,fill=white] (.325,0) circle (1.8pt);}
& CLY: balanced occupations\\
& Solid/dashed lines of all colors are upper/lower bonds for optimal recovery.
\end{tabular}}
\caption{Worst-case AQEC under the full AD channel of Example~\ref{ex:AD}, with $14$ modes, $28$ photons, and code dimension $128$. Both axes are logarithmic. The solid and dashed curves bound the same optimal worst-case error measured by the distance~\eqref{eq:BuresDist}, using the B\'eny--Oreshkov bounds for the normalized truncation at three losses. The shaded regions represent the gap between the upper and lower bounds. The codes are held fixed throughout. We evaluate loss probabilities for partition codes and CLY codes with uniform representatives and for CLY codes with balanced occupancy.%
}
\label{fig:NEpartition}
\end{figure}

As seen in Figure~\ref{fig:NEpartition}, the partition code outperforms the CLY code with uniform representatives over most of the tested range. At $\gamma=0.01$, the improvement is by a factor of at least $3.85$, allowing optimal recovery for both codes. CLY with balanced occupancy performs better at smaller tested loss probabilities, while its bounds overlap those of the partition code at larger probabilities, preventing an ordering of their optimal errors. In summary, a fixed, unoptimized instantiation of the i.i.d. construction can perform competitively with balanced-occupancy bosonic CLY codes, while outperforming a CLY code with uniform representatives under the worst-case criterion. 
\subsubsection*{Dual-rail codes under a photon-loss error set}
The dual-rail encoding represents a qubit by one excitation shared between two physical systems, with basis
\[
 \ket{0}\longmapsto\ket{1,0},\qquad
 \ket{1}\longmapsto\ket{0,1}.
\]
Losing the excitation takes the pair to $\ket{0,0}$, which is orthogonal to the encoded subspace. Thus, checking whether the pair still contains an excitation detects a loss without measuring the encoded binary value. This provides a physical way to identify the locations of loss errors. This encoding has been proposed for superconducting cavities~\cite{teoh2023dual} and demonstrated using coupled transmons~\cite{levine2024demonstrating}. Here we use its bosonic realization, where each pair consists of two modes containing one photon in total. Under independent AD with equal loss probabilities on both rails, photon loss takes the encoded state out of the dual-rail subspace in a detectable way. Duan et al.\ use this encoding as an inner code and a distance-$d$ quantum code as an outer code to obtain correction of up to $d-1$ photon losses~\cite[Sec.~III, Lemma~1 and Theorem~1]{duan2010multi}. We use the same dual-rail inner encoding for the partition and Shor codes.

We compare two outer codes of length $25$ and dimension $K=2$: a fixed i.i.d. partition code with $T=512$ independent uniform words in $[2]^{25}$ per basis state, and the square Shor code $Q_5$ of~\eqref{eq:NEshorCode}, with parameters $\qc{25,1,5}$. All $1024$ sampled words are distinct, and the realization is held fixed. With the same dual-rail inner encoding, both codes use $25$ photons in $50$ modes. The Shor code $Q_5$ exactly corrects up to four photon losses.

For $t=0,\ldots,6$, we use $\calE_{\leq t}^{\s{AD}}$ of~\eqref{eq:ADerrorSet}, with $q=50$, $N=25$, and the stated normalization. We upper-bound $\varepsilon_{\mathrm{ch}}^{\mathrm{opt}}(Q,\calE_{\leq t}^{\s{AD}})$ in \eqref{eq:NEchFamilyError} by $\zeta_{\s{H}}(\calE_{\leq t}^{\s{AD}},Q)$. Lower bounds follow by choosing the controlled channels
\[
 \calN_r(\rho)=\sum_{\ur\in\calS_{50,r}}
 \Tilde{A}_{\ur}\rho\Tilde{A}_{\ur}^{\dagger},\qquad 0\leq r\leq t,
\]
which erase a uniformly chosen set of $r$ pairs. The lower bounds for the partition code performances on that channel follow from the from the AQEC guarantee of Petz map (see Corollary~\ref{cor:HellingerChannelFidelity}). For the optimal performances of Shor's code on that channel we use an SDP bound: A recovery map can be represented by its Choi matrix, with the CPTP condition expressed as positivity and linear trace constraints. For a fixed code and noise channel, the channel fidelity is linear in this matrix, so maximizing it is a semidefinite program~\cite{fletcher2007optimum}. Feasible points of the corresponding dual give upper bounds on the fidelity attainable by any recovery, and hence lower bounds on the optimal correction error. For Shor codes under Pauli noise, the stabilizer syndrome structure reduces the recovery optimization to independent choices among logical error classes, significantly reducing the optimization space, and allows for a better bound (compared to the Petz recovery bound). The KL--Hellinger calculation retains all entries of the QEC matrices over all $245{,}506$ supports of at most six lost pairs.

\begin{figure}[t]
\centering
\begin{overpic}[width=.85\linewidth]{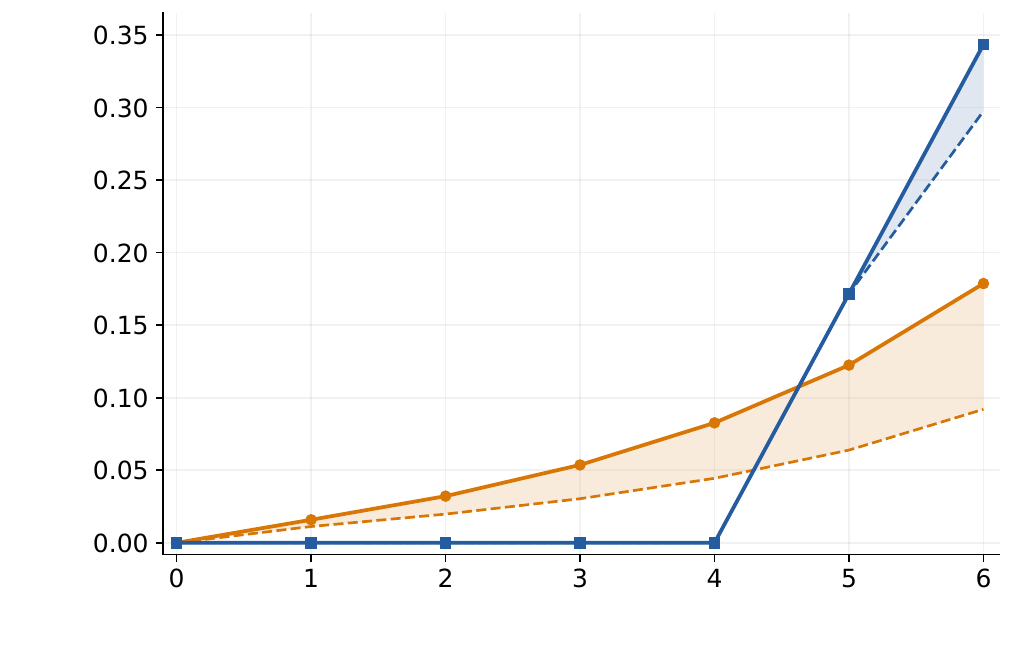}
 \put(55,1.5){\makebox(0,0){\small Maximum number of lost photons, $t$}}
 \put(1.5,34){\rotatebox{90}{\makebox(0,0){\small Average-case correction error, $d_{\mathrm{ch}}$}}}
\end{overpic}
\par\smallskip
{\small\begin{tabular}{@{}ll@{}}
\tikz[baseline=-.5ex]{\draw[orange!85!black,line width=.9pt] (0,0)--(.65,0);\fill[orange!85!black] (.325,0) circle (1.8pt);}
& Fixed i.i.d. partition code\\
\tikz[baseline=-.5ex]{\draw[blue!65!black,line width=.9pt] (0,0)--(.65,0);\fill[blue!65!black] (.275,-.05) rectangle (.375,.05);}
& Dual-rail square Shor code
\end{tabular}}
\caption{Average-case AQEC for two codes encoding one qubit into $25$ photons in $50$ modes under $\calE_{\leq t}^{\s{AD}}$. Solid and dashed curves give upper and lower bounds on the optimal family error, respectively; shading marks the resulting intervals. Both axes are linear. The partition-code upper bound is below the Shor lower bound at $t=5,6$, beyond Shor's exact correction of four photon losses. }
\label{fig:NEpartitionDualRail}
\end{figure}

As seen in Figure~\ref{fig:NEpartitionDualRail},  the partition code upper bounds on the worst-case optimal error at $t=5,6$ lie below the worst-case optimal error lower bound for the Shor code. In particular, a fixed (unoptimized, sampled once) i.i.d. realization of the partition code outperforms the Shor code beyond its exact photon-loss correction capability in the worst-case setting. 

\subsection{Approximate correction in a two-mode GNU code beyond exact distance}\label{sec:NumericalGNU}
We now consider the GNU permutation-invariant codes of \cite{ouyangPI} as constant-excitation Fock state codes against amplitude damping noise. We show that, under the error-set model, these codes approximately correct more photon losses than their exact photon-loss correction capability, with small error. Consequently, their approximate photon-loss distance can be strictly larger than its exact counterpart.

Recall that the Dicke basis $\ppp{\ket{D_{\un}}~:~\un\in\calS_{q,N}}$ of \eqref{eq:Dickestates} and the Fock basis $\ppp{\ket{\un}~:~\un\in\calS_{q,N}}$ of \eqref{eq:COnstExSpace} are indexed by the same simplex. The identification
 $\ket{D_{\un}}\longmapsto\ket{\un}$
therefore realizes a permutation-invariant code in $\mathrm{Sym}(q,N)$ (see~\eqref{eq:symmetricSpaceX}) as a code in a $q$-mode Fock space with  $N$ photons, $\calH_{q,N}$. Here we take $q=2$ and the $(3,3,u)$ GNU code, where, in the notation of~\cite{ouyangPI}, $g=3$ is the gap between consecutive occupied Dicke weights, $n=3$ is the occupancy parameter determining the number of Dicke levels and their binomial weights, and $u\geq1$ is the scaling parameter determining the code length $N=gnu=9u$. Its constant-excitation realization $Q_N^{\s{GNU}}\subseteq\calH_{2,N}$ has basis
\[
 \ket{c_0}=\frac{\ket{N,0}+\sqrt3\ket{N-6,6}}2,
 \qquad
 \ket{c_1}=\frac{\sqrt3\ket{N-3,3}+\ket{N-9,9}}2.
\]
Thus, $Q_N^{\s{GNU}}$ is a two-dimensional code in two modes with total excitation $N$. We examine how its approximate photon-loss correction capability changes with $N$.

We use the error set $\calE_{\leq t}^{\s{AD}}$ from \eqref{eq:ADerrorSet} with $q=2$ and the normalized operators $\Tilde{A}_{\ur}$ from \eqref{eq:ADnormalizedOP}. As shown in Example~\ref{ex:AD}, it controls every truncated AD channel $\calN_{\gamma,\leq t}$ as well as the channels that lose exactly $r\leq t$ photons. We vary the total excitation $N$ and compare fixed maximum loss numbers $t$, and the guarantees apply to every channel in $\mathscr{N}(\calE_{\leq t}^{\s{AD}})$. The known correction guarantee for the $(3,3,u)$ GNU codes and the equivalence with constant-excitation Fock-state codes imply correction of two photon losses~\cite{ouyangPI,aydin2026quantum}.  The question is whether the error-set AQEC conditions support correction of more than two photon losses with small $\varepsilon$.

For the average-case criterion, we calculate $\zeta_{\s{H}}(\calE_{\leq t}^{\s{AD}},Q_N^{\s{GNU}})$ for $N=9,18,36,\ldots,9216$ and $0\leq t\leq9$ and use Theorem~\ref{th:AvAQECErrorSetSuff}. For the worst-case criterion, we calculate a feasible upper bound $\bar{\zeta}$ on the environment-leakage distance $\zeta$ of Definition~\ref{def:DiamondNormProxConst} by semidefinite programming (see \cite{watrous2012simpler} for details on the SDP bound), for $N=9,36,144,\ldots,147456$ and $t=2,3,4$. We restrict the worst-case calculation to $t=2,3,4$ to keep the SDP computations small.  Theorem~\ref{th:BOerrorset} then gives an upper bound on the optimal error $\varepsilon_{\s{wc}}^{\s{opt}}\leq \sqrt{2\bar{\zeta}}$. Both calculations use the full error set with the normalization of \eqref{eq:ADnormalizedOP}.

\begin{figure}[t]
\centering
\begin{minipage}[t]{.49\linewidth}
\centering
{\small (a) Average-case guarantee}\par
\begin{overpic}[width=\linewidth]{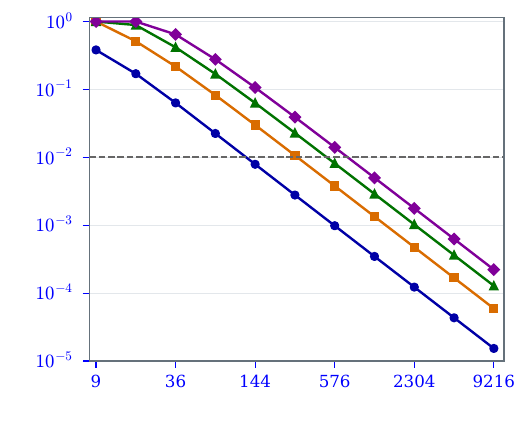}
 \put(56,1.2){\makebox(0,0){\scriptsize Total photon number, $N$}}
 \put(1.7,41){\rotatebox{90}{\makebox(0,0){\scriptsize Upper bound on $d_{\mathrm{ch}}$}}}
\end{overpic}
{\scriptsize
\textcolor{blue!65!black}{$\bullet$} $t=3$\quad
\textcolor{orange!85!black}{$\blacksquare$} $t=5$\quad
\textcolor{green!45!black}{$\blacktriangle$} $t=7$\quad
\textcolor{purple!80!black}{$\blacklozenge$} $t=9$}
\end{minipage}\hfill
\begin{minipage}[t]{.49\linewidth}
\centering
{\small (b) Worst-case guarantee}\par
\begin{overpic}[width=\linewidth]{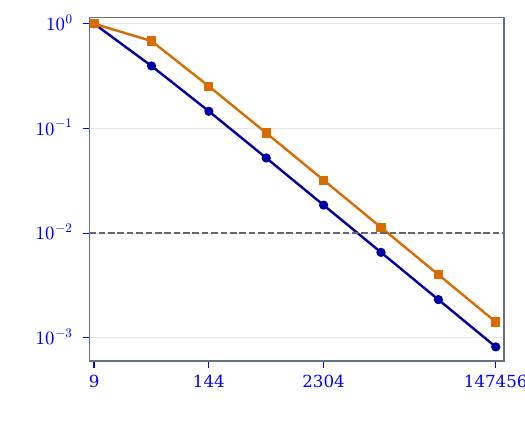}
 \put(56,1.2){\makebox(0,0){\scriptsize Total photon number, $N$}}
 \put(1.7,41){\rotatebox{90}{\makebox(0,0){\scriptsize Upper bound on $d$}}}
\end{overpic}
{\scriptsize
\textcolor{blue!65!black}{$\bullet$} $t=3$\qquad
\textcolor{orange!85!black}{$\blacksquare$} $t=4$}
\end{minipage}
\par\smallskip
{\small Dashed horizontal line: error tolerance $\varepsilon=0.01$.}
\caption{AQEC of $Q_N^{\s{GNU}}\subseteq\calH_{2,N}$ for the photon-loss error set $\calE_{\leq t}^{\s{AD}}$. The axes in both plots are logarithmic. Each of the curves is plotted for fixed $t$ and varying total excitation $N$. Figure (a) gives the average-case upper bound $\zeta_{\s{H}}(\calE_{\leq t}^{\s{AD}},Q_N^{\s{GNU}})$; Figure (b) gives the worst-case upper bound $\sqrt{2\bar{\zeta}}$. Both bounds are capped at one.} A point below the dashed line certifies $\varepsilon$-AQEC for the entire error set with $\varepsilon=0.01$, in the corresponding criterion. Every code exactly corrects two photon losses and fails the KL conditions at three losses.
\label{fig:NEgnu}
\end{figure}

Figure~\ref{fig:NEgnu}(a) shows that, as the total excitation increases, the GNU codes approximately correct progressively more photon losses at a fixed error tolerance, extending their exact photon-loss correction capability. The same effect occurs in the worst-case criterion.

These upper bounds on the optimal error translate to lower bounds on the approximate code distance (see Definition~\ref{def:dist}): at $N=9216$ and $t=3$, we obtain $\bar{\zeta}=2.14611\times10^{-5}$ and $\sqrt{2\bar{\zeta}}<0.006552$. At $N=36864$ and $t=4$, the corresponding values are $8.05115\times10^{-6}$ and $0.004013$. Thus, for the photon-loss family $\mathscr{E}^{\s{AD}}=(\calE_{\leq t}^{\s{AD}})_t$, Definition~\ref{def:dist} gives
\[
 d_{0.007}(Q_{9216}^{\s{GNU}},\mathscr{E}^{\s{AD}})\geq3,
 \qquad
 d_{0.005}(Q_{36864}^{\s{GNU}},\mathscr{E}^{\s{AD}})\geq4,
\]
while the exact photon-loss distance is two. Consequently, both criteria certify approximate correction beyond the exact capability, uniformly over the corresponding controlled families.


\vspace*{.2in}
This concludes the main text of the paper. Our thoughts about future research directions are presented in Sec.~\ref{sec: conclusion}, beginning on p.\pageref{sec: conclusion}.
\subsection*{Acknowledgments}
\addcontentsline{toc}{section}{Acknowledgments}
The authors are grateful to Zi-Wen Liu for helpful discussions with D.E. on subsystem variance.  D.E. acknowledges support from the Yad Hanadiv Foundation through the Rothschild Fellowship.   V.V.A. acknowledges NSF grant OMA2120757 (QLCI).   
  The research of A.B. was partially supported by NSF grants CIF-2330909 and CIF-2526035. All the ideas, mathematical statements and proofs in this document are human-generated.
  The completed manuscript was edited with the assistance of Claude and Chat GPT, developed by Anthropic and Open AI, respectively. 
  Claude and Chat GPT were used to refine language, improve clarity, and enhance readability in accordance with the authors’ instructions. All content, claims, and conclusions have been reviewed and verified by the authors to ensure accuracy and originality.   Certain products, commercial and otherwise, are mentioned in this publication. These mentions are for informational purposes only, and do not imply recommendation or endorsement by NIST.

\subsection*{Data availability} This is a purely mathematical work and no data was created or analyzed in this study. All figures can be reproduced directly from the presented equations.

\clearpage
\appendix
\makeatletter
\section*{Appendices}
\addcontentsline{toc}{section}{Appendices}

\section{Classical random codes}\label{app:classicalCodes}
 This section is devoted to presenting and developing the classical coding-theoretic results needed to establish the partition-code constructions of Section~\ref{sec:Exampels}. We focus on random constructions of two main families of classical codes. The first consists of $q$-ary codes with respect to several error metrics, including the Hamming metric, the $\ell_1$ metric, and deletion-type metrics. The second consists of simplex codes in the $\ell_1$ metric. Throughout, we are interested in the asymptotic regime in which both the relative distance and the rate are bounded away from zero.

We begin with a general random-coding principle for i.i.d. codes over metric spaces. This is a standard Gilbert--Varshamov-type argument, written in a distributional form that will be useful for our purposes. In classical coding theory, this argument is often applied to
the random code ensemble in the Hamming space, e.g., \cite{barg2002random}, but it is well recognized that it applies more broadly when
the codewords are sampled from an arbitrary probability measure on a discrete metric space. Following this route, we apply the same first-moment and expurgation argument to the different spaces and metrics arising from the noise models considered in Section~\ref{sec:Exampels}.

Let $(X,d)$ be a discrete metric space, and let $\mu$ be a probability measure on $X$. Consider a random code
\begin{equation}
     \s{C}=\ppp{\s{x}_0,\dots,\s{x}_{L-1}}\label{eq:RandomCode}
\end{equation}
generated by choosing $\s{x}_0,\dots,\s{x}_{L-1}$ independently according to $\mu$. Recall that $d(\s{C})$ denotes the minimum $d$-distance between two distinct codewords in $\s{C}$. Our goal is to identify conditions on $L$ under which $\s{C}$ has minimum distance at least $t$. We will apply the argument to sequences $(X_N,d)$ 
of metric spaces equipped with probability measures $\mu_N$ in the regime where the target distance $t=t_N$ grows linearly with $N$. The following theorem gives sufficient conditions for this general setting.

\begin{prop}[i.i.d.-type Gilbert--Varshamov bounds]\label{prop:GV} Let $X_N, N\ge 1$ be a sequence of discrete metric spaces, each equipped
with a probability measure $\mu_N$. Let $\s{C}_N=\ppp{\s{x}_0,\dots\s{x}_{L_N-1}}$ be a code in $X_N$ formed from i.i.d. codewords sampled  according to $\mu_N$, and let $\delta>0$ be fixed. Consider the lower tail probability:
    \[\calT_N(\mu,\delta) :=  \P\pp{d(\s{x},\s{y})\leq \delta N}, \quad (\s{x},\s{y})\sim \mu_N\otimes \mu_N,\]
    where $\mu_N\otimes \mu_N$ denotes two independent copies of $\mu_N$, and let 
    \[\calT_{\mu}(d,\delta)= -\limsup_{N\to \infty} \frac{\ln(\calT_N(\mu_N,\delta))}{N}. \]
Then for $N\to\infty$,if  
    $$
    \frac{\ln L_N}{N}\leq \frac{\calT_{\mu}(d,\delta)}{2}-o(1)
    $$
we have 
  $$\P_{\mu_N}\pp{d\p{\s{C}_N} >\delta N }=1-o(1). 
  $$
    Furthermore, if
\[
\frac{\ln L_N}{N}\leq \calT_{\mu}(d,\delta)-o(1)
\]
then, with probability $1-o(1)$, the code $\s{C}_N$ contains a subcode $\s{C}'_N\subseteq\s{C}_N$ of size $|\s{C}'_N|=(1-o(1))L_N$ such that $d(\s{C}'_N)\geq \delta N$.
\end{prop} 
More precisely, there exists an $o(1)$ function such that the claims of this proposition hold. We will follow this convention in other similar statements below. 

For better readability, below we drop $N$ from the notation $X_N,L_N$, $\mu_N$, and $\s{C}_N$, 
although the derivations apply to sequences of metric spaces and codes.
\begin{proof}
     For $0\leq i < j \leq  L-1$, let $A_{i,j}$ be the event that $d\p{\s{x}_i,\s{x}_j}\leq \delta N$.  By definition of $T_{\mu}(\delta)$ there is a decaying function $g(N)\to 0$ such that for all $i\neq j$ 
    \[\frac{\ln \P[A_{i,j}]}{N}=\frac{\ln \calT_N(\mu,\delta)}{N}\leq -\calT_{\mu}+g(N). \]
    Consider the function $\Tilde{g}(N)=g(N)/2+1/{\sqrt{N}}=o(1)$. Assuming that $\frac{\ln L }{N}\leq \calT_{\mu}(d,\delta)/2-\Tilde{g}(N)$, by the union bound we have
    \begin{align*}
        \P\pp{d(\s{C})\leq \delta N}&=\P\Big[{\bigcup_{i<j} A_{i,j}}\Big]\leq \sum_{i<j}\P\pp{A_{i,j}}=\binom{L}{2}\calT_N(\mu,\delta)\\
        &\leq L^2\exp\p{N\p{-\calT_\mu(d,\delta) +g(N)}}\\
        &=\exp\p{N\p{2\frac{\ln L}{N}-\calT_\mu(d,\delta) +g(N)}}\\
        &\leq \exp\p{-2\sqrt{N}}=o(1).
    \end{align*}
    This completes the first part of the proof.

    For the second part of the proof, we perform an expurgation process: Let $\s{B}$ be the set of all pairs $(\s{x}_i,\s{x}_j)$,  $0 \leq i<j\leq L-1$  such that $d(\s{x}_i,\s{x}_j)\leq \delta N$. We show that with probability $1-o(1)$, $|\s{B}|/L =o(1) $, for some $o(1)$ function, and then conclude by removing one element from each such pair in $|\s{B}|$ to obtain $\s{C}'$ that satisfies $d(\s{C}')>\delta N$. Indeed, note that 
    \[\E\pp{\frac{|\s{B}|}{L}}=\frac{1}{L}\sum_{0\leq i<j\leq L-1}\E\pp{\Ind_{\ppp{d(\s{x}_i,\s{x_j})\leq \delta N}}}=\frac{1}{L}\binom{L}{2}\calT_N(\mu,\delta)=\frac{(L-1)\calT_N(\mu,\delta)}{2}. \]
    Combining the above with Markov's inequality, for $\varepsilon>0$ we have 
    \begin{align*}
        \P\pp{\frac{|\s{B}|}{L} > \varepsilon}&\leq \frac{1}{\varepsilon}\E\pp{\frac{|\s{B}|}{L}}\leq \frac{L\calT_N(\mu,\delta)}{2\varepsilon}=\exp\p{N\p{\frac{\ln \calT_{N}(\mu,\delta)}{N}-\frac{\ln{2\varepsilon}}{N}+\frac{\ln L}{N} }}\\
        &\leq \exp\p{N\p{\frac{\ln L}{N} -\frac{\ln{2\varepsilon}}{N}-\calT_{\mu}(d,\delta)+g(N)  }}.
    \end{align*}
    In particular, for $\varepsilon=1/N =o(1)$ and $\hat{g}(N)=g(N)+1/\sqrt{N}=o(1)$, if $\frac{\ln L}{N}\leq \calT_{\mu}(d,\delta) - \hat{g}(N)$ we have 
    \[\P\pp{\frac{|\s{B}|}{L} > \frac{1}{N}}\leq \exp\p{-\sqrt{N}+\ln(N/2)}=o(1).\]
    This completes the proof.
\end{proof}
This proposition represents a ``meta-statement'' that we use several times below for distinct metric spaces (the differences are purely computational, related to the ball size).

A special case of interest is when $\mu$ is the uniform measure, where the lower tail probability is given by the average ball size.
\begin{lemma}\label{lem:AvBallSize}
     Assume that $\mu=\mu_{\s{U}}$ is the uniform measure on $X$. Let $B_d^{\s{Av}}(t)$ be the average ball size on $X$:
    \[B_d^{\s{Av}}(t)=\frac{1}{|X|}\sum_{x\in X}|B_d(t,x)|=\E_{\s{x}\sim\mu}\pp{|B_{d}(t,\s{x})|}, \quad B_d(t,x)=\abs{\ppp{y\in X ~:~ d(x,y)\leq t}}.\] 
    For $\delta>0$, 
    \[\calT_{\mu}(d,\delta)=\limsup_{N\to\infty} \frac{1}{N}\p{\ln|X|-\ln{B_d^{\s{Av}}(\delta N)}}.\]
\end{lemma}
\begin{proof}
    The proof is straightforward. Note that for uniform independent $\s{x},\s{y}$ we have
    \begin{align*}
        \P\pp{d(\s{x},\s{y})\leq \delta N }&=\frac{\abs{\ppp{(x,y)\in X^2 ~:~ d(x,y)\leq \delta N}}}{|X|^2}=\frac{1}{|X|^2}\sum_{x\in X}\sum_{y\in X} \Ind_{\ppp{d(x,y)\leq \delta N}}\\
        &=\frac{1}{|X|^2}\sum_{x\in X} |B_d(\delta N, x)|=\frac{B_{d}^{\s{Av}}(\delta N)}{|X|}.
    \end{align*}
    The result follows immediately from the definition of $\calT_{\mu}(d,\delta)$.
\end{proof}

\subsection{Random $q$-ary codes}

We now specialize the general i.i.d. random-coding statement to codes over the alphabet $[q]$. Throughout this subsection, let $\mu_{\s{U}}$ denote the uniform measure on $[q]^N$, and let
\begin{equation}
     \s{C}=\ppp{\s{x}_0,\dots,\s{x}_{L-1}}\label{eq:RandomQaryCode}
\end{equation}
be the random code obtained by choosing $\s{x}_0,\dots,\s{x}_{L-1}$ independently according to $\mu_{\s{U}}$. Since the ambient space has size $q^N$, we express rates through the normalized quantity $\frac{\log_q L}{N}$ using base-$q$ logarithms. We first consider the Hamming metric.

\subsubsection*{Random $q$-ary codes in the Hamming metric}

Throughout this subsection, distances are measured with respect to the scaled Hamming metric specified in \eqref{eq:HammingScaled}. Define
\begin{equation}
    R^{\s{H}}(\delta):=1-H_q(2\delta),\label{eq:GVrateHamming}
\end{equation}
where $H_q$ is the $q$-ary entropy function defined in \eqref{eq:qaryEntropy}. The following lemma recovers the standard Gilbert--Varshamov random-coding behavior in this setting: an i.i.d. uniform random code achieves one half of the Gilbert--Varshamov rate without expurgation, while the full Gilbert--Varshamov rate is obtained after deleting one endpoint from each close pair. The binary case was earlier studied in ~\cite[Theorem~2.1]{barg2002random}, and the extension to arbitrary alphabet size $q$ is straightforward.

\begin{lemma}[Random codes in the Hamming metric]
\label{lem:random-qary-hamming}
Let $\s{C}$ be the random code in \eqref{eq:RandomQaryCode}, with distance given by the scaled Hamming metric in \eqref{eq:HammingScaled}. Fix $\delta\in\p{0,\frac{1}{2}(1-\frac{1}{q})}$. 
For $N\to\infty$ if
\[
\frac{\log_q L}{N}\leq \frac{1}{2}R^{\s{H}}(\delta)-o(1),
\]
then $\P\pp{d(\s{C})\geq \delta N}=1-o(1)$. Furthermore, if
\[
\frac{\log_q L}{N}\leq R^{\s{H}}(\delta)-o(1),
\]
then, with probability $1-o(1)$, the code $\s{C}$ contains a subcode $\s{C}'\subseteq\s{C}$ of size $|\s{C}'|=(1-o(1))L$ such that $d(\s{C}')\geq \delta N$.
\end{lemma}
\begin{proof}
The proof follows as a straightforward application of Proposition~\ref{prop:GV}, Lemma~\ref{lem:AvBallSize}, and the standard estimate for $q$-ary Hamming balls: for $0\leq p\leq 1-\frac{1}{q}$, the volume of a Hamming ball of radius $pN$ in $[q]^N$ is $q^{N(H_q(p)+o(1))}$; see \cite[p.105]{roth2006intro}. Accounting for the $\frac12$ factor in \eqref{eq:HammingScaled}, we use this estimate
for the ball of radius $2\delta N$ whose size scales as $q^{N(H_q(2\delta)+o(1))}$; see \eqref{eq:GVrateHamming}.
\end{proof}
\subsubsection*{Random $q$-ary codes against deletions}

Codes against deletions are part of the broader theory of codes for synchronization errors, with motivations ranging from communication channels with insertions and deletions to DNA storage; see, e.g., the survey~\cite{cheraghchi2021overview}. Here, we focus on random codes in the asymptotic regime. Our analysis will be derived from the typical longest common subsequence behavior of pairs of independent random words.

\begin{definition}[Longest common subsequences]\label{def:Chvvatal-SankoffConst}
Let $\ux,\uy\in[q]^N$. A word $\uz\in[q]^\ell$ is a common subsequence of $\ux$ and $\uy$ if there exist indices $0\leq i_1< \cdots<i_\ell\leq N-1$ and $1\leq j_1< \cdots<j_\ell\leq N-1$ such that $z_r=x_{i_r}=y_{j_r}$ for every $r\in[\ell]$. We denote by $\lcs(\ux,\uy)$ the maximal such $\ell$ and define the deletion distance to be
\[d_{\s{Del}}(\ux,\uy)=N-\lcs(\ux,\uy).\] Let $\mu=(\mu_N)_{N\geq 1}$ be a sequence of probability measures with $\mu_N$ supported on $[q]^N$. We define
\[ \lcs_{\mu}(q,N):=\mathbb{E}_{\mu_N^{\otimes 2}}\pp{\lcs(\su{x},\su{y})}, \]
where $\mu_N^{\otimes 2}$ is the the joint law of two independent copies of $\mu_N$. Whenever the limit exists, we also define
\[ \lcs_{\mu}(q):=\lim_{N\to\infty}\frac{\lcs_{\mu}(q,N)}{N}. \]
\end{definition}

We will use this notation for $\mu_{\s{U}}$, the sequence of uniform measures on $[q]^N$, where the numbers $(\lcs_{\mu_{\s{U}}}(q))_q$ are usually called the Chv\'atal-Sankoff constants. The following standard properties of these constants serve as the basis for our random-code construction.

\begin{lemma}[LCS of uniform random words]
\label{lem:LCS-uniform}
For every fixed $q\geq 2$, the limit
\[ \gamma_q :=  \lcs_{\mu_{\s{U}}}(q)=\lim_{N\to\infty}\frac{\lcs_{\mu_{\s{U}}}(q,N)}{N} \]
exists and satisfies $ \gamma_q<1$. Moreover,
\[ \lim_{q\to\infty}\sqrt{q}   \gamma_q=2. \]
\end{lemma}

The existence of the limit and the proof that $\gamma_q$ is strictly smaller than $1$
go back to Chv\'atal and Sankoff~\cite[Theorem~1]{chvatal1975longest}. The large-alphabet asymptotic is due to Kiwi, Loebl, and Matou\v{s}ek~\cite[Corollary~2]{kiwi2005expected}. Let $H_q^*(x)=H_q(\min(x,1-\frac{1}{q}))$ be the modified $q$-ary entropy function. For a fixed $q\geq2 $ and $0<\delta<1$, consider the function 
\begin{equation}
    R^{\s{Del}}(\delta) :=  \max\Big(\frac{\max(1-\delta- \gamma_q,0)^2}{\ln{q}},1+\delta-2H_q^*(\delta)\Big),\label{eq:RateGVdel}
\end{equation}
where $\delta$ is the relative proportion of deletions in the code block of length $N$. For most values of $\delta$, the  term $1+\delta-2H_q^*(\delta)$ dominates the quadratic term and therefore determines the value of $R^{\s{Del}}(\delta)$. However, this term becomes nonpositive at a smaller value of $\delta$ than the quadratic term. The maximum in \eqref{eq:RateGVdel} therefore extends the  positive-rate guarantee to a larger segment of values of $\delta$.

\begin{lemma}\label{lem:iidCodesDel}
    Consider the uniform i.i.d. code of \eqref{eq:RandomQaryCode}. If
    \[ \frac{\log_{q}L}{N}\leq \frac{1}{2}R^{\s{Del}}(\delta)-o(1). \]
    then $\P\pp{d_{\s{Del}}(\s{C})>\delta N}=1-o(1)$. Furthermore, if
\[
\frac{\log_q L}{N}\le R^{\s{Del}}(\delta)-o(1),
\]
then, with probability $1-o(1)$, the code $\s{C}$ contains a subcode $\s{C}'\subseteq\s{C}$ of size $|\s{C}'|=(1-o(1))L$ such that $d_{\s{Del}}(\s{C}')\geq \delta N$.
\end{lemma}

\begin{proof}
The proof follows from Proposition~\ref{prop:GV} upon bounding the lower tail probability $\calT_{\mu}(d_{\s{Del}},\delta)$ in two different ways, which gives the maximum in \eqref{eq:RateGVdel}. First, note that if $\su{x}=(\su{x}_0,\dots ,\su{x}_{N-1})$ and $\su{y}=(\su{y}_0,\dots ,\su{y}_{N-1})$ are i.i.d. uniform vectors in $[q]^N$ then the random variables $\su{x}_0,\dots ,\su{x}_{N-1},\su{y}_0,\dots ,\su{y}_{N-1}$ are i.i.d. uniform over $[q]$. Additionally, the function $\lcs$ (as a function of $2N$ variables) has bounded-differences since in each coordinate the variance is bounded by $1$. Thus, by McDiarmid’s inequality (see Lemma~\ref{lem:McDiarmid}) and Lemma~\ref{lem:LCS-uniform} for $\delta \leq1-\gamma_q$ we have
\begin{align*}
    \P\pp{d_{\s{Del}}(\su{x},\su{y})\leq \delta N}&=\P\pp{\lcs(\su{x},\su{y})> (1-\delta) N}\\
    &=\P\pp{\lcs(\su{x},\su{y})-\gamma_q N> (1-\delta-\gamma_q) N}\\
    &\leq \P\pp{|\lcs(\su{x},\su{y})-\gamma_q N|> (1-\delta-\gamma_q) N}\\
    &\leq 2\exp\p{-N(1-\delta -\gamma_q)^2}.
\end{align*}
In particular
\begin{equation}
    \calT_{\mu}(d_{\s{Del}},\delta)\geq  (1-\delta -\gamma_q)^2.    \label{eq:Taudelta1}
\end{equation}

On the other hand, by Lemma~\ref{lem:AvBallSize}, $\calT_{\mu}(d,\delta)$ is given by the asymptotics of the average volume of the ``deletion ball''. 
Using the result of \cite[Eq. (6)]{levenshtein2002bounds},  we have
\[\limsup_{N\to\infty}\frac{1}{N}\log_q B^{\s{Av}}_{d_{\s{Del}}}(\delta N)\leq 2H_q^*(\delta)-\delta.\]
Combining these results with Lemma~\ref{lem:AvBallSize}, we obtain
\begin{equation}
    \calT_{\mu}(d_{\s{Del}},\delta)\geq \p{1+\delta-2H_q^*(\delta)}\ln q.\label{eq:Taudelta2}
\end{equation}
Taking together \eqref{eq:Taudelta1} and \eqref{eq:Taudelta2}, we obtain 
\[\calT_{\mu}(d_{\s{Del}},\delta)\geq \max\p{\p{1+\delta-2H_q^*(\delta)}\ln q,(1-\delta -\gamma_q)^2}.\]
Now Proposition~\ref{prop:GV} implies that, once 
   \begin{equation}
    \frac{\ln L}{N}\leq \frac{1}{2}\max\p{\p{1+\delta-2H_q^*(\delta)}\ln q,(1-\delta -\gamma_q)^2}-o(1), \label{eq:Shek1}
 \end{equation}
then the code distance satisfies $\P\pp{d_{\s{Del}}(\s{C})>\delta N}=1-o(1)$. Further, if
\begin{equation}
    \frac{\ln L}{N}\le \max\p{\p{1+\delta-2H_q^*(\delta)}\ln q,(1-\delta -\gamma_q)^2}-o(1),\label{eq:Shek2}
\end{equation}
then with probability $1-o(1)$, the code $\s{C}$ contains a subcode $\s{C}'\subseteq\s{C}$ of size $|\s{C}'|=(1-o(1))L$ such that $d_{\s{Del}}(\s{C}')\geq \delta N$. Changing the logarithm base from $e$ to $q$ in \eqref{eq:Shek1} and \eqref{eq:Shek2} matches the assumptions of the lemma and completes the proof. 
\end{proof}

\subsubsection*{Random $q$-ary codes in the $\ell_1$ metric}

We now present GV-type bounds for i.i.d. uniform classical codes with respect to the (scaled) $\ell_1$ metric \eqref{eq:ell1Dist}, where $\mu$ is the uniform distribution on $[q]^N$. The key input is the precise asymptotic evaluation of average $\ell_1$ ball sizes over   $[q]^N$, obtained in \cite[Section~V]{goyal2024gilbert}. We use this evaluation with the additional factor $1/2$ in our metric normalization. We fix $q\in \Z_0$ throughout this section and let $B_{d_1}^{\s{Av}}(\delta N)$ denote the average ball (of radius $\delta N$) size in $[q]^N$.

\begin{lemma}[Asymptotic volume of the $\ell_1$ ball, {\cite[Section~V]{goyal2024gilbert}}]\label{lem:ell1avBall}
Let
\[ S_q(y) :=  q+2\sum_{j=1}^{q-1}(q-j)y^j.\]
For $0\leq \delta\leq (q^2-1)/6q$, let $y_\delta\in[0,1]$ be the unique solution, with the convention $y_0=0$, of
\[ \sum_{j=1}^{q-1}(q-j)(j-2\delta)y_\delta^j=q\delta. \]
Define
\[ \Tilde{ T}_q(\delta) :=  \ln S_q(y_\delta)-2\delta\ln y_\delta. \]
Then,
\[ \lim_{N\to\infty}\frac{1}{N}\ln B_{d_1}^{\s{Av}}(\delta N)= \begin{cases}
\widetilde T_q(2\delta)-\ln q & 0\leq \delta\leq \frac{q^2-1}{6q},\\
\ln q  & \delta\geq \frac{q^2-1}{6q}.
\end{cases} \]
\end{lemma}
Switching to base-$q$ logarithms, we define:
\begin{equation}
    R^{\ell_1}(\delta)=2-\log_q(S_q(y_{\delta}))+2\delta \log_q y_{\delta}. \label{eq:ell1GVrate}
\end{equation}
Combining Lemma~\ref{lem:ell1avBall} with Proposition~\ref{prop:GV} and Lemma~\ref{lem:AvBallSize}, we obtain the following achievable rates for uniform i.i.d. codes under the $\ell_1$ metric.  

\begin{lemma}\label{lem:ell1randomCodes}
    Consider the uniform i.i.d. code of \eqref{eq:RandomCode}. If 
    \[ \frac{\log_{q}L}{N}\leq \frac{1}{2}R^{\ell_1}(\delta)-o(1), \]
    then $\P\pp{d_{1}(\s{C})>\delta N}=1-o(1)$. Furthermore, if
\[\frac{\log_q L}{N}\le R^{\ell_1}(\delta)-o(1), \]
then, with probability $1-o(1)$, the code $\s{C}$ contains a subcode $\s{C}'\subseteq\s{C}$ of size $|\s{C}'|=(1-o(1))L$ such that $d_1(\s{C}')\geq \delta N$.
\end{lemma}

Another quantity of interest in the context of the $\ell_1$ metric is the exponential growth rate of the size of the $\ell_1$ ball centered at $\underline{0}$, which turns out to be related to the number of amplitude damping error operators on $q$-ary codes (this is discussed in Section~\ref{sec:Exampels}). In the following lemma, we give an exact characterization of this quantity using large deviations theory. 

\begin{lemma}\label{lem:ell1BallZero}
    Let $q\geq 2$ be fixed and $B^{\ell_1}_{q,N}(\delta N,\uzero)$ denote the ball of radius $\delta N$ centered in $\uzero_N\in [q]^N$ with respect to the scaled $\ell_1$ distance \eqref{eq:ell1Dist}. 
    Then
    \begin{equation}
        B_q(\delta) :=  \lim_{N\to \infty} \frac{1}{N}\log_q|B^{\ell_1}_{q,N}(\delta N,\uzero)| = \begin{cases}
            \log_q f_q(x_\delta) -2\delta \log_q(x_\delta ) & \delta \leq \frac{q-1}{4},\\
            1 & \text{otherwise,}
        \end{cases}\label{eq:Bqfunction}
    \end{equation}
    where $f_q(x)=1+x+\cdots + x^{q-1}$ and $x_{\delta}$ is the unique solution of the equation 
    \[2\delta=\frac{xf_q'(x)}{f_q(x)}\]
\end{lemma}

\begin{proof}
    The proof follows a classical `method-of-types' argument. For a vector $\ux\in[ q]^N$, the empirical distribution of $\ux$, $\mu_{\ux}\in [0,1]^q $ is  
    \[ \mu_{\ux}(a) := \frac{1}{N}\abs{\ppp{i\in[N]~:~x_i=a}},\qquad a\in[q]. \]
    Define the set $A_\delta\subseteq [0,1]^q$ as  
    \begin{equation}
        A_\delta :=  \Big\{\mu\in [0,1]^q ~:~ \sum_{a\in[q]} \mu(a)=1,~\sum_{a\in [q]}a \mu(a)\leq 2 \delta\Big\}.\label{eq:BallAdelta}
    \end{equation}
    
 For any $\ux\in [q]^N$, the following statements are equivalent:
   \begin{align*}
        \ux\in B_{q,N}^{\ell_1}(\delta N,\uzero)& \iff  d_1(\ux,0)=\frac{1}{2}\sum_{i=0}^{N-1} x_i\leq \delta N  \iff \frac{1}{2}\sum_{a\in[q]}  a\mu_{\ux}(a) \leq  \delta \iff \mu_{\ux} \in A_{\delta},
    \end{align*}
   Let $(\s{X}_0,\dots, \s{X}_{N-1})=\su{X}\sim \mu_{\s{U}}$ be a uniform vector on  $[q]^N$.
    By \eqref{eq:BallAdelta} we have 
    \begin{equation}
        B_{q,N}^{\ell_1}(\delta N,\uzero)=q^N\P\pp{\mu_{\su{X}}\in A_{\delta}}.\label{eq:BallProbA}
    \end{equation}
    Note that its coordinates $\s{X}_0,\dots,\s{X}_{N-1}$ are i.i.d. uniform on $[q]$. Thus, by Sanov's theorem (see Lemma~\ref{lem:Sanov}) we have 
    \begin{align}
        -\inf_{\nu\in A_{\delta}^{\circ}}D(\nu\|\mu_{\s{U}})\leq\liminf_{N\to\infty}\frac{1}{N}\log_q\P\pp{\mu_{\su{X}}\in A_{\delta}}\leq\limsup_{N\to\infty}\frac{1}{N}\log_q\P\pp{\mu_{\su{X}}\in A_{\delta}}\leq-\inf_{\nu\in\overline{A}_{\delta}}D(\nu\|\mu_{\s{U}}).\label{eq:SanovEq}
    \end{align}
    Note that $A_\delta$ is a closed and convex set, on which the function 
    \begin{equation}
        f(\nu)=D(\nu||\mu_{\s{U}})=\sum_{a\in[q]}\nu(a)\log_q\frac{\nu(a)}{\mu_{\s{U}}(a)}=1-H(\nu) \label{eq:SanovFunction}
    \end{equation}
    is continuous (where $H(\nu)=-\sum_{a}\nu(a)\log_q\nu(a)$ is the Shannon entropy). Combining \eqref{eq:BallProbA}, \eqref{eq:SanovEq}, and \eqref{eq:SanovFunction}, we formulate the asymptotic growth rate of the ball volume as an entropy maximization problem:
    \begin{equation}
        \lim_{N\to\infty} \frac{1}{N}\log_q|B_{q,N}^{\ell_1}(\delta N,\uzero)| =\sup_{\nu\in A_{\delta}}H(\nu)= \sup\Big\{H(\nu) : \nu \in  [0,1]^q,~\sum_{a\in [q]} \nu(a)=1,~\sum_{a\in [q]}a \nu(a)\leq 2\delta\Big\}.\label{eq:BallOpti}
    \end{equation}
The global optimum of $H(\nu)$ is attained for the uniform distribution $\mu_{\s{U}}$, which is contained in $A_\delta$ if
$\delta\geq \frac{q-1}{4}$. In this case, \eqref{eq:BallOpti} becomes $H_{q}(\mu_{\s{U}})=1$. For $\delta\leq \frac{q-1}{4}$, using concavity of $H$ and convexity of $A_{\delta}$, we have that the supremum is achieved on the boundary of $A_{\delta}$, and the problem
takes the form
    \[
\begin{aligned}  & H(\mu) \to\max\\
\text{subject to}\quad & \sum_{a\in[q]}\mu(a)=1,\\
& \sum_{a\in[q]}a\mu(a)=2\delta,\\
& \mu(a)\geq 0,\qquad a\in[q].
\end{aligned}\]
Using Lagrange multipliers, we obtain that the optimal distribution has the form $\mu_{\delta}(a)={x_{\delta}^{a}}/{f_q(x_{\delta})}$ for $ a\in[q],
$
where $f_q(x) :=  \sum_{a\in[q]}x^a$.
The parameter $x_{\delta}\in[0,1]$ is determined by the constraint
\[ 2\delta=\sum_{a\in[q]}a\mu_{\delta}(a)=\frac{x_{\delta}f_q'(x_{\delta})}{f_q(x_{\delta})}. \]
The solution $x_{\delta}$ exists and is unique because 
$ \frac{x f_q'(x)}{f_q(x)} $ equals 0 at $x=0$ and is strictly increasing on $(0,1]$.
Substituting the maximizing distribution gives
   \[ 
   H(\mu_{\delta})
= \log_q f_q(x_{\delta})-2\delta\log_q x_{\delta}. \qedhere
\]
\end{proof}

Lemma~\ref{lem:ell1BallZero} gives the exponential size of an $\ell_1$ ball centered at $\uzero$ when the coordinates are restricted to the finite one-sided alphabet $[q]$. The natural two-sided infinite-alphabet analog is obtained by replacing $[q]$ with $\Z$, so that the ball consists of all signed integer vectors whose total $\ell_1$ weight is bounded. This is precisely the large-alphabet Lee-ball regime studied by Gardy and Sol\'e~\cite[Theorem~7]{gardy1991saddle}. We will use this exponent in Section~\ref{sec:Exampels} to describe the growth rate of the number-shift and phase-rotation error sets for constant-excitation Fock-state code.
\begin{lemma}[Signed scaled-$\ell_1$ balls over $\Z$ {\cite[Theorem~7]{gardy1991saddle}}]\label{lem:signedL1BallZero}
Let
\[
B_{\Z,n}^{\ell_1}(r,\uzero) := \ppp{\ux\in\Z^n~:~d_1(\ux,\uzero)=\frac{1}{2}\norm{\ux}_1\leq r}.
\]
Then, for $n=\alpha N$ and $r=\delta N$,
\[
B_{\Z,\alpha}(\delta) := \lim_{N\to\infty}\frac{1}{N}\ln\abs{B_{\Z,\alpha N}^{\ell_1}(\delta N,\uzero)}
=2\delta\ln\frac{2\delta}{\alpha}+\alpha\ln\p{\frac{2\delta}{\alpha}+\sqrt{1+\frac{4\delta^2}{\alpha^2}}}-2\delta\ln\p{\sqrt{1+\frac{4\delta^2}{\alpha^2}}-1}.
\]
\end{lemma}

\subsection{Codes for the golden mean shift}

We next consider binary runlength-limited codes, a classical family of constrained codes used in recording and storage systems; see, e.g., \cite{lind2021introduction}. The $(1,\infty)$-RLL system, also known as the \textit{golden mean shift}, is the set of binary words 
that do not contain adjacent ones:
\begin{equation}
    X_{N}^{\s{GM}} :=  \ppp{\ux\in\ppp{0,1}^{N}~:~(x_i, x_{i+1})\neq (1,1) \text{ for all } i\in [N-1]}.\label{eq:GMspace}
\end{equation}
We equip $X_{N}^{\s{GM}}$ with the scaled Hamming metric \eqref{eq:HammingScaled} and consider the uniform distribution on $X_{N}^{\s{GM}}$. Note that the literature references given in this section speak of the $(0,1)$-RLL system (of all strings with no two consecutive zeros), which is equivalent to the $(1,\infty)$-RLL system by flipping the bits $0\leftrightarrow1$. 

\begin{lemma}[Capacity and average Hamming ball size for the $(1,\infty)$-RLL system]\label{lem:RLLavBall}
The $(1,\infty)$-RLL system satisfies
\[ C_{\s{GM}} :=  \lim_{N\to\infty}\frac{1}{N}\ln|X_{N}^{\s{GM}}|= \ln\Big(\frac{1+\sqrt{5}}{2}\Big). \]
Additionally, for $z\in[0,1]$, let $\lambda(z)$ be the largest positive root of
\[ \lambda^3-(1+z)\lambda^2-(1+z)\lambda+z=0. \]
For $0\leq \delta\leq 1/5$, let $z_\delta\in[0,1]$ be the unique solution of the differential equation
\[ 2\delta=\frac{z_\delta\lambda'(z_\delta)}{\lambda(z_\delta)}.\]
 Then
\[\lim_{N\to\infty}\frac{1}{N}\ln B_{d}^{\s{Av}}(\delta N) =\begin{cases}\ln\lambda(z_\delta)-2\delta\ln z_\delta-C_{\s{GM}} & 0\leq \delta\leq \frac{1}{5},\\
C_{\s{GM}} & \delta\geq \frac{1}{5}.
\end{cases}
\]
\end{lemma}
The  capacity value is a well-known result in discrete dynamical systems \cite[Example 4.1.4.]{lind2021introduction}. The asymptotic average-ball calculation follows from the general average-ball evaluation method for constrained systems in \cite[Theorem~2]{goyal2024evaluating}, originally due to Kolesnik and Krachkovsky \cite{kolesnik1991generating}. 

Define the $(1,\infty)$-RLL rate function to be 
\begin{equation}
    \label{eq:RLLrateBinary}
    R^{\s{GM}}(\delta)=\begin{cases}
    2\log_2\p{\frac{1+\sqrt{5}}{2}}-\log_2\lambda(z_\delta)+2\delta\log_2 z_\delta & 0\leq \delta\leq \frac{1}{5},\\
    0 & \delta\geq \frac{1}{5}.
\end{cases}
\end{equation}
Combining the above evaluation with Proposition~\ref{prop:GV} and Lemma~\ref{lem:AvBallSize} we obtain the following result.
\begin{lemma}\label{lem:RLLcodes}
    Consider the uniform i.i.d.  code $\s{C}_N$ on $X_{N}^{\s{GM}}$ of size $L$. If
    \[ \frac{\log_{2}L}{N}\leq \frac{1}{2}R^{\s{GM}}(\delta)-o(1), \]
    then $\P\pp{d(\s{C})>\delta N}=1-o(1)$. Furthermore, if
\[\frac{\log_2 L}{N}\le R^{\s{GM}}(\delta)-o(1), \]
then, with probability $1-o(1)$, the code $\s{C}$ contains a subcode $\s{C}'\subseteq\s{C}$ of size $|\s{C}'|=(1-o(1))L$ such that $d(\s{C}')\geq \delta N$.
\end{lemma}

\subsection{Simplex codes}

We next review the i.i.d. simplex codes introduced in \cite{elimelech2026asymptotically}. In this setting, codewords are drawn independently from distributions over the discrete simplex $\calS_{q,N}$. We consider two natural choices: the uniform distribution $\mu_{\s{U}}$ on $\calS_{q,N}$ and the multinomial distribution $\mu_{\s{M}}$, both discussed in \cite{elimelech2026asymptotically}. Throughout this subsection, we focus on the linear-alphabet regime $q=\alpha N$ for a fixed $\alpha>0$, and present the Gilbert--Varshamov bounds obtained from Proposition~\ref{prop:GV} by evaluating the corresponding average ball sizes.

\subsubsection{Uniform simplex codes}\label{sec:uniformSimples}

Recall the definition of the discrete simplex:
\[
\calS_{q,N} := \ppp{\ux\in\mathbb Z_{0}^{q}:\sum_{i=1}^{q}x_i=N},
\]
and let $\mu_{\s{U}}$ be the uniform distribution on $\calS_{q,N}$. We assume that $q=\alpha N$ and consider the scaled $\ell_1$ metric $d_1$ in \eqref{eq:ell1Dist}, which equals one half of the standard $\ell_1$ distance (used here because the standard distance on $\calS_{q,N}$ is always even).

\begin{lemma}[Asymptotic $\ell_1$ ball size for uniform simplex codes, {\cite[Proposition~9]{goyal2024gilbert}}]\label{lem:SimplexUnifAvBall}
Let
\[ C_{\alpha}^{\s{U}} :=  (1+\alpha)\ln(2)H_2\p{\frac{\alpha}{1+\alpha}}=\lim_{N\to \infty}\frac{\ln |\calS_{\alpha N,N}|}{N}. \]
For $0\leq \delta\leq (1+\alpha)/(2+\alpha)$, define
\begin{align*}
    \Tilde{T}_{\alpha}^{\s{U}}(2\delta)& := 
-\alpha+2\delta\ln(2\delta)-\alpha\ln\p{\sqrt{\alpha^2+4\delta^2}-2\delta}-2\delta\ln\p{\sqrt{\alpha^2+4\delta^2}-\alpha}\\
&\qquad
+(1+\alpha-\delta)\ln(2+2\alpha-2\delta)-(1-\delta)\ln(2-2\delta).
\end{align*}
Then,
\[ \lim_{N\to\infty}\frac{1}{N}\ln B_{d_1}^{\s{Av}}(\delta N)=
\begin{cases}
\Tilde{T}_{\alpha}^{\s{U}}(2\delta)-C_{\alpha}^{\s{U}} & 0\leq \delta\leq \frac{1+\alpha}{2+\alpha},\\
C_{\alpha}^{\s{U}} & \delta\geq \frac{1+\alpha}{2+\alpha}.
\end{cases}\]
\end{lemma}

For a fixed $\alpha>0$, define 
\begin{equation}
    R_{ \alpha}^{\s{U}}(\delta)= \begin{cases}
2C_{\alpha}^{\s{U}} -\Tilde{T}_{\alpha}^{\s{U}}(2\delta)& 0\leq \delta\leq \frac{1+\alpha}{2+\alpha},\\
0 & \delta\geq \frac{1+\alpha}{2+\alpha}.
\end{cases}\label{eq:UniformSimplexrate}
\end{equation}

\subsubsection{Multinomial simplex codes}\label{sec:Multinomial}

We next consider the multinomial distribution on $\calS_{q,N}$. Let $\s{Y}_1,\ldots,\s{Y}_N$ be i.i.d. random variables uniformly distributed on $[q]$, and define $\s{\ux}\in\calS_{q,N}$ by
\[\s{x}_i :=  \sum_{j=1}^{N}\1\pp{\s{Y}_j=i},\quad i=1,\ldots,q. \]
We denote the distribution of $\su{x}$ by $\mu_{\s{M}}$. Equivalently, $\mu_{\s{M}}$ is obtained by independently placing $N$ balls into $q$ bins and recording the occupancy vector. We again assume $q=\alpha N$ and use the scaled $\ell_1$ metric $d_1$ in \eqref{eq:ell1Dist}.

\begin{lemma}[Lower tail for multinomial simplex codes, {\cite[Proposition~17 and Lemma~18]{elimelech2026asymptotically}}]\label{lem:MultTail}
Let $\mu_{\s{M}}$ be the multinomial distribution on $\calS_{\alpha N,N}$, 
and let
\[ \Delta_\alpha :=  \frac{1}{\pi}\intop_{0}^{\pi}e^{-\frac{2(1-\cos\theta)}{\alpha}}(1-\cos\theta)d\theta. \]
 Then, for every $0\leq\delta<\Delta_\alpha$,
\[
\calT_N(\mu_{\s{M}},\delta)\leq \exp\Big(-\frac{N}{4}\p{(\Delta_\alpha-\delta)^2+o(1)}\Big).
\]
Consequently,
\begin{equation}
    R_{\alpha}^{\s{M}}(\delta) :=  \calT_{\mu_{\s{M}}}(d_1,\delta)=-\limsup_{N\to\infty}\frac{1}{N}\ln\calT_N(\mu_{\s{M}},\delta)
\geq \frac{(\Delta_\alpha-\delta)^2}{4}. \label{eq:MultinomialRate}
\end{equation}
\end{lemma}

Taking together the results of the last two subsections, we obtain a bound on the attainable rate of codes in the simplex.
\begin{lemma}\label{lem:iidCodesSimplex}
    Let $\alpha>0$ be fixed and consider the random     
     simplex code generated by the measure $\mu_{\s{*}}$ on $\calS_{\alpha N, N }$, where $\s{*}\in \ppp{\s{M},\s{U}}$. If
    \[ \frac{\ln L}{N}\leq \frac{1}{2}R^{\s{*}}_{\alpha}(\delta)-o(1), \]
    then $\P\pp{d_{1}(\s{C})>\delta N}=1-o(1)$. Furthermore, if
\[
\frac{\ln L}{N}\le R^{\s{*}}_{\alpha}(\delta)-o(1),
\]
then, with probability $1-o(1)$, the code $\s{C}$ contains a subcode $\s{C}'\subseteq\s{C}$ of size $|\s{C}'|=(1-o(1))L$ such that $d_{1}(\s{C}')\geq \delta N$. 
\end{lemma}

The proof is a direct consequence of Proposition~\ref{prop:GV} and Lemmas~\ref{lem:AvBallSize}, \ref{lem:SimplexUnifAvBall}, and \ref{lem:MultTail}.

\section{Technicalities and auxiliary results} 

\subsection{State and channel norms}\label{app: Norms}
In this section, we briefly recall the definitions and results on state and channel norms used throughout this work.

\begin{definition}[State norms]\label{def:statenorms}
    Let $\calH$ and $\calH'$ be finite-dimensional Hilbert spaces, and let $X:\calH \to \calH'$ be a linear operator. For $p\geq 1$ we define the $p$-Schatten norm of $X$ as 
    \[
    \norm{X}_{p} :=  \tr\p{|X|^p}^{\frac{1}{p}}=
    \biggl({\sum_{\lambda\in \s{eig(X^\dag X)}}|\lambda|^{p/2}}\biggr)^{\frac{1}{p}},
    \]
where $|X|:=\sqrt{X^\dag X}$. For $p=\infty$  we define $\norm{ }_\infty$, a.k.a. the spectral norm, as
    \[\norm{X}_{\infty}=\lim_{p\to\infty} \norm{X}_p=\max_{\lambda\in \s{eig(X^\dag X)}}\sqrt{|\lambda|}=\sup_{0\neq \ket{x}\in \calH}\sqrt{\frac{\bra{x}X^\dag X\ket{x}}{\braket{x|x}}}.\]
    For $p=1$ and $p=2$, these norms are also called the {\em trace} and {\em Frobenius norm}, respectively.
\end{definition}

\begin{lemma}[Norm inequalities]\label{lem:stateNorms}
Let $\calH$ and $\calH'$ be a finite dimensional space, and let $X:\calH \to \calH'$ be a linear operator. Then the following statements hold:
\begin{enumerate}
    \item $\norm{ }_p$ is monotone non-increasing with $p$. In particular 
    \[\norm{X}_\infty\leq \norm{X}_2\leq \norm{X}_1.\]
    \item The duality principle {\rm (\cite[eq. (1.173)
, p.~33]{watrous2018theory})}: 
    
    For any $Y:\calH'\to \calH$ 
    \[\sup_{\norm{X}_1}\tr(XY)\leq \norm{Y}_\infty.\]
    \item The partial trace is a $\norm{ }_1$-contraction {\rm (\cite[Corollary 9.1.2]{wilde})}. Formally, if $\calH=\calH'=\calH_1\otimes \calH_2$ then 
    \[\max\p{\norm{\tr_{\calH_2}(X)}_1,\norm{\tr_{\calH_1}(X)}_1}\leq \norm{X}_1.\]
    \item H\"older's inequality for Schatten norms {\rm (e.g., \cite[p.33]{watrous2018theory})}: for $p,q\in [1,\infty]$ such that $\frac{1}{p}+\frac{1}{q}=1$ and an operator $Y:\calH''\to \calH$,
    \[\norm{XY}_1\leq \norm{X}_p\norm{Y}_q.\]
\end{enumerate}

\end{lemma}

 \begin{definition}[Diamond norm, completely bounded norm]\label{def:CbDiamondNorm}
    Let $\calH$ and $\calH'$ be finite dimensional Hilbert spaces, and let $\calN:L(\calH)\to L(\calH')$ be a superoperator. 
    \begin{enumerate}
        \item The diamond norm of $\calN$ is defined as 
        \[\norm{\calN}_{\diamond} :=  \sup_{\rho\in D(\calH^{\otimes 2})}\norm{I_{L(\calH)}\otimes \calN (\rho)}_{1}.\]
          \item The completely bounded (c.b.) norm of $\calN$ is defined as 
        \[\norm{\calN}_{\s{cb}} :=  \sup_{\substack{X\in L(\calH^{\otimes 2})\\ \norm{X}_1=1}}\norm{I_{L(\calH)}\otimes \calN (X)}_{1}.\]
    \end{enumerate}
    \end{definition}
    
    \begin{lemma}[Properties of the diamond and c.b. norms]\label{lem:cbDiamondProp}
    Let $\calN:L(\calH)\to L(\calH')$ be a superoperator. Then the following holds:
    \begin{enumerate}
        \item \label{item:HPdiamond} If $\calN$ is Hermitian preserving then {\rm (\cite[Theorem 3.51]{watrous2018theory})}
        \[\norm{\calN}_{\diamond}=\norm{\calN}_{\s{cb}}=\sup_{\substack{\ket{\psi}\in \calH^{\otimes 2}\\ \braket{\psi|\psi}=1}}\norm{I_{L(\calH)}\otimes \calN (\ket{\psi}\bra{\psi})}_{1}.\]
        
        \item \label{item:OptDia} Equivalent ancilla-dimension optimization {\rm (\cite[pp.241-242]{wilde} and \cite[Theorem 3.46]{watrous2018theory})}: 
        \[\norm{\calN}_{\diamond} =\sup_{n\in \N} \sup_{\rho\in D(\C^n\otimes \calH )}\norm{I_{L(\calH)}\otimes \calN (\rho)}_{1},\]
        and similarly
        \[\norm{\calN}_{\s{cb}} =\sup_{n\in \N} \sup_{\substack{X\in L(\C^{n}\otimes \calH)\\ \norm{X}_1=1}}\norm{I_{L(\calH)}\otimes \calN (\rho)}_{1}.\]
        \item \label{item:CPTPnorm}If $\calN$ is CPTP then $\norm{\calN}_{\diamond}=1$ {\rm (\cite[Proposition 3.44]{watrous2018theory})}.
        \item \label{item:CBmuliplicative}The completely bounded norm is submultiplicative with respect to compositions {\rm (\cite[Proposition 3.48]{watrous2018theory})}. Namely, if $\calM:L(\calH')\to L(\calH'')$ then  \[\norm{\calM \circ\calN}_{\s{cb}}\leq \norm{\calM }_{\s{cb}}  \norm{\calN}_{\s{cb}}.\]
    \end{enumerate}
        
    \end{lemma}
    \begin{lemma}\label{lem:compoDiamond}
        Let $U,V:\calH\to \calH'$ be operators, and consider the superoperator $\calT_{U,V}:L(\calH)\to L(\calH')$ defined by $\calT_{U,V}(X)= U X V^\dag$. Then 
        \[\norm{\calT_{U,V}}_{\s{cb}}\leq \norm{U}_\infty  \norm{V}_{\infty},\]
        and if $U=V$ equality holds. 
    \end{lemma}
    \begin{proof}
        Let $X\in \calH^{\otimes 2}$. Note that 
        \[I_{L(\calH)}\otimes\calT_{U,V}(X)=(I_{\calH}\otimes U)X (I_{\calH}\otimes V)^\dag.\]
        Thus, by the H\"older's inequality (Lemma~\ref{lem:stateNorms}) we have 
        \begin{align*}
            \norm{I_{L(\calH)}\otimes\calT_{U,V}(X)}_{1}&=\norm{(I_{\calH}\otimes U)X (I_{\calH}\otimes V)^\dag}_1\\
            &\leq \norm{I_{\calH}\otimes U}_\infty   \norm{(I_{\calH}\otimes V)^\dag}_\infty  \norm{X}_1\\
            &=\norm{U}_{\infty}  \norm{V}_{\infty}  \norm{X}_1,
        \end{align*}
        where the last equality follows since $I\otimes U$ and $I\otimes V$ have the same singular values as $U$ and $V$ respectively, and since the spectral norm is invariant under conjugation. If $U=V$, then equality is obtained by taking $X=\ket{\psi}\bra{\psi}\otimes \ket{\psi}\bra{\psi}$, where $\ket{\psi}$ is a normalized singular vector of maximal singular value of $U$.
    \end{proof}
\begin{lemma}\label{lem:diamondnormNACHS}
    Let $\calH,\calH'$ be Hilbert spaces and let $\calB:L(\calH)\to L(\calH')$ be a superoperator of the form $\calB(\rho)= \tr(\rho B)   \tau$, where $B\in L(\calH)$ and $\tau\in L(\calH')$ are linear operators. Then 
    \[\norm{\calB}_{\diamond}\leq \norm{B}_{\infty }  \norm{\tau}_1,\]
    and equality in the above holds if $B$ is a normal operator.
   \end{lemma}
   \begin{proof}
       Let $X$ be a linear operator on $\calH^{\otimes 2}$ with $\norm{X}_{1}\leq 1$. Note that if $X=X_1\otimes X_2$ is a product state, we have 
       \[I_{L(\calH)}\otimes \calB (X)=\tr(BX_2)  X_1\otimes \tau=\tr_2\big((I_{\calH}\otimes B )(X_1\otimes X_2)\big) \otimes \tau=\tr_2(I_{\calH}\otimes B )(X)) \otimes \tau.  \]
       Extending by linearity, the above holds for all $X$. In particular, using multiplicativity of the trace norm and the fact that the partial trace is a contraction of the trace norm (see Lemma~\ref{lem:stateNorms}), we have
       \begin{align*}
           \norm{I_{L(\calH)}\otimes \calB (X)}_1&=\norm{\tr_2(I_{\calH}\otimes B )X) \otimes \tau}\leq \norm{(I_{\calH}\otimes B )X}_1  \norm{\tau}_1\\
           &\leq \norm{I_{\calH}\otimes B }_{\infty}   \norm{X}_1  \norm{\tau}_1\leq \norm{I_{\calH}\otimes B }_{\infty}   \norm{\tau}_1=\norm{ B }_{\infty}   \norm{\tau}_1,
       \end{align*}
       where we used the duality principle (see Lemma~\ref{lem:stateNorms}) in the first inequality of the second line. This proves that 
       \[\norm{\calB}_{\diamond}\leq \norm{B}_{\infty}  \norm{\tau}_1.\]
       For the opposite inequality, assume that $B$ is normal, so there exists an eigenvector $\ket{\psi}$  of $B$ such that $\norm{B}_{\infty}=\bra{\psi}B\ket{\psi}$. Note that 
       \[\norm{\calB}_{\diamond}\geq \norm{\calB(\ket{\psi}\bra{\psi})}_1=\abs{\tr(B\ket{\psi}\bra{\psi})}  \norm{\tau }_1=\braket{\psi | B|\psi}  \norm{\tau}_1=\norm{B}_{\infty}  \norm{\tau}_1. \qedhere\]
   \end{proof}
   \begin{definition}[Block representation]\label{def:blockRep}
        Let $\calH$ be a Hilbert space such that $\calH$ is given by a direct sum $\calH=\calH_1\oplus\cdots \oplus\calH_n$, and let $X\in L(\calH)$. The block decomposition of $X$ is the set of operators $\p{X_{ij}}_{i,j}$, where $X_{ij}:\calH_i\to\calH_j$ is the orthogonal projection of $X$ on $L(\calH_i,\calH_j)$. Equivalently, let $\imath_i:\calH_i \to \calH$ be the natural embedding, then
        $X_{ij}=\imath_j^\dag X \imath_i$.
   \end{definition}
\begin{definition}[Direct sum maps]\label{def:directsumMap}
    Let $\calH$ and $\calH'$ be Hilbert spaces such that $\calH$ admits a direct sum decomposition $\calH=\calH_1\oplus \cdots \oplus\calH_n$. 
Let $X:\calH\to\calH$ be an operator and let  $(X_{ij})_{i,j}$ be its block representation. A superoperator $\calN:L(\calH)\to L(\calH')$ is called a {\em direct sum map} if it acts as 
    \begin{equation}\label{eq:direct sum}
    \calN(X)=\sum_{i=1}^n \calN_{i}(X_{ii}),
   \end{equation}
where each $\calN_i:L(\calH_i)\to    L(\calH')$ is a superoperator. Abusing notation, we write $\calN=\bigoplus_{i=1}^n \calN_i$.
\end{definition}

    \begin{lemma}\label{lem:directSumLemma}
        Let $\bigoplus_{i=1}^n \calN_i=\calN:L(\calH)\to L(\calH')$ be a direct sum map. Then, 
        \[\norm{\calN}_{\s{cb}}= \max_i \norm{\calN_i}_{\s{cb}}.\]
    \end{lemma}
\begin{proof}
    Let $\imath_i:\calH_i\to \calH$ be the embedding map as in Definition~\ref{def:blockRep}.  For an operator $X\in L(\calH^{\otimes 2})$ let $X_{i,i}$ denote the projection of $X$ on $L(\calH\otimes H_i)$. Note that for a tensor product element $X=W\otimes Y$, we have 
    \begin{align*}
        X_{i,i}=(I_\calH \otimes \imath_i)^\dag (W\otimes Y )(I_\calH \otimes \imath_i)=W\otimes Y_{i,i},
    \end{align*} where $Y_{i,i}$ is the projection of $Y$ on $L(\calH_i)$ (as in Definition~\ref{def:blockRep}). 
    \begin{align}
        I_{L(\calH)}\otimes \calN (X)&\nonumber =I_{L(\calH)}\otimes \calN (W\otimes Y)\\&\nonumber=\sum_i W\otimes \calN_i(Y_{i,i} )  \\
        &\nonumber=\sum_i   I_{L(\calH)}\otimes \calN_i(W\otimes Y_{i,i} )  \\
        &\nonumber=\sum_i  I_{L(\calH)}\otimes \calN_i\big( (I_{\calH}\otimes \imath_i^\dag) (W\otimes Y) (I_{\calH}\otimes \imath_i ) \big)  .\\
        &=\sum_i   I_{L(\calH)}\otimes \calN_i( X_{i,i} ) .\label{eq:Yashmin}
    \end{align}
    Extending \eqref{eq:Yashmin} using linearity, for any operator $X\in L(\calH^{\otimes 2})$ we obtain 
    \begin{align*}
        I_{L(\calH)}\otimes \calN (X)=\sum_i I_{L(\calH)}\otimes \calN_i (X_{i,i}),
    \end{align*}
    Using item \ref{item:OptDia} in Lemma~\ref{lem:cbDiamondProp} we obtain:
    \begin{align}
        \norm{I_{L(\calH)}\otimes \calN (X)}_1&\nonumber=\norm{\sum_i I_{L(\calH)}\otimes \calN_i (X_{i,i})}_1\\
        &\nonumber\leq \sum_i \norm{ I_{L(\calH)}\otimes \calN_i (X_{i,i})}_1\\
        &\leq \sum_i\norm{\calN_i}_{\s{cb}}   \norm{ (X_{i,i})}_1\label{eq:prodSum2}\\&
        \leq \p{\max_{j=1,\dots,n }\norm{\calN_j}_{\s{cb}} } \sum_i  \norm{ (X_{i,i})}_1.\label{eq:prodSum3} 
    \end{align}
    Here \eqref{eq:prodSum2} follows from item \ref{item:OptDia} of Lemma~\ref{lem:cbDiamondProp}.
    Next, we show that $\sum_{i}\norm{X_{i,i}}_1\leq \norm{X}_1$. To that end, consider the superoperator $\calM:L(\calH)\to L(\calH)$ which extracts the diagonal blocks: 
    \[\calM(X)=\sum_{i}Q_i XQ_i, \quad Q_i=I_{\calH}\otimes P_i,\]
    where $P_i$ is the orthogonal projection on $\calH_i$. Note that $\calM$ is a CPTP map since $\sum_{i}Q_i^\dag Q_i=\sum_iQ_i= I_{\calH^{\otimes 2}}$. Thus, by item \ref{item:CPTPnorm}. in Lemma~\ref{lem:cbDiamondProp}, $\norm{\calM}_{\s{cb}}=1$, and in particular 
    \begin{align}
        \norm{\calM(X)}_{1}\leq \norm{\calM}_{\s{cb}}  \norm{X}_1=\norm{X}_1.\label{eq:MchannelNOrm1}
    \end{align} 
    On the other hand, note that $Q_iXQ_i$ are supported on orthogonal spaces and therefore the singular values of $\calM(X)$, $\s{S}(\calM(X))$, are the union (in a multiset manner) of $\s{S}(Q_i X Q_i)$. In particular, 
    \begin{align}
        \norm{\calM(X)}_1=\sum_{\lambda\in \s{S}(\calM(X))}\lambda =\sum_{i}\sum_{\lambda_i\in \s{S}(Q_iXQ_i)}\lambda_i=\sum_{i}\norm{Q_iXQ_i}_1.\label{eq:normSumEq}
    \end{align}
    We observe that $Q_iXQ_i$ has the same singular values as $X_{i,i}$ and therefore has the same trace norm. Indeed, for $\ket{\psi}\in \calH_i$, a singular vector of $X_{i,i}$ with singular value $\lambda_i$, we have
    \begin{align*}
        Q_iXQ_iQ_iX^\dag Q_i \ket{\psi}&=(I_{\calH}\otimes \imath_i^\dag) XQ_iX^\dag (I_{\calH}\otimes \imath_i)  \ket{\psi}\\
        &=(I_{\calH}\otimes \imath_i^\dag) X (I_{\calH}\otimes \imath_i )(I_{\calH}\otimes \imath_i^\dag )X^\dag (I_{\calH}\otimes \imath_i)  X_{i,i}\ket{\psi}\\
        &=X_{i,i}X_{i,i}^\dag \ket{\psi}=\lambda_i^2 \ket{\psi},
    \end{align*}
    where the first equality follows since $I_{\calH}\otimes \imath_i$ and $I_{\calH}\otimes P_i=Q_i$ agree on $\calH\otimes \calH_i$, and since the adjoint of the embedding map is the projection on $\calH\otimes \calH_i$, and the second equality follows since $I_{\calH}\otimes \imath_i^\dag \imath_i=I_{\calH}\otimes P_i =Q_i$. Combining this observation with \eqref{eq:prodSum3}, \eqref{eq:MchannelNOrm1}, and \eqref{eq:normSumEq} we obtain: 
    \begin{align*}
        \norm{I_{L(\calH)}\otimes \calN (X)}_1&\leq \p{\max_{j=1,\dots,n }\norm{\calN_j}_{\s{cb}} } \sum_i  \norm{ (X_{i,i})}_1\\&=\p{\max_{j=1,\dots,n }\norm{\calN_j}_{\s{cb}} } \sum_i \norm{Q_iXQ_i}_1\\
        &=\p{\max_{j=1,\dots,n }\norm{\calN_j}_{\s{cb}} }  \norm{\calM(X)}_1\\
        &\leq \p{\max_{j=1,\dots,n }\norm{\calN_j}_{\s{cb}} }   \norm{X}_1.
    \end{align*}
    This proves the inequality \[\norm{\calN}_{\s{cb}}\leq \max_{j=1,\dots,n }\norm{\calN_j}_{\s{cb}} .\]

    For the opposite inequality, assume Without loss of generality that $\norm{\calN_1}_{\s{cb}}$ achieves the maximum, and let $X_1\in L(\calH_1^{\otimes2})$ be a norm-$1$ operator for which $\norm{\calN_1}_{\s{cb}}=\norm{I_{L(\calH_{1})}\otimes \calN_{1}(X_1)}_{1}$. Consider the operator $Y\in L(\calH^{\otimes 2})$ given by $P_1\otimes P_1 X_1 P_1\otimes P_1$. Using a technique similar to the first part of the proof, one can show that 
    \[ \norm{I_{L(\calH)}\otimes \calN(Y)}_1=\norm{I_{L(\calH_{1})}\otimes \calN_1(X_1)}_1\]
    and that $\norm{Y}_1=\norm{X_1}_1=1$. This shows that 
    \[\norm{\calN}_{\s{cb}}\geq  \norm{I_{L(\calH)}\otimes \calN(Y)}_1=\norm{I_{L(\calH_{1})}\otimes \calN_1(X_1)}_1=\norm{\calN_{1}}_{\s{cb}}.\]
\end{proof}
The next lemma shows that the Hellinger distance is nonincreasing under conjugation by a contraction.
\begin{lemma}\label{lem:HellingerContaction}
    Let $A,B\in L(\calH)$ be positive semidefinite operators and $C:\calH\to \calH'$ be an operator with $\norm{C}_\infty\leq 1$. Then:
    \[D_{\s{H}}(CAC^\dag, CBC^\dag)\leq D_{\s{H}}(A,B),\]
    where $D_{\s{H}}(A,B)=\|{\sqrt{A}-\sqrt{B}}\|_2$.
\end{lemma}
\begin{proof}
    We recall the data processing inequality for affinity (see \cite[Eq.\,(6)]{wilde2018recoverability}): if  $\calM$ is a CPTP map, then for all $A,B\succeq 0$ 
    \begin{equation}
        \tr(\sqrt{A}\sqrt{B})\leq  \tr\p{\sqrt{\calM(A)}\sqrt{\calM(B)}}. \label{eq:AffinityDataProcessing}
    \end{equation}
    Now consider a superoperator $\calM:L(\calH)\to L(\calH\oplus \calH')$ 
      (where $\calH\oplus \calH'$ is the direct sum of $\calH$ with $\calH'$) given by
    \[\calM(A)=CAC^\dag \oplus RAR^\dag, \quad R=\sqrt{I_{\calH}-C^\dag C}.\]
    First, note that $R$ is well defined since $\norm{C}_{\infty}\leq 1$ which implies that $I_{\calH}-C^{\dag}C \succeq 0$. We claim that $\calM$ is CPTP. Indeed, it has a Kraus representation with operators 
    \[\hat{C}=\begin{bmatrix}
        C \\ 0
    \end{bmatrix},\quad \hat{R}=\begin{bmatrix}
         0\\ R
    \end{bmatrix},\]
    and is therefore completely positive. On the other hand, for any $A$, 
    \begin{align*}
        \tr(\calM(A))&=\tr\p{CAC^\dag \oplus RAR^\dag}=\tr(CAC^\dag)+\tr(RAR^\dag)\\
        &=\tr(C^\dag C A)+\tr(R^\dag RA)=\tr((C^\dag C + R^\dag R)A)=\tr(A).
    \end{align*}
Next, we observe that $\calM$ maps any operator to a block-diagonal operator. In particular, for all positive $A,B\in L(\calH)$ we have 
\begin{align}
    D_{\s{H}}^2(\calM(A),\calM(B)) &\nonumber=\norm{\sqrt{\calM(A)}-\sqrt{\calM(B)}}_2^2\\&
    \nonumber=\norm{\sqrt{CAC^\dag \oplus RAR^\dag}-\sqrt{CBC^\dag \oplus RBR^\dag}}_2^2\\
    &=\norm{\sqrt{CAC^\dag} \oplus \sqrt{RAR^\dag}-\sqrt{CBC^\dag} \oplus\sqrt{ RBR^\dag}}_2^2\label{eq:oplusSqrt}\\
    &\nonumber=\norm{\p{\sqrt{CAC^\dag}-\sqrt{CBC^\dag}}\oplus \p{\sqrt{ RAR^\dag}-\sqrt{ RBR^\dag}}}_2^2\label{eq:aditivityNormOplus}\\
    &=\norm{\sqrt{CAC^\dag}-\sqrt{CBC^\dag}}_2^2+\norm{\sqrt{RAR^\dag}-\sqrt{RBR^\dag}}_2^2\\
    &\nonumber=D^2_{\s{H}}\p{CAC^\dag,CBC^\dag}+D^2_{\s{H}}\p{RAR^\dag,RBR^\dag}\\
    &\geq D^2_{\s{H}}\p{CAC^\dag,CBC^\dag},\label{eq:psitivityHer}
\end{align}
where in \eqref{eq:oplusSqrt} we used the additivity 
of square root under direct sum operation, in \eqref{eq:aditivityNormOplus} we used block-additivity of the squared Frobenius norm, and in \eqref{eq:psitivityHer} we used positivity of the Hellinger distance.  

On the other hand, 
\begin{align}
    D_{\s{H}}^2(\calM(A),\calM(B)) &\nonumber=\norm{\sqrt{\calM(A)}-\sqrt{\calM(B)}}_2^2=\tr\p{\p{\sqrt{\calM(A)}-\sqrt{\calM(B)}}^2}\\
    &\nonumber=\tr\p{\calM(A)}+\tr\p{\calM(B)}-2\tr\p{\sqrt{\calM(A)}\sqrt{\calM(B)}}\\
    &=\tr(A)+\tr(B)-2\tr\p{\sqrt{\calM(A)}\sqrt{\calM(B)}}\label{eq:TracePres}\\
    &\leq \tr(A)+\tr(B)-2\tr\p{\sqrt{A}\sqrt{B}}=D_{\s{H}}^2(A,B) \label{eq:Affinity}
\end{align}
where \eqref{eq:TracePres} follows since $\calM$ is trace-preserving and  \eqref{eq:Affinity} follows from \eqref{eq:AffinityDataProcessing}. The proof is concluded by combining \eqref{eq:psitivityHer} and \eqref{eq:Affinity}. \end{proof}

    \subsection{Auxiliary technical results}
    \begin{lemma}[Properties of the partial trace {\cite[pp.~107]{nielsen2010quantum}}] \label{lem:propPartTrace}
    Let $\calH=\calH_1\otimes \calH_2$ be a tensor product Hilbert space. For any  $\rho_1\in L(\calH_1) $ and $\tau\in L(\calH)$ we have 
    \[\tr((\rho_1\otimes I_{\calH_2}) 
    \tau)=\tr(\rho_1 \tau_1),\]  
    where $\tau_1$ is given by the partial trace:
    \[\tau_1=\tr_{\calH_2}(\tau).\]
    \end{lemma}

    \begin{lemma}[Hoeffding's inequality {\cite[Theorem 2]{hoeffding1963probability}; \cite[Lemma 2.2]{Boucheron2013}}]
\label{lem:Hoeffding}
Let $\s{X}_1,\dots,\s{X}_T$ be independent real-valued random variables. Assume that, for every $i\in[T]$, there exist constants $a_i,b_i\in\R$ such that
\[ a_i\leq \s{X}_i\leq b_i\]
almost surely. Then, for every $\varepsilon>0$,
\begin{equation*}
    \P\pp{\abs{\frac{1}{T}\sum_{i=1}^T \pp{\s{X}_i-\E[\s{X}_i]}}>\varepsilon} \leq 2\exp\p{ -\frac{2T^2\varepsilon^2}{\sum_{i=1}^T (b_i-a_i)^2}}.
\end{equation*}
\end{lemma}

\begin{lemma}[McDiarmid's inequality {\cite[Theorem 3.1]{mcdiarmid1989method}; \cite[Sec.~6.1]{Boucheron2013}}]\label{lem:McDiarmid}
Let $\s{X}_1,\dots,\s{X}_n$ be independent random variables with $\s{X}_i$ taking values in $A_i$, and let $f:A_1\times s\times A_n\to\R$ be measurable. Assume that $f$ satisfies the bounded differences condition with constants $c_1,\dots,c_n$, namely, for every $i\in[n]$ and every $\ux,\uy\in A_1\times \cdots\times A_n$ that differ only in the $i$-th coordinate, $\abs{f(\ux)-f(\uy)}\leq c_i$. Then, for every $\varepsilon>0$,
\[
\P\pp{\abs{\frac{1}{n}f(\s{X}_1,\dots,\s{X}_n)-\E\pp{\frac{1}{n}f(\s{X}_1,\dots,\s{X}_n)}}>\varepsilon}\leq 2\exp\p{-\frac{2n^2\varepsilon^2}{\sum_{i=1}^n c_i^2}}.
\]
\end{lemma}

\begin{lemma}[Sanov's Theorem {\cite[Theorem 2.1.10]{dembo2009large}}] \label{lem:Sanov}
Let $[q]=\ppp{0,1,\dots,q-1}$, let $\nu$ be a probability distribution on $[q]$, and let $\s{X}_1,\dots,\s{X}_n$ be i.i.d. random variables distributed according to $\nu$. For $\ux=(x_1,\dots,x_n)\in[q]^n$, let $\mu_{\ux}$ denote its empirical distribution, defined by
\[
\mu_{\ux}(a) := \frac{1}{n}\abs{\ppp{i\in[n]~:~x_i=a}},\qquad a\in[q].
\]
Then, for every set $A$ of probability distributions on $[q]$,
\[
-\inf_{\mu\in A^{\circ}}D(\mu\|\nu)\leq\liminf_{n\to\infty}\frac{1}{n}\ln\P\pp{\mu_{\s{X}_1,\dots,\s{X}_n}\in A}\leq\limsup_{n\to\infty}\frac{1}{n}\ln\P\pp{\mu_{\s{X}_1,\dots,\s{X}_n}\in A}\leq-\inf_{\mu\in\overline{A}}D(\mu\|\nu),
\]
where
\[
D(\mu\|\nu) := \sum_{a\in[q]}\mu(a)\ln\frac{\mu(a)}{\nu(a)},
\]
with the usual convention that $D(\mu\|\nu)=\infty$ if $\mu(a)>0$ for some $a$ satisfying $\nu(a)=0$. Here $A^{\circ}$ and $\overline{A}$ denote the interior and closure of $A$ with respect to the total-variation topology on the probability simplex over $[q]$.
\end{lemma}
    \section{Additional proofs}

    \subsection{Proof of Lemma~\ref{lem:replaceOpt}}\label{app:lemREPLACE}
        Let $\lambda$ be a matrix such that $\norm{\calB^{\calE}_{\lambda,Q}}_{\diamond}\leq \varepsilon$. Our goal is to show that we can replace the matrix $\lambda$ by $\lambda'$ given by $\lambda_{k,l}'=\bra{c_0}{E_l^\dag E_k} \ket{c_0}$, where $\ket{c_0}$ is an arbitrary (normalized) nonzero codeword such that $\norm{\calB^{\calE}_{\lambda',Q}}_{\diamond}\leq 2\varepsilon$. Indeed, let $B_{k,l}$ and $B_{k,l}'$ be defined as in \eqref{eq:BOoperators} with respect to $\lambda$ and $\lambda'$, respectively. We consider the superoperator $\calL$  defined as 
    \begin{equation}
        \calL(\rho)=\sum_{k,l}\tr\p{L_{k,l}P\rho}\ket{k}\bra{l}, \quad L_{k,l}=\lambda_{k,l}'-\lambda_{k,l},\label{eq:Amatrix}
    \end{equation}
    where $P$ is the projection on $Q$ as in \eqref{eq:BOoperators}.
    Note that 
    \begin{align*}
        B_{k,l}'&=P{E_l^\dag E_k} P -\lambda_{k,l}'P=P{E_l^\dag E_k} P -\lambda_{k,l}P-L_{k,l}P=B_{k,l}-L_{k,l}P,
    \end{align*}
    and in particular, for any $\rho$
    \[\calB^{\calE}_{\lambda',Q}(\rho)=\sum_{k,l}\tr(B_{k,l}'\rho)  \ket{k}\bra{l}=\sum_{k,l}\tr((B_{k,l}-L_{k,l}P)\rho)  \ket{k}\bra{l}=\calB^{\calE}_{\lambda,Q}(\rho)-\calL(\rho).\]

    Using the triangle inequality, we obtain 
\begin{align}
    \norm{\calB^{\calE}_{\lambda',Q}}_{\diamond}=\norm{\calB^{\calE}_{\lambda,Q}-\calL}_{\diamond}\leq \norm{\calB^{\calE}_{\lambda,Q}}_{\diamond}+\norm{\calL}_{\diamond}\leq \varepsilon+\norm{\calL}_{\diamond}, \label{eq:triangDiamond}
\end{align}
 Thus, it is sufficient to show that $\norm{\calL}_{\diamond}\leq \varepsilon$. Consider the matrix $L$ whose entries are given by $L_{k,l}$ defined in \eqref{eq:Amatrix}. We first prove that $\norm{\calL}_{\diamond}=\norm{L}_1$.  Note that for any product state $ \tau \otimes \eta \in \calH^{\otimes 2}$ we have 
 \begin{align*}
     I_{L(\calH)}\otimes \calL(\tau \otimes \eta)&=\sum_{k,l}\tr(P\eta)L_{k,l} \cdot     \tau\otimes \ket{k}\bra{l}=\tr(P \eta )  \tau \otimes\overset{L}{\overbrace{\p{\sum_{k,l}L_{k,l}   \ket{k}\bra{l}}}}\\&=  \tr_{\calH}\p{\tau  \otimes P \eta}\otimes L=\tr_{\calH}\p{\p{I_{\calH}\otimes  P} (\tau \otimes \eta)}\otimes L.
 \end{align*}
 Using linearity, we obtain that the above equality holds beyond pure tensor products, for any $\rho\in L(\calH^{\otimes 2})$. That is, 
 \[I_{L(\calH)}\otimes \calL(\rho)=\tr_{\calH}\p{\p{I_{\calH}\otimes  P} \rho}\otimes L.\]  In particular, we have
 \begin{align}
     \norm{\calL}_{\diamond}&\nonumber=\sup_{\rho\in D(\calH^{\otimes 2})}\norm{I_{L(\calH)}\otimes \calL(\rho)}_1\\
     &\nonumber=\sup_{\rho\in D(\calH^{\otimes 2})}\norm{\tr_{\calH}\p{\p{I_{\calH}\otimes  P} \rho}\otimes L}_1\\
     &\label{eq:traceDistMul}=\sup_{\rho\in D(\calH^{\otimes 2})}\norm{\tr_{\calH}\p{\p{I_{\calH}\otimes  P} \rho}}_1  \norm{L}_1\\
     &\label{eq:tracecChannel}\leq \sup_{\rho\in D(\calH^{\otimes 2})}\norm{\p{I_{\calH}\otimes  P} \rho}_1  \norm{L}_1\\
     &\label{eq:Proj}\leq \sup_{\rho\in D(\calH^{\otimes 2})}\norm{I_{\calH}\otimes  P }_\infty  \norm{\rho}_1  \norm{L}_1\\
     &\label{eq:calLbound}=\norm{L}_1.
 \end{align}
 Here \eqref{eq:traceDistMul} follows from the multiplicativity of the trace norm with respect to tensor product, \eqref{eq:tracecChannel} follows since the partial trace is a contraction for the trace norm (see Lemma~\ref{lem:stateNorms}), \eqref{eq:Proj} follows from H\"older's inequality for trace norm (see Lemma~\ref{lem:stateNorms}), and the \eqref{eq:calLbound} follows since $I_{\calH}\otimes P$ is a projection and therefore its spectral norm is $1$. The opposite inequality is obtained by lower bounding the diamond norm with the evaluation of $I_{L(\calH)}\otimes \calL$ on $\rho=\ket{c}\bra{c}\otimes \ket{c}\bra{c}$ for any normalized codeword $\ket{c}\in Q$. 

 In the next step, we show that \[\norm{L}_1\leq\norm{\calB^{\calE}_{\lambda,Q}}_{\diamond}\leq \varepsilon.\] Indeed, consider the state $\rho_0 =\ket{c_0}\bra{c_0}\otimes \ket{c_0}\bra{c_0} \in \calH^{\otimes 2}$. We have:
 \begin{align}
     \norm{\calB^{\calE}_{\lambda,Q}}_{\diamond}&\nonumber=\sup_{\rho \in D(\calH^{\otimes 2})}\norm{I_{L(\calH)}\otimes\calB^{\calE}_{\lambda,Q}(\rho)}_1\geq \norm{I_{L(\calH)}\otimes\calB^{\calE}_{\lambda,Q}(\rho_0)}_1\\
     &\nonumber=\norm{\ket{c_0}\bra{c_0}\otimes\calB^{\calE}_{\lambda,Q}(\ket{c_0}\bra{c_0})}_1=\norm{\ket{c_0}\bra{c_0}}_1  \norm{\calB^{\calE}_{\lambda,Q}(\ket{c_0}\bra{c_0})}_1\\
     &\nonumber=\norm{\sum_{k,l}\tr(B_{k,l}\ket{c_0}\bra{c_0})\ket{k}\bra{l}}_1=\norm{\sum_{k,l}\bra{c_0}(P{E_l^\dag E_k} P -\lambda_{k,l}P)\ket{c_0}  \ket{k}\bra{l}}_1\\
     &=\norm{\sum_{k,l}(\bra{c_0}{E_l^\dag E_k}\ket{c_0}-\lambda_{k,l})  \ket{k}\bra{l}}_1=\norm{\sum_{k,l} L_{k,l}  \ket{k}\bra{l}}_1=\norm{L}_1.\label{eq:L_1normbound}
 \end{align}
 Putting \eqref{eq:triangDiamond}, \eqref{eq:calLbound} and \eqref{eq:L_1normbound} together we obtain 
 \begin{equation}
     \norm{\calB^{\calE}_{\lambda',Q}}_{\diamond}\leq \varepsilon+\norm{\calL}_{\diamond}= \varepsilon+\norm{L}_{1}\leq \varepsilon+\norm{\calB^{\calE}_{\lambda,Q}}_{\diamond}\leq 2\varepsilon.\label{eq:replace}
 \end{equation}

 \begin{remark}\label{rem:generlizedLemma2}
     The proof above applies unchanged if instead the matrix $\lambda'_{k,l}=\bra{c_0}{E_l^\dag E_k}\ket{c_0}$ we take $\lambda^{\sigma}_{k,l}=\tr(\sigma {E_l^{\dag} E_k})$ where $\sigma\in D(Q)$ is a (possibly mixed) state in the code. The only minor change in the proof is in the chain of inequalities leading to \eqref{eq:L_1normbound}, where instead of $\rho_0=\ket{c_0}\bra{c_0}\otimes \ket{c_0}\bra{c_0}$ we take $\rho_0=\sigma\otimes \sigma$. The remaining arguments work exactly the same. Thus, we get the bound: 
     \[\norm{\calB_{\lambda^{\sigma},Q}^{\calE}}_{\diamond}\leq 2 \zeta(\calE,Q).\]
 \hfill$\triangleleft$ \end{remark}

    \subsection{Proof of Proposition~\ref{prop:exampleOptimality}}\label{app:optimalCodeExamle}
   Consider the qubit code $Q_N\subseteq \calH_2^{\otimes N}$ spanned by the basis 
   \[ \ket{c_0}=\ket{0}^{\otimes N} \quad \ket{c_1}=\sqrt{\alpha}  \ket{0}^{\otimes N-2}\otimes \ket{01}+ \sqrt{1-\alpha}\ket{0}^{\otimes N-2}\otimes \ket{10}, \]
   with projector $P=\ket{c_0}\bra{c_0}+\ket{c_1}\bra{c_1}$.
   Consider the Pauli set $\calE_N=\ppp{E_0,\dots, E_{N-1}}$ where $E_i=X_i$ (that is, $X$ acting on the $i$'th qudit with the identity on the remaining space)   for $0\leq i\leq N-3$, $E_{N-2}=I$ and $E_{N-1}=Z_{N-1}$. Consider the B\'eny-Oreshkov operator $\calB_{\lambda^N,Q_N}^{\calE_N}$ associated with the $N\times N$ matrix $\lambda$ given by 
   \[\lambda^N_{i,j}={\begin{cases}
       1-\alpha & (i,j)\in\ppp{(N-2,N-1),(N-1,N-2)},\\
       \bra{c_0}E_j^\dag E_i \ket{c_0} & \text{otherwise}.
   \end{cases}}\]
     A straightforward calculation reveals that:
     \begin{align*}
         \bra{c_0}E_i^\dag E_j\ket{c_0}&={\begin{cases}
             1 & i=j \text{ or } (i,j)\in\ppp{(N-2,N-1),(N-1,N-2)},\\
             0 & \text{otherwise};
         \end{cases}}\\
         \bra{c_1}E_i^\dag E_j\ket{c_1}&={\begin{cases}
             1 & i=j,\\
             1-2\alpha & (i,j)\in\ppp{(N-2,N-1),(N-1,N-2)},\\
             0 & \text{otherwise};
         \end{cases}}
         \\
          \bra{c_0}E_i^\dag E_j\ket{c_1}&=  \bra{c_1}E_i^\dag E_j\ket{c_0}=0.
     \end{align*}
     In particular, 
     \begin{align}
         B_{i,j}=P{E_j^{\dag}E_i}P-\lambda^N_{i,j}P&\nonumber=\sum_{k,l=0,1}\bra{c_k}{E_j^{\dag} E_i} \ket{c_l}   \ket{c_k}\bra{c_l} -\sum_{k=0,1} \lambda^N_{i,j}\ket{c_k}\bra{c_k}\\
         &\nonumber =\sum_{k=0,1}\bra{c_k}{E_j^{\dag} E_i} \ket{c_k}   \ket{c_k}\bra{c_k} -\sum_{k=0,1} \lambda^N_{i,j}\ket{c_k}\bra{c_k}\\
         &={\begin{cases}
             \alpha\p{\ket{c_0}\bra{c_0}-\ket{c_1}\bra{c_1}} & (i,j)\in\ppp{(N-2,N-1),(N-1,N-2)},\\
             0 & \text{otherwise},
         \end{cases}}\label{eq:Nachsim}
     \end{align}
     and therefore
     \begin{align*}
         \calB_{\lambda^N, Q_N}^{\calE_N}(\rho)&=\sum_{i,j=0}^{N-1} \tr(B_{i,j}\rho)  \ket{i}\bra{j}\\
         &={\alpha\tr\p{(\ket{c_0}\bra{c_0}-\ket{c_1}\bra{c_1})\rho}\p{\ket{N-2}\bra{N-1}+ \ket{N-1}\bra{N-2}}}. 
     \end{align*}
     By Lemma~\ref{lem:diamondnormNACHS} we have 
     \begin{align}
         \norm{\calB_{\lambda^N, Q_N}^{\calE_N}}_{\diamond}={\alpha\norm{\ket{c_0}\bra{c_0}-\ket{c_1}\bra{c_1}}_\infty   \norm{\ket{N-2}\bra{N-1}+ \ket{N-1}\bra{N-2}}_1=2\alpha},\label{eq:sameargument}
     \end{align}
     where we used the equality {$\|\ket{c_0}\bra{c_0}-\ket{c_1}\bra{c_1}\|_\infty=1$} and the fact that the singular values of the matrix $\ket{N-2}\bra{N-1}+ \ket{N-1}\bra{N-2}$ are $1$ and $0$ with multiplicity 2 and $N-2$, respectively.

     Now consider the channel $\calN_N$ with Kraus set $\calE_N^A=\ppp{A_0,A_1}$ where 
     \[A_0=\frac{1}{2}(E_{N-1}+E_{N-2}),\quad A_1=\frac{1}{2}(E_{N-1}-E_{N-2}).\quad \]
     A quick calculation shows that $\calE_{N}^A$ satisfies the completeness condition
    \begin{align}
         A_0A_0^\dag+ A_1A_1^\dag &=\nonumber \frac{1}{4}\p{I-Z_{N-1}-Z_{N-1}^\dag +Z_{N-1}^\dag Z_{N-1}  }+\frac{1}{4}\p{I+Z_{N-1}+Z_{N-1}^\dag +Z_{N-1}^\dag Z_{N-1}  }\\
         &=\frac{1}{4}\p{I-Z_{N-1}-Z_{N-1} +I  }+\frac{1}{4}\p{I+Z_{N-1}+Z_{N-1} +I  }=I.\label{eq:checkstam}
    \end{align}
     Let us calculate the B\'eny-Oreshkov operator $\calB_{\lambda^{N,A},Q_N}^{\calE_N^A}$, with respect to the matrix $\lambda^{N,A}$ given by $\lambda^{N,A}_{i,j}=\bra{c_0}{A_j^\dag A_i} \ket{c_0}$. Note that 
     \[A_0^\dag A_1=(I+Z_{N-1}^\dag)(I-Z_{N-1})/4=0,\]
     and therefore for $(i,j)\in\ppp{(0,1),(1,0)}$
     \[B^A_{i,j}=P{A_j^\dag A_i} P -\lambda^{N,A}_{i,j}P=0.\]
     On the other hand, using \eqref{eq:checkstam} and \eqref{eq:Nachsim} we have 
     \begin{align*}
         B_{0,0}^A&=P(A_0^\dag A_0)P-\bra{c_0}A_0^\dag A_0 \ket{c_0}  P\\
         &=\frac{1}{4}\p{P(2I+2E_{N-1})P-\bra{c_0}(2I+2E_{N-1}) \ket{c_0}  P}\\
         &={-\alpha   \ket{c_1}\bra{c_1}}.
         \end{align*}
         Similarly, we compute that
         \[ B_{1,1}^A={\alpha   \ket{c_1}\bra{c_1}},\]
         and therefore 
         \[\calB_{\lambda^{N,A},Q_N}^{\calE_N^A}(\rho)=\sum_{i,j=0,1}\tr(B_{i,j}\rho)  \ket{i}\bra{j}=\tr(\ket{c_1}\bra{c_1}\rho)  (-\alpha \ket{0}\bra{0}+\alpha\ket{1}\bra{1}).\]
         Using the same argument as in \eqref{eq:sameargument}, we have 
         \begin{align}
             \label{eq:samargumentagian}
             \norm{\calB_{\lambda^{N,A},Q_N}^{\calE_N^A}}_{\diamond}=\alpha\norm{\ket{c_1}\bra{c_1}}_\infty   \norm{\ket{0}\bra{0}-\ket{1}\bra{1}}_1=2\alpha.
         \end{align}

     We now have everything we need to take the final logical step. Note that by Proposition~\ref{prop:orthUni} $\calN_N\in\mathscr{N}(\calE_N)$ (as it is a quantum channel with Kraus operators spanned by $\calE_N$, and $\calE_N$ is a unitary Hilbert-Schmidt orthogonal set). Thus, if $Q_N$ is an $\varepsilon$-AQEC for $\calE_N$, then it is an $\varepsilon$-AQEC for $\calN_N$. In particular, by B\'eny-Oreshkov Theorem (see Theorem~\ref{th:BO}), the definition of the environment-leakage distance, and \eqref{eq:diamondBures}, 
     \[
     \zeta(\calE_N^A,Q_N)\leq {2\sqrt{2}\varepsilon}.
     \]
     On the other hand, by Lemma~\ref{lem:replaceOpt} and {\eqref{eq:samargumentagian}}, we have 
     \[2\alpha=\norm{\calB_{\lambda^{N,A},Q_N}^{\calE_N^A}}_{\diamond}\leq 2 \zeta(\calE_N^A,Q_N)\leq {4\sqrt{2} \varepsilon},\]
     Which implies $\alpha \leq {2\sqrt{2} \varepsilon} $. Combining the above with  {\eqref{eq:sameargument}} we obtain 
     \[\zeta(\calE_N,Q_N)\leq \norm{\calB_{\lambda^N, Q_N}^{\calE_N}}_{\diamond}={2\alpha \leq 4 \sqrt{2}\varepsilon}.\]
     
  \subsection{Proof of Lemma~\ref{lem:orthogonalProjectionKLspace}}\label{app:ProofLemGeometric}
  \begin{proof}
    We begin our proof by observing that for any  operator $X\in L(\C^{K}\otimes \C^M)$ we have \[ P_{\s{KL}}(X)=\frac{1}{K} I_K\otimes \tr_K(X),\]
    where $P_{\s{KL}}$ is the orthogonal projection on $\calH_{\s{KL}}$ with respect to the Hilbert-Schmidt inner product. It is sufficient to show that for any $\lambda \in L(\C^{M})$,
    \[\innerP{X-\frac{1}{K}I_K\otimes \tr_K(X),I_K\otimes \lambda }_{\s{HS}}=0.\]
    Indeed, 
    \begin{align*}
        \innerP{X-\frac{1}{K}I_K\otimes \tr_K(X),I_K\otimes \lambda }_{\s{HS}}^*&=\tr\p{\p{X-\frac{1}{K}I_K\otimes \tr_K(X)} I_K\otimes \lambda^\dag }\\&=\tr\p{X (I_K\otimes \lambda^\dag)}-\frac{1}{K}\tr\p{I_K \otimes (\tr_K (X)\lambda^\dag)}\\
        &=\tr\p{\tr_K(X)  \lambda^\dag}-\frac{1}{K}\tr\p{I_K } \tr(\tr_K (X)\lambda^\dag)=0,
    \end{align*}
    where the equality on the last line follows from Lemma~\ref{lem:propPartTrace} and the equality $\tr(I_K)=K$. Since the orthogonal projection $P_{\s{KL}}(X)$ is the closest point to $X$ in $\calH_{\s{KL}}$ we have 
    \begin{align}
        \min_{\lambda\in L(\C^{M})}\norm{X-I_K\otimes \lambda}_2^2&\nonumber=\norm{X-\frac{1}{K}I_K\otimes \tr_K(X)}_2^2=\norm{X-P_{\s{KL}}(X)}_2^2\\
        &=\norm{X}_2^2-\norm{P_{\s{KL}}(X)}_2^2=\norm{X}_2^2-\norm{\frac{1}{K}I_K\otimes \tr_K(X)}_2^2\label{eq:projec1}\\
        &=\norm{X}_2^2-\frac{1}{K}\norm{\tr_K(X)}_2^2,\label{eq:projec2}
    \end{align}
where the first equality in \eqref{eq:projec1} follows from the Pythagorean theorem for orthogonal projections and the last equality follows by straightforward calculation.
  
Next, note that $A_{\s{QEC}}\succeq 0$, so the square root $\sqrt{A_{\s{QEC}}}$ is well defined, and  let $X=\sqrt{A_{\s{QEC}}}$. Note that 
 $\norm{\sqrt{A_{\s{QEC}}}}_2^2=\tr\p{A_{\s{QEC}}}$ and therefore \eqref{eq:projec2} gives
   \begin{align}
      \tr(A_{\s{QEC}})-\frac{1}{K}\norm{\tr_K\sqrt{A_{\s{QEC}}}}_2^2&\nonumber=\norm{\sqrt{A_{\s{QEC}}}}_2^2-\frac{1}{K}\norm{\tr_K(\sqrt{A_{\s{QEC}}})}_2^2\\
      &\nonumber = \min_{\lambda\in L(\C^{M})}\norm{\sqrt{A_{\s{QEC}}}-I_K\otimes \lambda}_2^2\\
      & = \min_{\substack{\lambda\in L(\C^{M})\\
      \lambda \succeq 0}}\norm{\sqrt{A_{\s{QEC}}}-I_K\otimes \lambda}_2^2\label{eq:minOverPositive}\\
      &\nonumber = \min_{\substack{B\in \calH_{\s{KL}}\\
     B \succeq 0}}\norm{\sqrt{A_{\s{QEC}}}-B}_2^2\\
     & = \min_{\substack{B\in \calH_{\s{KL}}\\
     B \succeq 0}}\norm{\sqrt{A_{\s{QEC}}}-\sqrt{B}}_2^2\label{eq:minovverSqrt}\\
   \end{align}
Here, \eqref{eq:minOverPositive} follows since we have already shown that the minimizer is given by $\frac{1}{K}\tr_K(\sqrt{A_{\s{QEC}}})$, which is positive since the partial trace is CPTP, and \eqref{eq:minovverSqrt} follows since the map $\sqrt{  }$  is a bijection on the set of positive operators in $\calH_{\s{KL}}$.  
\end{proof}

    \subsection{Proof of Theorem~\ref{th:equivErGen}}
    \label{app:ThErGen}
    \subsubsection*{Proof of the upper bound}
    Let us start with the upper bound. The proof idea is as follows: We start with a B\'eny-Oreshkov operator for erasures, $\calB_{\lambda_{\s{Er}},Q}^{\calE^{\s{Er}}_{2t}}$, with small diamond norm. We then show that we can find a matrix $\lambda_{t}$ such that the corresponding B\'eny-Oreshkov operator for HW errors with weight at most $t$, $\calB_{\lambda_t,Q}^{\calE_t}$, can be represented as 
    \begin{align}
        \calB_{\lambda_t,Q}^{\calE_t}=\calF\circ \calB_{\lambda_{\s{Er}},Q}^{\calE^{\s{Er}}_{2t}}, \label{eq:FactGoal1}
    \end{align}
       where $\calF$ is a superoperator acting blockwise, {to be defined below.}
    We then bound the completely bounded norms of these factorization maps and use submultiplicativity of the completely bounded norm to bound the diamond norm of $\calB_{\lambda_t,Q}^{\calE_t}$.

    We begin with a simple observation: let $I_1,I_2\subseteq [N]$ be  two distinct sets of coordinates of size $2t$, $\ui,\uj\in [q]^{2t}$, and let $E_{\ui,I_1}$ and $E_{\uj,I_2}$ be the corresponding erasure operators \eqref{eq:erasure_operator}. For any basis state $\ket{\ux}\in \calH_{q}^{\otimes N}$ we have
    \begin{equation}
        \bra{\ux}E_{\ui,I_1}^\dag E_{\uj, I_2}\ket{\ux}=\braket{\ux_{[I_1]}|\ui}\braket{\ux_{[I_2]}|\uj}^*  (\bra{\ux_{[N]\setminus I_1}}\otimes \bra{\perp}_{I_1})( \ket{\perp}_{I_2}\otimes \ket{\ux_{[N]\setminus I_2}})=0,\label{eq:ijzero}
    \end{equation}
since $I_1\triangle I_2\ne\varnothing$. In particular, we have 
\begin{equation}
    E_{\ui,{I_1}}^\dag E_{\uj,I_2}=\delta_{I_1,I_2}  \ket{\ui}_{I_1}\bra{\uj}_{I_1}\otimes I_{[N]\setminus I_1}.\label{eq:deltaI1I2}
\end{equation}
   Consider the matrix $\lambda_{\s{Er}}$ given by 
    \begin{align}
        \lambda_{\s{Er}}(\ui,\uj,I_1,I_2)=\bra{c_0}{E_{\uj,I_2}^\dag E_{\ui,I_1}}\ket{c_0}=\delta_{I_1,I_2}  \bra{c_0} ({\ket{\uj}_{I_1}\bra{\ui}_{I_1}}\otimes I_{[N]\setminus I_1})\ket{c_0},\label{eq:lambdaER}
    \end{align}
    where $\ket{c_0}\in Q$ is an arbitrary code state, and the last equality follows from \eqref{eq:deltaI1I2}. By Lemma~\ref{lem:replaceOpt} we have  
    \begin{equation*}
      \norm{\calB_{\lambda_{\s{Er}},Q}^{\calE_{2t}^{\s{Er}}}}_{\diamond}\leq 2\zeta(\calE^{\s{Er}}_{2t},Q).
   \end{equation*}
  For better readability, let us write $\calB_{\s{Er}}$ for $\calB_{\lambda_{\s{Er}},Q}^{\calE_{2t}^{\s{Er}}}$ and define
    \begin{equation}
        E_{\ui,\uj}^{(I)}:={E_{\uj,I}^\dag E_{\ui,I}=\ket{\uj}_{I}\bra{\ui}_{I}}\otimes I_{[N]\setminus I}, \quad \lambda_{\ui,\uj}^{(I)}:=\lambda_{\s{Er}}(\ui,\uj,I,I).\label{eq:defLAMBDAIij}
    \end{equation}
By \eqref{eq:lambdaER} and \eqref{eq:deltaI1I2}, the action of $\calB_{\s{Er}}$ is given by 
    \[\calB_{\s{Er}}(X)=\sum_{I\in \binom{[N]}{2t}}\sum_{\ui,\uj\in [q]^{2t}}\tr\p{(PE_{\ui,\uj}^{(I)}P-\lambda_{\ui,\uj}^{(I)}P)X}  \ket{\ui,I}\bra{\uj,I}, \]
    where $P$ is the projection on $Q$ and $\ppp{\ket{\ui,I} ~:~ \ui\in[q]^{2t} ,I \in \binom{[N]}{2t}}$ is an orthonormal basis of the direct sum space:
    \[\calH_{\s{Er}}=\bigoplus_{I\in \binom{[N]}{2t}}\calH_I, \quad \calH_I=\C^{q^{2t}}.\]

    Let $\s{vec}_t(I)$ be the set of all vectors of weight at most $t$ in $[q^2]^{N}$ supported on the set $I$ and let $\s{P}_t(I)$ denote the set of all HW operators on $\calH_{q}^{\otimes N}$ indexed by elements in  $\s{vec}_t(I)$; see \eqref{eq:HWoperatorsT}. Consider the spaces
    \[
    \calH_{\s{HW}}^{\s{loc}}=\bigoplus_{I\in \binom{[N]}{2t}}\calH^{\s{HW}}_I, \quad \calH^{\s{HW}}_I=\C^{\s{vec}_t(I)}. 
    \]
   Any HW operator $W_{\ux}\in \s{P}_t(I)$ is given by
    \[W_{\ux}=W_{\ux,I}\otimes I_{[N]\setminus I}, \]
    where $W_{\ux,I}$ is the Pauli on $\calH_{q}^{\otimes 2t}$ indexed by the restriction of $\ux$ to $I$. Let $\s{P}_{N,t}=\cup_{I, |I|\le t}\s P_t(I)\otimes I_{[N]\backslash I}$ be the set of all (global) HW operators of weight at most $t$ by  
    and consider the space $\calH_{\s{HW}}^{\s{glob}}=\C^{\s{P}_{N,t}}$. For simplicity, we think of $\s{P}_{N,t}$ as the set of all vectors in $[q^2]^{N}$ with weight at most $t$. Note that $|\s{vec}_t(I)|=B_{q^2}(|I|,t)$ and $|\s P_{N,t}|=B_{q^2}(N,t)$.

Next, we define the local block maps. First, consider the superoperator $\calF^{\s{b}}_I:L(\calH_I)\to L(\calH_{I}^{\s{HW}})$ 
   \begin{equation}\label{eq: Fb}
   \calF_I^{\s{b}}(Z)=\sum_{\ux,\uy\in\s{vec}_t(I)} \tr({W_{\uy,I}^\dag W_{\ux, I}} Z)  \ket{\ux,I}\bra{\uy,I},
  \end{equation}
For defining the second block map, we introduce a map $\sigma:\s{P}_{N,t}\times \s{P}_{N,t}\to \binom{[N]}{2t}$ defined
as follows. Let $\ux,\uy\in \s{P}_{N,t}$ and let $U=\supp(\ux)\cup\supp(\uy)$. If $|U|=2t$, we put $\sigma(\ux,\uy)=U$,
otherwise $\sigma(\ux,\uy)$ is formed as a union of $U$ and the $2t-|U|$ smallest coordinates in $[N]$. The defining
property that we use below is that for any $\ux,\uy\in \s{P}_{N,t}$,
  \[
  \supp(\ux)\cup\supp(\uy)\subseteq \sigma(\ux,\uy).
  \]
Now consider the block map $\calF^{\s{g}}_I:L(\calH_{I}^{\s{HW}})\to L(\calH_{\s{HW}}^{\s{glob}})$, defined as: 
\begin{equation}
    \calF^{\s{g}}_I(Z)=\sum_{(\ux,\uy)\in \sigma^{-1}(I)  }\bra{\ux,I} Z \ket{\uy,I}  \ket{\ux}\bra{\uy}.\label{eq:globfBlock}
\end{equation}

We define the global factorization map 
$\calF:L(\calH_{\s{Er}})\to L(\calH_{\s{HW}}^{\s{glob}})$ as the direct sum map (see Def.~\ref{def:directsumMap}):
  \begin{equation}
\calF=\bigoplus_{I\in\binom{[N]}{2t}}\calF_I^{\s{g}}\circ \calF_I^{\s{b}}.\label{eq:Fglob}
  \end{equation}
The intuition behind this representation is as follows. The block $\calB_{\s{Er}}(X)_I$ measures how close erasures on the subsystem $I$ are to acting as scalars on the code space, expressing the approximate KL conditions on $I$.  Indeed, since the operators $\{E_{\ui,\uj}^{(I)}\}_{\ui,\uj}$ span the set of operators supported on $I$, the operators $P E_{\ui,\uj}^{(I)}P-\lambda_{\ui,\uj}^{(I)} P$ quantify the deviation of $I$-supported operators from scalar action on the code. In particular, this local erasure data determines how all HW operators supported on $I$, and hence in particular the products ${W_{\uy}^\dagger W_{\ux}}$ supported in $I$, act on the code space.  The role of $\calF_{I}^{\s{b}}$ is then to re-express the local data in the HW-pair coordinates relevant to the global B\'eny-Oreshkov superoperators: it converts the information of $\calB_{\s{Er}}(X)_I$ into the coordinates associated with the operators ${W_{\uy}^\dagger W_{\ux}}$ supported on $I$. In turn, $\calF_{I}^{\s{g}}$ is a bookkeeping map for the block $I$: it selects from these local HW-pair coefficients precisely the entries indexed by $\sigma^{-1}(I)$ and places them in the corresponding locations of the global output. 

\begin{lemma}[Correctness of factorization]\label{lem:FactCor}
    There exists a $B_{q^2}(N,t)\times B_{q^2}(N,t)$ matrix $\lambda_t$ such that 
    \[\calB_{\lambda_t,Q}^{\calE_t}=\calF\circ \calB_{\s{Er}},\]
    where $\calF$ is the superoperator defined in \eqref{eq:Fglob}.
\end{lemma}
\begin{proof}
    Define a $B_{q^2}(N,t)\times B_{q^2}(N,t)$ matrix $\lambda_t$ as 
\[\lambda_t(\ux,\uy)=\hat{\lambda}^{(I)}_{\ux,\uy}, \]
where $\sigma(\ux,\uy)=I$, {$|I|=2t$,} and 
\begin{equation}
    \hat{\lambda}^{(I)}_{\ux,\uy}=\tr\p{{W_{\uy,I}^\dag W_{\ux,I}} \lambda^{(I)}}, \quad \lambda^{(I)}=\sum_{\ui,\uj\in [q]^{2t}} \lambda_{\ui,\uj}^{(I)}   \ket{\ui,I}\bra{\uj,I},\label{eq:AnotherDef}
\end{equation}
where $\lambda_{\ui,\uj}^{(I)}$ is defined in \eqref{eq:defLAMBDAIij}.  
We start by showing that for any $I$,  $\calF^{\s{b}}_I$ acts on $ \calB_{\s{Er}}$ as
\begin{align}
    \calF^{\s{b}}_I \p{\calB_{\s{Er}}(X)_I}=\sum_{\ux,\uy\in\s{vec}_t(I)}  \tr\p{(P{W_{\uy}^\dag  W_{\ux}}P -\hat{\lambda}^{(I)}_{\ux,\uy}P)X} \ket{\ux,I}\bra{\uy, I},\label{eq:Composition1}
\end{align}
where {the action is written in the direct-sum form \eqref{eq:direct sum}}. Indeed, the $(I,I)$ block of $\calB_{\s{Er}}(X)$  (see Def.~\ref{def:blockRep}) is given by 
\[\calB_{\s{Er}}(X)_I=\sum_{\ui,\uj}\tr\p{(PE_{\ui,\uj}^{(I)}P-\lambda^{(I)}_{\ui,\uj}P)X}\ket{\ui,I}\bra{\uj,I}.\]
Therefore,
\begin{align}
    \calF^{\s{b}}_I \p{\calB_{\s{Er}}(X)_I}&\nonumber=\sum_{\ux,\uy\in \s{vec}_t(I)}\sum_{\ui,\uj\in[q]^{2t}}\tr\p{(PE_{\ui,\uj}^{(I)}P-\lambda^{(I)}_{\ui,\uj}P)X} \tr\p{{W_{\uy,I}^\dag W_{\ux,I}}\ket{\ui,I}\bra{\uj,I}} \ket{\ux,I} \bra{\uy,I}\\
    &\nonumber=\sum_{\ux,\uy\in \s{vec}_t(I)}\Bigg(\sum_{\ui,\uj\in[q]^{2t}}\tr\p{PE_{\ui,\uj}^{(I)}PX} \tr\p{{W_{\uy,I}^\dag W_{\ux,I}}\ket{\ui,I}\bra{\uj,I}}\\
    &\nonumber\hspace{2.6cm}-  \tr(PX) \Big(\sum_{\ui,\uj\in[q]^{2t}}\lambda_{\ui,\uj}^{(I)}  \tr\p{{W_{\uy,I}^\dag W_{\ux,I}}\ket{\ui,I}\bra{\uj,I}}\Big)\Bigg) \ket{\ux,I} \bra{\uy,I}.
\end{align}
Let us compute each term separately, starting with the second one. For any $\ux,\uy \in \s{vec}_t(I)$ we have
\begin{align}\label{eq: second term}
    \tr(PX) \Big(\sum_{\ui,\uj\in[q]^{2t}}\lambda_{\ui,\uj}^{(I)}  \tr\p{{W_{\uy,I}^\dag W_{\ux,I}}\ket{\ui,I}\bra{\uj,I}}\Big)=\tr\p{\hat{\lambda}_{\ux,\uy}^{(I)} P X},
\end{align}
immediately from the definition of $\hat{\lambda}^{(I)}$ \eqref{eq:AnotherDef}. For the other summand, denote by $(PXP)_I$ the operator $\tr_{[N]\setminus I}(PXP)$ obtained by tracing out the qudits in $[N]\setminus I$ in $PXP$. We have:
\begin{align}
    \sum_{\ui,\uj\in[q]^{2t}}&\tr\p{PE_{\ui,\uj}^{(I)}PX}  \tr\p{{W_{\uy,I}^\dag W_{\ux,I}}\ket{\ui,I}\bra{\uj,I}} \nonumber
    \\&\label{eq:trCyc}=\sum_{\ui,\uj\in[q]^{2t}}\tr\p{({\ket{\uj,I}\bra{\ui,I}}\otimes I_{[N]\setminus I})PX P}  \bra{\uj,I}{W_{\uy,I}^\dag W_{\ux,I}}\ket{\ui,I}\\
    &\label{eq:trProp}=\sum_{\ui,\uj\in[q]^{2t}}\tr\p{({\ket{\uj,I}\bra{\ui,I}}(PXP)_I }  \bra{\uj,I}{W_{\uy,I}^\dag W_{\ux,I}}\ket{\ui,I}
   \\
     &\nonumber=\sum_{\ui,\uj\in[q]^{2t}}{\bra{\ui,I}(PXP)_I\ket{\uj,I}}  {\bra{\uj,I}W_{\uy,I}^\dag W_{\ux,I}\ket{\ui,I}}\\
     &\nonumber=\tr\p{{W_{\uy,I}^\dag W_{\ux, I}}(PXP)_I }\\
     &\label{eq:trProp3} =\tr\p{({W_{\uy,I}^\dag W_{\ux, I}}\otimes I_{[N]\setminus I}) PX P }\\
     &\nonumber=\tr\p{{W_{\uy}^\dag W_{\ux}}PX P }\\
     &\nonumber=\tr\p{P{W_{\uy}^\dag W_{\ux}}PX  }.
\end{align}
where \eqref{eq:trCyc} follows from the cyclicity  of the trace and the definition of $E_{\ui,\uj}^{(I)}$ 
\eqref{eq:defLAMBDAIij}, and \eqref{eq:trProp}, \eqref{eq:trProp3} follow by Lemma~\ref{lem:propPartTrace}. {This 
calculation together with \eqref{eq: second term} proves \eqref{eq:Composition1}.}

Now, let us compute the composition $\calF^{\s{g}}_I\circ \calF^{\s{b}}_I$ on $\calB_{\s{Er}}(X)$. Using \eqref{eq:globfBlock} and \eqref{eq:Composition1}, we have
\begin{align*}
    \calF_{I}^{\s{g}}\circ \calF_{I}^{\s{b}}\p{\calB_{\s{Er}}(X)_I}&=\sum_{(\ux,\uy)\in \sigma^{-1}(I)} \sum_{\ux',\uy'\in \s{vec}_t(I)} \tr\p{(P{W_{\uy'}^\dag  W_{\ux'}}P -\hat{\lambda}^{(I)}_{\ux',\uy'}P)X} \\
    &\nonumber\hspace{6cm}  \times \braket{\ux,I| \ux',I}\braket{\uy, I|\uy',I}   \ket{\ux}\bra{\uy}\\
    &=\sum_{(\ux,\uy)\in \sigma^{-1}(I)} \tr\p{(P{W_{\uy}^\dag  W_{\ux}}P -\hat{\lambda}^{(I)}_{\ux,\uy}P)X}    \ket{\ux}\bra{\uy}\\
    &=\sum_{(\ux,\uy)\in \sigma^{-1}(I)} \tr\p{(P{W_{\uy}^\dag  W_{\ux}}P -\lambda_t\p{\ux,\uy}P)X}    \ket{\ux}\bra{\uy}. 
\end{align*}

Finally, we use the definition of $\calF$ and the direct sum representation \eqref{eq:direct sum} to obtain 
\begin{align*}
    \calF\p{\calB_{\s{Er}}(X)}&=\sum_{I\in \binom{[N]}{2t}} \calF_{I}^{\s{g}}\circ \calF_{I}^{\s{b}}\p{\calB_{\s{Er}}(X)_I}\\
    &=\sum_{I\in \binom{[N]}{2t}}\sum_{(\ux,\uy)\in \sigma^{-1}(I)} \tr\p{(P{W_{\uy}^\dag  W_{\ux}}P -\lambda_t\p{\ux,\uy}P)X}    \ket{\ux}\bra{\uy}\\
    &=\sum_{\ux,\uy\in \s{P}_{N,t}}\tr\p{(P{W_{\uy}^\dag  W_{\ux}}P -\lambda_t\p{\ux,\uy}P)X}    \ket{\ux}\bra{\uy}\\
    &=\calB_{\lambda_t,Q}^{\calE_t}(X).
\end{align*}
This proves correctness of the factorization map. 
\end{proof}
Let us bound the diamond norm of $\calB_{\lambda_t,Q}^{\calE_t}$. Using Lemmas~\ref{lem:FactCor}, \ref{lem:directSumLemma}, and \ref{lem:cbDiamondProp}, we have
\begin{align}
    \norm{\calB_{\lambda_t,Q}^{\calE_t}}_{\diamond}\nonumber&\leq \norm{\calB_{\lambda_t,Q}^{\calE_t}}_{\s{cb}}\leq \norm{\calF }_{\s{cb}}   \norm{  \calB_{\s{Er}}}_{\s{cb}}\\
    &\nonumber=\norm{  \calB_{\s{Er}}}_{\s{cb}}  \max_{I\in \binom{[N]}{2t}} \norm{\calF_{I}^{\s{g}}\circ \calF_{I}^{\s{b}}}_{\s{cb}}\\
    &\label{eq:ShownInth}\leq \norm{  \calB_{\s{Er}}}_{\diamond}   \max_{I\in \binom{[N]}{2t}}  \norm{\calF_{I}^{\s{b}}}_{\s{cb}}   \norm{ \calF_{I}^{\s{g}}}_{\s{cb}}\\
    &\leq 2\zeta(\calE_{2t}^{\s{Er}},Q)   \max_{I\in \binom{[N]}{2t}}  \norm{\calF_{I}^{\s{b}}}_{\s{cb}}   \norm{ \calF_{I}^{\s{g}}}_{\s{cb}},\label{eq:UpporERt}
\end{align}
where \eqref{eq:ShownInth} follows since 
for Hermitian preserving operators, the diamond and completely bounded norms are equivalent, and since the  specific form of $\lambda_{\ui,\uj}$ described above implies that $\calB_{\s{Er}}$ is Hermitian preserving (which is proved similarly to
the proof that $\calB_{\lambda,Q}^{\calE}$ is Hermitian preserving in Theorem~\ref{th:BOerrorset}). 
It now remains to bound the quantity
$ \max_{I}  \norm{\calF_{I}^{\s{b}}}_{\s{cb}}   \norm{ \calF_{I}^{\s{g}}}_{\s{cb}}$.

\begin{lemma} If $\ast\in\{\s b,\s g\}$, then
\begin{align}
    \norm{\calF_I^{\s{\ast}}}_{\s{cb}}\leq B_{q^2}(2t,t).\label{eq:boundFloc}
\end{align}
\end{lemma}
\begin{proof} Define a linear map $U_I:\calH_I \to \calH_I^{\s{HW}}\otimes \calH_I$ given by 
   \[
   U_I\ket{\psi}= \sum_{\ux \in \s{vec}_t(I)}\ket{\ux,I}\otimes {W_{\ux,I}} \ket{\psi}.
   \]
Starting with the definition in \eqref{eq: Fb}, for $Z\in L(\calH_I)$ we may write
\begin{align*}
    \calF_I^{\s{b}}(Z)
    =\tr_{\calH_{I}}(U_I Z U_I^\dag)=\tr_{\calH_I}\circ \calT_{U_I}(Z),
\end{align*}
with $\tr_{\calH_1}$ is the operation denotes the partial trace over $\calH_I$ and $\calT_{U_I}=\calT_{U_I,U_I}$ is the superoperator defined in Lemma~\ref{lem:compoDiamond} with $V=U=U_I$.

Recall that the partial trace operation is CPTP and therefore by Lemmas~\ref{lem:cbDiamondProp} and \ref{lem:compoDiamond} we have 
\begin{align}
    \norm{\calF_I^\s b}_{\s{cb}}=\norm{\tr_{\calH_I}\circ \calT_{U_I}}_{\s{cb}}\leq  \norm{\tr_{\calH_I}}_{\s{cb}}  \norm{ \calT_{U_I}}_{\s{cb}}= \norm{U_I}_{\infty}^2.\label{eq:FiComp}
\end{align}
Note that $U_I$ is a scaled partial isometry: 
\begin{align*}
    U_I^\dag U_I&=\sum_{\ux,\uy\in\s{vec}_t(I)} \braket{\ux,I |\uy,I} {W_{\ux,I}^\dag W_{\uy,I}}=\sum_{\ux\in\s{vec}_t(I)}{W_{\ux,I}^\dag W_{\ux,I}}=|\s{vec}_t(I)|   I_{\C^{q^{2t}}}.
\end{align*}
Combining the above with \eqref{eq:FiComp}, we obtain
\[\norm{\calF_I^\s{b}}_{\s{cb}}\leq \norm{U_I}_{\infty}^2\leq |\s{vec}_t(I)|=B_{q^2}(2t,t).\]

Next, we prove that $\calF_{I}^{\s{g}}$ satisfies the same norm inequality. 
Consider a bipartite graph $G_I({\mathscr V}_1\cup{\mathscr V}_2,{\mathscr E}) $ whose sets of left and right vertices are indexed by $\s{vec}_t(I)$ and $(\ux,\uy)\in{\mathscr E}$ place an edge iff
$\sigma(\ux,\uy)=I$. Note that $|\mathscr{V}_1|=|\mathscr{V}_2|=B_{q^2}(2t,t)$. By Konig's coloring theorem \cite[Prop.5.3.1]{Diestel2025}, 
any bipartite graph with maximum degree $\Delta_I$ admits a proper edge coloring with $\Delta_I$ colors. In other words, ${\mathscr E}$ admits a partition 
into $\Delta_I$ disjoint matchings $M_1,\dots,M_{\Delta_I}$, where for $k=1,\dots,\Delta_I$,
   \[
   M_k=\set{(\ux_{k,1},\uy_{k,1}),\dots,(\ux_{k,m_k},\uy_{k,m_k})}.
   \] 
 Let $\calM_{I,k}:L(\calH_{I}^{\s{HW}})\to L(\calH_{\s{HW}}^{\s{glob}})$ be the superoperator defined as 
    \[
    \calM_{I,k}(Z)=\sum_{(\ux,\uy)\in M_k}\bra{\ux,I} Z \ket{\uy,I} \ket{\ux}\bra{\uy}.
    \]
We will show that $\norm{\calM_{I,k}}_{\s{cb}}\leq 1$, which will imply the claimed bound for $\norm{\calF_{I}^{\s{g}}}_{\s{cb}}$.
To bound the norm of $\calM_{I,k}$ we follow the same ideas as those used to bound $\norm{\calF_I^{\s{b}}}_{\s{cb}}$: we decompose $\calM_{I,k}$ into a product of three superoperators of norm at most $1$ and then use submultiplicativity. This decomposition relies on the fact that each $M_k, k=1,\dots,\Delta_I$, is a matching, so $\ux_{k,i}\neq \ux_{k,j}$ and $\uy_{k,i}\neq \uy_{k,j}$ for distinct 
$i,j \in \set{1,\dots, m_k}$, and therefore $\{\ket{\ux_{k,i},I}\}_i$ and $\{\ket{\uy_{k,i},I}\}_i$ 
are orthonormal sets. 
Let $U_k,V_k:\calH_{I}^{\s{HW}}\to \C^{m_k}$ be operators defined as 
   \[
   U_k=\sum_{i=1}^{m_k}\ket{i}\bra{\ux_{k,i},I}, \quad V_k=\sum_{i=1}^{m_k}\ket{i}\bra{\uy_{k,i},I},
   \]
where $\ppp{\ket{i}}_i$ is an orthonormal basis of $\C^{m_k}$, and consider a superoperator
    \[
    \calT_{U_k,V_k}(Z)=U_kZV_k^{\dag}.
    \]
By Lemma~\ref{lem:compoDiamond}, we have
    \[
    \norm{\calT_{U_k,V_k}}_{\s{cb}}\leq \max\p{\norm{U_k}_{\infty},\norm{V_k}_{\infty}}^2=1, 
    \] 
where the equality follows since $\{\ket{\ux_{k,i},I}\}_i$ and $\{\ket{\uy_{k,i},I}\}_i$ are orthonormal sets and thus $U_k$ and $V_k$ are partial isometries. Next, consider the diagonal extraction superoperator $\calJ_k:L(\C^{m_k})\to L(\C^{m_k})$ given by 
\[\calJ_k(Z)=\sum_{j=1}^{m_k}\bra{j}Z\ket{j}  \ket{j}\bra{j}. \]
Note that $\calJ_k$ is CPTP since it has a Kraus representation with projection operators $\ppp{\ket{j}\bra{j}}$ that satisfy the completeness condition. In particular, by Lemma~\ref{lem:cbDiamondProp}, $\norm{\calJ_k}_{\s{cb}}=1$. We also consider the maps $\hat{U}_k,\hat{V}_k: \C^{m_k}\to \calH_{\s{HW}}^{\s{glob}}$ given by 
   \[
   \hat{U}_{k}=\sum_{i=1}^{m_k}\ket{\ux_{k,i}}\bra{i}, \quad  \hat{V}_{k}=\sum_{i=1}^{m_k}\ket{\uy_{k,i}}\bra{i},
   \]
and the corresponding superoperator $\calT_{\hat{U}_k,\hat{V}_k}$. Using the same arguments as above, we have 
$\norm{\calT_{\hat{U}_k,\hat{V}_k}}_{\s{cb}}\leq 1$.
Finally, let us show that $\calM_{I,k}=\calT_{\hat{U}_k,\hat{V}_k}\circ \calJ_k \circ \calT_{U_k,V_k}$: 
\begin{align*}
    \calT_{\hat{U}_k,\hat{V}_k}\circ \calJ_k \circ \calT_{U_k,V_k}(Z)&=\calT_{\hat{U}_k,\hat{V}_k}\circ \calJ_k\p{\sum_{i,j=1}^{m_k}\bra{\ux_{k,i},I}Z\ket{\uy_{k,j},I}  \ket{i}\bra{j}}\\
    &=\calT_{\hat{U}_k,\hat{V}_k}\p{\sum_{j=1}^{m_k}\bra{\ux_{k,j},I}Z\ket{\uy_{k,j},I}  \ket{j}\bra{j}}\\
    &=\sum_{j=1}^{m_k}\bra{\ux_{k,j},I}Z\ket{\uy_{k,j},I}  \ket{\ux_{k,j}}\bra{\uy_{k,j}}=\calM_{I,k}(Z).
\end{align*}
Putting everything together, we have 
    \[
\norm{\calM_{I,k}}_{\s{cb}}= \norm{\calT_{\hat{U}_k,\hat{V}_k}\circ \calJ_k \circ \calT_{U_k,V_k}}_{\s{cb}}\leq \|\calT_{\hat{U}_k,\hat{V}_k}\|_{\s{cb}}  \norm{\calJ_k }_{\s{cb}}  \norm{ \calT_{U_k,V_k}}_{\s{cb}}\leq 1.
\]
Finally, $\calF_{I}^{\s{g}}(Z)=\sum_{k=1}^{\Delta_I}\calM_{I,k}(Z)$ for all $Z$, and therefore by the triangle inequality
\begin{equation}
    \norm{\calF_{I}^{\s{g}}}_{\s{cb}}\leq\sum_{k=1}^{\Delta_I}\norm{\calM_{I,k}}_{\s{cb}}\leq \Delta_{I}\leq B_{q^2}(2t,t),\label{eq:Ftriangular}
\end{equation}
where the last inequality follows since the maximum degree in $G_I$ satisfies $\Delta_I\le \max(|{\mathscr V}|_1,|{\mathscr V}_2|)$. 
\end{proof}
Relying on \eqref{eq:boundFloc}, we can now establish the upper bound of the theorem:
    \[
    \zeta(\calE_t, Q)\leq \|\calB_{\lambda_t,Q}^{\calE_t}\|_{\diamond}\leq 2\zeta(\calE_{2t}^{\s{Er}},Q)  
            \max_{I\in \binom{[N]}{2t}}  \big(\|\calF_{I}^{\s{b}}\|_{\s{cb}} \cdot  \| \calF_{I}^{\s{g}}\|_{\s{cb}}\big)\leq 2B_{q^2}(2t,t)^2\ \zeta(\calE_{2t}^{\s{Er}},Q).
\]

\subsubsection*{Proof of the lower bound}
For the lower bound, we use a similar strategy: we  choose the matrix $\lambda_t=\lambda_t(\ux,\uy)$ defined by 
$\lambda_t(\ux,\uy)=\bra{c_0}{W_{\uy}^\dag W_{\ux}}\ket{c_0}$ for a codeword $\ket{c_0}$. As in the proof of the upper bound, the corresponding B\'eny-Oreshkov operator $\calB$ is Hermitian preserving and (by Lemmas~\ref{lem:replaceOpt} and \ref{lem:cbDiamondProp}) satisfies
\begin{equation}
\norm{\calB}_{\s{cb}}=\norm{\calB}_{\diamond}\leq 2\zeta(\calE_t, Q). \label{eq:normBaSSUMPTION}    
\end{equation}
We then find a superoperator $\calF_\sigma:L(\calH_{\s{HW}}^{\s{glob}})\to L(\calH_{\s{Er}})$ (depending on some embedding function $\sigma$ of  our choice) such that 
   \begin{equation}\label{eq: lfactor}
   \calB_{\s{Er}}:=\calB_{\lambda_{\s{Er}},Q}^{\calE_{2t}^{\s{Er}}}=\calF_{\sigma}\circ \calB
   \end{equation}
for some appropriate choice of a matrix $\lambda_{\s{Er}}$. We then bound $\|\calF_{\sigma}\|_{\s{cb}}$  in order to control 
$\|\calB_{\s{Er}}\|_{\diamond}$. 

We begin by defining the embedding function. Following the notation from the upper bound proof, let $\s{P}_{2t}(I)$ denote the set of all HW operators on $\calH_{q}^{\otimes N}$ of weight at most $2t$ supported on a subset $I\in \binom{[N]}{2t}$, 
and  let $\s{P}_{2t}$ denote a basis of HW operators on $\calH_q^{\otimes 2t}$.  
An HW operator $W\in \s{P}_{2t}(I)$ can be written as $W=W_I\otimes I_{[N]\setminus I}$, where $W_I\in \s{P}_{2t}$. 
Let $[q^2]^N_{\leq t}$ denote the set of $N$-vectors over $[q^2]$ of weight at most $t$. 
A mapping $\sigma:\binom{[N]}{2t}\times \s P_{2t}\to [q^2]^{N}_{\leq t}\times [q^2]^N_{\leq t}$ is called an embedding function if for any subset $I\in \binom{[N]}{2t}$ and HW operator $V\in \s P_{2t}$ we have  
    \[
    \sigma(I,V)=(\ux_{V,I},\uy_{V,I}), \quad \text{where}\quad W_{\ux_{V,I}}^\dag W_{\uy_{V,I}}=\varphi_{V,I}  V\otimes I_{[N]\setminus I}, 
    \] 
    and $\varphi_{V,I}$ is some global phase factor. 
    For fixed $I$ and $V$, we define  
\[\hat{V}_I=\varphi_{V,I}^*  V^\dag\in \s{P}_{2t}.\] Observe that potentially there are many choices of $\sigma$. 
We will later show that such an embedding function always exists, and we will optimize our bound over its choice. 

The map $\calF_\sigma$, which we will now define, quantifies how close the action of the partial trace on the code space is to being constant, using only the corresponding proximity-to-constant information for limited HW operators. The role of $\sigma$ in this map is to determine how a local HW operator $V$ supported on a subsystem $I$ is represented in terms of the global low-weight HW data: for each pair $(I,V)$, the value $\sigma(I,V)=(\ux,\uy)$ specifies low-weight labels $\ux,\uy$ whose associated HW operators are supported on $I$ and satisfy $W_{\ux}^\dag W_{\uy}=V$ on $I$ and 
$\s{id}$ in $I^c$. In this way, $\sigma$ lifts local subsystem data to the corresponding global HW data from which the local erasure characteristic is reconstructed: once all such local HW operators on $I$ are taken into account, they determine the action of the erasure on the code space.

Let
\begin{align}\label{eq:calFsigma}
\calF_{\sigma}(Z) :=  \frac{1}{   q^{2t}}\sum_{I\in \binom{[N]}{2t}} \sum_{V\in \s{P}_{2t}} \bra{\ux_{V,I}} Z \ket{\uy_{V,I}}   {\hat{V}_I^\dag},    
\end{align}
 Formally, $\calF_{\sigma}(Z)$ is a block-diagonal operator on $\calH_{\s{Er}}=\bigoplus_I \calH_I$ where the $I$-th block is given by
\[\calF_{\sigma}(Z)_I=\frac{1}{q^{2t}}\sum_{V\in \s{P}_{2t}} \bra{\ux_{V,I}} Z \ket{\uy_{V,I}}   {\hat{V}_I^\dag},\]
where $\hat{V}_I$ is the operator defined on $\calH_I$ by 
\begin{equation}\label{eq:VI}
   \hat{V}_I=\sum_{\ui,\uj\in [q]^{2t}} \bra{\ui} (\varphi_{V,I}^* V^\dag)\ket{\uj}   \ket{\ui,I}\bra{\uj,I}. 
\end{equation}
Observe that the image of $\hat V_I$ is the subspace $\calH_I\subset \calH_{\s{Er}}$.

\begin{lemma}[Correctness of factorization]\label{lem:FactCor2}
    Factorization \eqref{eq: lfactor} holds; in other words, for any $X\in L(\calH)$
\begin{equation}
    \calF_{\sigma}(\calB(X))=\sum_{I\in \binom{[N]}{2t}}\sum_{\ui,\uj\in [q]^{2t}}\tr\p{(PE_{\ui,\uj}^{(I)}P -\lambda_{\s{Er}}^{(I)}(\ui,\uj)P)X}  \ket{\ui,I}\bra{\uj,I},\label{eq:CompDesc}
\end{equation}
where $\lambda_{\s{Er}}^{(I)}$ is a $[q]^{2t}\times [q]^{2t}$ matrix given by
    \[
    \lambda_{\s{Er}}^{(I)}(\ui,\uj)=\bra{\ui,I}\Big(\frac{1}{q^{2t}}\sum_{V\in \s{P}_{2t}}\lambda_t(\ux_{V,I},\uy_{V,I})  {\hat{V}^{\dag}_I}\Big)\ket{\uj,I}.
    \]

\end{lemma}
 
{\em Remark}: Note that $\calF_\sigma \circ \calB $ is a valid B\'eny-Oreshkov operator for $\calE^{\s{Er}}_{2t}$, as ${E_{\uj,I'}^\dag E_{\ui,I}}=0$ for all $I\neq I'$ (which was shown in \eqref{eq:deltaI1I2}).

\begin{proof}[Proof of Lemma~\ref{lem:FactCor2}]
    We begin the proof by analyzing the summands on the right-hand side. For a fixed $I$ and $X$, recall that $(PXP)_I$ denotes the partial trace $\tr_{[N]\setminus I}(PXP)$. We have
\begin{align}
      \sum_{\ui,\uj\in[q]^{2t}}\tr(PE_{\ui,\uj}^{(I)}P X)  \ket{\ui,I}\bra{\uj,I}&=\nonumber\sum_{\ui,\uj\in[q]^{2t}}\tr({\ket{\uj,I}\bra{\ui,I}}\otimes I_{[N]\setminus I} P XP)  \ket{\ui,I}\bra{\uj,I}\\
      &\label{eq:useLemma} =\sum_{\ui,\uj\in[q]^{2t}}\tr({\ket{\uj,I}\bra{\ui,I}} (P XP)_I)  \ket{\ui,I}\bra{\uj,I}\\
      &\nonumber=\sum_{\ui,\uj\in[q]^{2t}} {\bra{\ui,I} (P XP)_I\ket{\uj,I}}   \ket{\ui,I}\bra{\uj,I}\\
      &\label{eq:TransposeDef}={(PXP)_I},
\end{align}
where \eqref{eq:useLemma} follows from Lemma~\ref{lem:propPartTrace} and \eqref{eq:TransposeDef} follows from the matrix representation with respect to the basis $\ppp{\ket{\ui,I}}$. For the second term in \eqref{eq:CompDesc} we compute
\begin{align}
    \sum_{\ui,\uj\in[q]^{2t}}&\nonumber\tr(\lambda_{\s{Er}}^{(I)}(\ui,\uj)P X)  \ket{\ui,I}\bra{\uj,I}\\&=\tr(PX) \sum_{\ui,\uj\in[q]^{2t}}\lambda_{\s{Er}}^{(I)}(\ui,\uj)  \ket{\ui,I}\bra{\uj,I}\\
    &\nonumber=\tr(PX) \sum_{\ui,\uj\in[q]^{2t}}\bra{\ui,I}\p{\frac{1}{q^{2t}}\sum_{V\in \s{P}_{2t}}\lambda_t(\ux_{V,I},\uy_{V,I})  {\hat{V}^{\dag}_I}}\ket{\uj,I}  \ket{\ui,I}\bra{\uj,I}\\
    &\nonumber=\frac{1}{q^{2t}}\tr(PX) \sum_{V\in \s{P}_{2t}}\lambda_t(\ux_{V,I},\uy_{V,I}) \sum_{\ui,\uj\in[q]^{2t}}\bra{\ui,I} {\hat{V}^{\dag}_I} \ket{\uj,I}  \ket{\ui,I}\bra{\uj,I}\\
    &=\frac{1}{q^{2t}}\tr(PX) \sum_{V\in \s{P}_{2t}}\lambda_t(\ux_{V,I},\uy_{V,I})  {\hat{V}_I^\dag}.\label{eq:part2}
\end{align}
where \eqref{eq:part2} follows since the last sum on the previous line is an expansion of ${\hat{V}_I^\dag}$ in the
basis $\ket{\ui,I}\bra{\uj,I}$.

Further, recall that $\calB$ is the B\'eny-Oreshkov operator defined by $\lambda_t$ and the error set $\calE_t$, and it is therefore given by
\[\calB(X)=\sum_{\ux,\uy\in P_{N,t}}\tr\p{(P{W_{\uy}^\dag W_{\ux}}P-\lambda_t(\ux,\uy)P)X}  \ket{\ux}\bra{\uy},\]
and therefore the summands defining $\calF_{\sigma}(\calB(X))$ can be rewritten as:
\begin{align}
    \frac{1}{q^{2t}}\sum_{V\in \s{P}_{2t}}&\nonumber \bra{\ux_{V,I}}\calB(X)\ket{\uy_{V,I}}  {\hat{V}_I^\dag}\\&\nonumber= \frac{1}{q^{2t}} \sum_{V\in \s{P}_{2t}} \sum_{\ux,\uy\in \s{P}_{N,t}}\braket{\ux_{V,I}|\ux} \braket{\uy| \uy_{V,I}} \tr\p{(P{W_{\uy}^{\dag }W_{\ux}}P-\lambda_t(\ux,\uy)P)X}  {\hat{V}_I^\dag}\\
    &\nonumber=\frac{1}{q^{2t}}\sum_{V\in \s{P}_{2t}}  \tr\p{(P{W_{\uy_{V,I}}^{\dag }W_{\ux_{V,I}}}P-\lambda_t(\ux_{V,I},\uy_{V,I})P)X}  {\hat{V}_I^\dag}
    \\&\nonumber=\frac{1}{q^{2t}}\sum_{V\in \s{P}_{2t}}  \tr\p{({\hat{V}_{I}}\otimes I_{[N]\setminus I}-\lambda_t(\ux_{V,I},\uy_{V,I}))PXP}  {\hat{V}_I^\dag}\\
    &=\frac{1}{q^{2t}}\sum_{V\in \s{P}_{2t}}  \tr{\p{\hat{V}_{I}(PXP)_I}}  {\hat{V}_I^\dag}-\frac{1}{q^{2t}}\sum_{V\in \s{P}_{2t}}  \lambda_t(\ux_{V,I},\uy_{V,I})\tr(PX)  {\hat{V}_I^\dag}\label{eq:LemmaAgain}\\
    &\label{eq:transPdef}=\frac{1}{q^{2t}}\sum_{V\in \s{P}_{2t}}  \innerP{{\hat{V}_{I}^\dag},(PXP)_I}_{\s{HS}}  {\hat{V}_I^\dag}-\frac{1}{q^{2t}}\sum_{V\in \s{P}_{2t}}  \lambda_t(\ux_{V,I},\uy_{V,I})\tr(PX)  {\hat{V}_I^\dag}\\
    &\label{eq:lastEqBAA}={(PXP)_I}-\frac{1}{q^{2t}}\sum_{V\in \s{P}_{2t}}  \lambda_t(\ux_{V,I},\uy_{V,I})\tr(PX)  {\hat{V}_I^\dag}.
\end{align}
In the above, \eqref{eq:LemmaAgain} follows from Lemma~\ref{lem:propPartTrace}, \eqref{eq:transPdef} follows from {the definition of the Hilbert--Schmidt inner product} and \eqref{eq:lastEqBAA} follows since the set $\{{\hat{V}_I^\dag}/q^{t}\}_{V\in \s{P}_{2t}}$ is a Hilbert--Schmidt orthonormal basis for $L(\calH_q^{2t})$. Indeed, recall that $\hat{V}_I=\varphi_{V,I}^*  V^\dag$, and note that the adjoint operation permutes the elements of $ \s{P}_{2t}$ (which is a Hilbert--Schmidt orthonormal basis for $L(\calH_q^{2t})$ when normalized by $1/q^{t}$) and adds a phase factor to each element. We then recall that multiplying vectors of an orthonormal basis by phase factors yields another orthonormal basis

We now combine \eqref{eq:calFsigma}, \eqref{eq:TransposeDef}, \eqref{eq:part2} and \eqref{eq:lastEqBAA} to conclude \eqref{eq:CompDesc}:
\begin{align*}
    \calF_{\sigma}(\calB(X))&=\sum_{I\in \binom{[N]}{2t}}\frac{1}{q^{2t}}\sum_{V\in \s{P}_{2t}}\bra{\ux_{V,I}}\calB(X)\ket{\uy_{V,I}}  {\hat{V}_I^\dag}\\
    &=\sum_{I\in \binom{[N]}{2t}}\p{{(PXP)_I}- \tr(PX)   \frac{1}{q^{2t}}\sum_{V\in \s{P}_{2t}}  \lambda_t(\ux_{V,I},\uy_{V,I})  {\hat{V}_I^\dag}} \\
    &=\sum_{I\in \binom{[N]}{2t}}\p{\sum_{\ui,\uj\in[q]^{2t}}\tr(PE_{\ui,\uj}^{(I)}P X)  \ket{\ui,I}\bra{\uj,I}-\sum_{\ui,\uj\in[q]^{2t}}\tr(\lambda_{\s{Er}}^{(I)}(\ui,\uj)P X)  \ket{\ui,I}\bra{\uj,I}}\\
    &=\sum_{I\in \binom{[N]}{2t}}\sum_{\ui,\uj\in [q]^{2t}}\tr\p{(PE_{\ui,\uj}^{(I)}P -\lambda_{\s{Er}}^{(I)}(\ui,\uj)P)X}  \ket{\ui,I}\bra{\uj,I}. \qedhere
\end{align*}
\end{proof}

We have proved the validity of \eqref{eq: lfactor}.
Next, let us bound $\norm{\calF_{\sigma}}_{\s{cb}}$ using a technique similar to the argument in the upper bound part of the proof. Consider a bipartite multigraph $G_{\sigma}$ whose left and right vertex sets are sets of indices labeled by global HW operators in $P_{N,2t}$ and represented by vectors in $[q]^{2t}$.  We connect a pair of vertices $\ux,\uy$ from the left and right sides of $G_{\sigma}$ by a labeled edge $(\ux,\uy,V,I)$ if $\sigma(I,V)=(\ux_{V,I},\uy_{V,I})=(\ux,\uy)$. Denote the maximal degree of $G_{\sigma}$ by $\Delta_\sigma$. As in the proof of the upper bound, we find a partition of the (labeled) edges of the graph into $\Delta_{\sigma}$ into disjoint matchings, $M_1,\dots,M_{\Delta_{\sigma}}$. For $k=1,\dots,\Delta_\sigma$ define a superoperator $\calM_k:L(\calH_{\s{HW}^{\s{glob}}})\to L(\calH_{\s{Er}})$ by
  $$
\calM_{k}(Z)=\frac{1}{q^{2t}}\sum_{I\in \binom{[N]}{2t}}\sum_{\substack{V\in P_{2t}\\ (\sigma(I,V),V,I)\in M_k}}\bra{\ux_{I,V}}Z\ket{\uy_{I,V}}{\hat{V}_{I}^{\dag}}.
  $$
We have $\calF_{\sigma}=\sum_{k}\calM_k$ and therefore
\begin{equation}
    \label{eq:calMopSumBound}
    \norm{\calF_{\sigma}}_{\s{cb}}\leq \sum_{k=1}^{\Delta_{\sigma}}\norm{\calM_k}_{\s{cb}}.
\end{equation}

We claim that $\norm{\calM_{k}}_{\s{cb}}\leq 1$ for all $k$. To prove this, we show that $\calM_k$ has a decomposition of the form 
\begin{equation}
    \calM_{k}= \calL_k\circ\calJ_k\circ\calT_{U_k,V_k}, \label{eq:composition100}
\end{equation}
where each of these maps has $\s{cb}$-norm bounded by $1$. Let $\{(\ux_i,\uy_i,(V_i,I_i))\}_{i=1}^{m_k}$ be the labeled edges in $M_k$. Define $U_k, V_k:\calH_{\s{HW}}^{\s{glob}}\to \C^{m_k}$ by 
\[U_k=\sum_{i=1}^{m_k}\ket{i}\bra{\ux_i},\quad V_k=\sum_{i=1}^{m_k}\ket{i}\bra{\uy_i},\]
where $\ppp{\ket{i}}_{i}$ is an orthonormal basis of $\C^{m_k}$, and consider the superoperator $\calT_{U_k,V_k}$ as defined in Lemma~\ref{lem:compoDiamond}, $\calT_{U_k,V_k}(Z)=U_kZV_k^{\dag}$. Next, we define $\calJ_k$ to be the diagonal extractor map:
\[\calJ_k(Z)=\sum_{i=1}^{m_k}\braket{i|Z|i}   \ket{i}\bra{i}.\]
 Note that 
\begin{equation}
    \calJ_k\circ \calT_{U_k,V_k}(Z)=\sum_{i=1}^{m_k}\bra{\ux_i}Z\ket{\uy_i}  \ket{i}\bra{i}.\label{eq:middleCalc}
\end{equation}
Finally, consider the maps $\calL_k:L(\C^{m_k})\to L(\calH_{\s{Er}})$ given by 
\begin{align*}
    \calL_{k}(Z)&=\frac{1}{q^{2t}}\sum_{i=1}^{m_k}\bra{i}Z\ket{i}\sum_{I\in{\binom{[N]}{2t}}}\sum_{V\in\s{P}_{2t}} \delta_{(I,V),(I_i,V_i)}   {\hat{V}_{I}^\dag}=\frac{1}{q^{2t}}\sum_{i=1}^{m_k}\bra{i}Z\ket{i}  {\hat{V}_{i,I_i}^\dag},
    \end{align*}
where in the above, $\delta$ denotes the Kronecker delta. A straightforward calculation using \eqref{eq:middleCalc} reveals that indeed $\calL_k\circ \calJ_k\circ \calT_{U_k,V_k}=\calM_{k}$ as desired. 

Let us now prove that 
\[\max\p{\norm{\calL_k}_{\diamond},\norm{\calJ_k}_{\diamond},\norm{\calT_{U_k,V_k}}_{\diamond}}\leq 1.\] Indeed, by Lemma~\ref{lem:compoDiamond}, 
\[\norm{\calT_{U_k,V_k}}_{\s{cb}}\leq \norm{V_k}_\infty  \norm{U_k}_{\infty}=1,\] 
where the last equality follows since $U_k$ and $V_k$ are partial isometries (as $M_k$ is a matching and therefore all $\ux_i$ are distinct and all $\uy_i$ are distinct). We now bound the norm of $\calJ_k$. As shown in the proof of the upper bound, $\calJ_k$ is CPTP and by Lemma~\ref{lem:cbDiamondProp}, $\norm{\calJ_k}_{\s{cb}}=1$. It remains to show that $\norm{\calL_k}_{\s{cb}}\leq 1$. Denote $P_i=\ket{i}\bra{i}$ and let $X\in L(\C^{m_k}\otimes \C^{m_k})$ be any linear operator. If $X=X_1\otimes X_2$, then
\begin{align}
    I_{L(\C^{m_k})}\otimes \calL_k (X)&\nonumber=\frac{1}{q^{2t}}\sum_{i=1}^{m_k}\bra{i}X_2\ket{i}  X_1 \otimes  {\hat{V}_{i,I_i}^\dag}\\
    &=\nonumber\frac{1}{q^{2t}}\sum_{i=1}^{m_k} \tr_2(X_1 \otimes P_i X_2 P_i  ) \otimes  {\hat{V}_{i,I_i}^\dag}\\
    &=\frac{1}{q^{2t}}\sum_{i=1}^{m_k}\tr_2\p{(I_{\C^{m_k}}\otimes P_i) X (I_{\C^{m_k}}\otimes P_i)^\dag}  \otimes  {\hat{V}_{i,I_i}^\dag}.\label{eq:linEx}
\end{align}
Extending \eqref{eq:linEx} using linearity, we obtain 
\begin{align}
    I_{L(\C^{m_k})}\otimes \calL_k (X)=\sum_{i=1}^{m_k}\tr_2\p{Q_iXQ_i^\dag}\otimes \frac{{\hat{V}_{i,I_i}^\dag}}{q^{2t}}, \quad Q_i=I_{\C^{m_k}}\otimes P_i.
\end{align}
Using multiplicativity of the trace norm with respect to the tensor product, the triangle inequality, and the fact that the partial trace is a contraction for the trace norm (see Lemma~\ref{lem:stateNorms}), we obtain:
\begin{align}
     \norm{I_{L(\C^{m_k})}\otimes \calL_k (X)}_1 \leq \sum_{i=1}^{m_k}\norm{Q_iXQ_i^\dag}_1  \frac{1}{q^{2t}} \norm{{\hat{V}_{i,I_i}^\dag}}_1=\sum_{i=1}^{m_k}\norm{Q_iXQ_i^\dag}_1,\label{eq:Lkfinal}
\end{align}
where the last equality follows since for all $i$, ${\hat{V}_{i,I_i}^\dag}$ is a unitary on $\calH_q^{2t}$ and therefore satisfies $\norm{{\hat{V}_{i,I_i}^\dag}}_1=q^{2t}$. By \eqref{eq:Lkfinal}, in order to complete the proof that $\norm{\calL_{k}}_{\s{cb}}\leq 1$, it suffices to show that for all $X\in L(\C^{m_k}\otimes \C^{m_k})$, we have
\begin{equation}
    \sum_{i=1}^{m_k}\norm{Q_iXQ_i^\dag}_1\leq \norm {X}_1.\label{eq:letsconclude}
\end{equation}
Indeed, consider the superoperator $\calR$ given by the Kraus operators $\ppp{Q_i}_{i=1}^{m_k}$. Since the $Q_i$'s satisfies the completeness relation we have that $\calR$ is CPTP and by Lemma~\ref{lem:cbDiamondProp} 
\begin{equation}
    \norm{\calR(X)}_1\leq\norm{\calR}_{\s{cb}}  \norm{X}_1=\norm{X}_1. \label{eq:Part1Rnorm}
\end{equation}
 On the other hand, note that $Q_iXQ_i$ are supported on orthogonal spaces and therefore the singular values of $\calR(X)$, $\s{S}(\calR(X))$, are given by the union (in a multiset manner) of $\s{S}(Q_i X Q_i)$. In particular, 
    \begin{align}
        \norm{\calR(X)}_1=\sum_{\lambda\in \s{S}(\calR(X))}\lambda =\sum_{i}\sum_{\lambda_i\in \s{S}(Q_iXQ_i)}\lambda_i=\sum_{i}\norm{Q_iXQ_i}_1.\label{eq:normSumEq1}
    \end{align}
Combining \eqref{eq:Part1Rnorm} and \eqref{eq:normSumEq1}, we conclude \eqref{eq:letsconclude}, as desired. 
To conclude, combining the norm bounds on $\calL_k,\calT_k$ and $\calJ_k$ with \eqref{eq:composition100} and \eqref{eq:calMopSumBound} we conclude that 
\[\norm{\calF_\sigma}_{\s{cb}}\leq \sum_{k=1}^{\Delta_\sigma} \norm{\calM_{k}}_{\s{cb}}\leq \sum_{k=1}^{\Delta_\sigma} \norm{\calL_{k}}_{\s{cb}}  \norm{\calJ_{k}}_{\s{cb}}\norm{\calT_{U_k,V_k}}_{\s{cb}}\leq \Delta_{\sigma}.\]
Combining the above with \eqref{eq:normBaSSUMPTION} we get
\[\zeta(\calE^{\s{Er}}_{2t},Q)\leq \norm{\calB_{\s{Er}}}_{\diamond}\leq  \norm{\calB_{\s{Er}}}_{\s{cb}}\leq \norm{\calF_\sigma}_{\s{cb}}  \norm{\calB}_{\s{cb}}\leq \Delta_{\sigma}  \norm{\calB}_{\diamond}\leq 2\Delta_{\sigma}  \zeta(\calE_{t},Q). \] 
\medskip
\noindent\textbf{Existence and maximal degree of embedding functions:}
 Let us now show that an embedding function always exists and bound the maximal degree of the corresponding graph. First, recall that each HW operator $W\in P_{2t}$ can be represented by a vector $\uv\in[q^2]^{2t}$ such that $\uv=(a_i,b_i)_{i\in [2t]}$, $a_i,b_i\in[q]$,
    \[
    W=W_{\uv}=\bigotimes_{i}X^aZ^b.
    \] 
Recalling \eqref{eq: Wab}, we note that
     \[
     W_{\uv}^{\dag}=\omega^{\ua  \ub}W_{-\uv}, \quad \omega=e^{2\pi i/q},
     \]
where $-\uv$ is computed $\operatorname{mod} q$. Additionally, for two HW operators given by $\uv=(a_i,b_i)_i$ and $\uv'=(a_i',b_i')_i$ we have:
\[W_{\uv}W_{\uv'}=\omega^{\uv\odot \uv'}W_{\uv+\uv'}, \quad \uv+\uv'=(a_i+a_i',b_i+b_i')_i,\]
where addition is again modulo $q$ and $\odot$ denotes the symplectic inner product. Thus, the problem of finding an embedding function reduces to the following combinatorial problem: Find a function $\sigma:\binom{[N]}{2t}\times [q^2]^{2t}\to [q^2]^{N}_{\leq t}\times [q^2]^{N}_{\leq t}$ (where $[q^2]^{N}_{\leq t}$ denotes the vectors of Hamming weight at most $t$ in $[q^2]^N$) such that for any $(I,\uv)$ we have 
\[\sigma(I,\uv)=(\ux,\uy), \quad\text{s.t.} \quad \hat{\uv}=\uy-\ux, \]
where $\hat{\uv}$ is the length-$N$ vector obtained by placing $\uv$ in the coordinates of $I\subset [N]$ and
$(0,0)$ in the remaining coordinates. 

Given $I\in \binom{[N]}{2t}$, let $I=I_1\bigcup I_2$ be its partition into two halves of size $t$ each. For $\uv\in [q^2]^{2t}$ let $\uv_1\in [q^2]^t$ and $\uv_2\in [q^2]^t$ such that $(\uv_1,\uv_2)=\uv$. Define $\hat{\uv}_1\in [q^2]^N$ and $\hat{\uv}_2\in [q^2]^N$ to be the vectors obtained by placing $-\uv_1$ and $\uv_2$ in $I_1$ and $I_2$, respectively,
and $(0,0)$ elsewhere. Now define $\sigma(I,\uv)=(\hat{\uv}_1,\hat{\uv}_2)$.
By construction, $\hat{\uv}_1$ and $\hat{\uv}_2$ are of weight at most $t$ and $\hat{\uv}=\hat{\uv}_2-\hat{\uv}_1$ as desired. Let us estimate the maximum degree of the resulting graph $G_{\sigma}$. Let $\ux$ be a vertex in $G_{\sigma}$.
Clearly,
\[\s{\mathrm{deg}}(\ux)\leq \binom{N}{2t} B_{q^2}(2t,t),\]
and since this bound does not depend on $\ux$, it is also valid for $\Delta_{\sigma}$.

\subsection{Proof of Proposition~\ref{prop:SSVzeta}}\label{app:Complexity} 
Throughout this proof, we denote $\calH_2^{\otimes N}$ by simply $\calH$. Consider the $\binom{N}{t}2^t\times \binom{N}{t}2^t$ matrix $\lambda$, indexed by pairs $((I_1,\ui),(I_2,\uj))$, $I_1,I_2\in\binom{N}{t}$, $\ui,\uj\in \ppp{0,1}^t$ (representing pairs of operators of $\calE_{t}^{\s{Er}}$, see \eqref{eq:erasure_operator}) given by 
    \[\lambda_{I_1, I_2, \ui, \uj}=\tr\p{\Gamma {E_{I_2,\uj}^\dag E_{I_1,\ui}}}=\delta_{I_1,I_2} \tr\p{\Gamma ({\ket{\uj}_{I_1}\bra{\ui}_{I_1}}\otimes I_{{[N]}\setminus I_1})},\]
    where the last equality follows from \eqref{eq:deltaI1I2}. We follow the notation of the proof of Theorem~\ref{th:equivErGen}  and define
    \begin{equation}
        E_{\ui,\uj}^{(I)}:={E_{\uj,I}^\dag E_{\ui,I}=\ket{\uj}_{I}\bra{\ui}_{I}}\otimes I_{[N]\setminus I}, \quad \lambda_{\ui,\uj}^{(I)}:={\lambda_{I,I,\ui,\uj}}.\label{eq:defLAMBDAIij1}
    \end{equation}
    Let $\calB_{\lambda,Q}^{\s{Er}}$ be the B\'eny-Oreshkov operator associated with $\lambda$:
\[\calB_{\lambda,Q}^{\s{Er}}(X)=\sum_{I\in \binom{[N]}{t}}\sum_{\ui,\uj\in [2]^t}\tr\p{(PE_{\ui,\uj}^{(I)} P-\lambda_{\ui,\uj}^{(I)}P) X )}  \ket{\ui,I}\bra{\uj,I},\]
where $P=P_Q$ is the projection on $Q$. To prove the upper bounds, it is sufficient to show that 
\[\norm{\calB_{\lambda,Q}^{\s{Er}}}_{\diamond}\leq 2K\binom{N}{t} V_{\s{sub}}(Q,t). \]
For a fixed $I\in \binom{[N]}{t}$, consider the superoperator: $\calB_{\lambda,Q,I}^{\s{Er}}:L(\calH)\to L(\C^{2^t})$ given by 
\[\calB_{\lambda,Q,I}^{\s{Er}}(X)=\sum_{\ui,\uj\in [2]^t}\tr\p{(PE_{\ui,\uj}^{(I)} P-\lambda_{\ui,\uj}^{(I)}P) X )}  \ket{\ui,I}\bra{\uj,I},\]
and note that $\calB_{\lambda,Q}^{\calE_t^{\s{Er}}}$ is exactly the direct sum of $\p{\calB_{\lambda,Q,I}^{\s{Er}}}_I$. Thus, by the triangle inequality 
\begin{equation}
    \norm{\calB_{\lambda,Q}^{\calE_t^{\s{Er}}}}_{\diamond}\leq \binom{N}{t}\max_{I\in \binom{[N]}{t}}\norm{\calB_{\lambda,Q,I}^{\s{Er}}}_\diamond.\label{eq:UnionI}
\end{equation}
Using the Cauchy--Schwarz argument, one can show that the above diamond norm can be dominated by the operator norm (with respect to the trace norm on $L(\C^{2^t})$):
\begin{equation}
    \norm{\calB_{\lambda,Q,I}^{\s{Er}}}_\diamond\leq K  \sup_{\norm{X}_1\leq 1}\norm{\calB_{\lambda,Q,I}^{\s{Er}}(X)}_1;\label{eq:normoperatordiamind}
\end{equation}
see \eqref{eq:DiamBound} for a complete derivation. Let us denote 
\[\norm{\calB_{\lambda,Q,I}^{\s{Er}}}_{\calD} :=  \sup_{\rho\in D(\calH)} \norm{\calB_{\lambda,Q,I}^{\s{Er}}(\rho)}_1.\]

 For a fixed operator $X$, we can always write $X=X_1+iX_2$ where $X_1, X_2$ are Hermitian operators. 
\[X_1=\frac{X+X^\dag}{2},\quad X_2=\frac{X-X^\dag}{2i}.\]
Note that by the triangle inequality $\norm{X_i}_1\leq \norm{X}_1$ for $i=1,2$. We can also decompose $X_i=X_i^+-X_i^-$, where $X_i^\pm$ are positive operators supported on orthogonal space and therefore satisfy:
\[\norm{X_i}_1=\norm{X_i^+}_1+\norm{X_i^-}_1.\]
We also note that $X_i^{\pm}/\norm{X_i^{\pm}}$ are positive with trace $1$ and therefore belong in $D(\calH)$. Substituting into  \eqref{eq:normoperatordiamind}, we obtain: 
\begin{align}    \norm{\calB_{\lambda,Q,I}^{\s{Er}}}_\diamond&\nonumber\leq K  \sup_{\norm{X}_1\leq 1}\norm{\calB_{\lambda,Q,I}^{\s{Er}}(X)}_1\\
&\nonumber\leq K  \sup_{\norm{X}_1\leq 1}\norm{\sum_{i=1,2}\p{\calB_{\lambda,Q,I}^{\s{Er}}(X_i^{+})-\calB_{\lambda,Q,I}^{\s{Er}}(X_i^{-})}}_1\\
&\nonumber\leq K  \sup_{\norm{X}_1\leq 1}\sum_{i=1,2}\p{\norm{\calB_{\lambda,Q,I}^{\s{Er}}(X_i^{+})}_1+\norm{\calB_{\lambda,Q,I}^{\s{Er}}(X_i^{-})}_1}\\
&\nonumber\leq K  \sup_{\norm{X}_1\leq 1}\sum_{i=1,2}\p{\norm{\calB_{\lambda,Q,I}^{\s{Er}}(X_i^{+})}_1+\norm{\calB_{\lambda,Q,I}^{\s{Er}}(X_i^{-})}_1}\\
&\nonumber= K  \sup_{\norm{X}_1\leq 1}\sum_{i=1,2}\p{\norm{X_i^+}_1  \norm{\calB_{\lambda,Q,I}^{\s{Er}}(X_i^{+}/\norm{X_i^+}_1)}_1+\norm{X_i^-}  \norm{\calB_{\lambda,Q,I}^{\s{Er}}(X_i^{-}/\norm{X_i^-})}_1}\\
&\nonumber\leq  K  \sup_{\norm{X}_1\leq 1}\sum_{i=1,2} \p{ \norm{X_i^-}_1  \norm{\calB_{\lambda,Q,I}^{\s{Er}}}_{\calD} + \norm{X_i^+}_1  \norm{\calB_{\lambda,Q,I}^{\s{Er}}}_{\calD} }\\
&\leq 2 K \norm{\calB_{\lambda,Q,I}^{\s{Er}}}_{\calD}.\label{eq:DiamondNormDnorm}
\end{align}

We furthermore note that whenever $I$ is fixed, $\calB_{\lambda,Q,I}^{\s{Er}}$ is the superoperator defined in \cite[eq. (A13)]{yi2024complexity}, for the special case where $\calN_I$ is the $I$-subsystem erasure channel given by Kraus operators $\ppp{E_{I,\ui}}_{\ui\in [2^t]}$. In \cite[Proposition 6]{yi2024complexity} it is shown that for a fixed state $\rho\in D(\calH)$ we have
 \begin{equation}
     \norm{\calB_{\lambda,Q,I}^{\s{Er}}(\rho)}_{1}=\norm{\sigma_I-\tr(\sigma)\Gamma_I}_1\label{eq:GotemanIden},
 \end{equation}
where $\sigma=P\rho P$. Combining \eqref{eq:UnionI}, \eqref{eq:GotemanIden} and \eqref{eq:DiamondNormDnorm} we obtain
\begin{align*}
    \zeta(\calE_{t}^{\s{Er}},Q)\leq \norm{\calB_{\lambda,Q}^{\s{Er}}}_{\diamond} &\leq 2K\binom{N}{t}\max_{I\in \binom{[N]}{t}}\norm{\calB_{\lambda,Q,I}^{\s{Er}}}_{\calD}\\
    &\leq 2K\binom{N}{t}\max_{I\in\binom{[N]}{t}}\sup_{\rho\in D(Q)}\norm{\rho_I-\Gamma_I}_1= 2K\binom{N}{t}V_{\s{sub}}(Q,t). 
\end{align*}
This proves the upper bound. 

We now prove the lower bound. By Remark~\ref{rem:generlizedLemma2} we have 
\begin{equation}
    \zeta(\calE_{t}^{\s{Er}},Q)\geq \frac{1}{2}\norm{\calB_{\lambda,Q}^{\s{Er}}}_{\diamond}.\label{eq:Lowlower}
\end{equation}
    On the other hand, let $\rho\in D(Q)$ and $I_{\max}\in\binom{[N]}{t}$ be such that $\norm{\rho_{I_{\max}}-\Gamma_{I_{\max}}}_1=V_{\s{sub}}(Q,t)$.  
By the definition of the diamond norm 
\begin{align}
    \norm{\calB_{\lambda,Q}^{\s{Er}}}_{\diamond}&\geq \norm{I_{L(\calH)}\otimes \calB_{\lambda,Q}^{\s{Er}} (\rho\otimes \rho)}_1=\norm{\rho}_1  \norm{\calB_{\lambda,Q}^{\s{Er}} (\rho)}_1\label{eq:multiplagain}\\
    &=\norm{\sum_{I\in \binom{[N]}{t}}\calB_{\lambda,Q,I}^{\s{Er}} (\rho)}_1=\sum_{I\in \binom{[N]}{t}}\norm{\calB_{\lambda,Q,I}^{\s{Er}} (\rho)}_1\label{eq:directsumTrace}\\
&\geq \biggl\|\calB_{\lambda,Q,I_{\max}}^{\s{Er}} (\rho)\biggr\|_1=\norm{\rho_{I_{\max}}-\Gamma_{I_{\max}}}_1=V_{\s{sub}}(Q,t).\label{eq:finale}
\end{align}
In the above, \eqref{eq:multiplagain} uses multiplicativity of the trace norm under tensor products, \eqref{eq:directsumTrace} uses additivity of the trace norm with respect to a direct sum of operators, and \eqref{eq:finale} follows from \eqref{eq:GotemanIden}. We conclude the proof by combining \eqref{eq:Lowlower}  and \eqref{eq:finale}.

\subsection{Proof of Proposition~\ref{prop:AQECcond}}\label{app:InnerProdCond}

        For a fixed normalized codeword $\ket{c_0}\in Q$ define
         \[ \varepsilon_{\max} :=  \max_{\substack{i,j\in [K]\\k,l\in [M]} } \abs{\bra{c_i} E_k^\dag E_l \ket{c_j}-\bra{c_0} E_k^\dag E_l \ket{c_0}  \delta_{i,j}}.\]
        For an $M\times M$ matrix $\lambda'$ define
        \[\varepsilon_{\lambda'} :=  \max_{\substack{i,j\in [K]\\k,l\in [M]} } \abs{\bra{c_i} E_k^\dag E_l \ket{c_j}-\lambda'_{k,l} \delta_{i,j}}.\]
        Fix some $k,l\in [M]$. Note that for $i\neq j$ we have 
        \begin{align}
             \abs{\bra{c_i} E_k^\dag E_l \ket{c_j}}=\abs{\bra{c_i} E_k^\dag E_l \ket{c_j}-\bra{c_0} E_k^\dag E_l \ket{c_0}  \delta_{i,j}}= \abs{\bra{c_i} E_k^\dag E_l \ket{c_j}-\lambda'_{k,l} \delta_{i,j}}\leq \varepsilon_{\lambda'}.\label{eq:epslambda1}
        \end{align}
        For $i=j$, by the triangle inequality, we have
        \begin{align}
            \abs{\bra{c_i} E_k^\dag E_l \ket{c_i}-\bra{c_0} E_k^\dag E_l \ket{c_0}  \delta_{i,i}}&\nonumber\leq \abs{\bra{c_i} E_k^\dag E_l \ket{c_i}-\lambda'_{k,l}  \delta_{i,i}}+\abs{\bra{c_0} E_k^\dag E_l \ket{c_0}\delta_{i,i}-\lambda'_{k,l}  \delta_{i,i}}\\
            &\nonumber=\abs{\bra{c_i} E_k^\dag E_l \ket{c_i}-\lambda'_{k,l}  \delta_{i,i}}+\abs{\bra{c_0} E_k^\dag E_l \ket{c_0}-\lambda'_{k,l}  \delta_{0,0}}\\
            &\leq 2\varepsilon_{\lambda'}.\label{eq:epslambda2}
        \end{align}
        Combining \eqref{eq:epslambda1} and \eqref{eq:epslambda2}, we have $\varepsilon_{\max}\leq 2 \varepsilon_{\lambda'}$ so the assumption of \eqref{eq:BOcond1} implies that
        \[ \varepsilon_{\max}\leq \frac{\varepsilon^2}{K^2M^2}.\]

Using virtually the same argument, it is shown that we may assume for the specific choice $\lambda_{k,l}=\bra{c_0} {E_l^\dag E_k} \ket{c_0}$, 
the assumptions of \eqref{eq:BOcond1}, \eqref{eq:BOcond2}, and \eqref{eq:BOcond3} hold, when we multiply the r.h.s. by $2$. We remark that 
we require this procedure to ensure that $\lambda$ is Hermitian, which assures that the corresponding B\'eny-Oreshkov operator 
$\calB_{\lambda,Q}^{\calE}$ is Hermitian preserving and therefore we can use Lemma~22 of \cite{elimelech2026asymptotically}, required
for our analysis.
      It was shown in the proof of Theorem~\ref{th:BOerrorset} that $\calB_{\lambda,Q}^{\calE}$ is a Hermitian preserving superoperator. As shown in the proof of Theorem~\ref{th:BOerrorset}, it is sufficient to show that 
      $\|\calB_{\lambda,Q}^{\calE}\|_{\diamond}\leq \varepsilon^2$. 
      We adapt the strategy of \cite{elimelech2026asymptotically}. Since $\calB_{\lambda,Q}^{\calE}$ eventually operates on $L(Q)$, by \cite[Theorem 3.51]{watrous2018theory}, there exists a pure state $\ket{\psi}\in Q^{\otimes 2}$ such that $\|{\calB_{\lambda,Q}^{\calE}}\|_{\diamond}=\norm{I_{L(Q)}\otimes \calB_{\lambda,Q}^{\calE} (\ket{\psi}\bra{\psi})}_{1}$ (see Lemma~22 of \cite{elimelech2026asymptotically} for a formal proof of this statement). Let $\ket{\psi}$ be given by the Schmidt decomposition $\ket{\psi}=\sum_i \sqrt{p_i} \ket{i}\otimes \ket{i'}$, where $\ket{i}$ and $\ket{i'}$ are orthonormal bases of $Q$. Using the triangle inequality, we have  
    \begin{align*}
        \norm{\calB_{\lambda,Q}^{\calE}}_{\diamond}&=\norm{I_{L(Q)}\otimes \calB_{\lambda,Q}^{\calE} (\ket{\psi}\bra{\psi})}_{1}\leq \sum_{i,j=0}^{k-1}\sqrt{p_ip_j}\norm{\ket{i}\bra{j}\otimes \calB_{\lambda,Q}^{\calE}(\ket{i'}\bra{j'})}_1.
    \end{align*}
    Since $\ppp{\ket{i}}_i$ are orthonormal, $\norm{\ket{i}\bra{j}}_1=1$, and therefore, by the multiplicativity of the trace norm with respect to the tensor product, we have:
    \begin{align}
        \norm{\calB_{\lambda,Q}^{\calE}}_{\diamond}&\nonumber
        \leq \sum_{i,j=0}^{k-1}\sqrt{p_ip_j}\norm{\ket{i}\bra{j}\otimes \calB_{\lambda,Q}^{\calE}(\ket{i'}\bra{j'})}_1\\
       &\nonumber \leq  \sum_{i,j=0}^{k-1}\sqrt{p_ip_j}\sup_{\norm{X}_1=1}\norm{\calB_{\lambda,Q}^{\calE}(X)}_1\\
        &\nonumber=\Big(\sum_{i=0}^{K-1}\sqrt{p_i}\Big)^2\sup_{\norm{X}_1=1}\norm{\calB_{\lambda,Q}^{\calE}(X)}_1\\
        &\nonumber \leq K\sup_{\norm{X}_1=1}\norm{\calB_{\lambda,Q}^{\calE}(X)}_1\\
        &=K\sup_{\norm{X}_1=1}\Big\|\sum_{k,l}\tr(XB_{k,l}) \ket{k}\bra{l}\Big\|_1,\label{eq:DiamBound}
    \end{align}
    where we have used the Cauchy--Schwarz inequality. Using \eqref{eq:DiamBound} and the triangle inequality, we obtain
            \begin{align}
                \norm{\calB_{\lambda,Q}^{\calE}}_{\diamond}\nonumber&\leq K \sup_{\norm{X}_1=1}\norm{\sum_{k,l}\tr(XB_{k,l})  \ket{k}\bra{l}}_1\leq K\sup_{\norm{X}_1=1}\sum_{k,l}|\tr(XB_{k,l})|\\
                &\leq K \sum_{k,l} \sup_{\norm{X}_1=1}|\tr(XB_{k,l})|\leq K \sum_{k,l} \norm{B_{k,l}}_\infty.\label{eq:diamondKinfty}
            \end{align}
            Here, the last inequality follows from the duality principle (see Lemma~\ref{lem:stateNorms})
            
            Let us now analyze the expression in \eqref{eq:DiamBound} under the assumptions in \eqref{eq:BOcond1}, \eqref{eq:BOcond2}, and \eqref{eq:BOcond3}:
           \begin{enumerate}
            \item Assume that \eqref{eq:BOcond1} holds. Using \eqref{eq:normsDef}, we obtain
            \begin{equation}
                \norm{B_{k,l}}_\infty=\max_{\lambda\in{\s{eig}(B_{k,l}B_{k,l}^\dag)}}\sqrt{\abs{\lambda}}\leq \sqrt{\sum_{\lambda\in{\s{eig}(B_{k,l}B_{k,l}^\dag)}}\abs{\lambda} }=\norm{B_{k,l}}_2.\label{eq:normIneq1}
            \end{equation}
            We also observe that $B_{k,l}$ is given by
            \begin{align}
                B_{k,l}=P {E_l^\dag E_k} P -\lambda_{k,l}P=\sum_{i\neq j} \bra{c_i}{E_l^\dag E_k}\ket{c_j }  \ket{c_i}\bra{c_j}+\sum_{i} (\bra{c_i}{E_l^\dag E_k}\ket{c_i }-\lambda_{k,l})  \ket{c_i}\bra{c_i}.
            \end{align}
            Estimating the 2-norm, we obtain
            \begin{align*}
                \norm{B_{k,l}}_2&=\sqrt{\tr\p{B_{k,l}B_{k,l}^\dag}}=\sqrt{\sum_{i\neq j} \abs{\bra{c_i}{E_l^\dag E_k}\ket{c_j }}^2+\sum_{i} \abs{(\bra{c_i}{E_l^\dag E_k}\ket{c_i }-\lambda_{k,l})}^2}\\
                &=\sqrt{\sum_{i, j=0}^{K-1} \abs{\bra{c_i}{E_l^\dag E_k}\ket{c_j }-\lambda_{k,l}  \delta_{i,j}}^2}\leq K   \varepsilon_{\max}.
                \end{align*}
                Combining the above with \eqref{eq:normIneq1} and \eqref{eq:diamondKinfty} we conclude that
                \[\norm{\calB_{\lambda,Q}^{\calE}}_{\diamond}\leq K^2M^2\varepsilon_{\max}
                \leq \varepsilon^2.\]

                \item The second case follows using similar arguments. Note that under the assumption that  $\bra{c_i} E_k^\dag E_l \ket{c_j}=0$ for all $j\neq i$, we have 
                 \[
                 \norm{B_{k,l}}_\infty =\Big\|{\sum_{i} (\bra{c_i}{E_l^\dag E_k}\ket{c_i }-\lambda_{k,l})  \ket{c_i}\bra{c_i}}\Big\|_\infty=\max_{\substack{i\in [K]\\k,l\in[M]}}\abs{\bra{c_i}{E_l^\dag E_k}\ket{c_i }-\lambda_{k,l}}\leq \frac{\varepsilon^2}{K M^2}.
                 \]
 Now, by \eqref{eq:normIneq1} and \eqref{eq:diamondKinfty} we have 
                \[\norm{\calB_{\lambda,Q}^{\calE}}_{\diamond}\leq K  \sum_{k,l} \frac{\varepsilon^2}{K M^2}=\varepsilon^2.\]
                \item Assuming that $\bra{c_i} E_k^\dag E_l \ket{c_j}=0$ for all $(i,k)\neq (j,l)$, we immediately conclude that $B_{k,l}=0$ for $k\neq l$ and therefore by the same argument used in the previous case 
                \[\norm{\calB_{\lambda,Q}^{\calE}}_{\diamond}\leq K  \sum_l \norm{B_{l,l}}_\infty \leq K  M   \max_{l\in[M]}\abs{\bra{c_i}E_l^\dag E_l\ket{c_i }-\lambda_{l,l}}\leq \varepsilon^2.\]
            
        \end{enumerate}

        \subsection{Proof of Theorem~\ref{th:exactQECpatition}}\label{app:tverberg}
        The proof relies on Tverberg's theorem, a classical result in convex geometry which states that any set of $(K-1)(m+1)+1$ points in $\R^m$ can be partitioned into $K$ subsets whose convex hulls have a nonempty intersection. In \cite{aydin2026quantum}, the authors showed that the existence of a partition for which the corresponding system of linear equations admits a nonnegative solution can be reduced to an application of Tverberg's theorem. This result is formulated as follows:
        \begin{lemma}[{\cite[Proposition~VII.4.]{aydin2026quantum}}]\label{lem:TvErequationsSys}
            Let $M\in \N$ be a positive integer, $C$ be a finite index set and let $\calX=\ppp{x_c ~:~c\in C}$ be a set of variables indexed by $C$. Assume that $\beta_{m,c}$, $m\in[M]$, $c\in C$, be real constants. Then, if $|C|\geq (K-1)(M+1)+1$ there exists a disjoint partition of $C$ to $C_1,\dots,C_K$ such that the system of equations 
            \begin{align*}
                \sum_{c\in C_1}\beta_{1,c}  x_c&=\sum_{c\in C_2}\beta_{1,c}  x_c =\cdots=\sum_{c\in C_K}\beta_{1,c}  x_c\\
                &~\vdots\\
                \sum_{c\in C_1}\beta_{M,c}  x_c&=\sum_{c\in C_2}\beta_{M,c}  x_c =\cdots=\sum_{c\in C_K}\beta_{M,c}  x_c\\
            \end{align*}
            has a nonnegative solution $(x_c)_{c\in C}$ which satisfies
            \[\sum_{c\in C_1}x_c=\sum_{c\in C_2}x_c= \cdots =\sum_{c\in C_K}x_c=1.\]
        \end{lemma}
Let us prove Theorem~\ref{th:exactQECpatition}.
         We begin with a slightly simpler case where $(\calH_X,d,\mathscr{E})\in \mathscr{L}_2$, which we then generalize to the first-level scenario. 
      \begin{proof}[
      {\underline{Proof for the case $(\calH_X,d,\mathscr{E})\in \mathscr{L}_2$}}]
        The assumption $(\calH_X,d,\mathscr{E})\in \mathscr{L}_2$ implies that $(\calH_X,d,\mathscr{E})\in \mathscr{L}_1$  (by definition). Thus, by item 1 of Lemma~\ref{lem:PartitionCodesProp}  a partition code $Q$ induced by $C$ (with $d(C)>t$) with a canonical basis $\ppp{\ket{c_i}}_{i=1}^K$ satisfies 
        \[\bra{c_i}E^\dag F \ket{c_j}=0 \] 
        for all $i\neq j$ and $E,F\in \calE_t$. In particular, the orthogonality KL conditions are fulfilled.  Item 2. of Lemma~\ref{lem:PartitionCodesProp} gives that
        \[\bra{c_i}E^\dag F \ket{c_i}=0 \] 
        for all $i$ and $E\neq F$. The conclusion is that a partition code $Q$ spanned by 
        \begin{equation}
            \ket{c_i}=\sum_{c\in C}\alpha_c  \ket{c}, \quad \sum_{c\in C_i}|\alpha_c|^2=1, \quad i=0,1,\dots, K\label{eq:Tver1}
        \end{equation}
        is QEC for $\calE_t$ if and only if the non-deformation conditions hold:
        \begin{equation}
            \bra{c_1}E^\dag E\ket{c_1}=\bra{c_2}E^\dag E\ket{c_2}= \cdots =\bra{c_K}E^\dag E\ket{c_K},\label{eq:Tver3}
        \end{equation}
        for all $E\in \calE_t$.
        Let us expand the above expressions:
        \begin{equation}
            \bra{c_i}E^\dag E\ket{c_i}=\sum_{c,c'\in C_i}\alpha_c^*\alpha_{c'}\bra{c}E^\dag E \ket{c'}=\sum_{c\in C_i}|\alpha_c|^2  \bra{c}E^\dag E \ket{c}
        \end{equation}
        where we used that $\bra{c}E^\dag E \ket{c'}=0$ for $c,c'$ such that $d(c,c')>t$ by the definition of $\mathscr{L}_1$. 
        Replacing $|\alpha_c|^2$ by $x_c$ denoting $\calE_t=\ppp{E_1,\dots,E_M}$ and $\beta_{i,c}=\bra{c} E_i^\dag E_i \ket{c}$, the system of equations given by \eqref{eq:Tver1} and \eqref{eq:Tver3} reduces to
        \begin{align}
        &\sum_{c\in C_i}x_c  \beta_{m,c}\nonumber=\sum_{c\in C_j}x_c  \beta_{m,c} \quad \forall i,j=1,\dots,K,~m=1,\dots,M \\
         &\sum_{c\in C_i}x_c=1  \quad \forall i=1,\dots,K\label{eq:Tver2}
        \end{align} which is at  form of the system of Lemma~\ref{lem:TvErequationsSys}. In particular, by Lemma~\ref{lem:TvErequationsSys}, under the assumption $|C|\geq (K-1)(|\calE_t|+1)+1
        $ there exists a partition $C=\bigcup_{i=1,\dots,K} C_i$ and a nonnegative solution $(x_c)_{c\in C}$, which satisfies \eqref{eq:Tver2}. In particular, the partition code obtained by the partition $C_1,\dots, C_K$ and coefficients $\alpha_c=\sqrt{x_c}$ (whose dimension is $K$) is QEC for $\calE_t$.
        
        \end{proof}
\begin{proof}[
      {\underline{Proof for the case $(\calH_X,d,\mathscr{E})\in \mathscr{L}_1$}}] The proof follows the same idea as the case $(\calH_X,d,\mathscr{E})\in \mathscr{L}_2$, with a minor difference:   While the equation $\bra{c_i}E^\dag F \ket{c_j}=0$ for $i\neq j$ remains  true, $\bra{c_i}E^\dag F \ket{c_i}$ may be nonzero for distinct $E$ and $F$. In particular, the constraint of \eqref{eq:Tver3}, extends to 
      \begin{equation*}
          \bra{c_1}E^\dag F\ket{c_1}=\bra{c_2}E^\dag F\ket{c_2}= \cdots =\bra{c_K}E^\dag F\ket{c_K}, \quad \forall E,F\in \calE_t.
      \end{equation*}
      This imposes more equations on our system: 
        \begin{align*}
            \bra{c_i}E^\dag F \ket{c_i}&=\sum_{c,c'\in C_i}\alpha_c^{*}\alpha_{c'} \bra{c}E^\dag F \ket{c'}=\sum_{c\in C_i}\abs{\alpha_c}^2   \bra{c}E^\dag F \ket{c}, 
        \end{align*}
where we used that $\bra{c}E^\dag F\ket{c'}=0$ for $c,c'$ such that $d(c,c')>t$ by the definition of $\mathscr{L}_1$. Reformulating in the form of \eqref{eq:Tver2}, we look for a partition that admits a nonnegative solution to the following system:
\begin{align}
        &\sum_{c\in C_i}x_c  \beta_{m,k,c}\nonumber=\sum_{c\in C_j}x_c  \beta_{m,k,c} \quad \forall i,j=1,\dots,K,~m,k=1,\dots,M \\
         &\sum_{c\in C_i}x_c=1  \quad \forall i=1,\dots,K\label{eq:Tver8}
        \end{align}
    where $\beta_{m,k,c}=\bra{c}E_m^\dag E_k \ket{c}$. Lemma~\ref{lem:TvErequationsSys} we conclude that a partition that admits a solution to the system \eqref{eq:Tver8} exists if $|C|\geq (K-1)(|\calE_t|^2 +1)+1$, as desired.
\end{proof}
\subsection{Proof of Theorem~\ref{th:AQECrandomPar}}\label{app:RandomPartit}

\begin{proof}[\underline{The case $(\calH_X,d,\mathscr{E})\in \mathscr{L}_2$}]  By Lemma~\ref{lem:SufCondHira} it is sufficient to find numbers $\ppp{\lambda_E}_{E\in\calE_t}$ such that with probability at least $p_t$ the canonical basis of the random code $\calQ_{K,L}$ satisfies 
\begin{equation}
    \max_{\substack{i\in [K]\\E\in \calE_t} } \abs{\bra{\s{c}_i} E^\dag E \ket{\s{c}_i}-\lambda_{E}}\leq \frac{\varepsilon^2}{2K |\calE_t|}  \qquad(={\rm Eq.}\eqref{eq:SecondLevAQEC}).\label{eq:secondLevAQEC1}
\end{equation}
Let $A_t$ be the event that the random underlying classical code $\s{C}=\ppp{\s{x_1},\dots,\s{x_L}}$ satisfies $d(\s{C})>t$. For $E\in \calE_t$ and $i\in[K]$ define $A_{E,i}$ to be the event that
\begin{equation}
    \Big|\frac{1}{T}\sum_{j=0}^{T-1} \bra{\s{x}_{i  T+j} }E^\dag E \ket{\s{x}_{i  T+j} }-\lambda_E\Big|\leq \frac{\varepsilon^2}{2KM_t}, \quad \lambda_E=\E_{\s{x}\sim \mu}\pp{\bra{\s{x} }E^\dag E \ket{\s{x} }}.\label{eq:events1}
\end{equation}
We claim that
\begin{align}
    \P\pp{\s{Q}^{\s{P}}_{K,L}(\mu) \text{ is }\varepsilon\text{-AQEC for }\calE_t, \dim(Q_{K,L})=K}&\nonumber\geq \P\biggl[{\biggl\{{\max_{\substack{i\in [K]\\E\in \calE_t} } \abs{\bra{\s{c}_i} E^\dag E \ket{\s{c}_i}-\lambda_{E}}\leq \frac{\varepsilon^2}{2K M_t}}\biggr\}\bigcap A_t}\biggr]\\
    &= \P\biggl[{A_t\bigcap \biggl({\bigcap_{\substack{E\in \calE_t\\ i\in [K]}}A_{E,i}}\biggr)}\biggr].\label{eq:whicgives}
\end{align}
Indeed, note that if  $A_t$ occurs and $d(\s{C})>t\geq 0$, the fact that $d$ is a metric function 
implies that the elements $\s{x}_0,\dots,\s{x}_{L-1}$ are all distinct, and by Lemma~\ref{lem:PartitionCodesProp}, $\dim(\s{Q})=K$. Additionally, whenever $A_t$ occurs (and $d(\s{C})>t$), and also
\[\max_{\substack{i\in[K]\\ E\in \calE_t}}\abs{\bra{\s{c}_i}E^\dag E \ket{\s{c}_i}-\lambda_E}\leq\frac{\varepsilon^2}{2KM_t},\]
we have that $\s{Q}^{\s{P}}_{K,L}(\mu)$ satisfies the condition \eqref{eq:BOcond2}, and by Lemma~\ref{lem:SufCondHira}, ${\s{Q}^{\s{P}}_{K,L}(\mu)}$ is an $\varepsilon$-AQEC code for $\calE_t$.
This proves the first inequality.  For the second inequality, note that if $d(\s{C})>t$ we have that for all $i\in [K]$ and $E\in \calE_t$ 
\begin{align*}
    \bra{\s{c}_i}E^\dag E \ket{{\s{c}_i}}-\lambda_E&=\sum_{j,k=0}^{T-1}\frac{1}{T} \bra{\s{x}_{T  i +j}}E^\dag E \ket{\s{x}_{T  i +k}}-\lambda_E\\
    &=\sum_{\substack{j, k=0\\ j\neq k}}^{T-1}\frac{1}{T} \bra{\s{x}_{T  i +j}}E^\dag E \ket{\s{x}_{T  i +k}}+\sum_{j=0}^{T-1}\frac{1}{T} \bra{\s{x}_{T  i +j}}E^\dag E \ket{\s{x}_{T  i +j}}-\lambda_E\\
    &=\sum_{j=0}^{T-1}\frac{1}{T} \bra{\s{x}_{T  i +j}}E^\dag E \ket{\s{x}_{T  i +j}}-\lambda_E,
\end{align*}
where the last equality follows since $d(\s{C})>t$ implies $\s{x}_{T  i+j}\neq \s{x}_{T  i+k}$ for all $j\neq k$ and therefore by the definition of $\mathscr{L}_2$, eq.~\eqref{eq:secondlavel},
$
    \bra{\s{x}_{T  i +j}}E^\dag E \ket{\s{x}_{T  i +k}}=0.
$
This implies that 
\begin{align*}
    A_t\bigcap \biggl({\bigcap_{\substack{E\in \calE_t\\ i\in [K]}}A_{E,i}}\biggr)&=A_t\bigcap \biggl({\bigcap_{\substack{E\in \calE_t\\ i\in [K]}}\ppp{\abs{\bra{\s{c}_i}E^\dag E \ket{{\s{c}_i}}-\lambda_E}\leq \frac{\varepsilon^2}{2KM_t}}}\biggr)\\
    &=A_t\bigcap \biggl\{\max_{\substack{i\in [K]\\E\in \calE_t} } \abs{\bra{\s{c}_i} E^\dag E \ket{\s{c}_i}-\lambda_{E}}\leq \frac{\varepsilon^2}{2K M_t}\biggr\},
\end{align*}
which gives \eqref{eq:whicgives}. 

Using the union bound, we have 
\begin{align}
    \P\biggl[{A_t\bigcap \biggl({\bigcap_{\substack{E\in \calE_t\\ i\in [K]}}A_{E,i}}}\biggr)\biggr]\geq 1- \P\pp{A_t^c}-\sum_{\substack{E\in \calE_t\\ i\in [K]}}\P\pp{A_{E,i}^c}=p_{\mu}(L,t)-\sum_{\substack{E\in \calE_t\\ i\in [K]}}\P\pp{A_{E,i}^c}. \label{eq:onebeforelast}
\end{align}
To upper-bound the probability of $A_{E,i}^c$, we use Hoeffding's inequality; see Lemma~\ref{lem:Hoeffding}. Note that the random variables $\bra{\s{x}_{i  T+j}}E ^\dag E \ket{\s{x}_{i  T+j}}$,  $j=0,\dots, T-1$ are i.i.d. with expectation $\lambda_E$.  Also observe that by the definition of the metric-basis error intensity $\kappa_t(\mathscr{E},X)$ (Def.~\ref{def:parametersForApprox}) we have
\[0\leq \bra{\s{x}_{i  T+j}}E ^\dag E \ket{\s{x}_{i  T+j}}\leq  \kappa_t(\mathscr{E},X)\] 
 with probability $1$. Thus, applying Lemma~\ref{lem:Hoeffding} with $a=0$ and $b=\kappa_t(\mathscr{E},X)$, we obtain 
\begin{align}
    \P\pp{A_{E,i}^c}&\nonumber=\P\biggl\{\biggl|\frac{1}{T}\sum_{j=0}^{T-1} \bra{\s{x}_{i  T+j} }E^\dag E \ket{\s{x}_{i  T+j} }-\lambda_E\biggr|>\frac{\varepsilon^2}{2KM_t}\biggr\}\\
    &\leq 2\exp\p{-\frac{2T \varepsilon^4}{4K^2M_t^2 \kappa_t(\mathscr{E},X)^2}}=2\exp\p{-\frac{L \varepsilon^4}{2K^3M_t^2 \kappa_t(\mathscr{E},X)^2}}.\label{eq:lastofus}
\end{align}
Combining \eqref{eq:whicgives}, \eqref{eq:onebeforelast}, and \eqref{eq:lastofus} concludes the proof.
\end{proof}

\begin{proof}[\underline{The case $(\calH_X,d,\mathscr{E})\in \mathscr{L}_1$}] We follow the same proof strategy with one minor difference: as implied by \eqref{eq:FirstLevAQEC} in Lemma~\ref{lem:SufCondHira}, conditions of the form \eqref{eq:events1} must now be satisfied for distinct operators $E,F\in\calE_t$. Accordingly, let $A_t$ be as above and let $A_{E,F,i}$ be the event that 
    \[
\biggl|{\frac{1}{T}\sum_{j=0}^{T-1} \bra{\s{x}_{i  T+j} }E^\dag F \ket{\s{x}_{i  T+j} }-\lambda_{E,F}}\biggr|\leq \frac{\varepsilon^2}{2KM_t^2}, \quad \lambda_{E,F}=\E_{\s{x}\sim \mu}\pp{\bra{\s{x} }E^\dag F \ket{\s{x} }}.
\]
Using the same arguments as in the first part of the proof, if $d(\s{C})>t$, we have that for all $E$ and $F$
\begin{align*}
    \bra{\s{c}_i}E^\dag F \ket{{\s{c}_i}}-\lambda_{E,F}=\sum_{j=0}^{T-1}\frac{1}{T} \bra{\s{x}_{T  i +j}}E^\dag F \ket{\s{x}_{T  i +j}}-\lambda_{E,F},
\end{align*}
as by the $\mathscr{L}_1$ assumption for all $j\neq k$
\begin{equation}
    \bra{\s{x}_{T  i +j}}E^\dag F \ket{\s{x}_{T  i +k}}=0.\label{eq:structureenable}
\end{equation}
Thus, by Lemma~\ref{lem:SufCondHira} we have
\begin{equation}
    \P\pp{\s{Q}^{\s{P}}_{K,L}(\mu) \text{ is }\varepsilon\text{-AQEC for }\calE_t, \dim(Q_{K,L})=K} \geq \P\biggl[{A_t\bigcap \biggl({\bigcap_{\substack{E,F\in \calE_t\\ i\in [K]}}A_{E,F,i}}}\biggr)\biggr].\label{eq:whicgives2}
\end{equation}
Using the union bound, we obtain 
    \begin{align}
         \P\pp{\s{Q}^{\s{P}}_{K,L}(\mu) \text{ is }\varepsilon\text{-AQEC for }\calE_t, \dim(Q_{K,L)}=K} &\geq p_\mu(L,t)-\sum_{\substack{E,F\in \calE_t\\ i\in [K]}}\P\pp{A_{E,F,i}^c}.\label{eq:againUnion}
    \end{align}
 Using Hoeffding's inequality once again, we bound the probability of $A_{E,F,i}^c$ as follows:
    \begin{align}
        \P\pp{A_{E,F,i}^c}&\nonumber=\P\biggl[{\bigg|{\sum_{j=0}^{T-1}\frac{1}{T} \bra{\s{x}_{T  i +j}}E^\dag F \ket{\s{x}_{T  i +j}}-\lambda_{E,F}}\bigg|>\frac{\varepsilon^2}{2M_t^2 K } }\biggr]\\
    &\leq \exp\p{-\frac{L \varepsilon^4}{8K^3M_t^4 \kappa_t(\mathscr{E},X)^2}},\label{eq:againConectration}
    \end{align}
     where we used the fact that for all $i$ and $j$, by the Cauchy--Schwarz inequality,
    \begin{align*}
    \abs{ \bra{\s{x}_{T  i +j}}E^\dag F \ket{\s{x}_{T  i +j}}}\leq \sqrt{\bra{\s{x}_{T  i +j}}E^\dag E \ket{\s{x}_{T  i +j}}  \bra{\s{x}_{T  i +j}}F^\dag F \ket{\s{x}_{T  i +j}}}\leq \kappa_t(\mathscr{E},X),
\end{align*}
so the range of each of the terms in the sum over $j$ satisfies
\[-\kappa_t(\mathscr{E},X)\leq  \bra{\s{x}_{T  i +j}}E^\dag F \ket{\s{x}_{T  i +j}} \leq \kappa_t(\mathscr{E},X) .\]
We conclude the proof by combining \eqref{eq:againUnion} and \eqref{eq:againConectration}.
    \end{proof}

\subsection{Proof of Theorem~\ref{th:TypicallityProperties}}\label{app:Typicallityproof}
Let $\s{C}=\set{\s{x_0},\dots,\s{x}_{L-1}}$ be the underlying i.i.d. classical code in the construction of $\s{Q}_{K,L}^{\s{P}}$, and let $\ket{\s{c}_0},\dots,\ket{\s{c}_{K-1}}$ be the canonical basis of $\s{Q}_{K,L}^{\s{P}}$ corresponding to the partition $\s{C}_0,\dots,\s{C}_{K-1}$. We need to prove that 
    \[\P\pp{\min_{\ket{\psi}\in \s{Q}_{K,L}^{\s{P}}} \bra{\psi}\Pi_{\mathscr{P}}\ket{\psi}\geq 1-\varepsilon}\geq  p_\mu(L,0)-2Ke^{-2T^{1-2\xi}}. \]
    Note that for any basis elements of $\calH_{X}$ of the form $\ket{x},\ket{x'}$, $x,x'\in X$, we have 
        \[
        \bra{x}\Pi_{\mathscr{P}}\ket{x'}=\delta_{x,x'}\Ind_{\ppp{x\in \mathscr{P}}},
        \]
    where $\Ind_{\ppp{x\in \mathscr{P}}}$ denotes the indicator function of the event that $x\in \mathscr{P}$. In particular, if $\s{x}_0, \dots \s{x}_{L-1}$ are all distinct, or equivalently $|\s{C}|=L$, for $i,j\in [K]$ we have
    \begin{align}
        \bra{\s{c}_i}\Pi_{\mathscr{P}}\ket{\s{c}_j}&\nonumber= \sum_{x\in \s{C}_i}\sum_{x'\in \s{C}_j} \frac{1}{T}\bra{x}\Pi_{\mathscr{P}}\ket{x'}= \sum_{x\in \s{C}_i}\sum_{x'\in \s{C}_j} \frac{1}{T}\delta_{x,x'}\Ind_{\ppp{x\in \mathscr{P}}}\\
        &=\begin{cases}
            0 & i\neq j,\\
            \frac{1}{T}\abs{\s{C}_i\bigcap \mathscr{P}} & i=j.
        \end{cases}=\begin{cases}
            0 & i\neq j,\\
            \frac{1}{T}\sum_{k=0}^{T-1}\Ind_{\ppp{\s{x}_{i  T+k}\in \mathscr{P}}} & i=j.
        \end{cases}\label{eq:HoeffAgain}
    \end{align}
    
    Let $A$ be the event that $\abs{\s{C}}=L$ and for $i\in [K]$ let $A_i$ be the event that $\abs{\s{C}_i\bigcap \mathscr{P} }/T\geq 1-\varepsilon$. For an arbitrary codeword $\sum_{i}\beta_i \ket{\s{c}_i}=\ket{\psi}\in \s{Q}_{K,L}^{\s{P}}$, as long as $\abs{\s{C}}=L$,   we have 
    \begin{align}
        \bra{\psi}\Pi_{\mathscr{P}}\ket{\psi}=\sum_{i=0}^{K-1}\sum_{j=0}^{K-1}\beta_{i}^*\beta_j \bra{\s{c}_i}\Pi_{\mathscr{P}}\ket{\s{c}_j}=\sum_{i=0}^{K-1}|\beta_i|^2 \frac{1}{T} \abs{\s{C}_i \bigcap \mathscr{P}}\geq \min_{i\in [K]} \frac{1}{T} \abs{\s{C}_i \bigcap \mathscr{P}},\label{eq:stamStateProp}
    \end{align}
    where we have used the normalization $\sum_i |\beta_i|^2=1$. By \eqref{eq:stamStateProp}, we conclude that
    \begin{equation}
        \P\Big[{\min_{\ket{\psi}\in \s{Q}_{K,L}^{\s{P}}} \bra{\psi}\Pi_{\mathscr{P}}\ket{\psi}\geq 1-\varepsilon\Big]}\geq \P\pp{A\cap A_0 \cap  \cdots \cap A_{K-1}}\geq \P[A]-\sum_{i=0}^{K-1}\P[A_i^c],\label{eq:unionb}
    \end{equation}
    where the last inequality follows since $\P[A\cap(\cap_{i=0}^{K-1}A_i)]\geq\P[A]-\P[\cup_{i=0}^{K-1} A_i^c]$ using the union bound.
    By definition of $\varepsilon$, for any $i$, we have
    \begin{align}
        \P\pp{A_i^c}&\nonumber=\P\pp{ \frac{1}{T}\sum_{k=0}^{T-1}\Ind_{\ppp{\s{x}_{i  T+k}\in \mathscr{P}}} < 1-\varepsilon}\\
        &\nonumber=\P\pp{ \frac{1}{T}\sum_{k=0}^{T-1}\Ind_{\ppp{\s{x}_{i  T+k}\in \mathscr{P}}}- p_{\mu}(\mathscr{P}) < -T^{-\xi} }\\
        &\nonumber\leq \P\pp{ \abs{\frac{1}{T}\sum_{k=0}^{T-1}\Ind_{\ppp{\s{x}_{i  T+k}\in \mathscr{P}}}- p_{\mu}(\mathscr{P})}\geq  T^{-\xi} }\\
        &\leq 2\exp\p{-2  T^{1-2\xi}},\label{eq:HoefShuv}
    \end{align}
where in \eqref{eq:HoefShuv} we apply Hoeffding's inequality (Lemma~\ref{lem:Hoeffding}) independent indicators  $\Ind_{\ppp{\s{x}_{i  T+k}\in \mathscr{P}}}$, which are bounded in $[0,1]$ and satisfy $\E[\Ind_{\ppp{\s{x}_{i  T+k}\in \mathscr{P}}}]= p_{\mu}(\mathscr{P})$. Now combining \eqref{eq:unionb} and \eqref{eq:HoefShuv} concludes the proof.

\bibliographystyle{abbrvurl}
\bibliography{quantum,PI}


\end{document}